\PassOptionsToPackage{dvipsnames,table}{xcolor}
\documentclass[aps,prx,reprint,superscriptaddress,longbibliography,nofootinbib,floatfix]{revtex4-2}

\usepackage[english]{babel}
\addto\captionsenglish{%
}
\usepackage{iftex}
\ifPDFTeX
  \usepackage[utf8]{inputenc}
\fi
\usepackage{amsmath,dsfont,mathtools}
\usepackage{tikz}
\usepackage{multirow}
\usepackage{amssymb}
\usepackage{graphicx}
\usepackage{pgf}
\usepackage{placeins}
\usepackage{dcolumn}
\usepackage{bm}
\usepackage{makecell}
\usepackage[normalem]{ulem}
\graphicspath{{figures/}}
\usepackage{adjustbox}
\usepackage{physics}
\usepackage{amsfonts}
\usepackage{soul} 
\usepackage{overpic}
\usepackage{lmodern}
\usepackage{fix-cm}
\usepackage{nowidow}
\usepackage{amsthm}
\usepackage{comment}
\usepackage[noend]{algpseudocode}
\usepackage{thmtools}
\usepackage{thm-restate}
\usepackage{xcolor}
\usepackage{mathdots}
\usepackage{yhmath}
\usepackage{ragged2e}
\usepackage[scr=rsfso]{mathalfa}
\usepackage{stmaryrd} 
\usepackage{booktabs}
\usepackage{array}

\newlength{\chmtablewidth}
\newlength{\chmtablekeywidth}
\newlength{\chmtablevaluewidth}
\newlength{\chmtablegap}
\newenvironment{chmkeytable}{%
  \setlength{\chmtablewidth}{\textwidth}%
  \setlength{\chmtablekeywidth}{0.26\chmtablewidth}%
  \setlength{\chmtablegap}{1em}%
  \setlength{\chmtablevaluewidth}{%
    \dimexpr\chmtablewidth-\chmtablekeywidth-\chmtablegap\relax}%
  \begin{tabular}{@{}l@{\hspace{\chmtablegap}}l@{}}%
}{%
  \end{tabular}%
}
\newcommand{\chmtablerow}[2]{%
  \parbox[t]{\chmtablekeywidth}{\raggedright\strut #1\strut}%
  &\parbox[t]{\chmtablevaluewidth}{\raggedright\strut #2\strut}%
}

\usepackage[skins,breakable]{tcolorbox}
\definecolor{darkblue}{RGB}{0,65,145}
\definecolor{darkorange}{RGB}{190,95,0}
\definecolor{keepgreen}{RGB}{204,236,210}
\definecolor{killred}{RGB}{248,215,215}
\definecolor{mergeblue}{RGB}{209,226,247}
\definecolor{panelgrey}{RGB}{247,247,247}
\definecolor{ruleframeblue}{RGB}{0,65,145}
\newtcolorbox{rulebox}[1][]{colback=panelgrey,colframe=ruleframeblue,
  boxrule=0.6pt,arc=2pt,left=5pt,right=5pt,top=4pt,bottom=4pt,
  fonttitle=\bfseries,coltitle=white,title={#1},breakable}
\newcounter{chmalg}

\newenvironment{chmalgorithm}[1]{%
  \par\medskip\noindent
  \refstepcounter{chmalg}%
  \begin{minipage}{\linewidth}%
  \hrule height .8pt \vspace{4pt}%
  {\raggedright\textbf{Algorithm \thechmalg:} #1\par}%
  \vspace{2pt}\hrule\vspace{4pt}%
}{%
  \vspace{4pt}\hrule height .8pt
  \end{minipage}\par\medskip
}

\usepackage{afterpage}

\usetikzlibrary{fadings}
\usetikzlibrary{patterns}
\usetikzlibrary{shadows.blur}
\usetikzlibrary{shapes}
\usetikzlibrary{positioning}
\usetikzlibrary{fit}
\usetikzlibrary{decorations.pathreplacing}
\usetikzlibrary{arrows.meta}
\usetikzlibrary{calc}
\usepackage{pgfplots}
\pgfplotsset{compat=1.17}
\theoremstyle{definition}
\newtheorem{definition}{Definition}[section]
\theoremstyle{plain}
\newtheorem{lemma}{Lemma}[section]
\newtheorem{proposition}{Proposition}[section]
\newtheorem{remark}{Remark}

\newtheorem{corollary}{Corollary}[section]

\definecolor{darkred}{RGB}{200,20,20}
\definecolor{darkblue}{RGB}{20,20,200}

\usepackage[colorlinks=true,citecolor=blue,linkcolor=blue,urlcolor=blue]{hyperref}
\usepackage{cleveref}
\usepackage[final]{pdfpages}

\begin{document}
\makeatletter
\let\LS@rot\@undefined
\let\orig@addcontentsline\addcontentsline
\def\addcontentsline#1#2#3{}
\makeatother

\title{Circuit Harmonic Matrices: A Spectral Framework for Quantum Machine Learning}

\author{Kyle J. S. Campbell}
\email{Kyle.J.S.Campbell@ed.ac.uk}
\affiliation{Quantum Software Lab, School of Informatics, The University of Edinburgh, United Kingdom}

\author{Luigi Del Debbio}
\affiliation{School of Physics and Astronomy, The University of Edinburgh, United Kingdom}

\author{Petros Wallden}
\affiliation{Quantum Software Lab, School of Informatics, The University of Edinburgh, United Kingdom}

\hyphenpenalty=50
\exhyphenpenalty=50
\begin{abstract}
Parametrised quantum circuits learn functions by adjusting gate parameters, but their ability to represent and learn a function also depends on how the circuit is designed. We introduce the \emph{circuit harmonic matrix} by organising the circuit's Fourier expansion in both its inputs and trainable parameters into a fixed matrix. This records the combined influence of the data encoding, trainable gates, initial state and measured observable. From it, we derive how the Fourier coefficients vary and are correlated across parameter choices, and how they respond to small parameter changes through the quantum neural tangent kernel. The same construction identifies restrictions on representable functions and gives a lower bound on fitting error. For circuits built from fixed Clifford gates and independently parametrised Pauli or controlled rotations, we develop a two-copy propagation method that computes coefficient covariance and the averaged tangent kernel for uniformly distributed parameters analytically without constructing the full exponentially large matrix. We apply the framework to 400 configurations spanning eight circuit families, two encoders, and different depths and qubit counts, and compare their coefficient statistics with performance on controlled learning tasks. Increasing depth and qubit count generally reduces total coefficient variance, yet their effects on learning differ. Greater depth usually improves early learning and final accuracy, while adding qubits usually slows early learning and has a frequency-dependent effect on final accuracy. Learning also improves with depth in circuits whose Fourier coefficients are uncorrelated across parameter choices. Comparisons of paired targets with equal contributions from individual coefficient variances find no robust overall learning advantage for targets favoured by covariance. These results establish a common framework for relating circuit design to representational constraints and learning, while distinguishing variation across parameter space from the local response available during training.
\end{abstract}

\maketitle

\section{Introduction}\label{sec:intro}

Parametrised quantum circuits (PQCs) are widely used as trainable models in quantum machine learning \cite{biamonte2017,schuld2018}. A circuit encodes an input \(x\), applies gates controlled by parameters \(\theta\), and returns an observable expectation value \(f(x;\theta)\). Training adjusts the parameters to fit data, but the encoder, trainable gates, reference state and observable already constrain the available functions. Understanding these constraints is important for explaining why circuits using the same data encoding can behave differently during learning.

Fourier analysis provides a natural starting point. For the phase encodings considered here, the model has a finite expansion
\[
f(x;\theta)=\sum_{\omega\in\Omega}a_\omega(\theta)e^{i\omega\cdot x}.
\]
The encoding generators and their repetition determine the allowed frequency set \(\Omega\)~\cite{schuld2021,perez-salinas2020}. The complete circuit determines how the coefficient functions \(a_\omega(\theta)\) can change during training. An allowed frequency can have an identically zero coefficient; a nonzero coefficient can be independent of the parameters; and coefficients that each vary need not be independently adjustable. Even when two circuits have the same varying frequencies, the scales and joint dependence of their coefficients can differ. Frequency support alone therefore does not specify the circuit's function class or its learning behaviour.

We study how the complete finite-depth circuit constrains coefficient variation and local learning response. Our approach collects established trigonometric circuit expansions \cite{wiedmann2024} into a fixed \emph{circuit harmonic matrix} (CHM), with input and parameter harmonics as its two indices. We first construct this matrix and derive coefficient variances, covariances and tangent kernels from it. We then evaluate those quantities numerically, using a direct two-copy calculation for the ensemble second-order moments. Finally, controlled Fourier-learning tasks examine the learning information supplied by variance, covariance and the QNTK, in that order.

\subsection{Relation to previous work}

Our construction builds on earlier work in quantum Fourier analysis. Nemkov \emph{et al.} developed computational-tree expansions for Clifford--Pauli variational circuits \cite{nemkovFourierExpansionVariational2023}. Wiedmann \emph{et al.} adapted this approach to supervised-learning circuits, expressing the input Fourier coefficients as trigonometric polynomials of the trainable parameters \cite{wiedmann2024}. Their construction identifies frequencies removed by the complete circuit and permits exact coefficient-covariance calculations under uniform parameter sampling. They also use an inverse-covariance (Mahalanobis) distance from target coefficients to the ensemble mean in a heuristic for architecture ranking. We take their coefficient construction and moment calculations as our starting point; we then ask how variance along a target direction, the sign of its covariance contribution, and local tangent response relate to learning.

Barthe and P\'erez-Salinas develop a harmonic representation of quantum states and study frequency profiles in re-uploading models, with Haar-ensemble results and numerical comparisons beyond those assumptions \cite{barthe2024}. Their matrix indexes computational-basis amplitudes and input frequencies, and remains parameter dependent. Our CHM instead expands the observable's Fourier-coefficient functions in parameter harmonics. The resulting matrix is fixed, although its size can grow exponentially with the number of parameters. In our numerical depth comparison, we also keep the number of encoding layers fixed. This separates the effects of trainable depth from encoder repetition, which is varied in their re-uploading analysis.

Mhiri \emph{et al.} relate coefficient variance to \emph{frequency redundancy}, the multiplicity of encoder eigenvalue-index paths producing the same frequency, and establish concentration results under specified assumptions on the trainable-unitary ensembles \cite{mhiri2025a}. Their prediction that variance is proportional to redundancy gives us an asymptotic reference against which to compare our finite-depth results.

Strobl \emph{et al.} construct Fourier fingerprints from coefficient correlations and show that lower average absolute correlation is associated with lower learning error across ansatzes on random Fourier targets \cite{strobl2025}. They also suggest that correlations aligned with target structure could provide a useful inductive bias. We investigate a related, target-dependent question: does a circuit learn the combination of Fourier components favoured by its covariance more easily? We compare paired targets with identical contributions from individual coefficient variances but opposite covariance contributions. Our study also examines how coefficient statistics change with trainable depth at fixed encoding and relates them to local tangent response.

Duffy and Jastrzebski study spectral bias through frequency-resolved gradient bounds and small-angle analysis, also building on constructive coefficient expansions \cite{duffy2026}. We complement that work with full-period coefficient moments and tangent kernels, assessed alongside controlled learning tests.

Quantum neural tangent kernels (QNTKs) supply a local description in terms of model derivatives, extending the classical tangent-kernel construction \cite{jacot2018,liu2022}. We express this local response through the CHM and evaluate it on sampled inputs. The frozen-kernel training model is the established linearisation; our separate gradient-descent control checks its accuracy in the tested regime.

\subsection{Contributions and main results}

For circuits whose parameter dependence also has a finite Fourier expansion, we write the coefficient vector as
\[
a(\theta)=C\psi(\theta),
\qquad
\psi_k(\theta)=e^{ik\cdot\theta}.
\]
Rows of \(C\) index input frequencies and columns index parameter harmonics. Once the gate parametrisation, encoder, reference state and observable are fixed, \(C\) is independent of the sampled parameters and the dataset. Figure~\ref{fig:chm_overview} summarises the questions this representation addresses. We make three connected contributions.

\begin{figure*}[t]
\centering
\begingroup
\hyphenpenalty=10000
\exhyphenpenalty=10000
\definecolor{chmBlue}{RGB}{0,65,145}
\definecolor{chmTeal}{RGB}{0,118,128}
\definecolor{chmOrange}{RGB}{190,95,0}
\definecolor{chmPurple}{RGB}{93,63,152}

\begin{tikzpicture}[
  >=Latex,
  font=\small,
  stage/.style={
    draw,
    rounded corners=2pt,
    minimum height=1.70cm,
    text width=3.50cm,
    align=center,
    inner sep=5pt,
    line width=0.6pt
  },
  question/.style={
    stage,
    text width=5.80cm,
    minimum height=1.10cm
  },
  flow/.style={->, line width=0.8pt, chmBlue!80}
]

\node[stage, draw=chmBlue!70, fill=chmBlue!5] (circuit) at (0,0)
  {\textbf{Circuit specification}\\[2pt]
   encoder and trainable gates\\
   reference state and observable};

\node[stage, text width=4.25cm, draw=chmTeal!75, fill=chmTeal!6] (chm) at (4.70,0)
  {\textbf{Circuit harmonic matrix}\\[2pt]
   $a(\theta)=C\psi(\theta)$\\[2pt]
   rows: input frequencies\\
   columns: parameter harmonics};

\node[question, draw=chmTeal!75, fill=chmTeal!6] (support) at (10.90,1.65)
  {\textbf{Coefficient support}\\[2pt]
   Which coefficients can occur or vary?};

\node[question, draw=chmOrange!80, fill=chmOrange!7] (covariance) at (10.90,0)
  {\textbf{Coefficient covariance}\\[2pt]
   How do coefficients vary together?\\
   Uniform full-period parameter sampling};

\node[question, draw=chmPurple!75, fill=chmPurple!6] (kernel) at (10.90,-1.65)
  {\textbf{Coefficient-space QNTK}\\[2pt]
   Which changes are locally available?\\
   At a specified parameter point};

\draw[flow] (circuit) -- (chm);
\coordinate (branch) at (7.40,0);
\draw[flow,-] (chm.east) -- (branch);
\draw[flow] (branch) |- (support.west);
\draw[flow] (branch) -- (covariance.west);
\draw[flow] (branch) |- (kernel.west);

\end{tikzpicture}%
\endgroup
\caption{\justifying From a circuit specification to coefficient variability and local response. The CHM records the dependence of input Fourier coefficients on parameter harmonics. Its rows distinguish coefficients that vanish, remain constant, or vary. A specified parameter ensemble gives coefficient covariances; independent uniform sampling over full parameter periods yields the row Gram identity used here. Derivatives at a specified parameter point instead give the coefficient-space QNTK. Evaluating those coefficient directions on selected inputs gives the data-space QNTK. These objects constrain and describe learning, but neither covariance nor an ensemble-averaged QNTK alone determines a training trajectory.}
\label{fig:chm_overview}
\end{figure*}

\begin{enumerate}
\item \emph{A common representation of coefficient statistics and local response.} Under independent uniform sampling over full parameter periods, coefficient covariance is \mbox{\(\operatorname{Cov}_\theta[a]=C^\circ C^{\circ\dagger}\)}. Here \(C^\circ\) omits the constant parameter-harmonic column. Differentiation gives the pointwise coefficient-space QNTK, whose parameter average is a positively weighted Gram matrix with the same range as covariance. A pointwise tangent kernel can nevertheless have smaller rank: variation across initialisations need not be independently available at one parameter point. Evaluating the input harmonics on a dataset recovers the data-space QNTK. The same representation gives an optimiser-independent training-error bound from the affine output span; the pointwise kernel instead supplies conditional forecasts from sampled residuals (Sections~\ref{sec:coeff_stats} and~\ref{sec:kernels}).

\item \emph{Direct evaluation of second-order quantities.} We derive a two-copy Pauli-propagation contraction that evaluates the coefficient covariance and parameter-averaged QNTK without materialising all parameter-harmonic columns of \(C\). For the independent single-use rotation parameters treated by the algorithm, averaging is performed locally during propagation. This avoids both parameter sampling and explicit enumeration of the full CHM, although the worst-case cost remains exponential. Small-system CHMs and independent Monte Carlo estimates match the exact two-copy results, which also agree with the Gram identity (Sections~\ref{sec:ensemble_computation} and~\ref{subsec:exact-validation}, and Appendix~\ref{appendix:twirl}).

\item \emph{A finite-depth comparison with learning.} We compute these quantities for eight circuit families with two encoder axes and compare them with learning performance on controlled Fourier targets under Adam. The \(R_X\)-encoded YZY-entangling circuit, with its fixed triangular all-to-all CNOT sequence, provides our main example. We also combine results across the full circuit bank to identify shared trends and exceptions~\cite{sim2019,strobl2025}. Across 400 combinations of circuit, qubit count, encoder and depth, we distinguish expansion of coefficient support from changes in variance, covariance magnitude and normalised correlation. Matrix sequences for each circuit and comparisons on fixed coefficient sets show which families decorrelate with depth. They also identify circuit--encoder combinations with no resolved off-diagonal covariance, whose learning performance provides a useful comparison. We train circuits on identical supported targets and compare target pairs that align or oppose the selected coefficient covariances to test for a learning preference. A separate small-step gradient-descent control assesses frozen-QNTK loss forecasts (Section~\ref{sec:numerical-results} and Appendix~\ref{subsec:appendix_frozen_qntk}).
\end{enumerate}

Increasing circuit depth generally reduces coefficient variances while bringing their relative distribution across frequencies closer to the redundancy prediction. It can also weaken the correlations between coefficients. For example, at six qubits, increasing depth from two to five reduces both the largest and the mean absolute off-diagonal covariance in five of the seven coupled circuit--encoder combinations; the largest normalised correlation falls in six. This behaviour depends on the circuit and encoder: seven other combinations already have diagonal matrices on their supported coordinates, while nonentangling YZY has only one positive frequency. Adding qubits does not produce the same decorrelation trend when we compare the same low-frequency coordinates [Figs.~\ref{fig:fixed-support-structure-depth} and~\ref{fig:fixed-support-structure-size}].

Depth and size also have different consequences for learning. Of the 96 single-mode circuit--target conditions supported at every tested depth, 84 learn faster initially and 90 reach a lower final MSE at depth five than at depth one. Increasing the number of qubits has the opposite tendency: 50 of the 58 conditions supported at every tested size learn more slowly at six qubits than at two. Deeper circuits can therefore learn faster even as their coefficient variances shrink.

The relationship between variance and learning also depends on which tasks are compared. Once differences in frequency, circuit and target phase are accounted for, variance provides only a modest additional description of the relative loss accumulated during training, and this relationship varies with circuit size. Likewise, the kernel results show that having access to all coefficient directions locally need not make them equally easy to change.

A related question is whether correlations between Fourier coefficients help or hinder learning. Across our circuit and depth settings, stronger mean correlation accompanies slower learning and higher final errors on targets containing several frequencies. This relationship alone does not tell us whether covariance is responsible, since the circuits also differ in their structure and in which target components they can represent.

To investigate this more directly, we compare pairs of target functions with the same frequencies and amplitudes, changing only the relative sign of two components. The individual coefficient variances contribute equally to both targets, while the covariance contribution changes sign. This lets us ask whether the combination favoured by covariance is also learned more easily. Covariance distinguishes these targets less strongly as depth increases, but their learning differences do not consistently shrink with it. Some circuits show a preference for covariance-aligned targets, yet we find no robust overall learning advantage across the benchmark.

\subsection{Organisation of the paper}

Section~\ref{sec:prelims} introduces the circuit model, its Fourier expansions in the input and parameters, and the parameter distribution used for averaging. Section~\ref{sec:interaction_matrix} constructs the circuit harmonic matrix \(C\). Section~\ref{sec:coeff_stats} then uses this matrix to derive the Fourier-coefficient means, variances, covariance matrix and normalised correlations. Section~\ref{sec:kernels} derives the QNTK and its parameter average, and explains how holding the QNTK fixed gives predictions of loss reduction under gradient descent. It also derives a lower bound on the training error that a circuit can achieve.

Section~\ref{sec:ensemble_computation} presents the two-copy method for computing the covariance matrix and averaged QNTK without constructing the full CHM. Section~\ref{sec:numerical-results} compares coefficient statistics across circuits, depths and qubit counts, before testing how variance, covariance and the QNTK relate to learning. Section~\ref{sec:discussion} discusses the implications and limitations of these results. The appendices give the experimental protocols, supporting numerical results and mathematical derivations.

\section{Preliminaries}\label{sec:prelims}

The CHM construction uses two Fourier expansions: one in the input and one in the trainable parameters. We introduce them separately and specify the parameter ensemble underlying the later covariance identities. The distinction between physical gate angles and Fourier coordinates also fixes the derivative convention needed later for the QNTK.

\subsection{Circuit model}
\label{subsec:pqc_model}

Let \(\rho_0\) be a fixed \(n\)-qubit state, \(O\) a Hermitian observable, and \(x\in\mathcal X\subseteq\mathbb R^{d_x}\) an input with \(d_x\) components. Let \(L\) be the number of encoding layers in the circuit. We consider circuits of the form
\begin{equation}
U(\theta,x)
=
W_L\bigl(\theta^{(L)}\bigr)S_L(x)
\cdots
W_1\bigl(\theta^{(1)}\bigr)S_1(x),
\label{eq:prelims_reupload}
\end{equation}
where \(S_\ell(x)\) encodes the input and \(W_\ell(\theta^{(\ell)})\) is an input-independent trainable block. Together, the blocks contain \(m\) independent parameters, collected in \(\theta\); their coordinate convention is specified below. The model output is the real expectation value
\begin{equation}
f(x;\theta)
=
\Tr\!\left[
O\,U(\theta,x)\rho_0 U(\theta,x)^\dagger
\right].
\label{eq:prelims_model}
\end{equation}
Each \(S_\ell\) is a complete encoding layer and may contain gates on several qubits. Re-uploading repeats these layers through the circuit \cite{perez-salinas2020}. When each layer applies one encoding gate to each of the \(n\) qubits, the total number of individual encoding gates (encoder insertions) is \(nL\). For example, encoding on \(n=5\) qubits across \(L=3\) layers gives \(nL=15\) encoding gates\footnote{The numerical study we look at in Section~\ref{sec:numerical-results} fixes \(L=1\) and varies instead the trainable-block depth \(d\) which allows for the study of increased parameterisation without changing the size of the encoding support. We control this support by increasing encoding qubit size. However, here we will be general as this construction is applicable to any number of encoding layers and trainable-block depth.}.

\subsection{Input frequencies set by the encoding blocks}
\label{subsec:qfm_input}

Suppose each encoding block has the form
\begin{equation}
S_\ell(x)
=
\exp\!\left[-i\sum_{q=1}^{d_x}x_qG_q^{(\ell)}\right],
\label{eq:prelims_encoder_diag}
\end{equation}
where \(G_q^{(\ell)}\) is the Hermitian generator associated with input component \(x_q\) in encoding block \(\ell\), and the generators within each block mutually commute. For a fixed parameter vector, the output is then a finite Fourier expansion, often called a quantum Fourier model \cite{schuld2021,perez-salinas2020,casas2023}:
\begin{equation}
f(x;\theta)
=
\sum_{\omega\in\Omega}
a_\omega(\theta)e^{i\omega\cdot x}.
\label{eq:prelims_qfm_expansion}
\end{equation}
The input frequency \(\omega\) labels a Fourier component, while \(a_\omega(\theta)\) describes how the trainable circuit sets its amplitude and phase.

To see where the frequencies come from, let \(\lambda_j^{(\ell)}\in\mathbb R^{d_x}\) denote the joint eigenvalue vectors of the generators \(G_q^{(\ell)}\) in encoding block \(\ell\). That block contributes phases \(e^{-i\lambda_j^{(\ell)}\cdot x}\); pairing a phase with a complex-conjugate phase in the expectation value gives an eigenvalue difference. A single encoding layer therefore permits frequencies in
\begin{equation}
\Omega^{(\ell)}
:=
\left\{
\lambda_{j'}^{(\ell)}-\lambda_j^{(\ell)}
:\ j,j'
\right\}.
\label{eq:prelims_diffset}
\end{equation}
Across encoding layers the phase factors multiply, so their frequencies add. We use the resulting envelope
\begin{equation}
\Omega
:=
\Omega^{(1)}\oplus\cdots\oplus\Omega^{(L)},
\label{eq:prelims_Omega_definition}
\end{equation}
where \(A\oplus B=\{a+b:a\in A,\ b\in B\}\) is the Minkowski sum. For example, the single-qubit layer \(S(x)=e^{-ixZ/2}\) permits \(\{0,\pm1\}\); \(L\) such layers give the envelope \(\{-L,\ldots,L\}\). The eigenbasis derivation is given in Appendix~\ref{appendix:qfms}.

The encoding blocks determine which frequencies the circuit can produce, but do not guarantee that all of them appear in its output. The trainable gates, initial state and observable also matter: their combination can make \(a_\omega(\theta)=0\) for every parameter choice, removing that frequency entirely. They cannot introduce frequencies outside \(\Omega\). For a frequency that does appear, its coefficient may vary with \(\theta\) or remain fixed. The CHM records this dependence and therefore tells us both which frequencies are present and which of their coefficients can change during training.

Every allowed frequency \(\omega\) is accompanied by \(-\omega\), since reversing an eigenvalue difference changes its sign. The output \(f\) is real, so the coefficients of these two frequencies must be complex conjugates:
\begin{equation}
a_{-\omega}(\theta)=\overline{a_\omega(\theta)}.
\label{eq:prelims_reality}
\end{equation}
We keep both \(\omega\) and \(-\omega\) in the expansion, with their coefficients linked by this relation.

\subsection{Parameter harmonics and angle periods}
\label{subsec:param_harmonics}

We next expand the coefficient functions in the trainable parameters. Let \(\vartheta_r\) be a physical gate angle and \(T_r\) a full period of the model in that coordinate. We use the dimensionless coordinates
\begin{equation}
\theta_r=\frac{2\pi}{T_r}\vartheta_r,
\qquad
\theta\in\mathbb T^m=(\mathbb R/2\pi\mathbb Z)^m.
\label{eq:prelims_period_rescaling}
\end{equation}
For ordinary Pauli rotations we write
\[
R_{P_r}(\vartheta_r)=\exp(-i\vartheta_rP_r/2),
\qquad P_r^2=I.
\]
A \(2\pi\) shift changes the unitary only by a global sign, which cancels from the expectation value. We take \(T_r=2\pi\) and \(\theta_r=\vartheta_r\) for these gates. A controlled rotation generically requires \(T_r=4\pi\), giving \(\theta_r=\vartheta_r/2\).

For the principal construction, each independent parameter occurs in one Pauli rotation, with fixed Clifford gates between rotations. Conjugating a Pauli observable by one such rotation introduces at most a constant, a cosine or a sine in that angle; fixed Cliffords introduce no parameter dependence. Iterating the gate rule given in Section~\ref{subsec:C_paths} makes each \(a_\omega(\theta)\) a finite trigonometric polynomial \cite{nemkovFourierExpansionVariational2023,wiedmann2024}. Expanding the sines and cosines into complex exponentials gives a common finite set of parameter harmonics:
\begin{equation}
\begin{aligned}
\psi_k(\theta)&:=e^{ik\cdot\theta},
\qquad k\in\mathcal K\subset\mathbb Z^m,\\
a_\omega&\in\operatorname{span}\{\psi_k:k\in\mathcal K\},
\qquad \omega\in\Omega.
\end{aligned}
\label{eq:prelims_param_fourier_generic}
\end{equation}
We refer to these exponentials as parameter characters. With each parameter used once, a term contains either no dependence on \(\theta_r\), a cosine, or a sine. Its harmonic index is therefore \(k_r=0\) or \(k_r=\pm1\), so
\begin{equation}
\mathcal K\subseteq\{-1,0,1\}^m.
\label{eq:prelims_k_support}
\end{equation}
As with \(\Omega\), membership in \(\mathcal K\) does not guarantee a nonzero harmonic coefficient. Reusing parameters, or decomposing a controlled rotation into shared-angle Pauli factors, can produce higher integer orders in the rescaled coordinates. The single-use construction is detailed in Appendix~\ref{appendix:pauli_propagation}; the controlled-rotation extension is given in Remark~\ref{rem:cr_gates}.

Changing from physical angles to the rescaled coordinates preserves the full-period ensemble and its coefficient moments. Shortening the averaging interval in the original angle would change that ensemble: a \(4\pi\) interval cannot in general be replaced by a \(2\pi\) one. For derivatives, the coordinate conversion must also be included explicitly:
\begin{equation}
\partial_{\vartheta_r}\psi_k(\theta(\vartheta))
=
i\frac{2\pi k_r}{T_r}\psi_k(\theta(\vartheta)).
\label{eq:prelims_physical_character_derivative}
\end{equation}
The initial QNTK factorisation uses derivatives in \(\theta\). Learning equations and numerical QNTKs use the physical angles \(\vartheta\) optimised in the experiments and include these factors. Thus coefficient variation is unchanged by this coordinate conversion, whereas parameter sensitivity must be expressed in the coordinates used for training.

\subsection{Parameter averages and covariance}
\label{subsec:centering}

Our coefficient statistics describe independent uniform sampling over each parameter's full period. In the coordinates of Eq.~\eqref{eq:prelims_period_rescaling}, define
\begin{equation}
\mathbb E_\theta[g]
:=
\int_{\mathbb T^m}g(\theta)\,\mathrm d\mu(\theta),
\label{eq:prelims_Etheta_def}
\end{equation}
with
\begin{equation}
\mathrm d\mu(\theta)
=
\prod_{r=1}^{m}\frac{\mathrm d\theta_r}{2\pi}.
\label{eq:DefMeasureOnTorus}
\end{equation}
This is the uniform (Haar) measure on the parameter torus. It should be distinguished from Haar measure on the unitary group, since uniformly sampling the angles does not in general produce uniformly random unitaries.

Under this measure the characters are orthogonal \cite{katznelson2004,stein2003}:
\begin{equation}
\mathbb E_\theta[\psi_k]=\delta_{k,\mathbf0},
\qquad
\mathbb E_\theta[\psi_k\overline{\psi_l}]=\delta_{k,l}.
\label{eq:prelims_character_orthogonality}
\end{equation}
Consequently, the mean of a Fourier expansion contains only its constant parameter harmonic. In a product with a complex conjugate, only terms with matching harmonic indices survive the average. We use these two properties to calculate coefficient moments directly from the CHM.

The Fourier coefficients can be complex even though \(f\) is real. For functions with finite second moments, we use Hermitian covariance. Writing \(\widetilde g=g-\mathbb E_\theta[g]\) and similarly for \(h\), define
\begin{equation}
\mathrm{Cov}_\theta[g,h]
:=
\mathbb E_\theta[\widetilde g\,\overline{\widetilde h}].
\label{eq:prelims_cov_def}
\end{equation}
For a complex coefficient vector \(z(\theta)\), this becomes
\begin{equation}
\begin{aligned}
\widetilde z(\theta)
&:=z(\theta)-\mathbb E_\theta[z(\theta)],\\
\mathrm{Cov}_\theta[z]
&:=\mathbb E_\theta[
\widetilde z(\theta)\widetilde z(\theta)^\dagger].
\end{aligned}
\label{eq:prelims_vector_covariance}
\end{equation}
The conjugate makes diagonal entries real, nonnegative variances. Centring separates coefficient fluctuations from their ensemble means.

Parameter averages are distinct from averages over observed inputs. For a dataset \(\mathcal D=\{x_i\}_{i=1}^N\), we write
\begin{equation}
\mathbb E_{\mathrm{data}}[g(x)]
:=
\frac1N\sum_{i=1}^N g(x_i),
\label{eq:prelims_Edata_def}
\end{equation}
with empirical centring and covariance using the same \(1/N\) normalisation. We keep the subscripts explicit to distinguish architecture-level parameter statistics from dataset-dependent quantities.

We use \(\omega,\mu\in\Omega\) for generic input frequencies and \(k,l\in\mathcal K\) for generic parameter frequencies. When enumerating selected frequencies in a diagram, we write \(\omega^{(1)},\omega^{(2)},\ldots\) and \(k^{(1)},k^{(2)},\ldots\). These label whole frequency vectors, not components; for example, \(k_r\) is the component associated with parameter \(\theta_r\). Appendix~\ref{appendix:notation} collects these conventions.

\section{Circuit Harmonic Matrices}
\label{sec:interaction_matrix}

The encoder's frequency set does not specify how its Fourier coefficients depend on the trainable parameters. The circuit harmonic matrix records this dependence in one fixed array. We first define the matrix and explain how to read it, then obtain its entries from the established Pauli-propagation construction and illustrate the procedure on one qubit.

\subsection{Definition and interpretation}
\label{subsec:C_def}

Collect the input Fourier coefficients into \(a(\theta)\in\mathbb C^{|\Omega|}\) and the parameter characters into \(\psi(\theta)\in\mathbb C^{|\mathcal K|}\). Their relation defines the \emph{circuit harmonic matrix} \(C\):
\begin{equation}
a(\theta)=C\psi(\theta),
\qquad C\in\mathbb C^{|\Omega|\times|\mathcal K|}.
\label{eq:C_aCpsi}
\end{equation}
Here \(\psi_k(\theta)=e^{ik\cdot\theta}\), as in Section~\ref{subsec:param_harmonics}. Rows index input frequencies \(\omega\); columns index parameter harmonics \(k\). Thus
\begin{equation}
a_\omega(\theta)=\sum_{k\in\mathcal K}C_{\omega k}e^{ik\cdot\theta},
\label{eq:C_param_series}
\end{equation}
and the complete model is
\begin{equation}
f(x;\theta)
=\sum_{\omega\in\Omega}\sum_{k\in\mathcal K}
C_{\omega k}e^{i\omega\cdot x}e^{ik\cdot\theta}.
\label{eq:C_joint_harmonic}
\end{equation}
We include \(k=\mathbf0\) in \(\mathcal K\) by convention and allow individual entries, rows or columns to be zero. The row and column sets describe different variables and generally have different sizes: \(C\) need not be square, symmetric or Hermitian.

For fixed harmonic conventions, \(C_{\omega k}\) is the coefficient of the joint harmonic \(e^{i\omega\cdot x}e^{ik\cdot\theta}\). Parameter-character orthogonality makes this definition explicit:
\begin{equation}
C_{\omega k}
=\mathbb E_\theta\!\left[a_\omega(\theta)
\overline{\psi_k(\theta)}\right].
\label{eq:C_fourier_coeff}
\end{equation}
The complete ordered circuit, reference state, observable and gate parametrisation determine \(C\). During training \(C\) remains fixed: changing \(\theta\) changes \(\psi(\theta)\), and hence \(a(\theta)\). Although Eq.~\eqref{eq:C_aCpsi} is linear in the characters, their values cannot be varied independently; they are functions of the same \(m\) parameters.

A row records the entire parameter dependence of one input coefficient. An all-zero row means that the frequency is absent from the model. A row with a nonzero entry only at \(k=\mathbf0\) gives a nonzero coefficient that is fixed throughout training. A row with any nonzero entry at \(k\neq\mathbf0\) gives a coefficient that varies with \(\theta\). Figure~\ref{fig:chm_anatomy} summarises these distinctions.

\begin{figure}[t]
\centering
\begingroup
\hyphenpenalty=10000
\exhyphenpenalty=10000
\definecolor{chmBlue}{RGB}{0,65,145}
\definecolor{chmTeal}{RGB}{0,118,128}
\definecolor{chmOrange}{RGB}{190,95,0}
\begin{tikzpicture}[
  font=\small,
  >=Latex,
  meaning/.style={anchor=west,align=flush left,text width=2.8cm,inner sep=2pt},
  arrow/.style={->,draw=black!50,line width=.6pt}
]
\node at (2.35,1.95) {Parameter harmonics (columns)};
\node[chmTeal] at (.925,1.48) {$k=\mathbf0$};
\node[chmOrange] at (2.825,1.48) {$k\neq\mathbf0$};
\node[anchor=west] at (4.72,1.48) {Resulting coefficient};
\draw[chmTeal!70,line width=.7pt] (.47,1.10)--(1.38,1.10);
\draw[chmOrange!70,line width=.7pt] (1.42,1.10)--(4.23,1.10);
\foreach \x/\lab in {.925/{\mathbf0},1.875/{k^{(1)}},2.825/{k^{(2)}},3.775/{\cdots}}
  \node at (\x,.77) {$\lab$};
\node[anchor=east] at (.18,.77) {$k$};

\fill[chmTeal!7] (.45,-.425) rectangle (1.40,-2.125);
\fill[chmOrange!10] (1.40,-1.275) rectangle (4.25,-2.125);
\foreach \x in {.45,1.40,2.35,3.30,4.25}
  \draw[black!25,line width=.4pt] (\x,.425)--(\x,-2.125);
\foreach \y in {.425,-.425,-1.275,-2.125}
  \draw[black!25,line width=.4pt] (.45,\y)--(4.25,\y);
\draw[chmBlue!65,line width=.7pt] (.45,.425) rectangle (4.25,-2.125);
\draw[chmTeal!70,line width=.8pt] (1.40,.425)--(1.40,-2.125);

\foreach \y/\lab in {0/{\omega^{(1)}},-.85/{\omega^{(2)}},-1.70/{\omega^{(3)}}}
  \node[anchor=east] at (.18,\y) {$\lab$};
\foreach \x in {.925,1.875,2.825,3.775}
  \node at (\x,0) {$0$};
\node[chmTeal] at (.925,-.85) {$\ast$};
\foreach \x in {1.875,2.825,3.775}
  \node at (\x,-.85) {$0$};
\node[chmTeal] at (.925,-1.70) {$\ast$};
\foreach \x in {1.875,2.825}
  \node[chmOrange] at (\x,-1.70) {$\ast$};
\node[chmOrange] at (3.775,-1.70) {$\cdots$};

\foreach \y in {0,-.85,-1.70}
  \draw[arrow] (4.30,\y)--(4.65,\y);
\node[meaning] at (4.72,0)
  {\textbf{Absent}\\$a_{\omega^{(1)}}(\theta)=0$};
\node[meaning] at (4.72,-.85)
  {\textbf{Fixed}\\$a_{\omega^{(2)}}(\theta)=C_{\omega^{(2)}0}$};
\node[meaning] at (4.72,-1.70)
  {\textbf{Varying}\\depends on $\theta$};
\node[align=center] at (3.80,-2.64)
  {$\ast$: fixed nonzero entry of $C$\\[3pt]
   The characters, not the entries of $C$, depend on $\theta$.};
\end{tikzpicture}
\endgroup
\caption{\justifying Reading CHM rows. Selected rows are shown schematically; a star denotes a fixed nonzero entry, not a prescribed value. The zero blocks in the first two rows extend over all \(k\neq\mathbf0\) columns. A frequency can be encoder-permitted but absent, present but fixed, or parameter-dependent. The last row illustrates the third case; its \(k=\mathbf0\) entry could also be zero. Every entry of \(C\) is fixed. The \(k=\mathbf0\) column multiplies the constant character \(1\) and gives \(\mathbb E_\theta[a_\omega]\); the other columns multiply \(\theta\)-dependent characters. It is \(a_\omega(\theta)\), not \(C_{\omega k}\), that can change during training.}
\label{fig:chm_anatomy}
\end{figure}

A column shows which input coefficients depend on the same parameter character. Since \(\psi_{\mathbf0}=1\), the \(k=\mathbf0\) column contains the means \(C_{\omega\mathbf0}=\mathbb E_\theta[a_\omega]\). Here ``constant'' refers to the parameters. The input-frequency row \(\omega=\mathbf0\) has a different role: it describes a term that is constant in the input but may still change during training. When two rows share \(k\neq\mathbf0\) columns, those characters contribute to their covariance. The contributions can reinforce or cancel according to their complex weights, as Section~\ref{sec:coeff_stats} explains.

\subsection{Construction by Pauli propagation}
\label{subsec:C_paths}

The entries of \(C\) can be obtained directly from the circuit's gates. Here we consider Pauli-rotation encoders, trainable Pauli rotations in which each parameter appears once, and fixed Clifford gates. For these circuits, the computational trees of Nemkov \emph{et al.}, together with the separation of input and parameter dependence developed by Wiedmann \emph{et al.}, show how each gate contributes to the joint Fourier expansion \cite{nemkovFourierExpansionVariational2023,wiedmann2024}. The definition of \(C\) also applies beyond this gate set, provided the finite harmonic expansions assumed above exist.

The construction follows the observable backwards through the circuit. Writing the observable as a sum of Pauli strings allows each term to be treated separately. A Clifford gate maps a Pauli string to another Pauli string, possibly changing its sign \cite{gottesman1998,aaronson2004}. A trainable rotation can also split a term into two, depending on whether its generator \(P_r\) commutes with the current string \(Q\):
\begin{equation}
\begin{aligned}
&R_{P_r}(\theta_r)^\dagger Q R_{P_r}(\theta_r)\\
&\quad=
\begin{cases}
Q, & [P_r,Q]=0,\\[2pt]
Q\cos\theta_r+(iP_rQ)\sin\theta_r, & \{P_r,Q\}=0.
\end{cases}
\end{aligned}
\label{eq:prelims_pauli_conjugation_dichotomy}
\end{equation}
When the two strings anticommute, \(iP_rQ\) is itself a Hermitian Pauli string up to a sign. The two terms therefore form branches weighted by \(\cos\theta_r\) and \(\sin\theta_r\) \cite{rall2019}. Encoding rotations produce the same branching, with the encoded input angle in place of \(\theta_r\).

A path \(p\) follows one branch at each split until the observable has passed through every gate. It accumulates a product of trigonometric factors and ends in a Pauli string \(P_p\), whose expectation in the initial state determines its contribution to the output. The factors from encoding gates depend on \(x\), while those from trainable gates depend on \(\theta\), so the sum over paths takes the form
\begin{equation}
f(x;\theta)=\sum_{p\in\mathcal P}
\alpha_p\,t_p^{(x)}(x)t_p^{(\theta)}(\theta).
\label{eq:path_expansion_fx}
\end{equation}
Here \(t_p^{(x)}\) and \(t_p^{(\theta)}\) contain the input and parameter sine and cosine factors, respectively. The remaining weight is \(\alpha_p=\sigma_p\Tr(\rho_0P_p)\), where \(\sigma_p\) combines the original observable coefficient with the signs acquired along the path. If \(P_p\) has zero expectation in the initial state, the whole path makes no contribution.

The connection to the CHM follows from the finite exponential expansions of these two products:
\begin{align}
t_p^{(x)}(x)
&=\sum_\omega\widehat t_p^{(x)}(\omega)e^{i\omega\cdot x},
\label{eq:C_txp_fourier_expansion}\\
t_p^{(\theta)}(\theta)
&=\sum_k\widehat t_p^{(\theta)}(k)e^{ik\cdot\theta}.
\label{eq:C_tthetap_fourier_expansion}
\end{align}
The hatted coefficients follow algebraically from the sine and cosine products, so these expansions require no numerical Fourier sampling. The input frequencies need not be integers, and each sum runs over the full set of frequencies in its path factor. An individual input path factor can contain frequencies outside the encoder envelope \(\Omega\). These extra frequencies cancel when the paths are combined, leaving the complete output supported within \(\Omega\) (Appendix~\ref{app:pp_intermediate_support}). For a given pair \((\omega,k)\), the CHM entry is the total contribution of all paths to the corresponding joint harmonic:
\begin{equation}
C_{\omega k}
=\sum_{p\in\mathcal P}\alpha_p\,
\widehat t_p^{(x)}(\omega)\widehat t_p^{(\theta)}(k).
\label{eq:C_path_factorisation}
\end{equation}

This expression recovers the joint Fourier coefficient defined in Eq.~\eqref{eq:C_fourier_coeff} and shows how it depends on the circuit. Paths contributing to the same entry can have different signs or complex phases, so their contributions may reinforce one another or cancel completely. The number of paths alone therefore does not determine which entries of \(C\) are nonzero, as the example below illustrates. The full derivation of the path factors, gate recursion and finite Fourier coefficients is given in Appendix~\ref{appendix:pauli_propagation}.

Constructing \(C\) in this way retains the full dependence of every input coefficient on the trainable parameters. Its ensemble moments require less information: the calculation in Section~\ref{sec:ensemble_computation} obtains them directly by contracting pairs of propagation terms, without storing their parameter-harmonic labels.

\subsection{A one-qubit example}
\label{subsec:C_example}

Consider \(\rho_0=\ket{0}\!\bra{0}\), \(O=Z\), and the circuit \(U(x;\theta)=R_X(\theta_2)R_Y(\theta_1)R_Y(x)\), with two independent trainable angles. Its encoder permits \(\Omega=\{-1,0,1\}\). The first backward step maps \(Z\) to \(Z\cos\theta_2+Y\sin\theta_2\). Continuing through the two \(Y\) rotations and contracting with \(\ket{0}\) leaves two scalar path contributions:
\begin{equation}
\begin{aligned}
f(x;\theta)
&=\cos x\cos\theta_1\cos\theta_2\\
&\quad-\sin x\sin\theta_1\cos\theta_2\\
&=\cos(x+\theta_1)\cos\theta_2.
\end{aligned}
\label{eq:C_example_output}
\end{equation}
Terms ending in \(X\) or \(Y\) have zero state expectation for this input state. The surviving parameter harmonics have \(k=(k_1,k_2)\in\{-1,1\}^2\). Including \(\mathbf0\) by convention, we take \(\mathcal K=\{-1,1\}^2\cup\{(0,0)\}\), giving a \(3\times5\) CHM. Figure~\ref{fig:k_branching} follows the two surviving terms into their CHM contributions.

\begin{figure}[t]
\centering
\begingroup
\hyphenpenalty=10000
\exhyphenpenalty=10000
\definecolor{chmBlue}{RGB}{0,65,145}
\definecolor{chmTeal}{RGB}{0,118,128}
\definecolor{chmOrange}{RGB}{190,95,0}
\begin{tikzpicture}[
  font=\small,
  >=Latex,
  gate/.style={draw=chmBlue!75,fill=white,minimum width=1.25cm,
    minimum height=.62cm,inner sep=3pt},
  pathterm/.style={draw=chmBlue!50,fill=chmBlue!4,rounded corners=2pt,
    text width=3.2cm,align=center,inner sep=5pt},
  arrow/.style={->,draw=chmBlue!75,line width=.7pt}
]
\draw[black!70,line width=.7pt] (.60,0)--(6.48,0);
\node at (.22,0) {$\ket{0}$};
\node[gate] at (1.5,0) {$R_Y(x)$};
\node[gate] at (3.4,0) {$R_Y(\theta_1)$};
\node[gate] at (5.3,0) {$R_X(\theta_2)$};
\node[anchor=west] at (6.50,0) {measure $Z$};

\node[align=center] at (4,-.92) {After propagation and state contraction};
\node[pathterm] (leftpath) at (2,-1.72)
  {\textbf{Path \(p_1\)}\\[2pt]$\cos x\cos\theta_1\cos\theta_2$};
\node[pathterm] (rightpath) at (6,-1.72)
  {\textbf{Path \(p_2\)}\\[2pt]$-\sin x\sin\theta_1\cos\theta_2$};

\draw[arrow] (leftpath.south)--(2,-2.60);
\draw[arrow] (rightpath.south)--(6,-2.60);
\node (leftmatrix) at (2,-3.24) {$\displaystyle
  \frac18\begin{pmatrix}
    1&1&0&\textcolor{chmOrange}{1}&\textcolor{chmOrange}{1}\\
    0&0&0&0&0\\
    \textcolor{chmOrange}{1}&\textcolor{chmOrange}{1}&0&1&1
  \end{pmatrix}$};
\node at (4,-3.24) {$+$};
\node (rightmatrix) at (6,-3.24) {$\displaystyle
  \frac18\begin{pmatrix}
    1&1&0&\textcolor{chmOrange}{-1}&\textcolor{chmOrange}{-1}\\
    0&0&0&0&0\\
    \textcolor{chmOrange}{-1}&\textcolor{chmOrange}{-1}&0&1&1
  \end{pmatrix}$};

\node[draw=chmTeal!70,fill=chmTeal!5,rounded corners=2pt,
  inner xsep=13pt,inner ysep=9pt] (result) at (4,-5.45)
  {$\displaystyle C=\frac14
    \begin{pmatrix}
      1&1&0&0&0\\
      0&0&0&0&0\\
      0&0&0&1&1
    \end{pmatrix}$};
\draw[arrow] (leftmatrix.south)--(2,-4.25)
  -| ([xshift=-.55cm]result.north);
\draw[arrow] (rightmatrix.south)--(6,-4.25)
  -| ([xshift=.55cm]result.north);

\node[anchor=north,align=center,text width=8.1cm,inner sep=0pt] at (4,-6.53)
  {Row order: $\omega=-1,0,1$.\\[3pt]
   Column order for $k=(k_1,k_2)$:\\[2pt]
   $(-1,-1),(-1,1),(0,0),(1,-1),(1,1)$.\\[4pt]
   \textcolor{chmOrange}{Orange entries cancel}; the other nonzero entries add.\\[3pt]
   $C$ has $3$ rows and $5$ columns.};
\end{tikzpicture}
\endgroup
\caption{\justifying Constructing a CHM from two surviving paths. For \(U=R_X(\theta_2)R_Y(\theta_1)R_Y(x)\), \(\rho_0=\ket{0}\!\bra{0}\), and \(O=Z\), each path populates eight joint harmonics. Orange entries cancel when added; the remaining entries add. The displayed row and column orders apply to all three matrices. The final CHM is rectangular: three input-frequency rows and five two-component parameter-frequency columns. Its four nonzero entries are \(C_{-1,(-1,\pm1)}=C_{1,(1,\pm1)}=1/4\). CHMs have no general square-matrix or transpose-symmetry requirement.}
\label{fig:k_branching}
\end{figure}

Each path separately contributes to all eight triples \((\omega,k_1,k_2)\in\{-1,1\}^3\). The contributions cancel for \(\omega=-k_1\) and add for \(\omega=k_1\), for either sign of \(k_2\). Consequently,
\begin{equation}
a_{\pm1}(\theta)=\tfrac12e^{\pm i\theta_1}\cos\theta_2,
\qquad a_0(\theta)=0.
\label{eq:C_example_coefficients}
\end{equation}
The zero-frequency term is absent despite being permitted by the encoding. The two nonzero coefficients are conjugate, with the same magnitude \(|\cos\theta_2|/2\). They do not represent two independent complex coefficients.

Indeed, the absence of a constant term also occurs in the nonentangling YZY circuit used in our numerical study. With our initial state and a single \(R_X\) or \(R_Y\) encoding on each qubit, the local trainable gates give an output of the form \(A(\theta)\cos x+B(\theta)\sin x\), regardless of trainable depth. Thus \(a_0(\theta)\equiv0\), so its variance and its covariances with all other coefficients vanish. This is visible in the ``YZY (no entanglers)'' row of Fig.~\ref{fig:variance-profile-atlas}: the \(\omega=0\) markers lie on the zero-variance line for both encoders at every plotted depth. Appendix~\ref{subsec:appendix_yzy_dashboards} derives this restriction.

The same fixed rows now determine how these coefficients fluctuate across the parameter ensemble.

\section{Coefficient statistics from the circuit harmonic matrix}
\label{sec:coeff_stats}

Having constructed \(C\), we now ask what a typical parameter choice produces: what is the mean of each Fourier coefficient, how strongly does it vary, and how are fluctuations at different frequencies related? Throughout this section, the character coordinates \(\theta\) are independent and uniform on \(\mathbb T^m\), corresponding to the full physical-angle periods defined in Section~\ref{subsec:param_harmonics}.

\subsection{Means and centred covariance}
\label{subsec:centering_moments}
\label{subsec:covariance_CCdag}

Recall \(a(\theta)=C\psi(\theta)\), where \(\psi_k(\theta)=e^{ik\cdot\theta}\). The constant parameter character contributes to the mean; the remaining characters describe fluctuations about it. Write
\begin{equation}
\mathcal K^\circ:=\mathcal K\setminus\{0\},
\qquad
C^\circ:=\bigl[C_{\omega k}\bigr]_
{\omega\in\Omega,\,k\in\mathcal K^\circ}.
\label{eq:coeffstats_P_def}
\end{equation}
With \(\bar a:=\mathbb E_\theta[a(\theta)]\), the centred coefficient vector is
\begin{equation}
\tilde a(\theta):=a(\theta)-\bar a
=C^\circ\psi^\circ(\theta),
\qquad
\psi^\circ_k(\theta)=e^{ik\cdot\theta}.
\label{eq:coeffstats_centered_a}
\end{equation}
A bullet denotes a complete row or column: \(C_{\omega\bullet}\) is row \(\omega\), and \(C_{\bullet k}\) is column \(k\).

\begin{proposition}[Coefficient moments from \(C\)]
\label{prop:coeff_moments_from_C}
Under the uniform parameter-torus ensemble,
\begin{align}
\bar a=\mathbb E_\theta[a(\theta)]
&=C_{\bullet0},
\label{eq:coeffstats_mean_vec}\\
S:=\mathbb E_\theta[a(\theta)a(\theta)^\dagger]
&=CC^\dagger,
\label{eq:coeffstats_second_moment_CCdag}\\
\Sigma:=\mathbb E_\theta[\tilde a(\theta)\tilde a(\theta)^\dagger]
&=C^\circ C^{\circ\dagger}.
\label{eq:coeffstats_cov_CPCdag}
\end{align}
Thus the centred coefficient covariance \(\Sigma=S-\bar a\bar a^\dagger\) is the row Gram matrix of the \(k\neq\mathbf0\) column submatrix of \(C\). Unlike \(C\), it is square and Hermitian, with input frequencies indexing both axes.
\end{proposition}

\begin{proof}
Insert \(a=C\psi\) and use the character orthogonality in Eq.~\eqref{eq:prelims_character_orthogonality}. Only the constant character survives the mean, and only equal character indices survive the second moment. Centring removes the constant-column contribution. The componentwise calculation is in Appendix~\ref{app:cov_CCdag}.
\end{proof}

These are standard Fourier-orthogonality identities; the CHM makes their circuit dependence explicit. Section~\ref{sec:ensemble_computation} evaluates these moments without materialising \(C\) by utilising the fact that only equal character indices survive, reorganising the exact coefficient-covariance calculations of Wiedmann \emph{et al.}~\cite{wiedmann2024}.

\subsection{Coefficient variance and exact zeros}
\label{subsec:variance_correlation}

The diagonal of \(\Sigma\) gives the coefficient variance:
\begin{equation}
v_\omega:=\mathbb E_\theta[|\tilde a_\omega|^2]
=\Sigma_{\omega\omega}
=\sum_{k\in\mathcal K^\circ}|C_{\omega k}|^2
=\|C^\circ_{\omega\bullet}\|_2^2.
\label{eq:coeffstats_variance}
\end{equation}
Thus \(\sqrt{v_\omega}\) is the root-mean-square fluctuation of \(a_\omega\) about its mean under the stated ensemble. For a zero-mean coefficient it is also the root-mean-square magnitude. More generally,
\begin{equation}
\mathbb E_\theta[|a_\omega|^2]
=|\bar a_\omega|^2+v_\omega
=\|C_{\omega\bullet}\|_2^2.
\label{eq:coeffstats_parseval}
\end{equation}
A variance profile therefore tells us how much each coefficient varies across the parameter ensemble. Figure~\ref{fig:variance-profile-atlas} in Appendix~\ref{subsec:appendix_variance_profiles} shows these profiles for all eight circuit families and both encoding axes in our numerical study. We examine the implications for learning in Section~\ref{sec:kernels}, where parameter derivatives describe how strongly the coefficients respond to updates and which combinations can be adjusted locally.

\begin{corollary}[Characterisation of zero variance]
\label{cor:zero_variance_characterisation}
A coefficient has \(v_\omega=0\) if and only if \(C_{\omega k}=0\) for every \(k\neq0\), equivalently \(a_\omega(\theta)=C_{\omega0}\) for every \(\theta\). It is identically absent only if \(C_{\omega0}=0\) as well.
\end{corollary}

\begin{proof}
Equation~\eqref{eq:coeffstats_variance} is a sum of nonnegative terms, so it vanishes exactly when every entry in a \(k\neq\mathbf0\) column of that row vanishes. Uniqueness of the character expansion gives the equivalent constant-coefficient statement.
\end{proof}

In the numerical calculations, we treat a coefficient as varying when its computed variance exceeds a small tolerance. The resulting support is approximate, since coefficients with very weak parameter dependence can fall below the threshold.

\subsection{Covariance and normalised correlation}
\label{subsec:coefficient_correlations}

Off-diagonal entries of the same Gram matrix give
\begin{equation}
\Sigma_{\omega\mu}
=\mathbb E_\theta[\tilde a_\omega\overline{\tilde a_\mu}]
=\sum_{k\in\mathcal K^\circ}
C_{\omega k}\overline{C_{\mu k}}.
\label{eq:coeffstats_cov_components}
\end{equation}
Covariance measures how two CHM rows overlap through their nonconstant parameter characters. If the rows have no such characters in common, their covariance is zero. Shared characters can give a nonzero result, but their complex contributions may also cancel (Fig.~\ref{fig:shared_k_columns}). It is the complete character index that must match; dependence on the same physical parameter alone is insufficient.

\begin{figure}[t]
\centering
\begingroup
\definecolor{chmSharedA}{HTML}{007C91}
\definecolor{chmSharedB}{HTML}{B45A19}
\definecolor{chmSeparate}{HTML}{70558F}
\begin{tikzpicture}[font=\small, x=1cm, y=1cm]
\def\cw{1.10}
\def\rh{0.66}
\node at (2.20,0.94) {Parameter harmonics $k\neq\mathbf0$};
\node[chmSharedA] at (0.5*\cw,0.37) {$k^{(1)}$};
\node[chmSharedB] at (1.5*\cw,0.37) {$k^{(2)}$};
\node[chmSeparate] at (2.5*\cw,0.37) {$k^{(3)}$};
\node[chmSeparate] at (3.5*\cw,0.37) {$k^{(4)}$};
\node[anchor=east] at (-0.18,-0.5*\rh) {$\omega^{(1)}$};
\node[anchor=east] at (-0.18,-1.5*\rh) {$\omega^{(2)}$};
\node[anchor=east] at (-0.18,-2.5*\rh) {$\omega^{(3)}$};
\node at (-1.40,-1.5*\rh) {$C^\circ:$};
\fill[chmSharedA!12] (0,0) rectangle (\cw,-2*\rh);
\fill[chmSharedB!12] (\cw,0) rectangle (2*\cw,-2*\rh);
\fill[chmSeparate!12] (2*\cw,-2*\rh) rectangle (4*\cw,-3*\rh);
\foreach \i in {0,1,2,3,4}
  \draw[gray!50, line width=0.35pt]
    (\i*\cw,0) -- (\i*\cw,-3*\rh);
\foreach \j in {0,1,2,3}
  \draw[gray!50, line width=0.35pt]
    (0,-\j*\rh) -- (4*\cw,-\j*\rh);
\node[chmSharedA] at (0.5*\cw,-0.5*\rh) {$C_{1,1}$};
\node[chmSharedB] at (1.5*\cw,-0.5*\rh) {$C_{1,2}$};
\node at (2.5*\cw,-0.5*\rh) {$0$};
\node at (3.5*\cw,-0.5*\rh) {$0$};
\node[chmSharedA] at (0.5*\cw,-1.5*\rh) {$C_{2,1}$};
\node[chmSharedB] at (1.5*\cw,-1.5*\rh) {$C_{2,2}$};
\node at (2.5*\cw,-1.5*\rh) {$0$};
\node at (3.5*\cw,-1.5*\rh) {$0$};
\node at (0.5*\cw,-2.5*\rh) {$0$};
\node at (1.5*\cw,-2.5*\rh) {$0$};
\node[chmSeparate] at (2.5*\cw,-2.5*\rh) {$C_{3,3}$};
\node[chmSeparate] at (3.5*\cw,-2.5*\rh) {$C_{3,4}$};
\node at (1.60,-2.39) {$C_{a,b}:=C_{\omega^{(a)},\,k^{(b)}}$};
\node[anchor=north, align=left] at (1.60,-2.78) {%
  $\begin{aligned}
  v_{\omega^{(1)}}&=|C_{1,1}|^2+|C_{1,2}|^2
       &&\text{(squared row norm)}\\[3pt]
  \Sigma_{\omega^{(1)}\omega^{(2)}}&=C_{1,1}\overline{C_{2,1}}+C_{1,2}\overline{C_{2,2}}
       &&\text{(shared columns)}\\[3pt]
  \Sigma_{\omega^{(1)}\omega^{(3)}}&=\Sigma_{\omega^{(2)}\omega^{(3)}}=0
       &&\text{(disjoint support)}
  \end{aligned}$};
\end{tikzpicture}
\endgroup
\caption{\justifying Variance and covariance as row geometry. This schematic shows three rows of \(C^\circ\), using the local shorthand \(C_{a,b}:=C_{\omega^{(a)},k^{(b)}}\) for the selected rows and columns. The labelled entries are illustrative fixed nonzero complex coefficients. All four displayed columns have \(k\neq\mathbf0\), and all omitted entries in these rows are zero. Squared magnitudes add to the variance. Covariance instead sums products in shared columns, so the two contributions between \(\omega^{(1)}\) and \(\omega^{(2)}\) can reinforce or cancel. Row \(\omega^{(3)}\) has disjoint support and hence zero covariance with both other rows.}
\label{fig:shared_k_columns}
\end{figure}

To compare this overlap independently of marginal scale, restrict to
\begin{equation}
\Omega_+:=\{\omega\in\Omega:v_\omega>0\}
\label{eq:coeffstats_corr_D}
\end{equation}
and define the Hermitian correlation matrix
\begin{equation}
R_{\omega\mu}:=
\frac{\Sigma_{\omega\mu}}{\sqrt{v_\omega v_\mu}},
\qquad \omega,\mu\in\Omega_+.
\label{eq:coeffstats_corr_def}
\end{equation}
This is the Gram matrix of the unit-normalised rows of \(C^\circ\); it is positive semidefinite, with \(R_{\omega\omega}=1\) and \(|R_{\omega\mu}|\leq1\). Correlation is undefined for a zero-variance coefficient. Indeed one must take care when interpreting numerical correlation values, as they can be misleading if the underlying variances are small. Unlike covariance, it removes marginal scale: a large correlation can link two coefficients whose absolute fluctuations are both small.

An ensemble covariance describes the linear spread of coefficient vectors across parameter choices. Those vectors can lie on a curved set whose affine span has a larger dimension than the set itself. For example, the two-qubit circuit developed in Section~\ref{subsec:kernel_parameter_average} has coefficients \(c_1=\cos\theta\) and \(s_2=-\tfrac12\sin\theta\) multiplying \(\cos x\) and \(\sin(2x)\). Uniform sampling of \(\theta\) gives zero covariance between these two distinct frequencies, while \(c_1^2+4s_2^2=1\) confines their values to an ellipse. Their correlation matrix is the identity even though their values are linked by this nonlinear constraint.

Training starts from a particular point on this ellipse. A small parameter update moves the coefficients along its tangent, giving one available direction within the two-dimensional ensemble span. This distinction motivates the pointwise QNTK introduced in Section~\ref{sec:kernels} and why looking at covariances alone is insufficient. Training is local and understanding how the sensistivity of locally available directions change across parameter space is crucial for understanding the learning dynamics induced by particular circuit choices. Indeed the tangent kernel can have a smaller effective dimension than the ensemble itself. Figure~\ref{fig:ensemble-tangent-geometry} develops the example and shows how both the ensemble covariance and the local kernel follow from the same CHM.

\subsection{Redundancy and finite-depth circuit structure}
\label{subsec:expressivity_redundancy}
\label{subsec:torus_vs_design}

Circuits with the same encoding block can have different variance profiles. The encoder spectrum identifies the permitted input frequencies~\cite{schuld2021}, but the variance of each coefficient depends on the complete circuit through the CHM row norm in Eq.~\eqref{eq:coeffstats_variance}. Different pairs of encoder eigenstates can give the same eigenvalue difference and hence the same input frequency. With several encoding layers, these differences add across the layers. The set \(\mathcal{R}(\omega)\) contains all pairs of eigenstate paths that produce \(\omega\), and its size \(|\mathcal{R}(\omega)|\) is the \emph{frequency redundancy}.

Redundancy counts these contributions without recording their weights. The CHM retains the complex weights, including any reinforcement or cancellation when contributions are added. Consequently, the same encoder redundancy can lead to different coefficient variances and correlations. Duffy and Jastrzebski~\cite{duffy2026} also count the nonzero terms in a circuit's trigonometric expansion. That count incorporates circuit structure, but it need not equal the encoder-path count and still does not retain the full weights. The relation between the two counts is given in Appendix~\ref{app:redundancy_dictionary}.

The variance becomes proportional to redundancy for certain ensembles of trainable blocks. For a single encoding layer and a pure input state, Mhiri \emph{et al.}~\cite{mhiri2025a} obtain \(v_\omega\propto|\mathcal{R}(\omega)|\) for \(\omega\neq0\) when the blocks before and after encoding are sampled independently from exact unitary \(2\)-designs. Such ensembles match the second moments of uniformly random unitaries (Haar measure), which are sufficient for this variance calculation \cite{dankert2009a,mele2024}. The proportionality factor depends on the Hilbert-space dimension and observable, so large redundancy can still accompany small absolute variance as the number of qubits grows. Appendix~\ref{app:mhiri_chm} gives the full expression and explains how the assumptions change for re-uploading and approximate designs.

Our numerical circuits have a fixed input state and no trainable block before encoding. For the state and \(R_X/R_Y\) encoding used here, however, an exact \(2\)-design after encoding alone would give the same relative variances at nonzero frequencies, with a different overall scale (Appendix~\ref{app:fixed_input_design_reference}). This provides a reference for comparing our circuits as trainable depth increases: the encoder redundancy stays fixed, while the weights in \(C\) change, with the variance profile tending towards redundancy proportionality as in the 2-design case. In Section~\ref{sec:numerical-results}, we examine how closely the relative variance profiles follow this reference, alongside changes in their absolute scale. A profile can become more nearly proportional to redundancy even while its variances decrease.

The next section examines how parameter updates move through these attainable coefficient sets. Differentiating \(a(\theta)=C\psi(\theta)\) gives the coefficient changes available to first order locally at the current parameters. Evaluating them on the training inputs then gives the tangent kernel governing the local output response.

\section{Tangent kernels from the circuit harmonic matrix}
\label{sec:kernels}

Training follows a path through the coefficient vectors that the circuit can produce. At each step, the current parameters determine which nearby changes are available and how strongly each parameter affects them. A curved attainable set can span several coefficient directions over the ensemble while offering fewer directions at any one point, as in the ellipse described in Section~\ref{subsec:coefficient_correlations}. Neural tangent kernels express this local geometry through model derivatives \cite{jacot2018,liu2022}. Their range identifies the changes available to first order, and their nonzero eigenvalues describe the sensitivity along those directions.

The CHM provides a common starting point for both descriptions. The Gram matrix \(C^\circ C^{\circ\dagger}\) gives coefficient covariance, while differentiating \(a=C\psi\) gives the local derivatives from which we construct the coefficient-space QNTK. Evaluation on the training inputs then converts the coefficient response into the data-space kernel that enters the learning dynamics. Figure~\ref{fig:ensemble-tangent-geometry} follows these operations for a concrete circuit. We also average the local kernels to connect their sensitivities with ensemble variation and to understand how averaging combines tangent directions from different parameter points.

Rescaling a parameter changes its derivative and therefore its contribution to the kernel. The derivation below uses the Euclidean metric in the \(2\pi\)-periodic character coordinates \(\theta\). For the learning equations and numerical calculations, we convert the derivatives to the physical angles \(\vartheta\) used by the optimiser, including the period factors introduced in Section~\ref{subsec:param_harmonics}.

\subsection{Local response in coefficient and data space}
\label{subsec:qntk_coefficient_space}
\label{subsec:kernel_chm_factorisation}

For inputs \(x_1,\ldots,x_N\), the vector of real model outputs \(F(\theta)\in\mathbb R^N\) is obtained by evaluating the Fourier expansion at each input:
\begin{equation}
F(\theta)=Va(\theta),
\qquad V_{i\omega}=e^{i\omega\cdot x_i}.
\label{eq:kernels_F_Va}
\end{equation}
The Fourier design matrix \(V\in\mathbb C^{N\times|\Omega|}\) depends only on the inputs and the chosen frequency set. Parameter changes affect the coefficients and sampled outputs through their respective Jacobians,
\begin{equation}
X_{\omega r}(\theta):=\partial_{\theta_r}a_\omega(\theta),
\qquad
J_{ir}(\theta):=\partial_{\theta_r}f(x_i;\theta).
\label{eq:kernels_jacobians}
\end{equation}
To first order, \(\delta a=X\delta\theta\) and \(\delta F=J\delta\theta\). Since \(V\) is independent of the parameters, these responses are related by
\begin{equation}
J=VX.
\label{eq:kernels_J_VX}
\end{equation}

The overlaps between coefficient gradients form the coefficient-space QNTK,
\begin{equation}
H(\theta):=X(\theta)X(\theta)^\dagger,
\qquad
H_{\omega\mu}=\sum_r X_{\omega r}\overline{X_{\mu r}}.
\label{eq:kernels_H_definition}
\end{equation}
Its diagonal measures the squared local sensitivity of each coefficient, while an off-diagonal entry measures the overlap between two coefficient gradients. Unlike the covariance \(\Sigma\), which describes joint fluctuations across the ensemble, \(H\) describes sensitivity at the current parameter point. Taking the same gradient overlaps between sampled outputs gives the standard data-space QNTK \cite{liu2022},
\begin{equation}
K(\theta):=J(\theta)J(\theta)^{\mathsf T},
\qquad K_{ij}=\sum_r J_{ir}J_{jr}.
\label{eq:kernels_K_definition}
\end{equation}
Because \(J\) is real, \(J^{\mathsf T}=J^\dagger\), so
\begin{equation}
K(\theta)=VH(\theta)V^\dagger.
\label{eq:kernels_K_VHV}
\end{equation}
Both kernels are positive semidefinite. The matrix \(H\) describes how the Fourier coefficients respond to parameter changes independently of the selected inputs. The map \(V\) then converts that response into the changes visible on a particular dataset.

The dependence of both kernels on the circuit can be made explicit by differentiating the CHM expansion \(a=C\psi\):
\begin{equation}
X=CG,\qquad
G_{kr}(\theta):=\partial_{\theta_r}\psi_k(\theta)
=ik_r e^{ik\cdot\theta}.
\label{eq:kernels_G_definition}
\end{equation}
Here \(G\) contains the gradients of the parameter characters. Their Gram matrix is
\begin{equation}
M(\theta):=G(\theta)G(\theta)^\dagger,
\label{eq:kernels_M_definition}
\end{equation}
whose entries are
\begin{equation}
M_{kl}(\theta)=(k\cdot l)e^{i(k-l)\cdot\theta}.
\label{eq:kernels_M_entries}
\end{equation}
Consequently,
\begin{equation}
H(\theta)=CM(\theta)C^\dagger.
\label{eq:kernels_full_factorisation}
\end{equation}
For a fixed character set and choice of parameter coordinates, all parameter dependence in this expression lies in \(M\). The fixed circuit weights in \(C\) turn the character-gradient overlaps into the coefficient kernel \(H\), and evaluation through \(V\) gives the data kernel \(K\). Figure~\ref{fig:kernel_hierarchy} shows how these steps follow the original model expansion.

\begin{figure}[t]
\centering
\begingroup
\definecolor{chmKernelBlue}{HTML}{005E78}
\definecolor{chmKernelOrange}{HTML}{A85318}
\begin{tikzpicture}[
  font=\small, >=Latex,
  box/.style={draw=black!50, rounded corners=2pt,
    text width=2.90cm, minimum height=0.88cm, align=center,
    inner xsep=4pt, inner ysep=4pt},
  map/.style={->, line width=0.7pt},
  derive/.style={->, dashed, draw=black!65, line width=0.65pt}
]
\node[font=\small\bfseries] at (0,0.88) {Model representation};
\node[font=\small\bfseries] at (4.0,0.88) {Jacobian Gram matrix};
\node[box, fill=black!3] (psi) at (0,0)
  {$\psi(\theta)$\\characters};
\node[box, fill=chmKernelBlue!6] (a) at (0,-1.65)
  {$a=C\psi$\\Fourier coefficients};
\node[box, fill=chmKernelOrange!6] (f) at (0,-3.30)
  {$F=Va$\\sampled outputs};
\node[box, fill=black!3] (m) at (4.0,0)
  {$M=GG^\dagger$\\character space};
\node[box, fill=chmKernelBlue!6] (h) at (4.0,-1.65)
  {$H=XX^\dagger$\\coefficient space};
\node[box, fill=chmKernelOrange!6] (k) at (4.0,-3.30)
  {$K=JJ^{\mathsf T}$\\data space};
\draw[map] (psi) -- node[left=3pt] {$C$} (a);
\draw[map] (a) -- node[left=3pt] {$V$} (f);
\draw[map] (m.south) -- node[right=3pt] {$C(\,\cdot\,)C^\dagger$} (h.north);
\draw[map] (h.south) -- node[right=3pt] {$V(\,\cdot\,)V^\dagger$} (k.north);
\draw[derive] (psi.east) -- (m.west);
\draw[derive] (a.east) -- (h.west);
\draw[derive] (f.east) -- (k.west);
\end{tikzpicture}
\endgroup
\caption{\justifying The CHM connects coefficient and data-space tangent kernels. On the left, the fixed maps \(C\) and \(V\) take parameter characters to Fourier coefficients and sampled outputs. On the right, the same maps relate their Jacobian Gram matrices through \(H=CMC^\dagger\) and \(K=VHV^\dagger\). Each horizontal arrow represents differentiation followed by contraction over parameter indices. All quantities are evaluated at the same parameter point and in the same parameter coordinates.}
\label{fig:kernel_hierarchy}
\end{figure}

Only nonconstant parameter characters contribute to the derivatives, so \(X=C^\circ G^\circ\). Consequently, \(\operatorname{col}H(\theta)\subseteq\operatorname{col}C^\circ\) and \(\operatorname{rank}H(\theta)\leq m\), where \(\operatorname{col}\) denotes column space. In independent real Fourier coordinates, the Jacobian and its Gram matrix have the same range and rank. At a regular parameter point, where the rank is constant in a neighbourhood, this rank is the local dimension of the attainable coefficient set. It can be smaller than the dimension of the ensemble's affine span. A finite set of inputs can hide further changes: any coefficient change in \(\ker V\) leaves all sampled outputs unchanged. Appendix~\ref{app:kernel_rank_geometry} gives the full rank relations and their proofs.

The numerical calculations use independent real Fourier coefficients. Their normalisation affects the eigenvalues of the coefficient kernel, so it is part of the specification needed to interpret its spectrum. Appendix~\ref{app:kernel_real_coordinates} gives the conversion from the complex coefficients used here.

\subsection{Parameter averages and coefficient variability}
\label{subsec:kernel_parameter_average}

To compare local sensitivity with coefficient variation, we average \(H(\theta)\) over the same full-period parameter ensemble used for the covariance. Orthogonality of the parameter characters removes the terms with \(k\neq l\), leaving
\begin{equation}
\mathbb E_\theta[M_{kl}(\theta)]
=\|k\|_2^2\delta_{kl}.
\label{eq:kernels_M_average}
\end{equation}
The averaged kernels are therefore
\begin{align}
\bar H:=\mathbb E_\theta[H(\theta)]
&=C^\circ W_{\mathcal K}C^{\circ\dagger},
\label{eq:kernels_H_average}\\
\bar K:=\mathbb E_\theta[K(\theta)]
&=V\bar H V^\dagger,
\label{eq:kernels_K_average}\\
W_{\mathcal K}
&:=\operatorname{diag}_{k\in\mathcal K^\circ}(\|k\|_2^2).
\label{eq:kernels_weight_matrix}
\end{align}
Both \(\Sigma=C^\circ C^{\circ\dagger}\) and \(\bar H\) are built from the same CHM rows. Their difference is the weight assigned to each parameter character: differentiation introduces the factor \(\|k\|_2^2\) in \(\bar H\). For single-use Pauli rotations, each \(k_r\) is \(0\) or \(\pm1\), so this factor counts how many parameters occur in the harmonic. Harmonics involving more parameters therefore contribute more strongly to the averaged sensitivity. Reusing a parameter can produce higher harmonic orders, whose squares enter the weight. Section~\ref{sec:ensemble_computation} shows how to calculate this weighted matrix product alongside the covariance, without constructing \(C\) explicitly.

For a particular coefficient, the diagonal entry is its mean squared gradient:
\begin{equation}
\begin{aligned}
\bar H_{\omega\omega}
&=\mathbb E_\theta[\|\nabla_\theta a_\omega\|_2^2]\\
&=\sum_{k\in\mathcal K^\circ}\|k\|_2^2|C_{\omega k}|^2.
\end{aligned}
\label{eq:kernels_average_diagonal_interpretation}
\end{equation}
Dividing by the coefficient variance compares this sensitivity with the overall amount of variation. For \(v_\omega>0\), we define
\begin{equation}
\gamma_\omega:=\frac{\bar H_{\omega\omega}}{v_\omega}.
\label{eq:kernels_sensitivity_variation_ratio}
\end{equation}
This is a weighted average of the squared harmonic order, with weights \(|C_{\omega k}|^2\). A large value means that the mean squared sensitivity is large relative to the variance. It does not by itself give a learning rate, which also depends on the dataset, the residual and the kernel at the current parameter point.

To express sensitivity in the physical angles, the chain rule [Eq.~\eqref{eq:prelims_physical_character_derivative}] multiplies column \(r\) of \(G\) by \(2\pi/T_r\). The weight in the averaged kernel then becomes
\begin{equation}
w_k^{(\vartheta)}=\sum_r\left(\frac{2\pi k_r}{T_r}\right)^2.
\label{eq:kernels_physical_weight}
\end{equation}
A local rotation with \(T_r=2\pi\) contributes \(k_r^2\), whereas a controlled rotation with \(T_r=4\pi\) contributes \(k_r^2/4\). All numerical QNTKs and sensitivity-to-variation ratios use these physical-angle weights. The conversion preserves the ensemble of outputs being averaged, but changes the sensitivity associated with a unit displacement in each parameter coordinate.

Every nonconstant character has positive weight in either coordinate system. Changing these weights therefore preserves the column space:
\begin{equation}
\operatorname{col}\bar H
=\operatorname{col}\Sigma
=\operatorname{col}C^\circ.
\label{eq:kernels_covariance_average_col}
\end{equation}
Although \(\bar H\) and \(\Sigma\) describe the same linear span, their different weights generally give different eigenvalues and eigenvectors. Moreover, an ensemble average can combine directions that are available at different parameter points, even when they cannot all be reached by a small update at any one point.

To make the distinction concrete, consider two qubits initially in \(\lvert00\rangle\), with \(R_Y(x)\) applied to each. Let \(V_{\mathrm{CX}}=\operatorname{CNOT}_{1\to0}\) and apply the trainable block \(W(\theta)=V_{\mathrm{CX}}R_Y^{(0)}(\theta)V_{\mathrm{CX}}\), measuring \(Z_0\). Conjugation gives \(W^\dagger Z_0W=Z_0\cos\theta-X_0Z_1\sin\theta\), so
\begin{equation}
\begin{aligned}
f(x;\theta)&=\cos\theta\cos x-\tfrac12\sin\theta\sin(2x),\\
q(\theta)&=(c_1,s_2)^{\mathsf T}
=(\cos\theta,-\tfrac12\sin\theta)^{\mathsf T}.
\end{aligned}
\label{eq:kernels_two_frequency_model}
\end{equation}
These two real coefficients belong to distinct input frequencies. The constraint \(c_1^2+4s_2^2=1\) makes their attainable set an ellipse. Its affine span is the entire \((c_1,s_2)\) plane.

With parameter characters ordered as \(\psi=(e^{-i\theta},e^{i\theta})^{\mathsf T}\), the same CHM expressed in these two active real input coordinates gives
\begin{equation}
q=C_q\psi,\qquad
C_q=\begin{pmatrix}
\tfrac12&\tfrac12\\
-\tfrac{i}{4}&\tfrac{i}{4}
\end{pmatrix}.
\label{eq:kernels_two_frequency_chm}
\end{equation}
Figure~\ref{fig:ensemble-tangent-geometry}(a) gives the original CHM \(C\) and the combinations of its rows forming \(C_q\). For uniform \(\theta\), averaging and differentiating this expansion give
\begin{equation}
\begin{aligned}
\Sigma_{\mathbb R}=C_qC_q^\dagger
&=\begin{pmatrix}\tfrac12&0\\0&\tfrac18\end{pmatrix},\\
H_{\mathbb R}(\theta)
=C_q\psi'\psi'^\dagger C_q^\dagger
&=\begin{pmatrix}
\sin^2\theta&\tfrac12\sin\theta\cos\theta\\
\tfrac12\sin\theta\cos\theta&\tfrac14\cos^2\theta
\end{pmatrix}.
\end{aligned}
\label{eq:kernels_two_frequency_geometry}
\end{equation}
The covariance has rank two and the correlation matrix is the identity. At each parameter point, \(H_{\mathbb R}=q'q'^{\mathsf T}\) has rank one, with range along the tangent \(q'=(-\sin\theta,-\tfrac12\cos\theta)^{\mathsf T}\). A parameter update moves both coefficients along this direction. At \(\theta_0=\pi/4\), for example, their infinitesimal changes obey \(\delta s_2=\tfrac12\delta c_1\). Figure~\ref{fig:ensemble-tangent-geometry}(b,c) contrasts the two-dimensional ensemble span with this local freedom.

Averaging the kernel gives \(\bar H_{\mathbb R}=\Sigma_{\mathbb R}\) in this example, since both parameter characters have unit squared frequency. The average has rank two because the tangent rotates around the ellipse. The equality is specific to these character weights, while equality of the covariance and averaged-kernel ranges holds generally under the full-period ensemble. Their common span describes the directions accumulated across parameter choices. The pointwise kernel identifies which of those directions are available at the current point.

The changing tangent also explains the scope of a frozen-kernel description. Starting at \(q_0=q(\theta_0)\), the linearised coefficients follow \(\widehat q=q_0+q'_0\delta\), while the circuit follows \(q(\theta_0+\delta)\). Figure~\ref{fig:ensemble-tangent-geometry}(d) compares these paths for the same parameter displacement. Their separation reflects the curvature of the attainable set, motivating the kernel-drift and forecast-error measurements in Section~\ref{subsec:numerical-frozen-qntk}.

\begin{figure*}[t]
\centering
\begingroup
\definecolor{chmTangentBlue}{HTML}{005E78}
\definecolor{chmTangentOrange}{HTML}{A85318}
\begin{tikzpicture}[
  font=\small, x=1cm, y=1cm,
  line cap=round, line join=round,
  axis/.style={-{Latex[length=1.4mm,width=1mm]},draw=black!40,line width=.5pt},
  flow/.style={-{Latex[length=1.6mm,width=1.1mm]},draw=black!65,line width=.65pt},
  tangent/.style={{Latex[length=1.5mm,width=1mm]}-{Latex[length=1.5mm,width=1mm]},
    draw=chmTangentOrange,line width=.95pt},
  heading/.style={font=\small\bfseries,align=center},
  every node/.style={inner sep=2pt}
]
\path[use as bounding box] (0,-.40) rectangle (17.40,9.82);
\node[heading,anchor=west] at (.12,9.55)
  {(a) Covariance and the QNTK from the same CHM};
\node at (8.70,8.95)
  {$f(x;\theta)=\cos\theta\cos x-\tfrac12\sin\theta\sin(2x),
    \qquad q=(c_1,s_2)^{\mathsf T}=(\cos\theta,-\tfrac12\sin\theta)^{\mathsf T}$};

\node[draw=black!18,fill=black!2,rounded corners=2pt,
  minimum width=3.65cm,minimum height=1.90cm] (chm) at (2.25,7.35)
  {$C=\dfrac18\begin{pmatrix}1&-1\\2&2\\2&2\\-1&1\end{pmatrix}$};
\node[align=center,text=black!75] at (2.25,5.98)
  {rows: $\omega=(-2,-1,1,2)$\\columns: $k=(-1,1)$};

\node[draw=black!18,fill=black!2,rounded corners=2pt,
  minimum width=3.65cm,minimum height=1.35cm] (chmq) at (8.00,7.35)
  {$C_q=\begin{pmatrix}\tfrac12&\tfrac12\\-\tfrac{i}{4}&\tfrac{i}{4}\end{pmatrix}$};
\draw[flow] (chm.east) -- (chmq.west);
\node[align=center,text=black!75] at (5.12,7.85) {real input\\coordinates};
\node[align=center] at (8.00,6.18)
  {$q=C_q\psi,\quad\psi=(e^{-i\theta},e^{i\theta})^{\mathsf T}$\\[3pt]
   $c_1=a_{-1}+a_1,\quad s_2=i(a_2-a_{-2})$};

\node[anchor=west,text=chmTangentBlue] (cov) at (11.65,8.02)
  {$\Sigma_{\mathbb R}=C_qC_q^\dagger$};
\node[anchor=west,text=black!70] at (11.65,8.53)
  {Average: $\mathbb E[\psi\psi^\dagger]=I_2$};
\node[anchor=west,text=chmTangentOrange] (kernel) at (11.65,6.65)
  {$H_{\mathbb R}(\theta)=C_q\psi'\psi'^\dagger C_q^\dagger$};
\node[anchor=west,text=black!70] at (11.65,6.12)
  {Differentiate: $q'=C_q\psi'$};
\draw[flow] (chmq.east) -- (10.65,7.35) |- (cov.west);
\draw[flow] (10.65,7.35) |- (kernel.west);

\draw[black!18,line width=.45pt] (.12,5.38) -- (17.28,5.38);
\node[heading] at (2.75,4.93) {(b) Across the ensemble};
\node[heading] at (8.70,4.93) {(c) At one parameter point};
\node[heading] at (14.65,4.93) {(d) Following a finite update};
\node[text=black!70] at (2.75,4.52) {$c_1^2+4s_2^2=1$};
\node[text=black!70] at (8.70,4.52) {$\theta_0=\pi/4$};
\node[text=black!70] at (14.65,4.52) {$\theta_0=\pi/4,\quad\delta=1$};

\begin{scope}[shift={(2.75,2.82)}]
  \filldraw[fill=chmTangentBlue!4,draw=black!16,line width=.4pt]
    (-2.10,-1.43) rectangle (2.10,1.43);
  \draw[axis] (-1.99,0) -- (2.21,0) node[right,text=black!75] {$c_1$};
  \draw[axis] (0,-1.33) -- (0,1.48);
  \node[anchor=west,text=black!75] at (.10,1.23) {$s_2$};
  \draw[chmTangentBlue,line width=1.15pt] (0,0) ellipse[x radius=1.70,y radius=.85];
  \foreach \theta in {0,20,...,340} {
    \fill[chmTangentBlue] ({1.70*cos(\theta)},{-.85*sin(\theta)}) circle[radius=1.25pt];
  }
  \foreach \theta in {30,150,270} {
    \pgfmathsetmacro{\tx}{-1.70*sin(\theta)}
    \pgfmathsetmacro{\ty}{-.85*cos(\theta)}
    \pgfmathsetmacro{\tnorm}{sqrt(\tx*\tx+\ty*\ty)}
    \coordinate (sample) at ({1.70*cos(\theta)},{-.85*sin(\theta)});
    \draw[tangent] ($(sample)+({-.37*\tx/\tnorm},{-.37*\ty/\tnorm})$)
      -- ($(sample)+({.37*\tx/\tnorm},{.37*\ty/\tnorm})$);
    \fill[chmTangentBlue] (sample) circle[radius=1.8pt];
  }
  \node[text=black!65,anchor=north west,font=\footnotesize] at (-2.03,-1.06)
    {two-dimensional affine span};
\end{scope}

\begin{scope}[shift={(8.70,2.82)}]
  \draw[axis] (-1.99,0) -- (2.21,0) node[right,text=black!75] {$c_1$};
  \draw[axis] (0,-1.33) -- (0,1.48);
  \node[anchor=west,text=black!75] at (.10,1.23) {$s_2$};
  \draw[chmTangentBlue!55,line width=.85pt] (0,0) ellipse[x radius=1.70,y radius=.85];
  \coordinate (local) at ({1.70/sqrt(2)},{-.85/sqrt(2)});
  \draw[tangent] ($(local)+(.62,.31)$) -- ($(local)+(-1.42,-.71)$);
  \fill[chmTangentBlue] (local) circle[radius=2pt];
  \node[text=chmTangentBlue,anchor=south west] at ($(local)+(.07,.13)$) {$q_0$};
  \node[text=chmTangentOrange,anchor=east] at (-.18,-1.14) {one tangent direction};
\end{scope}

\begin{scope}[shift={(14.65,2.82)}]
  \draw[axis] (-1.99,0) -- (2.21,0) node[right,text=black!75] {$c_1$};
  \draw[axis] (0,-1.33) -- (0,1.48);
  \node[anchor=west,text=black!75] at (.10,1.23) {$s_2$};
  \draw[chmTangentBlue!45,line width=.75pt] (0,0) ellipse[x radius=1.70,y radius=.85];
  \coordinate (start) at ({1.70/sqrt(2)},{-.85/sqrt(2)});
  \coordinate (exact) at ({1.70*cos(102.2957795)},{-.85*sin(102.2957795)});
  \coordinate (frozen) at (0,{-1.70/sqrt(2)});
  \draw[-{Latex[length=1.7mm,width=1.1mm]},chmTangentBlue,line width=1.45pt,
    domain=45:102.2957795,samples=40,variable=\t]
    plot ({1.70*cos(\t)},{-.85*sin(\t)});
  \draw[-{Latex[length=1.7mm,width=1.1mm]},chmTangentOrange,dashed,line width=1.1pt]
    (start) -- (frozen);
  \draw[black!55,densely dotted,line width=.8pt] (exact) -- (frozen);
  \fill[black!75] (start) circle[radius=2pt];
  \node[anchor=south west] at ($(start)+(.07,.13)$) {$q_0$};
  \fill[chmTangentBlue] (exact) circle[radius=2pt];
  \node[text=chmTangentBlue,anchor=east] at ($(exact)+(-.12,-.13)$) {$q(\theta_0+\delta)$};
  \filldraw[fill=white,draw=chmTangentOrange,line width=1pt] (frozen) circle[radius=2.3pt];
  \node[text=chmTangentOrange,anchor=west] at ($(frozen)+(.12,-.02)$) {$\widehat q$};
\end{scope}

\node at (2.75,.88) {$\Sigma_{\mathbb R}=\bar H_{\mathbb R}=\operatorname{diag}(\tfrac12,\tfrac18)$};
\node[text=chmTangentBlue] at (2.75,.30) {$\operatorname{rank}\Sigma_{\mathbb R}=\operatorname{rank}\bar H_{\mathbb R}=2$};
\node at (8.70,.77) {$H_{\mathbb R}(\theta_0)=\dfrac18\begin{pmatrix}4&2\\2&1\end{pmatrix}$};
\node[text=chmTangentOrange] at (8.70,.02) {$\operatorname{rank}H_{\mathbb R}(\theta_0)=1$};
\node at (14.65,.88) {$\widehat q=q_0+q'_0\delta$};
\node[text=black!70,align=center] at (14.65,.26) {Same starting point\\and parameter displacement};
\end{tikzpicture}
\endgroup
\caption{\justifying The same CHM describes ensemble variation and local response in a two-qubit circuit [Eq.~\eqref{eq:kernels_two_frequency_model}]. (a) The CHM has rows \(\omega=(-2,-1,1,2)\) and columns \(k=(-1,1)\). Combining its rows into the active real coefficients \(q=(c_1,s_2)^{\mathsf T}\) gives \(C_q\). Averaging the character products gives covariance, while differentiating the characters gives the pointwise kernel. (b) The blue ellipse contains all attainable coefficient vectors. The shading shows a portion of its two-dimensional affine span. For uniform \(\theta\), both covariance and the averaged QNTK have rank two. (c) At \(\theta_0=\pi/4\), the orange tangent is the single direction available to first order. The kernel has rank one and a nonzero off-diagonal entry, although ensemble covariance is diagonal. (d) The full circuit (blue) and its frozen linearisation (orange) start at the same \(q_0\) with the same derivative. The endpoints show a parameter displacement \(\delta=1\), with their separation marked by the dotted segment. All three geometric panels use the same coefficient coordinates and scale.}
\label{fig:ensemble-tangent-geometry}
\end{figure*}

\subsection{Predictions from sampled labels}
\label{subsec:kernel_predictions}

For a dataset with observed labels \(y\in\mathbb R^N\), training aims to reduce the residual \(r=F-y\). We use the squared-error loss
\begin{equation}
\mathcal L=\frac{1}{2N}\|F-y\|_2^2.
\label{eq:kernels_empirical_mse}
\end{equation}
The CHM already imposes a restriction on how closely these labels can be fitted, regardless of the initialisation or optimiser. Writing \(B:=VC^\circ\) and \(\bar F:=V\bar a\), the identity \(F-\bar F=B\psi^\circ\) places every model output in the affine space \(\bar F+\operatorname{col}B\). Any part of the label vector outside this space remains beyond the model's reach. The squared distance of \(y\) from this space, divided by \(2N\), therefore gives a lower bound on the loss.

\begin{corollary}[Affine-span restriction on training error]
\label{cor:kernels_affine_error}
Let \(P_B:=BB^+\) be the orthogonal projector onto \(\operatorname{col}B\), where \(+\) denotes the Moore--Penrose pseudoinverse. For any sampled labels \(y\),
\begin{equation}
\inf_\vartheta\mathcal L(\vartheta)
\geq\frac{1}{2N}\|(I-P_B)(y-\bar F)\|_2^2.
\label{eq:kernels_affine_loss_bound}
\end{equation}
\end{corollary}

Figure~\ref{fig:affine-loss-bound} illustrates the geometry behind this bound. The closest point to the labels in the affine span is \(y_{\mathrm{proj}}=\bar F+P_B(y-\bar F)\). For every attainable output, the error perpendicular to the span has the same length \(d_\perp=\|y-y_{\mathrm{proj}}\|_2\). Changing the parameters can only change the part of the error within the span.

\begin{figure}[tb]
\centering
\begingroup
\definecolor{chmAffineBlue}{HTML}{005E78}
\definecolor{chmAffineOrange}{HTML}{A85318}
\begin{tikzpicture}[
  font=\small,
  line cap=round, line join=round,
  every node/.style={inner sep=2pt}
]
\coordinate (mean) at (3.50,0.76);
\coordinate (model) at (1.25,0.76);
\coordinate (projection) at (4.55,0.76);
\coordinate (target) at (4.55,3.05);

\node[font=\small\bfseries] at (3.90,3.68)
  {Sampled-output space $\mathbb R^N$};

\filldraw[fill=chmAffineBlue!5, draw=black!35, line width=0.55pt]
  (0.30,-0.10) -- (6.70,-0.10) -- (7.60,1.65) -- (1.20,1.65) -- cycle;
\draw[chmAffineBlue, line width=1.2pt]
  (mean) ellipse[x radius=2.25cm, y radius=0.62cm];

\draw[black!60, line width=0.75pt]
  (model) -- node[pos=0.53, above=3pt, sloped, text=black!75]
  {total error} (target);
\draw[chmAffineBlue, dashed, line width=0.85pt]
  (model) -- (projection);
\draw[chmAffineOrange, line width=1.35pt]
  (projection) -- node[pos=0.65, right=5pt, align=left,
  text=chmAffineOrange] {unavoidable\\distance $d_\perp$} (target);
\draw[black!65, line width=0.55pt]
  (4.31,0.76) -- (4.31,1.00) -- (4.55,1.00);

\fill[chmAffineBlue] (model) circle[radius=2pt];
\node[anchor=south, text=chmAffineBlue] at (1.25,1.78) {$F(\vartheta)$};
\draw[chmAffineBlue!75, line width=0.55pt]
  (1.25,1.76) -- (1.25,0.86);
\fill[chmAffineOrange] (target) circle[radius=2.2pt];
\node[above right=2pt] at (target) {labels $y$};
\filldraw[fill=white, draw=black!75, line width=0.7pt]
  (projection) circle[radius=2pt];
\node[below left=3pt] at (projection) {$y_{\mathrm{proj}}$};
\draw[black!65, line width=0.7pt]
  (3.44,0.70) -- (3.56,0.82)
  (3.44,0.82) -- (3.56,0.70);
\node[below left=2pt, text=black!75] at (mean) {$\bar F$};

\node[text=chmAffineBlue, align=center] at (1.90,-0.45)
  {attainable outputs};
\draw[chmAffineBlue!75, line width=0.55pt]
  (2.25,-0.23) -- (2.25,0.25);
\node[align=center] at (5.65,-0.55)
  {affine span\\$\bar F+\operatorname{col}B$};
\draw[black!45, line width=0.55pt]
  (5.65,-0.22) -- (5.65,-0.10);

\node at (3.90,-1.30)
  {$\displaystyle
    2N\,\mathcal L(\vartheta)
    = {\color{chmAffineOrange}d_\perp^2}
    + {\color{chmAffineBlue}\|F(\vartheta)-y_{\mathrm{proj}}\|_2^2}
    \;\geq\; {\color{chmAffineOrange}d_\perp^2}$};
\end{tikzpicture}
\endgroup
\caption{\justifying The affine span gives a lower bound on training loss. In this schematic, attainable outputs form a blue curve within the shaded affine span \(\bar F+\operatorname{col}B\). The right triangle separates the total error into the perpendicular distance \(d_\perp=\|(I-P_B)(y-\bar F)\|_2\) (orange) and the remaining distance within the span (dashed blue). The first is fixed for every parameter choice, giving \(\mathcal L\geq d_\perp^2/(2N)\). The projected labels \(y_{\mathrm{proj}}\) need not be attainable, so the minimum loss can exceed the bound. If \(y\) lies in the span, the bound is zero even when \(y\) lies off the attainable curve.}
\label{fig:affine-loss-bound}
\end{figure}

A positive bound rules out an exact fit. This can happen because a frequency is absent, a coefficient is fixed, or a linear relation between coefficients is incompatible with the labels. A zero bound places the labels in the affine span, but they may still lie outside the set of attainable circuit outputs. The projector can also be obtained from the exact data-space covariance \(V\Sigma V^\dagger\), without constructing \(C\); Appendix~\ref{app:kernel_affine_restrictions} gives the proof and covers singular cases. Numerically discarding directions with small but nonzero variance can shrink the estimated span and artificially raise the bound, so such a truncation alone does not establish a minimum training error.

To describe the learning dynamics, the derivatives must be taken in the coordinates that the optimiser updates. We therefore use the physical angles \(\vartheta\) from here on, with \(X=\partial_\vartheta a\), \(J=VX\), \(H=XX^\dagger\), and \(K=VHV^\dagger\). The argument \(\vartheta\) distinguishes these kernels from those derived in character coordinates above. Under Euclidean gradient flow with scale \(\eta>0\), the residual and coefficient vector obey
\begin{align}
\dot r&=-\frac{\eta}{N}K(\vartheta)r,
\label{eq:kernels_residual_dynamics}\\
\dot a&=-\frac{\eta}{N}H(\vartheta)V^\dagger r.
\label{eq:kernels_coefficient_dynamics}
\end{align}
These are exact instantaneous identities obtained from the chain rule, using \(J^{\mathsf T}=X^\dagger V^\dagger\) for the coefficient update. In the second equation, \(V^\dagger r\) takes inner products of the residual with the sampled Fourier functions, and \(H\) determines the resulting coefficient response. This calculation uses the observed labels directly; it requires neither an inverse Fourier transform nor knowledge of the continuous target's Fourier series.

How quickly the loss falls depends on both the kernel and the direction of the current error. At a reference point \(\vartheta_{\mathrm{fr}}\), which may be the initialisation or a later checkpoint, write \(K_{\mathrm{fr}}=K(\vartheta_{\mathrm{fr}})\) and \(r_{\mathrm{fr}}=F(\vartheta_{\mathrm{fr}})-y\). The kernel's action along this residual is measured by the Rayleigh quotient
\begin{equation}
\mathcal R_{K_{\mathrm{fr}}}(r_{\mathrm{fr}})
:=\frac{r_{\mathrm{fr}}^{\mathsf T}K_{\mathrm{fr}}r_{\mathrm{fr}}}
{r_{\mathrm{fr}}^{\mathsf T}r_{\mathrm{fr}}},
\qquad r_{\mathrm{fr}}\neq0.
\label{eq:kernels_data_rayleigh}
\end{equation}
For the loss in Eq.~\eqref{eq:kernels_empirical_mse}, this gives the instantaneous relative rate of loss reduction, \(-\dot{\mathcal L}/\mathcal L=(2\eta/N)\mathcal R_{K_{\mathrm{fr}}}\). A larger Rayleigh quotient therefore means faster relative loss reduction at that point for the same \(\eta\). This relation holds under gradient flow even when the kernel changes during training.

Predicting the subsequent trajectory requires an approximation to that changing kernel. Linearising the model at the reference point and holding its Jacobian fixed gives the frozen-QNTK model of Liu \emph{et al.}~\cite{liu2022}. The approximation is a first-order expansion in the trainable parameters:
\[
\widehat F(\vartheta)=F(\vartheta_{\mathrm{fr}})
+J(\vartheta_{\mathrm{fr}})(\vartheta-\vartheta_{\mathrm{fr}}).
\]
The full circuit model generally depends nonlinearly on \(\vartheta\), but this fixed linearisation can still remain accurate along a training trajectory. The frozen model's residual evolves as
\begin{equation}
\widehat r(\tau)
=\exp\!\left(-\frac{\eta\tau}{N}K_{\mathrm{fr}}\right)
r_{\mathrm{fr}},
\label{eq:kernels_frozen_residual_solution}
\end{equation}
where \(\tau\) is the time elapsed since the reference point. The hat distinguishes the prediction of this linearised model from the trajectory of the full circuit model. The corresponding sampled outputs are \(\widehat F=y+\widehat r\), so the prediction depends on both the model at the reference point and the observed labels. For full-batch gradient descent on the same linearised model, the residual after \(j\) steps is
\begin{equation}
\widehat r_j
=\left(I-\frac{\eta}{N}K_{\mathrm{fr}}\right)^j r_{\mathrm{fr}}.
\label{eq:kernels_frozen_residual_discrete}
\end{equation}
Our numerical experiments minimise MSE without the factor of \(1/2\), replacing \(\eta/N\) by \(2\eta/N\) in these dynamics. Appendix~\ref{app:kernel_frozen_predictions} gives the condition on the step size for convergence and the corresponding predictions for test inputs and Fourier coefficients.

The eigenvectors of \(K_{\mathrm{fr}}\) separate the residual into components that evolve independently in the frozen model. Under continuous flow, a component with eigenvalue \(\lambda\) decays as \(e^{-\eta\lambda\tau/N}\), while a component in the nullspace persists. A decaying component can combine several Fourier frequencies. Its decay is governed by \(K_{\mathrm{fr}}=V H_{\mathrm{fr}}V^\dagger\) because the loss measures changes in the sampled outputs. The map \(V\) determines how strongly each coefficient change appears on those inputs; it can change lengths and angles, or make some changes invisible. Consequently, it need not preserve the eigenvalues of \(H_{\mathrm{fr}}\).

For integer frequencies on a uniform \([0,2\pi)\) grid with enough points to distinguish all modes in \(\Omega\), the complex Fourier columns satisfy \(V^\dagger V=NI\). This is the case for our benchmark grids. Then \(V/\sqrt N\) preserves lengths, and \(H_{\mathrm{fr}}\) and \(K_{\mathrm{fr}}/N\) have exactly the same nonzero eigenvalues. Thus the complex coefficient kernel does determine the decay rates in this setting. In the unscaled real cosine--sine basis used for the numerical coefficient spectra, the constant column has squared norm \(N\), whereas each sine or cosine column has squared norm \(N/2\). These different normalisations must be included before comparing that spectrum with the data-space decay rates (Appendix~\ref{app:kernel_real_coordinates}).

Knowing \(H_{\mathrm{fr}}\) also lets us predict the model's Fourier coefficients directly by integrating Eq.~\eqref{eq:kernels_coefficient_dynamics} with the frozen Jacobians. If instead we retain only the predicted output changes, recovering the coefficient changes means solving \(\Delta F=V\Delta a\). The Moore--Penrose pseudoinverse \(V^+\) gives the solution with the smallest coefficient norm:
\[
\Delta a_{\mathrm{rec}}=V^+\Delta F=V^+V\Delta a.
\]
When the columns of \(V\) are independent, \(V^+V=I\), so the coefficient change is recovered exactly. Our benchmark grids satisfy this condition. If the columns are dependent, a nonzero coefficient change in \(\ker V\) leaves every sampled output unchanged. Output changes alone cannot determine this invisible component, and \(V^+V\) projects it out. The forecast using \(H_{\mathrm{fr}}\) can still specify it because it also uses the circuit's parameter response. Appendix~\ref{app:kernel_coefficient_recovery} gives an example and discusses sensitivity to numerical errors when inverting \(V\).

This reconstruction uses the model's known frequency set \(\Omega\). Applied to exact labels, it also recovers a target's coefficients if the target is known to lie in that Fourier span and \(V\) has full column rank. Without a restriction on the target's frequencies, different continuous functions can agree at all sampled inputs yet differ between them. The kernel forecast does not require this target reconstruction: it predicts the known circuit's response directly from the sampled residuals.

\subsection{What can we infer before training?}
\label{subsec:kernel_before_training}

Before training, the CHM can reveal why a circuit will be unable to fit a particular dataset exactly. A frequency may be absent, a coefficient may be fixed, or coefficients may be linked by a linear relation. These restrictions hold for every parameter choice. Corollary~\ref{cor:kernels_affine_error} measures the part of the labels that lies outside the resulting affine span and gives a lower bound on the loss. A positive bound therefore rules out an exact fit regardless of the initialisation or optimiser.

These restrictions do not tell us how quickly training will reduce the remaining error. The covariance \(\Sigma\) and averaged kernel \(\bar H\) provide further information: the former describes how coefficients vary across a chosen parameter distribution, while the latter measures the mean squared sensitivity of coefficient combinations to parameter changes. Both can be calculated before selecting an initialisation, allowing circuits to be compared through their variation and sensitivity. Section~\ref{subsec:variance-learning} tests how informative the variation in a target's real Fourier coordinates is about subsequent learning.

A prediction for a particular run also needs the initial parameters and the current error. The coefficient kernel \(H\), evaluation matrix \(V\) and residual then determine the instantaneous response under Euclidean gradient flow. In particular, Eq.~\eqref{eq:kernels_data_rayleigh} gives the rate of loss reduction using only the observed labels and the current model. No knowledge of the continuous target's Fourier series is required. This relation remains exact as the kernel changes during training, provided it is evaluated at the current parameter point.

To predict a whole interval of training from one parameter point, the frozen-kernel approximation holds the Jacobian fixed at that point. Its usefulness depends on how accurately the resulting linearised model follows the full circuit model along the training trajectory. Using the ensemble mean \(\bar K\) in place of the kernel at the chosen initialisation makes an additional assumption about how well that mean represents the individual run. In our forecast experiment, both the full circuit and its linearised model use full-batch gradient descent (Section~\ref{subsec:numerical-frozen-qntk}). The separate variance and covariance learning comparisons use Adam. The same linearised model could be trained with Adam, but its predictions would then have to include Adam's momentum and adaptive scaling of parameter updates.

\section{Direct computation of ensemble moments}
\label{sec:ensemble_computation}

The preceding sections express coefficient variation and parameter sensitivity as matrix products involving \(C\). These expressions explain how the two quantities are related, but constructing the full CHM can be costly: even a circuit in which each Pauli rotation has its own parameter permits \(3^m\) parameter characters. For the numerical study, we can instead calculate the ensemble moments directly. The coefficient mean \(\bar a\) and second moment \(S\) determine the covariance \(\Sigma\), while the averaged kernel \(\bar H\) gives the corresponding measure of parameter sensitivity.

We obtain these quantities by combining Pauli propagation \cite{nemkovFourierExpansionVariational2023} with the coefficient-moment method of Wiedmann \emph{et al.}~\cite{wiedmann2024}. The idea is to average over each trainable angle as we propagate the observable backwards through the circuit. Products of outputs or derivatives require two copies of the observable, which are propagated together. This gives the required averages without reconstructing the full coefficient map as a function of the parameters. Figure~\ref{fig:two-copy-worked} follows this calculation for the one-qubit circuit of Section~\ref{subsec:C_example}, carrying both the second moment and the sum of derivative products through the same sequence of gates.

We consider independent physical angles, each sampled uniformly over a full period and used in only one trainable Pauli or controlled rotation. Fixed Clifford gates may appear between these rotations. The encoding gates are Pauli rotations with integer input frequencies; their input dependence is retained throughout the calculation. Appendix~\ref{appendix:twirl_assumptions} gives the detailed assumptions and local propagation maps.

\begin{figure*}[t]
\centering
\begingroup
\definecolor{chmMoment}{HTML}{005E78}
\definecolor{chmDerivative}{HTML}{A85318}
\begin{tikzpicture}[
  font=\small, x=1cm, y=1cm,
  line cap=round, line join=round,
  pair/.style={draw=chmMoment!45,fill=chmMoment!3,rounded corners=2pt,
    minimum width=1.65cm,minimum height=.90cm,inner sep=4pt,align=center},
  retained/.style={draw=chmMoment!35,fill=chmMoment!4,rounded corners=2pt,
    minimum width=2.25cm,minimum height=1.00cm,inner sep=4pt,align=center},
  removed/.style={retained,draw=black!18,fill=black!2,text=black!55},
  flow/.style={-{Latex[length=1.7mm,width=1.1mm]},draw=black!60,line width=.7pt},
  stage/.style={align=center,font=\small\bfseries},
  every node/.style={inner sep=2pt}
]
\path[use as bounding box] (0,-2.25) rectangle (17.40,7.50);
\node at (8.70,7.28)
  {$U=R_X(\vartheta_2)R_Y(\vartheta_1)R_Y(x),\qquad
    \rho_0=|0\rangle\langle0|,\quad O=Z,\qquad
    f(x;\vartheta)=\cos(x+\vartheta_1)\cos\vartheta_2$};
\node[stage] at (2.55,6.54) {(a) Average the shared angle};
\node[stage] at (11.63,6.54) {(b) Propagate both weights backwards};
\draw[black!18,line width=.45pt] (5.40,.02)--(5.40,6.17);

\node at (2.55,5.86) {$R_X^\dagger ZR_X=cZ+sY$};
\node[text=black!70] at (2.55,5.35)
  {$c=\cos\vartheta_2,\quad s=\sin\vartheta_2$};
\node[retained] at (1.28,4.45)
  {$c^2 Z\otimes Z$\\[3pt]$\longmapsto\ \tfrac12 Z\otimes Z$};
\node[removed] at (3.82,4.45)
  {$cs Z\otimes Y$\\[3pt]$\longmapsto\ 0$};
\node[removed] at (1.28,3.15)
  {$cs Y\otimes Z$\\[3pt]$\longmapsto\ 0$};
\node[retained] at (3.82,3.15)
  {$s^2 Y\otimes Y$\\[3pt]$\longmapsto\ \tfrac12 Y\otimes Y$};
\node[text=black!70,align=center] at (2.55,2.19)
  {For the derivative product,\\differentiate each copy first:};
\node at (2.55,1.56)
  {$\partial_{\vartheta_2}(cZ+sY)=-sZ+cY$};
\node[align=center] at (2.55,.82)
  {$\mathcal D_{R_X}(Z\otimes Z)$\\[3pt]
   $=\tfrac12(Z\otimes Z+Y\otimes Y)$};
\node[text=black!70] at (2.55,.08)
  {Both retained pairs have $b_S=b_H=\tfrac12$.};

\node[stage] at (6.55,5.80) {Start};
\node[stage] at (9.12,5.80) {Average\\$R_X(\vartheta_2)$};
\node[stage] at (12.18,5.80) {Average\\$R_Y(\vartheta_1)$};
\node[stage] at (15.62,5.80) {Encode and take\\state expectations};
\node[pair] (start) at (6.55,3.38)
  {$Z\otimes Z$\\[3pt]
   $(\textcolor{chmMoment}{1},\textcolor{chmDerivative}{0})$};
\node[pair] (rxzz) at (9.12,4.25)
  {$Z\otimes Z$\\[3pt]
   $(\textcolor{chmMoment}{\tfrac12},\textcolor{chmDerivative}{\tfrac12})$};
\node[pair] (rxyy) at (9.12,1.55)
  {$Y\otimes Y$\\[3pt]
   $(\textcolor{chmMoment}{\tfrac12},\textcolor{chmDerivative}{\tfrac12})$};
\node[pair] (ryzz) at (12.18,4.82)
  {$Z\otimes Z$\\[3pt]
   $(\textcolor{chmMoment}{\tfrac14},\textcolor{chmDerivative}{\tfrac12})$};
\node[pair] (ryxx) at (12.18,3.18)
  {$X\otimes X$\\[3pt]
   $(\textcolor{chmMoment}{\tfrac14},\textcolor{chmDerivative}{\tfrac12})$};
\node[pair] (ryyy) at (12.18,1.55)
  {$Y\otimes Y$\\[3pt]
   $(\textcolor{chmMoment}{\tfrac12},\textcolor{chmDerivative}{\tfrac12})$};
\draw[flow] (start.east)--(7.84,3.38)|-(rxzz.west);
\draw[flow] (start.east)--(7.84,3.38)|-(rxyy.west);
\draw[flow] (rxzz.east)--(10.65,4.25)|-(ryzz.west);
\draw[flow] (rxzz.east)--(10.65,4.25)|-(ryxx.west);
\draw[flow] (rxyy.east)--(ryyy.west);
\node[align=center] (outzz) at (15.62,4.82)
  {$\cos x\cos x'$};
\node[align=center] (outxx) at (15.62,3.18)
  {$\sin x\sin x'$};
\node[align=center,text=black!55] (outyy) at (15.62,1.55)
  {$0$\\[2pt]zero state expectation};
\draw[flow] (ryzz.east)--(outzz.west);
\draw[flow] (ryxx.east)--(outxx.west);
\draw[flow,draw=black!30] (ryyy.east)--(outyy.west);
\node at (11.63,.54)
  {Weights: $(\textcolor{chmMoment}{\Gamma},\textcolor{chmDerivative}{\Gamma_H})$
   $=$ (second moment, derivative sum)};
\node at (11.63,-.13)
  {$\textcolor{chmMoment}{\Gamma'=b_S\Gamma},\qquad
    \textcolor{chmDerivative}{\Gamma_H'
      =\underbrace{b_S\Gamma_H}_{\text{propagated}}
       +\underbrace{b_H\Gamma}_{\text{new}}}$};

\draw[black!18,line width=.45pt] (.10,-.74)--(17.30,-.74);
\node[stage] at (8.70,-1.11)
  {(c) Combine the surviving contributions and read their Fourier coefficients};
\node[align=center,text=chmMoment] at (4.28,-1.85)
  {$\tfrac14(\cos x\cos x'+\sin x\sin x')$
   $\quad\longrightarrow\quad S=\Sigma=\tfrac18 I_2$};
\node[align=center,text=chmDerivative] at (13.02,-1.85)
  {$\tfrac12(\cos x\cos x'+\sin x\sin x')$
   $\quad\longrightarrow\quad\bar H=\tfrac14 I_2$};
\end{tikzpicture}
\endgroup
\caption{\justifying Computing covariance and the averaged QNTK by propagating two observable copies. The example reuses the circuit of Fig.~\ref{fig:k_branching}, with independent physical angles \(\vartheta_1,\vartheta_2\) uniform on \([0,2\pi)\). Tensor products label the two copies of the observable; the circuit itself has one qubit. (a) Both copies share \(\vartheta_2\). Joint averaging retains the \(\cos^2\vartheta_2\) and \(\sin^2\vartheta_2\) terms and removes the mixed terms. Differentiating both copies before averaging gives the same two branches for this gate and pair. (b) Each retained pair carries an ordinary weight \(\Gamma\) (blue) and a derivative weight \(\Gamma_H\) (orange), summed over processed angles. Averaging \(R_Y(\vartheta_1)\) propagates the contribution from \(\vartheta_2\) and adds one from \(\vartheta_1\). The encoder is evaluated at \(x\) in one copy and \(x'\) in the other; this input dependence is retained. The final state expectations are taken in \(\rho_0\otimes\rho_0\). (c) The surviving functions sum to \(\tfrac14\cos(x-x')\) for the second moment and \(\tfrac12\cos(x-x')\) for the derivative sum. Reading the coefficients of \(e^{i\omega x}e^{-i\mu x'}\) gives the displayed matrices on \(\omega,\mu\in\{-1,1\}\); the zero-frequency row and column vanish. Here \(\bar a=0\), so \(\Sigma=S\). The general propagation rules are given in Appendix~\ref{appendix:twirl}.}
\label{fig:two-copy-worked}
\end{figure*}

\subsection{Averaging two copies of the observable.}
The second moment \(S=\mathbb E[aa^\dagger]\) describes how Fourier coefficients vary together within the same circuit. To compute it from circuit outputs, we therefore consider \(f(x;\vartheta)\) and \(f(x';\vartheta)\) at different inputs but the \emph{same} parameter vector. Since the output is real, the entries of \(S\) are the Fourier coefficients of the averaged product of these two outputs, with the second input character written as \(e^{-i\mu\cdot x'}\).

In the Heisenberg picture, conjugating an observable by the gates in reverse order expresses the output as an expectation in the fixed input state. For the product of two outputs, we apply this procedure to \(O\otimes O\). Each intermediate term contains a pair of Pauli strings and an input-frequency label for each copy. An encoding gate changes these frequency labels as it introduces input characters. A trainable gate instead introduces angle-dependent factors, whose product we average over the shared angle. For example, two factors of \(\cos\vartheta\) give \(\langle\cos^2\vartheta\rangle=1/2\). Averaging them separately would give \(\langle\cos\vartheta\rangle^2=0\), losing the dependence that the second moment is intended to measure.

The assumption that each angle is independent and used only once allows us to perform its average as soon as its gate is reached. After that step, the angle cannot appear again, so we need only retain the resulting Pauli strings, input frequencies and scalar weights. Contributions with identical Pauli and frequency labels can be added together, regardless of how they arose. This avoids storing every propagation path or the parameter-harmonic labels \(k\) that index the columns of \(C\).

If an angle is reused in several gates, its dependence must instead be retained until all occurrences have been processed. A related issue arises when a controlled rotation is written as a product of Pauli rotations: its factors share one physical angle. We therefore treat the controlled gate as a whole and average over its full \(4\pi\) period, as described in Appendix~\ref{app:twirl_controlled}.

\subsection{Accumulating parameter sensitivity.}
To obtain \(\bar H\), we need the sum of derivative products over all trainable parameters. The same propagation can accumulate this sum by differentiating both observable copies at each trainable gate. For a gate \(g(\vartheta_r)\) and a pair of Pauli strings \(Q_1,Q_2\), write \(Q_i^g=g^\dagger Q_i g\) and let \(\langle\cdot\rangle_r\) denote the normalised full-period average over its angle. The ordinary and differentiated propagation maps are
\begin{equation}
\begin{aligned}
\mathcal T_g(Q_1\otimes Q_2)
&=\langle Q_1^g\otimes Q_2^g\rangle_r,\\
\mathcal D_g(Q_1\otimes Q_2)
&=\langle\partial_{\vartheta_r}Q_1^g\otimes
       \partial_{\vartheta_r}Q_2^g\rangle_r.
\end{aligned}
\label{eq:computation_pair_maps}
\end{equation}
Here \(\mathcal T_g\) is the average used for the second moment, whereas \(\mathcal D_g\) supplies the contribution to parameter sensitivity from the current gate. In the latter, both copies are differentiated with respect to the same physical angle \emph{before} their product is averaged. Differentiating an already averaged observable would lose this information. Using physical angles also retains the period factors in Eq.~\eqref{eq:kernels_physical_weight}.

Each intermediate term now carries two scalar weights. The first, \(\Gamma\), contributes to the ordinary second moment; the second, \(\Gamma_H\), collects derivative products from gates already processed. Initially, \(\Gamma\) is the coefficient of the Pauli pair in \(O\otimes O\), and \(\Gamma_H=0\) because no gate has yet been differentiated. If a resulting Pauli pair has coefficients \(b_S\) and \(b_H\) in the two maps, its weights become
\begin{equation}
\Gamma'=b_S\Gamma,\qquad
\Gamma_H'=b_S\Gamma_H+b_H\Gamma.
\label{eq:computation_accumulators}
\end{equation}
The term \(b_S\Gamma_H\) propagates the derivative contributions already collected, while \(b_H\Gamma\) adds the contribution from differentiating the current gate. Thus all parameters contribute during one traversal of the circuit; we do not need to repeat the full calculation for each parameter. Fixed gates and encoding gates propagate both weights but add no derivative contribution, so \(b_H=0\). When terms with the same labels are combined, their two weights are added separately. A term is removed only if both weights vanish, since cancellation in the second moment need not imply cancellation in the derivative sum.

Averaging \(R_X(\vartheta_2)\) in Fig.~\ref{fig:two-copy-worked}(b) leaves two Pauli pairs, each carrying weights \((\Gamma,\Gamma_H)=(1/2,1/2)\). The next rotation, \(R_Y(\vartheta_1)\), splits the \(Z\otimes Z\) pair again. Each new branch has second-moment weight \(1/4\) and derivative weight \(1/2\), with a contribution of \(1/4\) from each trainable angle. The \(Y\otimes Y\) pair commutes with this gate, so its weights pass through unchanged; it vanishes when the input state expectations are taken. Combining the remaining input-frequency contributions gives the covariance and averaged kernel shown in panel (c), without retaining any parameter-harmonic labels.

Once propagation reaches the input, each Pauli pair contributes through its expectation \(\Tr(\rho_0Q_1)\Tr(\rho_0Q_2)\) in the two copies of the input state. Combining these expectations with the ordinary weights gives \(S_{\omega\mu}\) as the coefficient of \(e^{i\omega\cdot x}e^{-i\mu\cdot x'}\). Using the derivative weights gives \(\bar H_{\omega\mu}\) in the same way. A separate propagation of a single observable gives the mean coefficient vector \(\bar a\), from which the covariance follows as \(\Sigma=S-\bar a\bar a^\dagger\). The averaged kernel already has the required derivative Gram form, so it requires no mean subtraction. Appendix~\ref{app:twirl_algorithm} gives the algorithm and proves these relations by induction over the gates.

\subsection{Computational scope.}
The saving comes from averaging out each parameter before proceeding and combining terms that have the same remaining labels. The cost therefore depends on the gate count and the number of Pauli and input-frequency pairs retained during propagation. This avoids enumerating every parameter character, although the retained set can still grow exponentially with qubit count. Appendix~\ref{app:twirl_cost} gives the corresponding cost bounds. The averaging identities are exact; the implementation evaluates them with floating-point weights and discards small intermediate terms. The size of the discarded terms alone does not give a bound on the final error, so the numerical checks reported below remain necessary.

This procedure gives ensemble averages. The initial-kernel tests and frozen-kernel gradient-descent comparison also require the kernel at a particular parameter choice, which we calculate separately. We first recover the real cosine--sine coefficients from circuit outputs on a periodic input grid fine enough to distinguish all allowed frequencies. Automatic differentiation of these coefficients with respect to the physical angles gives their Jacobian, whose Gram matrix is the coefficient kernel. The real-coordinate version of \(K=VHV^\dagger\) then gives the kernel on the training inputs [Appendix~\ref{app:kernel_real_coordinates}]. This calculation also avoids constructing \(C\), but retains the dependence on the chosen parameter point. The ensemble-moment checks consequently test the averaging procedure, while the frozen-kernel comparisons test how well a kernel at one initialisation predicts subsequent training.

\section{Numerical results}
\label{sec:numerical-results}

We now use the CHM description to examine what circuit structure reveals about learning. The preceding sections have distinguished the frequencies allowed by an encoding block from the coefficient vectors a particular circuit can produce. They have also shown how the rows of \(C\) determine coefficient variances and covariances, and how differentiating the same construction gives the tangent kernel. These quantities describe different aspects of the model: variation across parameter choices need not imply an equally strong response to a small parameter update. The numerical study brings these descriptions together for a collection of circuits whose depth, qubit count and encoding axis we can vary. To make the examples reproducible and connect them to existing benchmarks, we use the circuit families studied by Strobl \emph{et al.}~\cite{strobl2025}. These comprise Circuits~15--19 from the catalogue of Sim \emph{et al.}~\cite{sim2019}, together with a hardware-efficient ansatz (HEA) and YZY circuits with and without entanglers.

We begin with the coefficient statistics, calculated by the method of Section~\ref{sec:ensemble_computation}. A YZY circuit provides a concrete example of how depth changes frequency support, the distribution of variance and the correlations between coefficients. We then examine which of these changes persist across other circuits. With this structure established, we train the circuits on Fourier targets to assess what the statistics tell us about learning. Single-mode targets probe the information in individual variances; targets combining modes allow us to ask whether covariance favours particular combinations. Comparisons that keep the supported targets fixed help distinguish changes in learning from the acquisition of new frequencies.

The final part returns to the distinction between ensemble variation and local response. We use numerical QNTKs to examine which coefficient changes are available at an initialisation and how strongly the circuit responds in those directions. Combining the initial kernel with the training residual then lets us test the frozen-kernel loss predictions derived in Section~\ref{sec:kernels}.

\subsection{Experimental setting and validation}
\label{subsec:numerical-setting}
\label{subsec:exact-validation}

All models share the same order of operations. Starting from $\ket{0}^{\otimes n}$, they encode one scalar input in a single encoding layer and then apply $d$ trainable blocks. Again, following Strobl \emph{et al.}~\cite{strobl2025}, we measure the mean $Z$ magnetisation,
\begin{equation}
    O_Z=\frac{1}{n}\sum_{j=1}^{n}Z_j.
    \label{eq:numerical-observable}
\end{equation}
Our main illustrative circuit uses the encoding block
\begin{equation}
    S_X(x)=R_X(x)^{\otimes n},
    \label{eq:rx-feature-map}
\end{equation}
followed by YZY-entangling blocks, each containing local $R_YR_ZR_Y$ rotations and a fixed triangular all-to-all CNOT sequence. The wider comparison includes eight circuit families, each paired with $R_X$ and $R_Y$ encoding. Their gate orderings are shown in Fig.~\ref{fig:appendix-circuit-families}. Evaluating the resulting 16 circuit--encoder combinations at $n=2,\ldots,6$ and $d=1,\ldots,5$ gives 400 circuit settings.

Each physical angle is independent and uniform over its full period: $2\pi$ for a local Pauli rotation and $4\pi$ for a controlled rotation. The moment calculation integrates every angle over $4\pi$. For local rotations, this simply covers the $2\pi$ conjugation period twice and leaves the normalised average unchanged. Training initialisations are sampled over the respective minimal periods. The numerical QNTKs use derivatives with respect to these same physical angles, including the period factors explained in Section~\ref{sec:kernels}. Appendix~\ref{subsec:appendix_protocol} gives the complete settings.

The two-copy method of Section~\ref{sec:ensemble_computation} gives \(\Sigma\) and \(\bar H\) directly. The diagonal of \(\Sigma\) then gives the coefficient variances. Dividing each supported covariance entry by the corresponding standard deviations gives the correlations in Eq.~\eqref{eq:coeffstats_corr_def}. To decide which coefficients have numerically resolved variation, we use the threshold
\begin{equation}
    v_\omega>\tau_v,\qquad
    \tau_v=\max\{10^{-14},10^{-12}\max_\nu v_\nu\}.
    \label{eq:support-threshold}
\end{equation}
This estimates the positive-variance set defined earlier. The averaging identities are exact, but the calculation uses floating-point arithmetic and discards intermediate terms below $10^{-14}$. A coefficient excluded by the support threshold is therefore not proved to vanish identically. The computed mean vectors are zero in all 400 runs. More generally, the distinction made in Section~\ref{sec:coeff_stats} still applies: zero variance means that a coefficient is independent of the parameters, and a nonzero mean could then contribute a fixed term to the output.

We check this calculation on small circuits for which an explicit CHM is also available. In the three tested YZY-entangling cases, the CHM matrix products and direct two-copy calculation agree to relative Frobenius errors no larger than $3.1\times10^{-16}$ for $\Sigma$ and $4.4\times10^{-16}$ for $\bar H$. Independently sampled Monte Carlo averages of the coefficients and their Jacobian Gram matrices provide a second check, including a controlled-rotation circuit averaged over its full $4\pi$ period. Appendix~\ref{subsec:appendix_validation} reports the cases, convergence plots and numerical limitations.

\subsection{The YZY-entangling example: support, variance and correlation}
\label{subsec:support-growth}
\label{subsec:variance-profiles}
\label{subsec:correlation-geometry}

The YZY-entangling example lets us examine separately what depth does to the set of available frequencies and to the coefficients within that set. As established in the preceding sections, the encoding block fixes the frequencies that are possible, while the rest of the circuit determines which of them can vary. Here, one parallel Pauli encoding layer on $n$ qubits allows
\begin{equation}
    \Omega_{\mathrm{enc}}^{(n)}=\{-n,\ldots,n\}.
    \label{eq:encoder-accessible-spectrum}
\end{equation}
We compare this allowed set with the numerically supported set
\begin{equation}
    \mathcal S=\{\omega\in\Omega_{\mathrm{enc}}^{(n)}:v_\omega>\tau_v\}
    \label{eq:trainable-support}
\end{equation}
defined using Eq.~\eqref{eq:support-threshold}. This is our numerical estimate of the positive-variance set \(\Omega_+\) in Section~\ref{sec:coeff_stats}, subject to the threshold qualification above. Since the output is real, $a_{-\omega}=\overline{a_\omega}$. We therefore use the nonnegative representatives $\mathcal S_{\geq0}$ for this example's support counts and Hermitian-correlation summaries.

Figure~\ref{fig:yzy-exact-overview}(a) shows that the full allowed spectrum is not immediately available. With one trainable block, YZY-entangling supports only $|\omega|\leq2$ throughout the tested size range. From $d=2$ onward, it supports every allowed frequency for $n=2,\ldots,6$ [Fig.~\ref{fig:appendix-yzy-support}]. The role of the trainable circuit is particularly clear when we remove the entanglers. The nonentangling YZY control cannot generate $|\omega|>1$ at any depth in this setting, because the local encoding and mean-magnetisation readout give a sum of single-qubit functions [Eqs.~\eqref{eq:nonentangling-decomposition}--\eqref{eq:nonentangling-support}]. Thus encoding on more qubits does not by itself ensure access to higher harmonics: the trainable circuit must also make those contributions available in the measured output.

The depth needed to reach the full encoder-allowed spectrum varies between circuit families. The variance profiles in Fig.~\ref{fig:variance-profile-atlas} show this directly at $n=6$: Circuits~15, 17 and~19 already support all allowed frequencies at $d=1$, HEA and YZY-entangling reach full support at $d=2$, and Circuits~16 and~18 require $d=4$. These depths are the same for both encoding axes in this comparison. Nonentangling YZY remains restricted to $\omega=\pm1$. With the encoding block and observable held fixed, the trainable circuit therefore determines whether the full allowed spectrum is reached and how many blocks are needed.

\begin{figure*}[t]
    \centering
    \includegraphics[width=\textwidth]{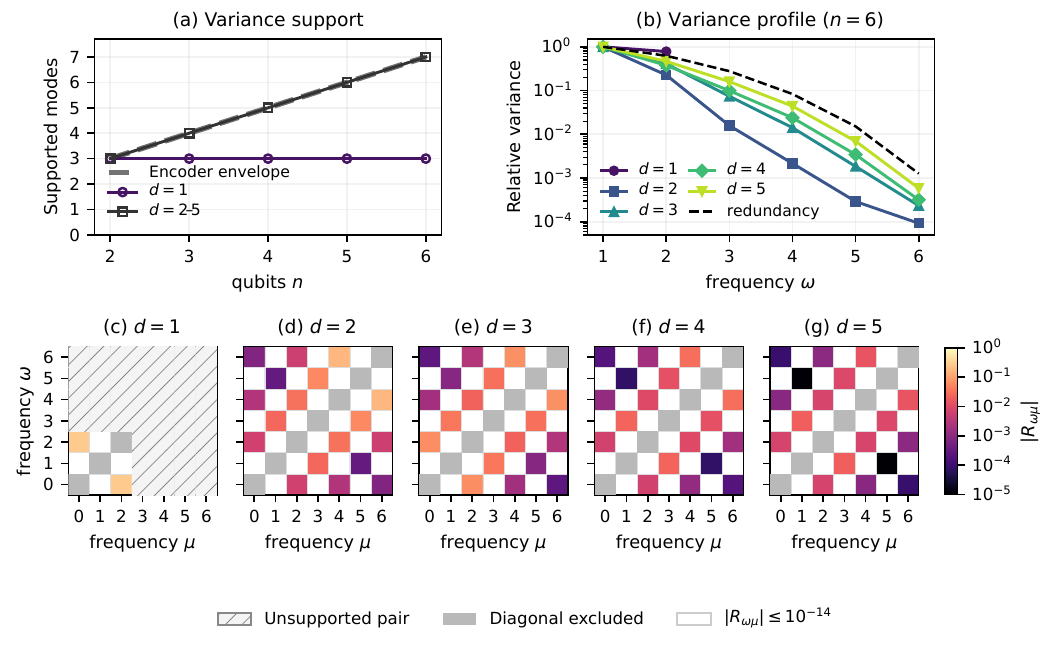}
    \caption{\justifying{Correlations persist after support saturates in $R_X$-encoded YZY-entangling. (a) Supported nonnegative-frequency counts, including zero: $d=1$ supports three modes, whereas $d=2,\ldots,5$ saturate the encoder envelope $n+1$. (b) Positive-frequency variances at $n=6$, each profile divided by its maximum, alongside similarly normalised redundancy. Zero variances are omitted from the log axis but retained in the profile distance. (c--g) Off-diagonal correlation magnitudes at $n=6$, $d=1,\ldots,5$, on a common grid and logarithmic scale, shared with Figs.~\ref{fig:aggregate-architecture-trends}(e) and~\ref{fig:appendix-yzy-correlation}. Hatching marks pairs involving an unsupported mode, for which correlation is undefined; grey marks excluded supported diagonal entries. White supported off-diagonal cells have $|R_{\omega\mu}|\leq10^{-14}$.}}
    \label{fig:yzy-exact-overview}
\end{figure*}

Once a coefficient can vary, its variance measures the extent of that variation across the parameter ensemble. We first ask how this variation is distributed among frequencies, before comparing its absolute scale across circuits. Encoder redundancy provides a useful reference for this relative profile. Under the independent pre- and post-encoding exact-$2$-design assumptions of Mhiri \emph{et al.}~\cite{mhiri2025a}, nonzero-frequency variances have the form
\begin{equation}
    v_\omega^{(2)}=\kappa(n,O)\,|\mathcal{R}(\omega)|,\qquad \omega\ne0,
    \label{eq:two-design-variance}
\end{equation}
where $\kappa$ is independent of frequency and the encoder redundancy is $|\mathcal{R}(\omega)|=\binom{2n}{n+\omega}$ for this encoding [Appendix~\ref{app:qfms_examples}]. Applying this theorem directly would require a random pre-encoding block, whereas our input state is fixed. For that setting, the calculation in Appendix~\ref{app:fixed_input_design_reference} gives the same frequency profile, with a different prefactor, if the post-encoding block forms an exact \(2\)-design. We use this shared profile as a reference for the \emph{shape} of the measured variances, without assuming either prediction for their absolute scale. The comparison excludes the zero-frequency term, whose variance has a separate correction.

We quantify agreement with this reference by normalising both profiles over the same positive frequencies $\omega=1,\ldots,n$ and calculating their Hellinger distance:
\begin{equation}
\begin{aligned}
    p_\omega&=\frac{v_\omega}{\sum_{\nu=1}^n v_\nu},
    &b_\omega&=\frac{|\mathcal{R}(\omega)|}{\sum_{\nu=1}^n|\mathcal{R}(\nu)|},\\
    D_{\mathrm{Hel}}&=
    \left[\frac12\sum_{\omega=1}^n
      \bigl(\sqrt{p_\omega}-\sqrt{b_\omega}\bigr)^2\right]^{1/2}.
\end{aligned}
\label{eq:numerical-hellinger-distance}
\end{equation}
This distance lies between zero and one, with zero indicating proportional profiles. Retaining frequencies with zero variance means that the comparison captures both missing spectral weight and its redistribution among coefficients that are already present.

At $n=6$, $D_{\mathrm{Hel}}$ decreases from $0.317$ at $d=1$ to $0.0763$ at $d=5$. Figure~\ref{fig:yzy-exact-overview}(b) shows the changing profiles, each divided by its maximum for visual comparison. The intermediate profiles still differ from redundancy, and a decrease under one measure of distance need not be monotonic under another; Appendix~\ref{subsec:appendix_yzy_dashboards} gives the alternative least-squares comparison. The closer profile shape is consistent with the redundancy reference.

Support saturation also makes it possible to follow correlations on an unchanged set of coefficients. All seven nonnegative modes are supported from $d=2$ onward in Fig.~\ref{fig:yzy-exact-overview}(c--g), yet their off-diagonal correlations continue to change. The largest magnitude decreases from $0.174$ at $d=2$ to $0.0207$ at $d=5$. Thus the matrix becomes more nearly diagonal, although individual entries need not decrease at every depth.

In the CHM interpretation, nonzero correlations indicate nonorthogonal centred rows. The plots show the magnitudes of these overlaps; their complex phases carry further information. Opposite-parity frequencies have no resolved coupling in the displayed YZY matrices, while same-parity entries can be nonzero. Replacing the encoder by $R_Y$ gives a diagonal supported matrix in the stored calculation which can be seen in Figure \ref{fig:correlation-depth-atlas-ry}. As discussed in Section~\ref{sec:coeff_stats}, zero Hermitian correlation does not establish statistical independence. The result here is that support saturation at depth two leaves room for substantial further change in how the same coefficients vary and correlate.

\subsection{Depth and size trends across the circuit bank}
\label{subsec:aggregate-trends}

We next examine how representative the YZY example is of the wider circuit collection. Figure~\ref{fig:aggregate-architecture-trends}(a) shows that support at shallow depth does not generally keep pace with the growing encoder spectrum, which is expected if the circuit is not mixing these encoders sufficiently. At $d=1$, the median fraction of supported signed frequencies falls from $1$ at $n=2$ to $0.462$ at $n=6$. This fraction measures how much of the available spectrum the circuit can vary; its decline need not imply a smaller number of supported frequencies.

Figure~\ref{fig:aggregate-architecture-trends}(b) turns to a quantity hidden by the normalised profiles above: the absolute variance scale. At $d=5$, the median across combinations of the largest nonzero-frequency variance falls from $3.59\times10^{-2}$ at $n=2$ to $3.31\times10^{-3}$ at $n=6$. This suppression with qubit count is consistent with the theoretical coefficient-variance concentration derived by Mhiri \emph{et al.}~\cite{mhiri2025a}. Under their global unitary-$2$-design assumptions, bounded-norm observables give variances that are exponentially suppressed with $n$ [Appendix~\ref{app:design_scope}]. Our finite-depth circuits are not assumed to form those ensembles, so the observed decrease does not by itself establish the same asymptotic scaling but does suggest that the same underlying mechanics still apply to finite depth circuits.

The median of the largest nonzero-frequency variances also falls with depth [Fig.~\ref{fig:aggregate-architecture-trends}(c)]. At $n=6$, it falls from $1.53\times10^{-2}$ at $d=1$ to $3.31\times10^{-3}$ at $d=5$. The common logarithmic axes in panels (b,c) make the magnitudes of these size and depth changes directly comparable. The individual tracks nevertheless show different responses across circuits, including nonmonotonic changes with depth.

\begin{figure*}[t]
    \centering
    \includegraphics[width=\textwidth]
    {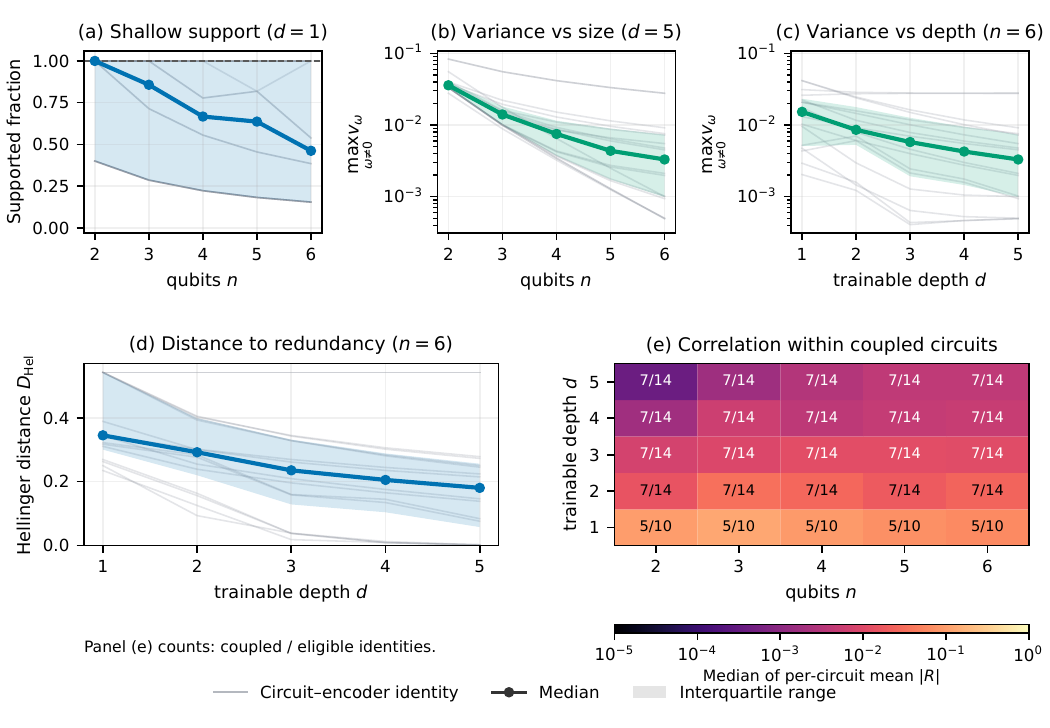}
    \caption{\justifying{Across the circuit bank, median variance contracts with size and depth while relative profiles approach redundancy. Grey tracks are circuit--encoder combinations; coloured lines and bands give their medians and interquartile ranges, not confidence intervals. (a) Variance-support fraction $|\mathcal S|/(2n+1)$ at $d=1$, counting the full signed frequency set. (b,c) Largest nonzero-frequency coefficient variance at fixed $d=5$ and fixed $n=6$, on identical logarithmic axes. (d) Hellinger distance to redundancy at $n=6$, using all positive encoder-allowed frequencies, including zero variances. (e) Median of the per-circuit mean absolute off-diagonal correlation $\overline{|R|}_{\mathrm{off}}$ over the coupled subset, on the same logarithmic magnitude scale as Fig.~\ref{fig:yzy-exact-overview}(c--g). Counts show coupled/eligible combinations; the subset criterion is $\overline{|R|}_{\mathrm{off}}>10^{-10}$. The plots distinguish support, absolute variance, relative profile shape and correlation structure. Neither profile agreement nor small correlations certify a unitary design or local tangent controllability.}}
    \label{fig:aggregate-architecture-trends}
\end{figure*}

The change in relative profile is also shared by many of the circuits. In Fig.~\ref{fig:aggregate-architecture-trends}(d), the median $D_{\mathrm{Hel}}$ at $n=6$ falls from $0.345$ to $0.181$ between depths one and five. Across all five sizes, 70 of the 80 circuit--encoder--size combinations have a smaller distance at depth five. The tendency towards a redundancy-shaped profile therefore extends beyond YZY-entangling, although the amount of agreement remains circuit dependent. Together, panels (c,d) show that the relative allocation of variance can approach the redundancy profile even while the absolute variances become smaller.

To compare correlations across circuits, we average their magnitudes over all distinct pairs of supported nonnegative frequencies:
\begin{equation}
    \overline{|R|}_{\mathrm{off}}
    =\frac{2}{s(s-1)}
    \sum_{\substack{\omega<\mu\\\omega,\mu\in\mathcal S_{\geq0}}}
    |R_{\omega\mu}|,\qquad s=|\mathcal S_{\geq0}|\geq2.
    \label{eq:mean-offdiagonal-correlation}
\end{equation}
The requirement $s\geq2$ ensures that at least one pair is available. The heatmap in Fig.~\ref{fig:aggregate-architecture-trends}(e) summarises the circuits with resolved coupling: we retain combinations whose mean exceeds $10^{-10}$ and plot the median of those means. At each size, five of ten eligible combinations meet this criterion at $d=1$, increasing to seven of fourteen from $d=2$ onward, as the cell annotations show. The same seven circuit--encoder combinations contribute throughout $n=2,\ldots,6$ and $d=2,\ldots,5$.

Within this coupled subset, the median mean correlation at $n=6$ falls from $0.0238$ at $d=2$ to $0.00410$ at $d=5$. The depth dependence is therefore still visible after dividing out the marginal standard deviations. Increasing size gives a different result: at $d=5$, the median rises from $0.000393$ at $n=2$ to $0.00410$ at $n=6$, even though the absolute variance falls in panel (b). Smaller variances therefore need not accompany weaker normalised correlations.

Although the contributing combinations are unchanged from depth two onward, their supported frequency sets can still grow. A falling mean correlation could consequently reflect both changes in existing couplings and the addition of weakly correlated pairs. To distinguish these effects before relating covariance to learning, we next follow individual circuits on coordinate sets held fixed across depth.

\subsection{Variance and covariance across depth}
\label{subsec:covariance-depth-structure}

The aggregate results leave open whether covariance weakens throughout a matrix or only in a few large entries. Figure~\ref{fig:fixed-support-structure-depth} therefore compares both the largest and the mean variance with the largest and the mean absolute off-diagonal covariance for each combination. The variance panels share one logarithmic colour scale and the covariance panels share another, so that changes with depth remain visible while the two colour bars allow comparison of their absolute magnitudes.

For these summaries, we express the covariance in the real cosine and sine coordinates $(c_1,s_1,\ldots,c_6,s_6)$, excluding the constant coefficient. These are the coordinates of the learning targets considered below. A complex frequency can be supported even when only one of its two real components varies, as Circuit~15/$R_X$ illustrates; separating cosine and sine therefore shows which target phases are available. The appendix atlases in Figs.~\ref{fig:correlation-depth-atlas-rx} and~\ref{fig:correlation-depth-atlas-ry} display the corresponding real correlation matrices. They group all cosine coordinates before all sine coordinates to make the couplings within and between these groups easier to see. This grouping only changes the display order. Figure~\ref{fig:yzy-exact-overview} uses complex-frequency coefficients. Both representations are obtained from the same full coefficient covariance, but converting between complex and real coordinates can change the normalised correlation values.

We also hold the coordinate set fixed within each combination: a coordinate is retained only if it passes the directional support rule at every depth $d=2,\ldots,5$. The variance summaries use the diagonal of this block, and the covariance summaries use all its off-diagonal entries, including cosine--sine pairs at the same frequency. Taking absolute values before averaging prevents positive and negative covariances from cancelling. Keeping zero entries in the mean makes it sensitive both to the number of coupled pairs and to their strength. With the same coordinates at every depth, a falling mean cannot be explained simply by adding newly supported pairs. The full correlation matrices, including changes in support over $d=1,\ldots,5$, are shown separately in Figs.~\ref{fig:correlation-depth-atlas-rx} and~\ref{fig:correlation-depth-atlas-ry}.

\begin{figure*}[t]
  \centering
  \includegraphics[width=\textwidth]{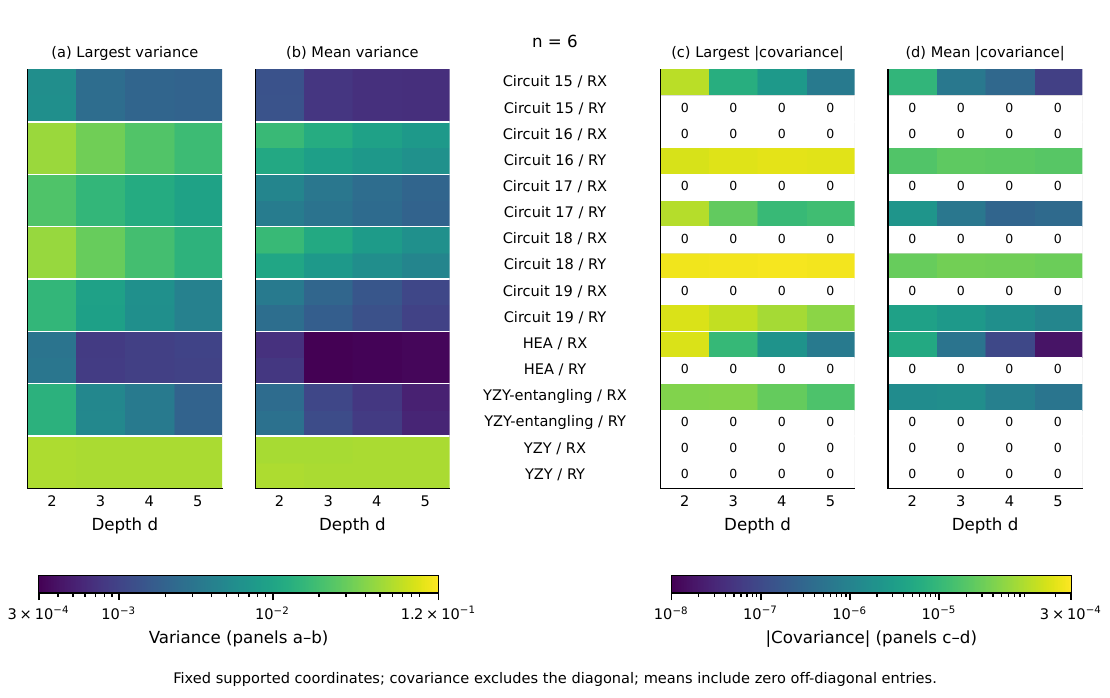}
  \caption{\justifying{Largest and mean fluctuations across depth ($n=6$), with separate logarithmic colour scales shared within the variance and covariance pairs. Circuit--encoder labels are centred between the variance and covariance panels. (a,b) Largest and mean diagonal entry. (c,d) Largest and mean absolute off-diagonal entry, including zero entries in the mean. Every statistic uses the same supported real cosine--sine coordinates within a combination, held fixed over depths two through five. The constant coefficient and unsupported coordinates are excluded. Means describe the complete supported block; covariance includes within-frequency quadrature pairs as well as cross-frequency pairs. White zero cells mean that this block is diagonal in the saved calculation, including YZY's frequency-one block. These deterministic summaries carry no sampling error bars.}}
  \label{fig:fixed-support-structure-depth}
\end{figure*}

On these fixed coordinate sets, all 16 combinations have a smaller maximum variance at depth five than at depth two, and the mean falls in 15 of them. The exception is nonentangling YZY/$R_X$, whose mean rises slightly from $0.0547$ to $0.0555$. Its largest variance still falls, illustrating why the maximum alone does not describe how variation is distributed throughout the matrix.

Of the seven combinations with resolved off-diagonal covariance, five also show a decrease between these depths in both the largest and mean absolute covariance. HEA/$R_X$ gives a pronounced example: the maximum falls from $1.65\times10^{-4}$ to $6.68\times10^{-7}$, while the mean falls from $5.24\times10^{-6}$ to $1.75\times10^{-8}$. Circuits~16 and~18 with $R_Y$ instead increase in both summaries. These changes describe the matrix as a whole over a fixed set of pairs; they allow individual entries to increase and intermediate depths to depart from a monotonic trend.

To determine whether the coefficients also become less correlated, we must divide out their marginal standard deviations. Otherwise, smaller covariance could simply reflect the contraction of all coefficient fluctuations. The normalised matrices in Figs.~\ref{fig:correlation-depth-atlas-rx} and~\ref{fig:correlation-depth-atlas-ry} address this question. Restricting to fixed cross-frequency pair sets, the largest correlation decreases in six of the seven coupled combinations between depths two and five.

The clearest decreases occur in Circuit~15/$R_X$ and HEA/$R_X$, whose maximum correlations fall from $0.0783$ to $0.000802$ and from $0.162$ to $0.00127$, respectively. The corresponding changes for Circuit~17/$R_Y$, Circuit~19/$R_Y$ and YZY-entangling/$R_X$ are $0.0565$ to $0.0119$, $0.0753$ to $0.0412$, and $0.174$ to $0.0231$. Circuit~16/$R_Y$ changes only slightly, from $0.0997$ to $0.0926$, while Circuit~18/$R_Y$ increases from $0.124$ to $0.143$. These last two are also the combinations whose raw covariance maxima increase. The root-mean-square (RMS) correlation decreases in all seven combinations, although only slightly for Circuit~18/$R_Y$. Thus a reduction in typical coupling can coexist with a persistent or growing largest correlation.

The remaining seven combinations with eligible cross-frequency pairs already have diagonal supported real matrices at every tested depth. The choice of encoding axis determines which pattern appears. Circuit~15, HEA and YZY-entangling have coupling with $R_X$ and diagonal controls with $R_Y$; Circuits~16--19 show the reverse. Nonentangling YZY has no cross-frequency comparison under either encoder. The circuits therefore span several behaviours, from pronounced decorrelation to persistent coupling or a matrix that is diagonal from the outset.

The CHM gives a direct interpretation of these changes. Covariance is an inner product of centred rows, and correlation measures their normalised overlap. This remains true after transforming the rows to real coefficient coordinates. A decrease in correlation therefore means that distinct coefficients depend on less aligned combinations of parameter harmonics, beyond any reduction in the sizes of their fluctuations. Independent full-period averaging retains only matching parameter harmonics, and additional trainable blocks change these surviving overlaps. This explains what the observed decorrelation means in the construction, without implying that every extra block must reduce every overlap. The coordinate choice also matters: decorrelation of independent real coefficients is a different statement from literal diagonalisation of the full signed complex matrix, which includes conjugate partners.

Increasing qubit count produces a different pattern. At $d=5$, we compare common supported coordinates at $\omega=1,2,3$ across $n=3,\ldots,6$ [Fig.~\ref{fig:fixed-support-structure-size}]. Every combination has smaller maximum and mean variance at the larger endpoint, consistent with the qubit-dependent variance suppression predicted under the design assumptions of Mhiri \emph{et al.}~\cite{mhiri2025a}. Yet five of the seven coupled combinations have larger maximum and mean absolute covariance, and five have larger maximum correlation. This comparison holds the frequencies themselves fixed, rather than their positions relative to the expanding encoder bandwidth. Within that comparison, the fall in variance with qubit count provides no general evidence that the coefficients become less correlated.

These structural results suggest a specific learning comparison. If two targets receive equal contributions from individual coefficient variances, covariance is the part of their projection-variance scores that can distinguish them. As the relevant normalised overlaps weaken, the available relative score difference is limited. We test whether this difference accompanies a learning preference in Section~\ref{subsec:covariance-learning-control}, after first examining what variance alone tells us.

\subsection{Learning information from support and variance}
\label{subsec:variance-learning}

We now ask whether the unequal variation assigned to different coefficients is reflected in learning. Single-mode targets make this comparison particularly direct, because each target requests only one cosine or sine coefficient. We train all 16 combinations on these targets at frequencies $\omega=1,\ldots,n$. The learning study uses nine size--depth settings: $n=6$ at $d=1,\ldots,5$, and $d=5$ at $n=2,\ldots,5$. Each circuit--target condition is trained from ten matched initialisations for 500 full-batch Adam steps with learning rate $0.02$.

The periodic training and evaluation grids contain 64 and 256 points, respectively. They resolve the full-period squared error of these noiseless Fourier-regression tasks. The denser evaluation grid therefore checks the fitted function over the period; it does not test generalisation to an unknown data distribution [Appendix~\ref{subsec:appendix_protocol}].

We collect the real coefficients in $q_{\mathrm{cs}}=(a_0,c_1,s_1,\ldots,c_n,s_n)^{\mathsf T}$, where $c_\omega=2\operatorname{Re}a_\omega$ and $s_\omega=-2\operatorname{Im}a_\omega$, and write \(\Sigma_{\mathbb R}=\operatorname{Cov}(q_{\mathrm{cs}})\). All targets in this comparison have zero constant component. For a target coefficient vector \(t\), the target-weighted variance score sums the marginal variances in its requested coordinates:
\begin{equation}
    A_{\mathrm{diag}}(t)
    =t^{\mathsf T}\operatorname{diag}(\Sigma_{\mathbb R})t.
    \label{eq:learning-marginal-score}
\end{equation}
Here \(\operatorname{diag}(\Sigma_{\mathbb R})\) is the diagonal matrix of individual coefficient variances. For example, the target \(0.8\cos(\omega x)\) has score \(A_{\mathrm{diag}}=0.8^2\operatorname{Var}_\theta[c_\omega]\). Our single-mode targets all have this amplitude, so their scores are simply \(0.64\) times the variance of the requested real coefficient. For targets combining several modes, the score adds these individual contributions and omits covariance between them. It then need not equal the variance of the full coefficient projection along \(t\).

The target coefficients are known here because we construct the tasks with fixed amplitudes and a fixed coefficient normalisation. The score describes the circuit's ensemble spread in those coordinates; it does not measure how far an initial output lies from the target or how a parameter update will change that output. Its association with Adam learning must therefore be measured empirically, as must that of the averaged sensitivity: neither specifies Adam's adaptive updates. For predictions from general sampled labels, Section~\ref{subsec:kernel_predictions} instead uses the residual and the kernel at the current parameter point. Appendix~\ref{subsec:appendix_covariance_learning} details the conversion to real coordinates and its implementation.

We distinguish initial learning speed, final accuracy and progress over the whole training run. Early decay is minus the fitted slope of log relative MSE over steps 0--25, so larger values mean faster initial error reduction; the optimiser learning rate itself remains fixed. Final MSE is the absolute error after 500 steps. Time-averaged relative loss, denoted AULC, averages each trajectory's MSE divided by its own initial MSE over steps 0--500 [Eq.~\eqref{eq:appendix_normalised_aulc}]. It measures how much relative error accumulates during training, which can differ even between runs reaching similar final accuracy. We compute each outcome for each seed and then take the median over the ten seeds.

\begin{figure*}[t]
    \centering
    \includegraphics[width=\textwidth]
    {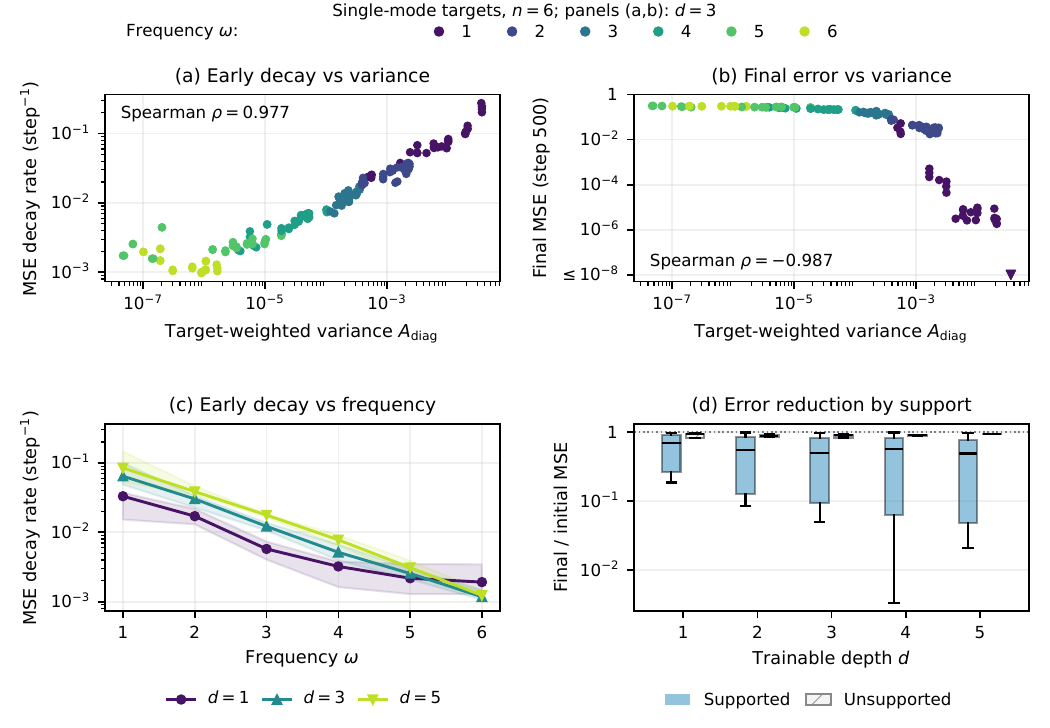}
    \caption{\justifying{Larger target-weighted variance accompanies faster early learning and lower final error in the single-mode Adam benchmark. These associations combine frequency and circuit differences; categorical adjustment is reported in Table~\ref{tab:appendix-variance-adjustment}. (a,b) The same 158 circuit--encoder--quadrature conditions with supported frequency and $A_{\mathrm{diag}}>10^{-14}$ at $n=6,d=3$, including cosine and sine targets, coloured by frequency. Learning values are ten-seed medians; the score is $0.64$ times requested real-coefficient variance, Eq.~\eqref{eq:learning-marginal-score}. (a) Early decay: minus the fitted log-relative-MSE slope over steps 0--25. (b) Absolute MSE at step 500. Four values below $10^{-8}$ appear as overlapping downward triangles at the floor; correlations use uncapped values. (c) Early-decay medians and interquartile ranges across frequency-supported conditions at $n=6$; colours and markers denote depths 1, 3 and 5. Both quadratures are included, even negligible-score conditions excluded from (a,b). (d) Final/initial MSE by numerical frequency support at $n=6$. Boxes show quartiles and medians of ten-seed condition medians; whiskers reach the most extreme values within 1.5 interquartile ranges on the logarithmic ratio scale. Outliers are not drawn; the dotted line denotes no reduction.}}
    \label{fig:variance-spectral-learning}
\end{figure*}

Figure~\ref{fig:variance-spectral-learning}(a,b) shows that the variance score is strongly associated with learning across the pooled conditions. At $(n,d)=(6,3)$, the same 158 supported conditions with non-negligible scores have Spearman associations of $0.977$ with early decay and $-0.987$ with final MSE. Directions with larger variance therefore tend to learn faster initially and reach lower error within this training budget. The association with AULC is consistent with this picture: across the five $n=6$ depths, it ranges from $-0.981$ to $-0.953$.

In panels (a,b), every condition has the same size and depth, $n=6,d=3$. The association therefore compares circuits, frequencies and target phases; it does not imply that making a circuit shallower will improve learning. We test changes in depth separately on fixed supported targets in Fig.~\ref{fig:matched-depth-size-learning}. To assess how much information variance adds beyond those differences, we fit log AULC using log score together with circuit--encoder combination, frequency and phase. The score coefficient remains negative at all five \(n=6\) depths, but adding it to the model containing only these categories increases in-sample \(R^2\) by \(0.007\)--\(0.024\). Thus the additional association is modest once the categories are accounted for, despite the strong pooled correlations. It is also not uniform across sizes: the adjusted coefficient is positive at \((n,d)=(2,5)\). Appendix~\ref{subsec:appendix_variance_adjustment} reports all nine fits. These regressions describe the observed conditions; they do not isolate a causal variance effect or test predictions for unseen circuits and targets.

Panel (c) makes the frequency dependence explicit by grouping early decay at depths one, three and five. Higher supported frequencies typically learn more slowly. Both real quadratures of each supported frequency are included, so the depth-three comparison contains six negligible-score conditions in addition to the 158 in panels (a,b). This distinction matters because support of a complex frequency need not imply variation in both its cosine and sine components. The observed frequency dependence is compatible with the active-term redundancy mechanism of Duffy and Jastrzebski~\cite{duffy2026}, although it does not directly test their bounds. CHM variances include signed weights and cancellations as well as the number of contributing terms.

Panel (d) connects the learning results to the output restrictions derived earlier. A structurally absent coefficient cannot fit the corresponding target component. Numerically unsupported targets retain most of their initial error in this experiment, although training can still reduce total loss by removing unwanted output at other frequencies. Conversely, support alone does not guarantee an exact fit: a coefficient may vary without being independently adjustable or able to reach the required value. The frequency-resolved diagnostics and multimode tests in Appendix~\ref{subsec:appendix_spectral_bias_dynamics} give further details of these learning behaviours.

\subsubsection{Depth and size on fixed supported targets.}
The comparisons above concern differences between circuits and targets at a given size and depth. To understand what changes when the same circuit is made deeper or larger, we also compare an unchanged set of supported targets. Figure~\ref{fig:matched-depth-size-learning} follows 96 circuit--frequency--phase conditions across $n=6,d=1,\ldots,5$, and 58 across $d=5,n=2,\ldots,6$. Each requested real coefficient has non-negligible variance at every setting in its sweep. The target amplitudes, training budget and optimiser step size remain fixed. Consequently, any change in learning in this comparison cannot be explained by a target frequency first becoming supported [Appendix~\ref{subsec:appendix_matched_depth_size}].

\begin{figure*}[t]
  \centering
  \includegraphics[width=\textwidth]{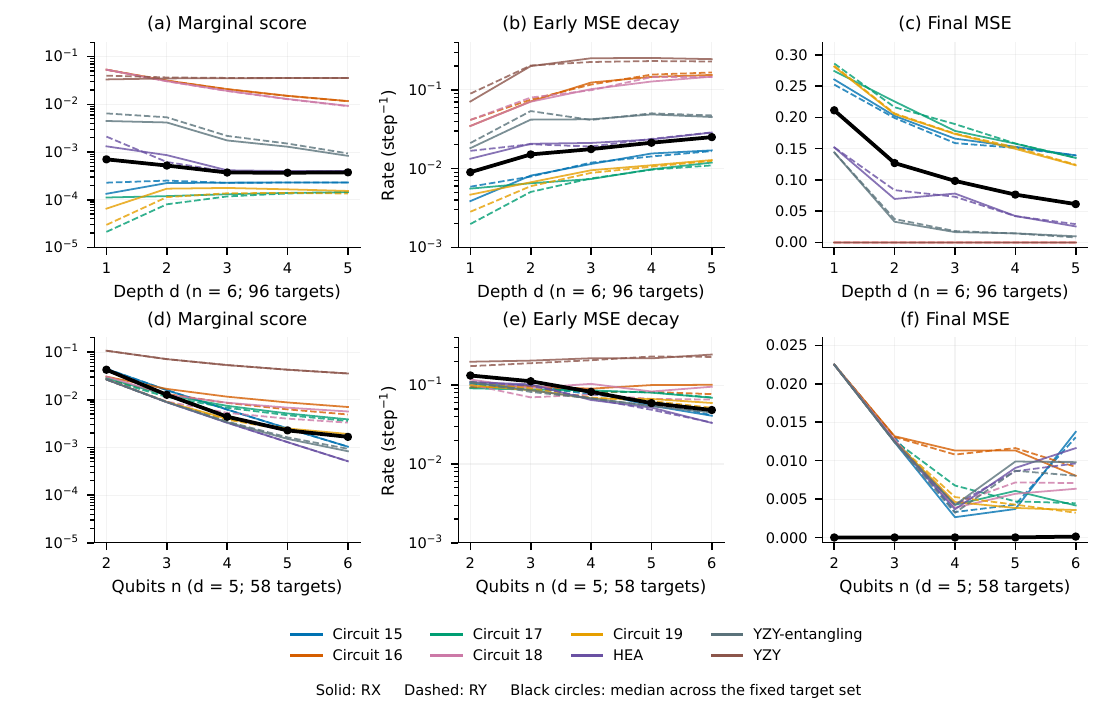}
  \caption{\justifying{Depth and size have different learning consequences on identical supported single-mode targets. (a--c) The fixed intersection of 96 circuit--target conditions at $n=6$, across all five depths. (d--f) The fixed intersection of 58 conditions across all five sizes at $d=5$; only frequencies one and two are shared. Columns show the marginal score, early MSE decay rate over steps 0--25, and absolute MSE after 500 Adam steps. Coloured curves are medians over the retained conditions within each circuit--encoder combination; solid and dashed lines denote $R_X$ and $R_Y$. Black circles are medians across all retained conditions, so combinations with more eligible targets contribute more conditions. Learning outcomes are first summarised by the ten-seed condition median. Lines join deterministic summaries, not confidence intervals. Score and rate axes are logarithmic; final-MSE axes are linear with different ranges in (c) and (f). The optimiser learning rate is fixed at $0.02$ throughout.}}
  \label{fig:matched-depth-size-learning}
\end{figure*}

Increasing depth generally improves learning on this fixed target set. The pooled median early decay rises from $0.00900$ to $0.0251$ per step, while median final MSE falls from $0.211$ to $0.0613$. Between the endpoints, 84 of the 96 conditions acquire faster early decay and 90 reach lower final MSE. Every combination's median early decay increases, although individual high-frequency conditions can slow down; the frequency-six median decreases.

The variance of each requested direction also need not follow the decline in the circuit's largest variance. Only 51 of these 96 marginal scores are smaller at the deeper endpoint. Thus added trainable blocks can reduce the largest fluctuations while increasing variation in initially weak directions, alongside an improvement in learning. The strong association between variance and performance across conditions is therefore compatible with better learning in deeper circuits whose maximum variance is smaller.

Increasing qubit count at $d=5$ produces a different outcome. All 58 matched marginal scores decrease from $n=2$ to $n=6$, consistent with the coefficient-variance concentration predicted under the design assumptions of Mhiri \emph{et al.}~\cite{mhiri2025a}. In training, 50 conditions have slower early decay. The pooled median rate falls from $0.133$ to $0.0484$ per step. Most circuit medians also fall, with exceptions for Circuit~16/$R_X$ and nonentangling YZY under both encoders. Final MSE, however, has no uniform size trend. Frequency-one targets are already fitted to near numerical precision at two qubits, whereas the median final MSE for frequency-two targets improves from $0.0450$ to $0.0169$ as size increases. Within this finite training budget, a smaller variance can therefore accompany slower initial learning without implying a larger final error.

\subsection{Learning information from off-diagonal covariance}
\label{subsec:covariance-learning-control}

\subsubsection{Mean correlation across circuits.}
Variance describes each coefficient separately. We now turn to their joint variation, first asking whether circuits with more strongly correlated coefficients tend to learn differently. This follows the motivation of Strobl \emph{et al.}~\cite{strobl2025}, whose FCC statistic relates coefficient correlations to learning performance.

At $n=6$, we use the mean absolute off-diagonal Hermitian correlation over supported nonnegative complex frequencies. Unsupported modes are excluded before normalisation, and at least two supported modes are needed to define a pair mean. This leaves 66 circuit--encoder--depth settings, of which 33 have zero resolved mean correlation. We train each setting on the same three multimode targets using frequencies one through six. Each learning outcome is first summarised by its median over ten initialisations per target, then by its median over the three targets. Some targets contain unavailable coefficient directions, so this comparison includes differences in support as well as correlation.

\begin{figure*}[t]
  \centering
  \includegraphics[width=\textwidth]{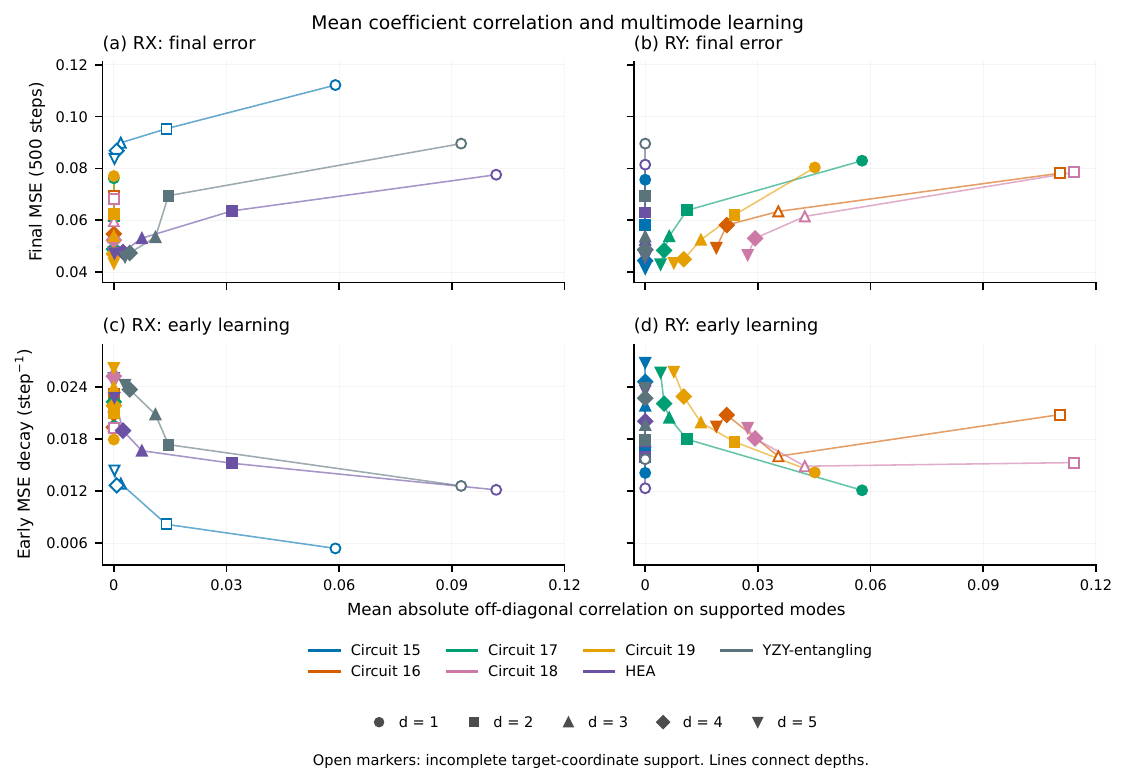}
  \caption{\justifying{Across circuit and depth settings, larger mean coefficient correlation accompanies higher final error and slower early learning; within-circuit trends vary. Columns separate $R_X$ and $R_Y$ encoders; colour identifies the circuit and marker shape the trainable depth $d$. (a,b) Final MSE after 500 Adam steps. (c,d) Early MSE decay rate over steps 0--25; larger values mean faster learning, not a larger optimiser step size. Each panel shows 33 settings at $n=6$. Points take the median across three fixed in-band target phases after the ten-initialisation median for each target. Lines join depths within a combination, not fitted regression curves. Open markers indicate incomplete support of the twelve positive-frequency cosine--sine target coordinates. The horizontal statistic averages absolute correlations over supported nonnegative complex modes, including the constant mode when supported; zero correlations are retained. Fourteen settings with only one supported nonnegative mode have an undefined pair mean and are omitted. Appendix~\ref{subsec:appendix_mean_correlation_learning} defines the support rules, aggregation and adjusted associations.}}
  \label{fig:correlation-learning-overview}
\end{figure*}

Across the pooled settings, larger mean correlation accompanies higher final MSE and slower early decay [Fig.~\ref{fig:correlation-learning-overview}]. The corresponding rank associations are $\rho=0.285$ and $\rho=-0.445$, respectively. The direction agrees with the FCC heuristic, but the paths joining depths within each circuit show that the relationship is not uniform. HEA/$R_X$ and YZY-entangling/$R_X$ move towards lower correlation, lower MSE and faster early decay with depth. Circuit~15/$R_X$ instead retains relatively high error and incomplete support of the target coordinates, while Circuit~16/$R_Y$ has nonmonotonic early decay. Even settings with zero mean correlation have different learning outcomes.

We therefore examine how the pooled associations change when depth and target-coordinate coverage are taken into account. After this adjustment, the final-MSE association falls to $0.146$, whereas the early-decay association remains $-0.404$. Including total variance changes these estimates little. Further adjustment for circuit--encoder combination changes the sign of both associations; intervals obtained by resampling circuit families are broad and span zero [Table~\ref{tab:correlation-learning-adjustment}]. The mean correlation is thus associated with performance across this collection, but these exploratory adjustments do not establish a circuit-independent effect or prediction for unseen targets. Our supported-mode statistic and training protocol also differ from the FCC benchmark of Strobl \emph{et al.}

\subsubsection{Matched-target preference.}
The association across circuits leaves a more specific question unanswered: within a given circuit, do the covariances favour some combinations of coefficients over others during learning? We test this using two-mode targets $t_+$ and $t_-$ with the same frequencies and amplitudes but opposite relative sign. Their marginal variance scores are identical, while the selected covariance cross term changes sign. Training both targets from matched initialisations under the same Adam protocol then lets us compare their learning outcomes with the preference in the full covariance score.

Matching the amplitudes, marginal score and initialisation removes several sources of difference between the targets. Changing the relative sign still changes the training residual and its alignment with the local kernel, however, so the comparison does not isolate a causal covariance effect. Figure~\ref{fig:covariance-learning-controls} summarises the result: the available score contrast becomes smaller with depth, but a robust pooled learning advantage for the covariance-favoured target does not emerge. The circuit-specific responses below show why that pooled result needs closer examination.

\begin{figure*}[t]
    \centering
    \includegraphics[width=\textwidth]{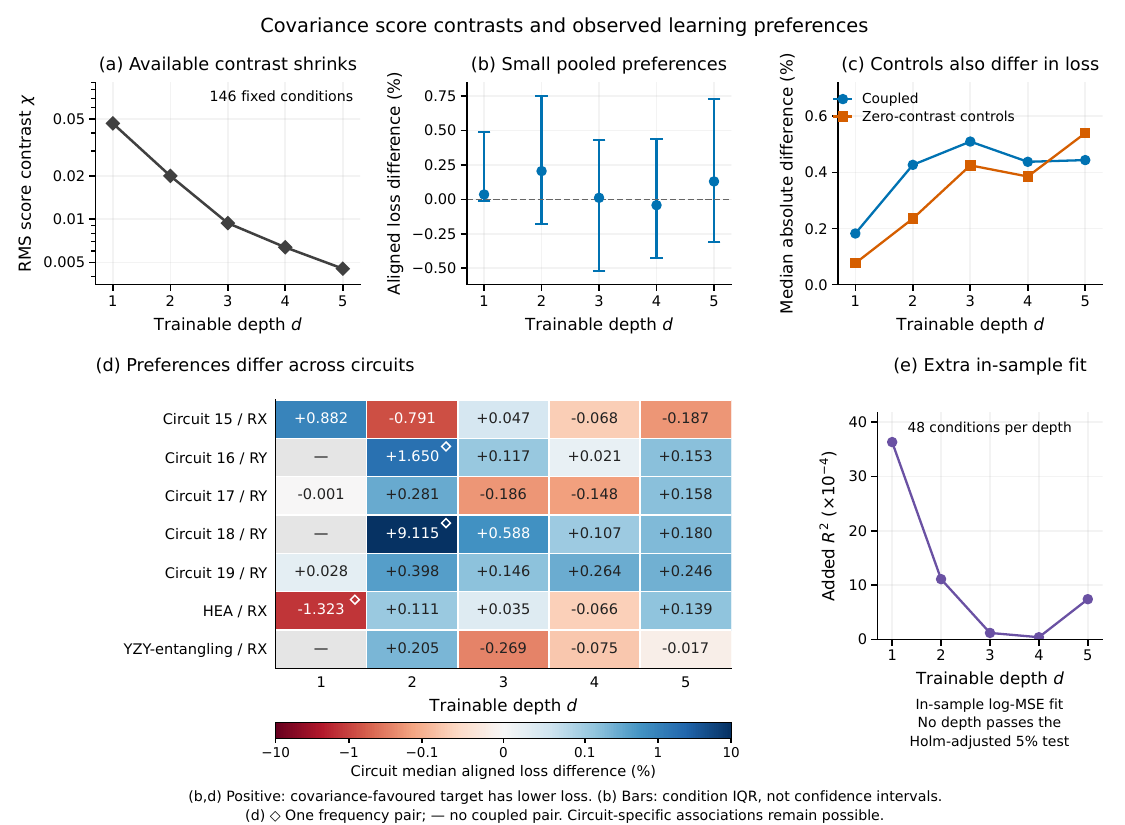}
    \caption{\justifying{No robust pooled covariance-aligned learning advantage is established in this benchmark ($n=6$). Pooled preferences are small and change sign, controls also differ in loss, and circuit-specific associations remain possible. (a) RMS contrast $\chi$ over 146 conditions supported at every depth, including zero-contrast controls. (b) Median and interquartile spread of the aligned response $100\operatorname{sign}(\chi)\Delta_{\mathcal L}$ over the 32, 58, 70, 78 and 78 coupled conditions at depths one through five. Positive values favour the covariance score. (c) Median absolute response $100|\Delta_{\mathcal L}|$ for coupled conditions and trained controls with a below-threshold selected contrast. (d) The same aligned response, summarised separately for every combination with resolved coupling. Cells give circuit medians in percent; open diamonds mark one frequency pair and grey dashes no coupled pair. The signed colour scale is linear within $\pm0.1\%$ and logarithmic outside; all medians are retained. (e) Additional in-sample $R^2$ from signed covariance in a regression of multimode log final MSE, beyond marginal score, circuit--encoder combination and phase (48 conditions per depth). In (b)--(d), each condition first takes the median over ten paired-seed responses in time-averaged MSE over steps 0--500. Condition populations change with depth; interquartile bars are not confidence intervals. Appendix~\ref{subsec:appendix_covariance_learning} gives individual responses, circuit spreads, support distinctions and reference tests.}}
    \label{fig:covariance-learning-controls}
\end{figure*}

We include the covariance cross terms by using the variance of the coefficient projection along the target. The resulting full score and its relative contrast between the two targets are
\begin{equation}
\begin{aligned}
    A_{\mathrm{full}}(t)&=t^{\mathsf T}\Sigma_{\mathbb R}t,\\
    \chi&=
    \frac{A_{\mathrm{full}}(t_+)-A_{\mathrm{full}}(t_-)}
    {A_{\mathrm{full}}(t_+)+A_{\mathrm{full}}(t_-)}.
\end{aligned}
\label{eq:learning-covariance-contrast}
\end{equation}
The full score is \(A_{\mathrm{full}}(t)=\operatorname{Var}_\theta[t^{\mathsf T}q_{\mathrm{cs}}]\): it measures how much the circuit's coefficient projection along \(t\) varies across initialisations. For a single-quadrature target, it agrees with the marginal score. When several modes contribute, their covariances can increase or decrease this projection variance. Both scores describe ensemble variation in the chosen coordinates; neither gives an MSE decay rate or the inverse-covariance distance used by Wiedmann \emph{et al.}~\cite{wiedmann2024}.

A positive $\chi$ means that $t_+$ has the larger full score. Because the targets have equal marginal contributions, their covariance cross terms determine the difference in the numerator, while their sum normalises it by the marginal variance. Thus $\chi$ expresses a preference for this particular target pair. It is distinct from both the mean covariance magnitude of a circuit and an individual Fourier-coefficient correlation $R_{\omega\mu}$.

We measure the learning preference using the time-averaged MSE over steps 0--500. For each matched initialisation $j$, let these averages be $\overline{\mathcal L}_{+,j}$ and $\overline{\mathcal L}_{-,j}$. Here we use the MSE itself, rather than dividing each trajectory by its initial loss as in AULC. The relative loss difference and its median over seeds are
\begin{equation}
\begin{aligned}
    \Delta_{\mathcal L,j}
    &=\frac{\overline{\mathcal L}_{-,j}-\overline{\mathcal L}_{+,j}}
    {\tfrac12(\overline{\mathcal L}_{-,j}+\overline{\mathcal L}_{+,j})},\\
    \Delta_{\mathcal L}&=\operatorname{median}_j\Delta_{\mathcal L,j}.
\end{aligned}
\label{eq:learning-loss-asymmetry}
\end{equation}
A positive $\Delta_{\mathcal L}$ means that $t_+$ has lower time-averaged MSE; a negative value favours $t_-$. Its absolute value measures the size of the difference regardless of which target learns better. We normalise the difference for each seed before taking the median. If the larger covariance score accompanies better learning, $\Delta_{\mathcal L}$ should have the same sign as $\chi$, so that $\operatorname{sign}(\chi)\Delta_{\mathcal L}>0$.

The decorrelation observed in Section~\ref{subsec:covariance-depth-structure} places a useful limit on the score contrast. For two supported real coefficient directions, $\chi$ is their correlation multiplied by a factor of magnitude at most one, determined by the balance of their marginal contributions. Weaker correlation therefore limits the available relative preference. A fall in raw covariance would not be enough if the variances fell in the same proportion. Whether the resulting preference is reflected in Adam learning remains a separate empirical question.

Panel (a) of Fig.~\ref{fig:covariance-learning-controls} first measures the change in available contrast while keeping the targets fixed. Across 146 conditions supported at all five depths, including zero-contrast controls, the RMS contrast falls from $0.0466$ to $0.00453$, a factor of $10.3$. Panel (b) then examines the learning response, with positive values of $100\operatorname{sign}(\chi)\Delta_{\mathcal L}$ favouring the covariance score. Its pooled medians at depths one through five are $0.035\%$, $0.205\%$, $0.011\%$, $-0.042\%$ and $0.130\%$. At every depth, the interquartile range contains both signs. These bars describe variation between conditions, rather than uncertainty on a shared effect. The score contrast contracts clearly, but the measured learning preferences do not show a corresponding monotonic decline.

Panel (d) separates the circuits contributing to those pooled values. Circuit~19/$R_Y$ favours the covariance score at every depth, whereas Circuit~15/$R_X$ changes sign with depth: its median is positive at depths one and three, and negative at depths two, four and five. Circuit~17/$R_Y$, HEA/$R_X$ and YZY-entangling/$R_X$ have mixed signs or medians near zero. Circuits~16 and~18 with $R_Y$ retain positive medians wherever coupled targets are available. Their largest responses occur at depth two, however, where each has only one frequency pair.

All cells based on one frequency pair are marked in the same way. Ten paired initialisations measure seed-to-seed variation for that pair, but do not supply additional frequency pairs with which to assess target dependence. Moreover, the eligible target set changes with depth in these response comparisons. Figures~\ref{fig:architecture-covariance-heterogeneity} and~\ref{fig:architecture-covariance-heterogeneity-ry} show the individual conditions, the circuit interquartile ranges and the corresponding matrices. Grouping responses by contrast in Fig.~\ref{fig:covariance-aligned-response} also shows that larger contrasts do not consistently accompany larger learning advantages.

The trained zero-contrast controls help interpret the size of these differences. Panel (c) compares absolute responses: their medians range from $0.182\%$ to $0.510\%$ for coupled conditions and from $0.076\%$ to $0.540\%$ for controls whose selected contrast is below threshold. A zero score preference therefore does not produce a flat zero learning response. This is also true for the seven combinations whose supported matrices are diagonal at every depth. Of the 1,038 control conditions, 60 contain an unsupported real quadrature; the appendix profiles distinguish them. These controls illustrate the limitation of interpreting a loss difference through ensemble covariance alone. Zero covariance leaves the training residual and local kernel free to differ, while comparisons across circuits also combine differences in architecture and target support.

Panel (e) asks a complementary question about multimode final MSE: does including covariance explain additional variation after marginal score, circuit--encoder combination and target phase have been accounted for? At each depth we compare a regression containing these terms with one that also includes the signed covariance contribution. The increase in in-sample $R^2$ ranges from $3.96\times10^{-5}$ to $3.63\times10^{-3}$ across the five depths. The largest increase occurs at depth one, where $R^2$ rises from $0.98866$ to $0.99230$. Although small relative to the total variation, this increase accounts for about $32\%$ of the baseline model's remaining unexplained log-loss variation. At depth five, the smaller increase from $0.99439$ to $0.99513$ accounts for about $13\%$ of the remaining variation. The small total increments therefore do not by themselves make covariance irrelevant. This comparison concerns a different target set and final accuracy, rather than the paired time-averaged losses, and it assesses fits to the observed data rather than prediction on held-out conditions.

The statistical reference tests provide limited evidence for these associations. At $d=2$ and $d=5$, the rank associations between signed contrast and loss difference are $\rho=0.344$ and $0.347$, respectively. Both have raw $p=0.0625$, increasing to $p_{\mathrm{Holm}}=0.3125$ after adjustment for testing five depths. The matched-sign tests also have too few circuit groups to reach the adjusted $5\%$ threshold: their smallest possible two-sided reference value is $2/128=0.015625$, giving $p_{\mathrm{Holm}}\geq0.078125$. Their adjusted result is therefore inconclusive. The multimode tests likewise have no significant depth under their separate adjustment. The reference distributions also rely on assumptions that the training experiment does not ensure: symmetry under reversing the responses of whole circuit groups for the matched-sign tests, and exchangeability of regression residuals within groups for the multimode tests. Table~\ref{tab:covariance-rank-reference} and Appendix~\ref{subsec:appendix_covariance_learning} give the depth-wise rank results, intervals and test definitions. These tests concern the associations just described, rather than the aligned medians or their trend with depth.

Taken together, the comparisons do not establish a robust pooled covariance-aligned learning advantage under this protocol. There are suggestive differences between circuits, and neither the mixed signs in the interquartile ranges nor the adjusted test results proves that an effect is absent. The conclusion concerns the target preferences and scores tested here, rather than every possible use of covariance in learning.

\subsubsection{Reconciling the learning comparisons.}
The association between mean correlation and learning, and the weaker evidence for a matched-target preference, can coexist because they address different questions. Figure~\ref{fig:correlation-learning-overview} compares overall performance across circuits and depths. Its direction agrees with the lower-FCC, lower-error association reported by Strobl \emph{et al.}~\cite{strobl2025}, using a different supported-mode statistic and training benchmark. Figure~\ref{fig:covariance-learning-controls} asks instead which of two targets learns better within one circuit. A small difference between those targets would still allow an influence shared by both, so it cannot establish that the aggregate association is incidental.

Nevertheless, some depth improvements occur without any change in mean correlation or target support. The multimode data contain three combinations with zero resolved mean correlation and complete support of the target coordinates at every depth: Circuit~15/$R_Y$, Circuit~17/$R_X$ and Circuit~19/$R_X$. All three reach lower final MSE and faster early decay at depth five than at depth one [Table~\ref{tab:diagonal-depth-learning}]. Their improvement requires neither a further decrease in this correlation statistic nor newly supported target coordinates. The single-mode comparison on an unchanged supported target set leads to a related observation [Fig.~\ref{fig:matched-depth-size-learning}].

These cases show that decorrelation is not needed for every depth benefit, but they do not identify a unique explanation for the pooled association. In particular, the association with early decay survives adjustment for depth and target coverage and becomes unstable only when circuit--encoder combination is also included. We therefore cannot attribute the entire relationship to depth alone. The remaining possibility to examine is the one developed in the tangent-kernel section: even when the same coefficients can vary across an ensemble, depth may change which combinations are accessible near an initialisation and how sensitive they are to parameter updates. The numerical QNTKs let us investigate these local properties.

\subsection{Numerical QNTKs: ensemble sensitivity and local directions}
\label{subsec:numerical-qntk-capacity}

The depth comparisons bring us back to the distinction between a coefficient that varies across parameter choices and one that can be changed independently at a particular initialisation. Marginal support establishes the former. The tangent kernel describes the latter and also measures how strongly each available direction responds. We examine these local properties as possible contributors to the learning improvements that occur even when correlation and support remain unchanged.

We first compare coefficient variation with sensitivity averaged over the parameter ensemble. Figure~\ref{fig:averaged-qntk-geometry} shows these quantities for YZY-entangling. As derived in Section~\ref{sec:kernels}, the averaged kernel uses the same CHM columns as the covariance, but weights them by squared parameter frequency, with the appropriate physical-angle factors. The ratio $\gamma_\omega$ in Eq.~\eqref{eq:kernels_sensitivity_variation_ratio} compares the resulting mean squared derivative strength with the variance of each coefficient. Its dependence on frequency shows that variance and sensitivity do not give proportional profiles. Interpreting the ratio requires both quantities: a large value can result from a small variance denominator and need not imply a large absolute gradient or fast learning.

The rise with depth in Fig.~\ref{fig:averaged-qntk-geometry}(c) therefore means greater sensitivity \emph{relative to variance}. The absolute sensitivity is their product, $\bar H_{\omega\omega}=\gamma_\omega v_\omega$, so it can increase or decrease as the variance contracts. Both behaviours occur in this example. Between depths one and five at $n=6$, the variances of frequencies one and two both fall. Over the same interval, $\bar H_{11}$ rises from $0.00955$ to $0.0229$, whereas $\bar H_{22}$ falls from $0.0149$ to $0.0128$, even though both ratios increase. The absolute sensitivities also need not change monotonically at intermediate depths.

In CHM terms, differentiation gives greater weight to parameter harmonics involving more angles. This allows a coefficient's mean squared gradient to grow even when its ensemble variance shrinks. The kernel sums these derivative contributions over all trainable parameters, whose number also grows with depth; it is not a per-parameter sensitivity. These changes offer a possible contribution to the learning improvements, but the ratio alone does not explain them. Learning depends on the kernel at the current parameters, its action on the residual, and the optimiser's update rule.

\begin{figure*}[t]
    \centering
    \includegraphics[width=0.90\textwidth]
    {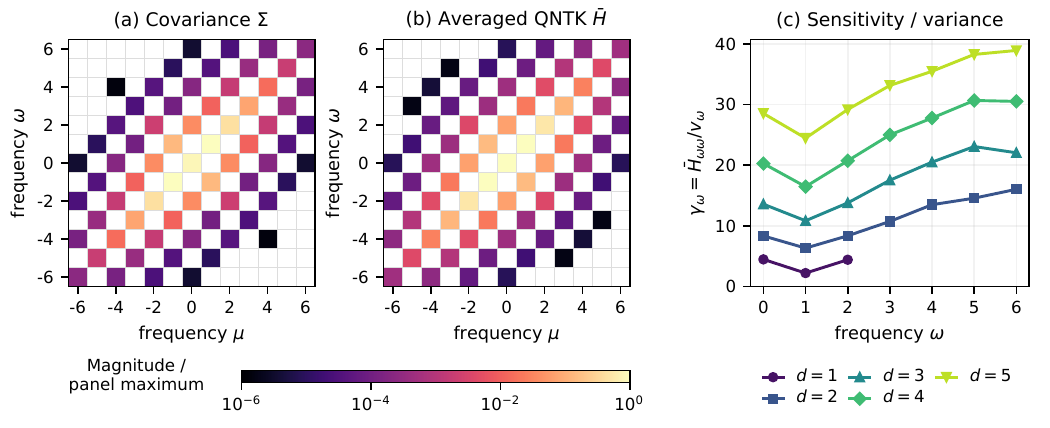}
    \caption{\justifying{Coefficient variance and mean tangent sensitivity are not proportional across frequencies in $R_X$-encoded YZY-entangling. (a,b) Magnitudes of $\Sigma$ and $\bar H$ at $n=6,d=3$, each divided by its own maximum: relative structure, not absolute scale. All signed frequencies and diagonals are retained; all modes are supported. White denotes normalised magnitude below $10^{-6}$, not necessarily exact zero. (c) Sensitivity-to-variance ratio $\gamma_\omega=\bar H_{\omega\omega}/v_\omega$ on supported nonnegative frequencies at $n=6$. The $d=1$ curve ends at $\omega=2$ because higher modes are unsupported. The rise with depth concerns sensitivity relative to variance; absolute sensitivity can rise or fall, as discussed in the text. Sensitivity sums squared derivative magnitudes over all parameters, without dividing by their number.}}
    \label{fig:averaged-qntk-geometry}
\end{figure*}

We next examine the kernel at individual initialisations, where its rank counts the coefficient directions available to first order. Using the real Fourier vector introduced above, including its constant coefficient, the full encoder-allowed dimension is $D=2n+1$. The Jacobian $X_{\mathbb R}=\partial_\vartheta q_{\mathrm{cs}}$ gives $H_{\mathbb R}=X_{\mathbb R}X_{\mathbb R}^{\mathsf T}$. We retain the unscaled cosine--sine convention throughout. Rank therefore refers to directions in this real coefficient space, rather than to the number of supported complex frequencies in the ensemble covariance.

For each combination at each of the nine learning settings, we evaluate the kernel at 50 independent uniform full-period parameter draws (seeds 0--49). These settings are $n=6$ at every depth and $d=5$ at $n=2,\ldots,5$, giving 7,200 kernels in total. The draws include the ten training initialisations, so the comparison describes the same collection of circuits and is not an independent holdout. We count numerical rank using eigenvalues above $\max\{10^{-13},10^{-10}\lambda_{\max}\}$.

A direction can pass this cutoff yet be much less sensitive than another. To measure how evenly sensitivity is distributed, we also use the normalised effective rank,
\begin{equation}
    \frac{r_{\mathrm{eff}}(H_{\mathbb R})}{D}
    =\frac{(\sum_i\lambda_i)^2}{D\sum_i\lambda_i^2}.
    \label{eq:numerical-qntk-effective-rank}
\end{equation}
A small value means that most of the squared-gradient response is concentrated in relatively few eigendirections, even if many directions have nonzero eigenvalues.

\begin{figure*}[t]
    \centering
    \includegraphics[width=.68\textwidth]{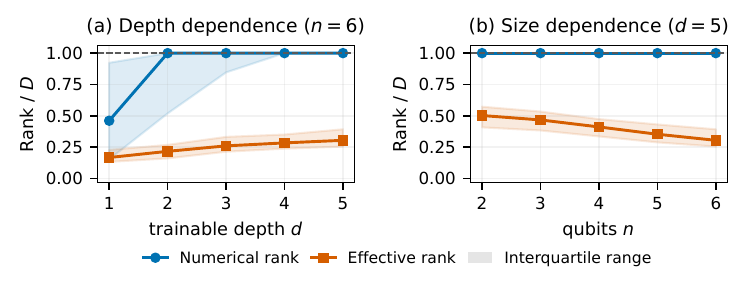}
    \caption{\justifying{Full median numerical tangent rank coexists with unequal directional sensitivity in real Fourier coordinates. (a) Depth dependence at $n=6$; (b) size dependence at $d=5$. Lines and bands show medians and interquartile ranges across 800 initial kernels per setting, not confidence intervals. Both ranks are divided by the full encoder-allowed real dimension $D=2n+1$; dashed lines mark unity. Numerical rank counts directions above the stated eigenvalue cutoff; effective rank measures the balance of tangent strength. Full median numerical rank neither implies isotropy nor full rank at every initialisation.}}
    \label{fig:qntk-capacity}
\end{figure*}

Figure~\ref{fig:qntk-capacity} shows why these two rank measures are useful together. At $n=6,d=1$, the median parameter-to-coefficient ratio is $m/D=1.38$, yet the median numerical-rank fraction is only $0.462$ and the median effective-rank fraction is $0.168$. Having more parameters than coefficient coordinates therefore does not ensure that all coefficient directions are locally available.

The median numerical rank reaches the full dimension $D$ at $d=2$. Sensitivity remains uneven, however: even at $d=5$, the median effective-rank fraction is only $0.305$. Increasing qubit count at fixed $d=5$ also changes this balance. From $n=2$ to six, the median effective-rank fraction falls from $0.503$ to $0.305$, while the median numerical rank remains full. Thus added depth reduces local rank deficiency across this collection, but the available directions still differ substantially in sensitivity. A full median rank also leaves room for individual circuits or initialisations with deficient rank.

This behaviour is consistent with the exact relation derived in the complex coordinates of Section~\ref{sec:kernels}:
\begin{equation}
\begin{aligned}
    \operatorname{col}H(\vartheta_0)
    &\subseteq\operatorname{col}C^\circ\\
    &=\operatorname{col}\Sigma
     =\operatorname{col}\bar H.
\end{aligned}
\label{eq:numerical_pointwise_average_range}
\end{equation}
The covariance and averaged kernel describe the same linear space of ensemble variation, while the kernel at \(\vartheta_0\) can access only a subspace of it. Even an accessible direction affects the loss according to its overlap with the training residual and the sampling map $K=VHV^\dagger$. Appendix~\ref{subsec:appendix_kernel_checks} verifies this map and reports checks on the numerical rank threshold.

The rank results identify one way depth can change local learning, but they do not isolate the cause of improvement in the zero-correlation controls. Nor can reaching full median numerical rank at depth two explain all improvements at later depths. To go beyond the number of available directions, we must examine how the kernel acts on the actual residual. The next experiment tests whether that response predicts loss reduction over a training interval where the local approximation can be checked directly.

\subsection{Learning forecasts from the initial QNTK}
\label{subsec:numerical-frozen-qntk}

Figure~\ref{fig:ensemble-tangent-geometry}(d) illustrates how the full circuit can depart from its fixed linearisation as the parameters change. We now test how accurately that linearisation forecasts the loss along a training run. The initial kernel describes the available response, and the initial residual specifies which part of that response is relevant to the task. We combine them in the frozen-QNTK prediction of Section~\ref{subsec:kernel_predictions}, following Liu \emph{et al.}~\cite{liu2022}. This experiment uses full-batch Euclidean gradient descent for both the full circuit and the frozen model, as required by the update being tested. The Adam comparisons above address separate empirical questions about the coefficient statistics.

We use YZY-entangling/$R_X$ at $n=6,d=3$, with learning rate $0.002$ and 500 steps. Six target frequencies, two phases and ten initialisations give 120 trajectories. At each initialisation, we map its coefficient kernel to the training inputs and iterate the frozen residual update in Eq.~\eqref{eq:kernels_frozen_residual_discrete}. The factor $2\eta/N$ corresponds to MSE defined without a factor of one half. The prediction uses the kernel of the individual run and its sampled residual; it requires neither the target's Fourier coefficients nor substitution of the ensemble average $\bar H$.

\begin{figure*}[t]
  \centering
  \includegraphics[width=.90\textwidth]{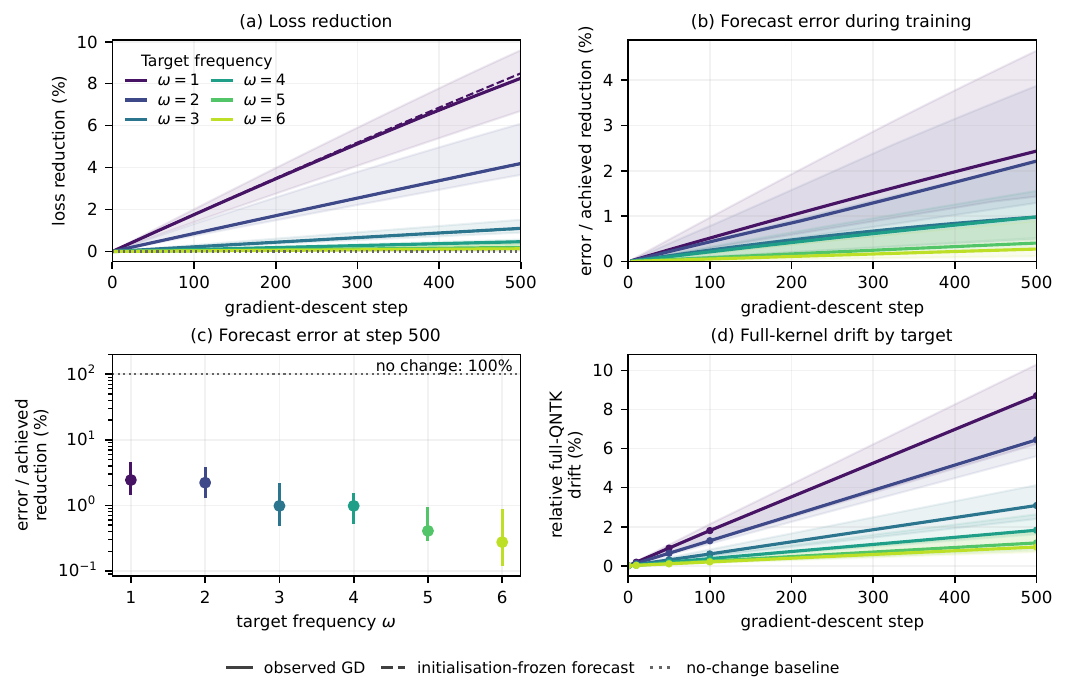}
  \caption{\justifying{The initial QNTK forecasts loss reduction closely while both forecast error and kernel drift develop during training. Full-batch gradient descent for YZY-entangling/$R_X$ at $n=6,d=3$, with learning rate $0.002$ and 500 steps. (a) Median observed loss reduction $100(\mathcal L_0-\mathcal L_t)/\mathcal L_0$ (solid) and the median of the individual initialisation-frozen forecasts (dashed); no change gives zero (dotted). The shaded bands show the spread of observed reductions across runs, not their forecast errors. (b) Absolute loss-forecast error divided by the achieved loss reduction, $100e_t$ from Eq.~\eqref{eq:frozen-error-over-reduction}, on a linear scale that makes the discrepancy visible. The undefined value at $t=0$ is omitted. (c) The same error at step 500 on a logarithmic scale, including the no-change baseline at $100\%$. Frequency-wise medians span $0.28\%$--$2.44\%$. (d) Relative Frobenius drift $100\|H_{\mathbb R}(\vartheta_t)-H_{\mathbb R}(\vartheta_0)\|_{\mathrm F}/
  \|H_{\mathbb R}(\vartheta_0)\|_{\mathrm F}$, evaluated at saved kernel
  checkpoints in unscaled cosine--sine coordinates. Colours group runs by target frequency; they do not project the drift onto individual frequencies. Lines and points show medians, and bands and bars show interquartile ranges over twenty seed--phase trajectories per frequency. These are descriptive spreads, not confidence intervals. Ratios are computed per trajectory before aggregation. Losses use the periodic evaluation grid.}}
  \label{fig:frozen-qntk-validation}
\end{figure*}

Figure~\ref{fig:frozen-qntk-validation}(a) compares the median observed loss reduction with the median of twenty individual forecasts at each target frequency. Each forecast uses that run's own initial kernel and residual. The shaded bands describe the spread of observed reductions across runs, so their width does not measure departure from the frozen model. The median curves almost coincide, with larger reductions for the lower-frequency targets. After 500 steps, the median reductions range from $8.26\%$ of the initial loss at $\omega=1$ to only $0.129\%$ at $\omega=6$. Because the loss changes so little, an error measured relative to the remaining loss could make even a prediction of no change appear accurate. We therefore divide each trajectory's absolute forecast error by its achieved loss reduction:
\begin{equation}
  e_t=\frac{|\mathcal L_t^{\mathrm{fr}}-\mathcal L_t|}
  {\mathcal L_0-\mathcal L_t},
  \qquad \mathcal L_t<\mathcal L_0.
  \label{eq:frozen-error-over-reduction}
\end{equation}
Here $\mathcal L_t^{\mathrm{fr}}$ is the frozen forecast and $\mathcal L_t$ is the observed evaluation-grid MSE. An exact loss forecast gives zero, while predicting no change gives one.

Panels (b,c) of Fig.~\ref{fig:frozen-qntk-validation} assess agreement for the individual runs. We take the absolute error between each trajectory and its own forecast before computing the medians, so over- and under-predictions cannot cancel. The median errors grow during training, but remain small compared with the learning achieved. At the final step, the frequency-wise medians are $0.28\%$--$2.44\%$ of the achieved loss reduction, and the median across all 120 trajectories is $1.02\%$. Every trajectory improves on the no-change baseline, with a largest individual error ratio of $10.87\%$. The close aggregate curves are therefore accompanied by accurate loss forecasts for individual initialisations. This supports using the initial kernel and residual to predict the local loss change of a particular run, beyond their agreement with a typical learning curve.

This accuracy does not require the kernel to remain exactly fixed. Figure~\ref{fig:frozen-qntk-validation}(d) shows that it changes during training, with a final median relative drift of $3.33\%$ across all 120 trajectories and a maximum of $15.15\%$. The drift measures change in the whole coefficient kernel, whereas loss evolution depends on how the corresponding data-space kernel acts along the residual. The initial kernel thus captures the local reduction in total loss well here despite measurable changes in the kernel itself. Appendix~\ref{subsec:appendix_frozen_qntk} gives the full protocol and frequency-resolved values.

The scope of this agreement follows from the experiment. It concerns the relatively small loss reductions achieved here; accuracy after larger reductions remains untested. Agreement in total loss also need not imply agreement in the full residual trajectory, and the experiment does not test frozen dynamics under Adam. Finally, the smaller global kernel drift observed for high-frequency targets in Fig.~\ref{fig:frozen-qntk-validation}(d) measures change in the whole kernel. To establish a lazy-training regime specific to a frequency, one would need to examine the drift projected onto the directions relevant to that frequency.

\subsection{Learning diagnostics and predictions}
\label{subsec:numerical-summary}

The numerical results separate differences between learning tasks from changes observed when adding depth to a given circuit. At fixed size and depth, targets in directions with larger marginal variance tend to learn faster and reach lower MSE. Yet making the same circuit deeper generally improves learning on targets that were already supported, even as its largest variance shrinks. The variance in the particular requested direction need not follow that largest value, and its local sensitivity need not scale in proportion to either. These are compatible observations: the association across tasks does not determine how a fixed task changes with depth.

Increasing qubit count has a different effect. On common supported targets, all marginal scores fall and most conditions learn more slowly at first, while final MSE remains frequency and circuit dependent. The variance suppression is consistent with the design-based concentration result of Mhiri \emph{et al.}~\cite{mhiri2025a}. It does not imply the same learning consequences as the variance reduction observed when trainable blocks are added.

Covariance also needs to be interpreted in relation to the comparison being made. Lower mean correlation accompanies better learning across the pooled settings, and several circuit families decorrelate with depth. However, depth also improves learning in controls whose correlation is already zero and whose target support is unchanged. The matched-target experiments likewise find no robust pooled advantage for the covariance-favoured target, despite differences between circuits. Covariance describes a real structural preference in the ensemble, but that preference alone does not determine the observed learning preference.

The QNTKs show what changes beyond this ensemble description. Added depth can make more coefficient directions locally available, and the averaged sensitivity can increase relative to variance without increasing in absolute terms at every frequency. These findings make changes in local accessibility and sensitivity plausible contributors to the Adam improvements, without establishing their cause. A direct prediction requires the response along the actual residual. Combining the initial kernel with that residual gives accurate local loss forecasts in the separate small-step gradient-descent experiment, even as forecast error and kernel drift grow during training. The construction thus connects restrictions on possible outputs, ensemble variation and conditional learning predictions, while the experiments show why none of these descriptions can stand in for the others.

\FloatBarrier

\section{Discussion and conclusions}
\label{sec:discussion}

\subsection{Contribution in context}
\label{subsec:discussion_interpretation}

The CHM connects three aspects of a fixed circuit: which coefficient vectors it can produce, how those coefficients fluctuate across a parameter ensemble, and how they respond locally to parameter changes. These questions require more than the encoder spectrum. Their common representation follows from the Fourier constructions of Schuld, Nemkov and Wiedmann \emph{et al.} \cite{schuld2021,nemkovFourierExpansionVariational2023,wiedmann2024}. The circuit-specific contribution is the connection between these coefficient moments and tangent response, together with the two-copy calculation and the finite-depth learning comparison. The distinction from covariance-based architecture heuristics \cite{wiedmann2024} and parameter-dependent harmonic state matrices \cite{barthe2024} is set out in Section~\ref{sec:intro}.

Given sampled labels, this connection lets us address two different questions. The ensemble mean and exact data-space covariance determine the affine span containing every circuit output. If the labels lie outside it, Corollary~\ref{cor:kernels_affine_error} gives an error that no optimiser can remove. Within that span, the circuit's attainable outputs can follow a lower-dimensional curved set, as Figure~\ref{fig:ensemble-tangent-geometry} illustrates. Once an initialisation has been selected, the pointwise QNTK describes the immediate response to the residual and, under the frozen-model approximation, forecasts subsequent learning. Both calculations use the sampled labels directly. Together they connect an exact restriction on possible fits with a local description of how training changes the output. The variance-based scores tested under Adam provide a further empirical comparison with learning performance.

\subsection{Implications for spectral learning}
\label{subsec:discussion_findings}

The depth sweeps show that a circuit can become easier to train even while its ensemble coefficient fluctuations become smaller. To understand this result, we must distinguish the amount of variation across parameter choices from the changes available near a particular parameter point, and compare the same learning tasks across depths.

Increasing trainable depth can broaden coefficient support without broadening the encoder spectrum. It can also redistribute relative variance while reducing its absolute scale [Fig.~\ref{fig:aggregate-architecture-trends}]. The tendency towards redundancy-shaped profiles and the suppression of variance with increasing qubit count are consistent with the unitary-design predictions of Mhiri \emph{et al.}~\cite{mhiri2025a}. For our fixed pre-encoding state, however, the relevant conditional design calculation has a different absolute prefactor from their independent-two-block result [Appendix~\ref{app:fixed_input_design_reference}], while retaining exponential suppression with qubit count under its design assumption [Appendix~\ref{app:design_scope}]. Thus closer agreement in shape is informative about the allocation of coefficient variation, but establishes neither design convergence nor a larger available gradient.

Largest and mean absolute covariances both decline with depth in five coupled families [Fig.~\ref{fig:fixed-support-structure-depth}]. The matrix sequences add a separate structural finding: several families also become less correlated after fixing coefficient support and dividing out the marginal scales [Figs.~\ref{fig:correlation-depth-atlas-rx} and~\ref{fig:correlation-depth-atlas-ry}]. In CHM terms, their coefficient rows become less aligned in parameter-harmonic space. Other encoder choices already give diagonal matrices, and some retain appreciable coupling. On fixed low-frequency directions, increasing qubit count generally does not reproduce this decorrelation. Absolute variance, raw covariance and normalised correlation therefore require separate numerical comparisons.

At a fixed size and depth, marginal variance helps describe which single-mode tasks are easier to learn. Across the pooled conditions, larger variance accompanies faster early decay and lower final MSE [Fig.~\ref{fig:variance-spectral-learning}]. Much less additional information remains in the time-averaged relative loss after accounting for frequency, circuit--encoder combination and phase, and that adjusted association changes sign in the two-qubit cohort [Table~\ref{tab:appendix-variance-adjustment}]. Variance is therefore informative within this benchmark, but these observations do not support a learning rule independent of circuit choice. Duffy and Jastrzebski's frequency-resolved gradient bounds \cite{duffy2026} address a related mechanism, but their small-angle predictions are not tested by our full-period initialisation and Adam protocol. Their active-term redundancy and the encoder-path count also summarise different stages of the coefficient expansion [Appendix~\ref{app:redundancy_dictionary}].

The matched depth sweep instead follows the same supported targets as more trainable blocks are added. Early learning becomes faster in 84 of the 96 retained conditions, and final MSE falls in 90, between depths one and five [Fig.~\ref{fig:matched-depth-size-learning}]. These improvements cannot be attributed to the target first becoming supported. They also coexist with smaller maximum variances, but a circuit's largest variance is not necessarily the variance requested by a given target: only 51 of the 96 marginal target scores decrease between these endpoints. Depth can redistribute variation among existing directions and change how parameter updates move the output. The association between variance and performance across targets at fixed depth therefore does not imply that shallower circuits should learn better.

Adding qubits produces a different pattern on common supported targets. All 58 marginal scores decrease, consistent with the coefficient-variance suppression predicted under the design assumptions of Mhiri \emph{et al.}~\cite{mhiri2025a}. Early learning slows in 50 conditions, while final accuracy remains frequency and circuit dependent [Fig.~\ref{fig:matched-depth-size-learning}]. The theory concerns coefficient fluctuations, while these learning results also depend on the target, optimiser and finite training budget. Variance suppression with size and with depth consequently need not have the same implications for learning.

Mean absolute correlation supplies a second ensemble diagnostic. Our association between lower mean correlation and lower multimode MSE agrees in direction with Strobl \emph{et al.}'s FCC observation \cite{strobl2025}, without reproducing their statistic or protocol [Fig.~\ref{fig:correlation-learning-overview}]. A plausible explanation is that depth and circuit structure change both coefficient overlap and the ease of optimisation. Learning also improves with depth in three circuit--encoder combinations whose mean correlation is already zero and whose target-coordinate support is unchanged [Table~\ref{tab:diagonal-depth-learning}]. Thus decreasing mean correlation is not necessary for the observed depth benefit. The final-MSE association weakens after depth and coverage adjustment, but the early-decay association persists; neither has a stable sign after additionally accounting for circuit--encoder combination. This supports treating mean correlation as a potentially indirect diagnostic, while leaving its independent contribution unresolved.

The matched-target controls ask whether covariance predicts a preference once the marginal score is fixed. They reveal different responses across circuits, but no robust pooled learning advantage for the target favoured by the covariance score [Figs.~\ref{fig:covariance-learning-controls}, \ref{fig:architecture-covariance-heterogeneity} and~\ref{fig:architecture-covariance-heterogeneity-ry}]. Where normalised correlations weaken, the available covariance-score preference should also be limited; the fixed-target score contraction supports this structural prediction. It does not establish a monotonic decay of the actual learning preference. This makes shallow depths a useful place to investigate the proposed bias, provided the required frequency pairs are supported. For some families those pairs first appear at depths two and three. The small number of available pairs and the differences between circuit matrices must still be accounted for when interpreting these comparisons. The present tests establish neither equivalence between the targets nor an absence of useful covariance information in other tasks. The matched-sign test concerns relative target preference, so its inconclusive pooled result does not rule out an influence common to both targets or explain away an across-circuit association. In the separate multimode regressions, the signed covariance contribution adds little to the total explained variation in final log MSE, but at depth one it accounts for about $32\%$ of the variation left unexplained by the marginal score and categorical terms, and about $13\%$ at depth five. This leaves room for useful covariance information in particular settings, while prediction on unseen conditions remains untested.

The tangent results explain how ensemble variation and local response can change differently. Covariance and the averaged QNTK have the same range, but differentiation weights their parameter harmonics differently. In the YZY-entangling example, the increasing ratio $\gamma_\omega=\bar H_{\omega\omega}/v_\omega$ shows greater mean squared sensitivity per unit variance as depth grows [Fig.~\ref{fig:averaged-qntk-geometry}(c)]. Absolute sensitivity is the product of this ratio and the variance, so it can rise in some frequency directions and fall in others. Moreover, the kernel sums derivatives over all parameters, whose number increases with depth. Sensitivity and ensemble variance therefore need not fall together as depth increases, but the ratio alone does not specify how quickly a target will learn.

At individual initialisations, depth also changes which coefficient directions can be adjusted independently. At $n=6$, the median numerical tangent-rank fraction rises from $0.462$ at depth one to unity at depth two, although directional sensitivities remain unequal [Fig.~\ref{fig:qntk-capacity}]. A pointwise kernel can still have smaller rank than the ensemble covariance or averaged kernel, and full median numerical rank does not imply evenly distributed sensitivity. Greater local accessibility and changes in sensitivity are plausible contributors to better learning, but reaching full median rank at depth two cannot explain all later improvements. A more direct dynamical quantity is the pointwise kernel's response along the current residual: Eq.~\eqref{eq:kernels_data_rayleigh} gives the exact instantaneous relative decay rate for Euclidean gradient flow. Testing whether changes in this response explain the depth benefit under Adam would require accounting for its adaptive updates and the evolving residual and Jacobian; that attribution is not established here.

The frozen-QNTK control supplies a more direct, but conditional, prediction. Following Liu \emph{et al.}~\cite{liu2022}, the initial kernel forecasts loss under small-step full-batch gradient descent. The forecasts improve on no change for all 120 trajectories at the final step, with frequency-wise median errors of $0.28\%$--$2.44\%$ of the loss reduction actually achieved [Fig.~\ref{fig:frozen-qntk-validation}]. Each error compares a run with its own forecast, taking the absolute value before aggregation. The agreement therefore demonstrates predictive accuracy for individual initialisations and cannot be explained by cancellation between over- and under-predictions. Forecast error and kernel drift both grow during training, while the initial kernel remains useful over this limited amount of learning. The test concerns one circuit setting and does not establish accuracy after large loss reductions or explain the Adam depth sweeps. It demonstrates the practical distinction between describing an ensemble's spectral structure and predicting a particular run from its initial model and sampled residuals.

\subsection{Scope and limitations}
\label{subsec:limitations}

The CHM representation applies when the joint harmonic expansion is finite. The direct computational method covers the narrower class of Clifford--Pauli circuits and controlled rotations specified in Section~\ref{sec:ensemble_computation}. It averages each angle locally because the angles are independent and used once. Reused angles require their joint dependence to be retained, while nonuniform ensembles require non-diagonal character moments in place of the full-period Gram contraction. The encoder argument has a separate restriction: its joint eigenvalues require commuting generators within each block, so it does not cover arbitrary noncommuting encoders.

Even when the averaging identities are exact, their numerical evaluation can be costly and approximate. The propagation dictionary has an exponential worst-case bound, and the implementation prunes intermediate terms without a systematic study of sensitivity to that pruning [Appendices~\ref{subsec:appendix_numerical_limits} and~\ref{app:twirl_cost}]. The selected explicit-CHM and Monte Carlo checks support the implementation, but numerical support and rank thresholds do not certify exact zeros or the affine-span error bound.

The empirical scope is eight circuit families, two encoder axes, up to six qubits and five trainable blocks, with one scalar input, a fixed observable and noiseless Fourier targets. The periodic grids measure full-period squared error, not generalisation to an unknown data distribution. Arbitrary input sampling can change the rank and conditioning of \(V\); finite-shot noise and other optimisers can change the learning comparison. Moreover, the mean and covariance determine an affine output span, not the complete nonlinear reachable set. Full tangent rank does not guarantee global reachability, and uncorrelated coefficients need not be statistically independent. All kernel comparisons must retain the stated parameter metric and Fourier-coordinate normalisation [Appendix~\ref{app:kernel_real_coordinates}].

The statistical comparisons are also restricted to related conditions from a small, fixed circuit bank. The adjusted models are fitted and assessed on those same conditions, and their categorical terms do not account for every circuit--frequency--phase interaction. The covariance reference tests make additional assumptions of cluster-sign symmetry or residual exchangeability that the experimental design does not ensure [Appendix~\ref{subsec:appendix_covariance_learning}]. Adjusting for multiple comparisons does not remove that dependence on the assumptions. We therefore interpret the results as associations and qualified negative findings within this benchmark, without claiming causal covariance effects or predictive performance on unseen architectures.

\subsection{Outlook}
\label{subsec:outlook}

The variance results motivate testing whether coefficient statistics improve predictions on held-out circuit--target combinations beyond frequency and architecture effects. A comparison with covariance-based architecture heuristics or Fourier fingerprints should first match the coefficient convention, support, ensemble and learning task; the present paired-sign control does not supply that comparison.

A further direction is to study how the geometry inside the covariance span affects learning. Figure~\ref{fig:ensemble-tangent-geometry} gives a simple instance in which the affine span is two-dimensional and the circuit moves along a one-dimensional curve. Characterising analogous constraints in larger circuits could sharpen the affine-span loss bound. Tracking how tangent directions and their sensitivities change during training could also help explain when added depth improves access to a target and when frozen-kernel forecasts remain accurate.

For the QNTK, the next question is how forecast accuracy changes at matched amounts of learning. Resolving frequency-dependent validity would require target-mode loss and projected kernel-drift measurements, rather than grouping whole-kernel drift by target frequency. Variance-based error bounds and high-frequency lazy-training claims remain open. Extending the circuit sizes would also require controlled approximations to two-copy propagation, with pruning errors assessed separately from final support and rank thresholds.

\subsection{Conclusion}
\label{subsec:conclusion}

The circuit harmonic matrix connects which Fourier coefficients a circuit can vary, how they fluctuate together, and how parameter changes affect them locally. It gives exact affine restrictions on the outputs and relates coefficient covariance to tangent response. For the circuit class studied here, the two-copy algorithm evaluates the ensemble moments without constructing the full matrix. Agreement with explicit CHM calculations and independent Monte Carlo estimates on the tested circuits supports this computational approach.

The numerical results show why the distinction between ensemble variation and local response matters for learning. Added depth can broaden coefficient support and bring relative variance profiles closer to redundancy while reducing the largest variances. At a fixed size and depth, directions with larger marginal variance tend to learn faster and reach lower error. Yet the same supported targets generally learn faster and reach lower final MSE when the circuit is made deeper, even as its maximum variance contracts. Increasing qubit count instead suppresses all marginal scores on the common target set, consistent with Mhiri \emph{et al.}'s design-based concentration prediction~\cite{mhiri2025a}, and usually slows early learning; its effect on final accuracy depends on the circuit and frequency. An association across tasks at fixed depth consequently does not determine the effect of changing the architecture.

Covariance adds further structure, with several circuit families becoming less correlated as depth increases and other encoder choices already giving diagonal matrices. Lower mean correlation accompanies better learning across the pooled settings, but depth also improves learning when correlation is already zero and target support is unchanged. The matched-target tests find circuit-dependent responses without establishing a robust pooled advantage for the covariance score. Its contrast contracts with depth more consistently than the observed learning preference. These results leave room for covariance to matter in particular tasks, while showing that decorrelation alone does not account for the benefits of depth.

The kernel experiments identify changes in local accessibility and sensitivity that could contribute to those benefits. More coefficient directions become available to first order, although their sensitivities remain unequal. In the YZY example, sensitivity grows relative to variance, while its absolute value can rise or fall at different frequencies. These observations do not yet establish the cause of the Adam improvements. For an individual run, however, the initial kernel and sampled residual give accurate local loss forecasts in the separate small-step gradient-descent experiment, despite growing forecast error and kernel drift. Ensemble moments describe the variations favoured by an architecture; learning also depends on the local response to the residual and the optimiser.

\section*{Code Availability}
{\raggedright The code associated with the results of this paper is available at
\url{https://github.com/quantumsoftwarelab/Circuit_Harmonic_Matrices}.
\par}

\begin{acknowledgments}

We thank Mario Herrero-Gonz\'alez for insightful conversations.

KJSC was supported by the Engineering and Physical Sciences Research Council (grant number EP/W524384/1) and the University of Edinburgh. PW was supported by EPSRC grants EP/X026167/1 and EP/Z53318X/1. LDD is supported by an STFC Consolidated Grant (ST/T000600/1, ST/X000494/1).
\end{acknowledgments}

\newpage

\bibliography{references}

\clearpage
\appendix

\makeatletter
\let\addcontentsline\orig@addcontentsline
\makeatother

\onecolumngrid
\counterwithin*{equation}{section}
\renewcommand\theequation{\thesection\arabic{equation}}

\section{Notation and coordinate conventions}
\label{appendix:notation}

The tables collect the circuit, ensemble and tangent-kernel notation. Appendices~\ref{appendix:numerical_methods}--\ref{appendix:learning_results} specify the experiments, their validation and supporting results; Appendices~\ref{appendix:qfms}--\ref{appendix:twirl} give the derivations and algorithms. Physical angles, parameter characters and real Fourier coefficients use distinct coordinate conventions, as indicated below.

\begin{table}[!htp]
\centering
\footnotesize
\renewcommand{\arraystretch}{0.94}
\caption{Notation summary for this paper. Preliminaries are discussed in Section \ref{sec:prelims}, circuit harmonic matrices in Section \ref{sec:interaction_matrix}, coefficient variances and correlations in Section \ref{sec:coeff_stats}, and numerical results in Section \ref{sec:numerical-results}.}
\label{tab:notation_prelims}
\vspace{1mm}
\begin{chmkeytable}
\hline\hline
\chmtablerow{\textbf{Symbol}}{\textbf{Meaning / definition}}\\
\hline
\chmtablerow{$x\in\mathbb{R}^{d_x}$}{Input variable; $d_x$ is the input dimension.}\\
\chmtablerow{$\theta\in\mathbb{T}^m$}{Trainable character-coordinate vector on the $m$-torus; $m$ is the number of trainable parameters.}\\[2pt]
\chmtablerow{$\vartheta_r,T_r,\theta_r$}{Physical gate angle, its full output period, and the rescaled $2\pi$ character coordinate $\theta_r=2\pi\vartheta_r/T_r$, cf.\ Eq.~\eqref{eq:prelims_period_rescaling}.}\\[2pt]
\chmtablerow{$\rho_0,\ O$}{Reference state (also denoted $\rho$ in the appendices) and measured observable defining $f(x;\theta)=\Tr[\rho_0\,U^\dagger O U]$.}\\[2pt]
\chmtablerow{$U(x;\theta)$}{Full circuit unitary with encoder blocks and trainable blocks.}\\[2pt]
\chmtablerow{$L$}{Number of encoding layers. With one encoding gate per qubit in each layer, the total number of encoding gates is $nL$.}\\[2pt]
\chmtablerow{$S_\ell(x)$}{Complete encoding layer $\ell$, with $\ell=1,\ldots,L$, cf.\ Eq.~\eqref{eq:prelims_encoder_diag}.}\\[2pt]
\chmtablerow{$W_\ell(\theta^{(\ell)})$}{Trainable block following encoding layer $\ell$, cf.\ Eq.~\eqref{eq:prelims_reupload}; its Clifford--Pauli canonical form is given in Eq.~\eqref{eq:app_pp_circuit}.}\\[2pt]
\chmtablerow{$f(x;\theta) \in \mathbb{R}$}{Model output.}\\[2pt]
\chmtablerow{$\Omega$}{Encoder-accessible input frequency set in the QFM expansion, cf.\ Eq.~\eqref{eq:prelims_qfm_expansion}.}\\[2pt]
\chmtablerow{$\Omega^{(\ell)}$}{Encoder difference set at layer $\ell$, cf.\ Eq.~\eqref{eq:prelims_diffset}; $\Omega=\bigoplus_\ell \Omega^{(\ell)}$.}\\[2pt]
\chmtablerow{$a_\omega(\theta)$}{Trainable Fourier coefficient for input harmonic $\omega$, cf.\ Eq.~\eqref{eq:prelims_qfm_expansion}.}\\[2pt]
\chmtablerow{$a(\theta)\in\mathbb{C}^{|\Omega|}$}{Vector of coefficients with entries $a_\omega(\theta)$.}\\[2pt]
\chmtablerow{$\mathcal{K}\subset\mathbb{Z}^m$}{Finite parameter-harmonic index set on $\mathbb{T}^m$; includes the realised support but can contain zero-coefficient harmonics, cf.\ Eq.~\eqref{eq:prelims_param_fourier_generic}.}\\[2pt]
\chmtablerow{$\omega,\mu;\ k,l$}{Generic input-frequency vectors in $\Omega$ and parameter-frequency vectors in $\mathcal K$, respectively. Subscripts such as $k_r$ denote vector components.}\\[2pt]
\chmtablerow{$\omega^{(a)},\ k^{(b)}$}{Enumerated whole input- and parameter-frequency vectors, respectively; used for selected rows and columns in schematics.}\\[2pt]
\chmtablerow{$\psi_k(\theta)$}{Parameter character $\psi_k(\theta)=e^{ik\cdot\theta}$, cf.\ Eq.~\eqref{eq:prelims_param_fourier_generic}.}\\[2pt]
\chmtablerow{$\psi(\theta)\in\mathbb{C}^{|\mathcal{K}|}$}{Vector of characters with entries $\psi_k(\theta)$.}\\[4pt]
\chmtablerow{$C\in\mathbb{C}^{|\Omega|\times|\mathcal{K}|}$}{Fixed, generally rectangular circuit harmonic matrix in $a(\theta)=C\psi(\theta)$, cf.\ Eq.~\eqref{eq:C_aCpsi}.}\\[2pt]
\chmtablerow{$C_{\omega k}$}{Joint input--parameter Fourier coefficient (row $\omega$, column $k$ of $C$), cf.\ Eqs.~\eqref{eq:C_param_series} and \eqref{eq:C_joint_harmonic}.}\\[2pt]
\chmtablerow{$p\in\mathcal P$}{Complete root-to-leaf Pauli-propagation path, determining $P_p$, $\sigma_p$, $t_p^{(x)}$, $t_p^{(\theta)}$, and $\alpha_p=\sigma_p\Tr(\rho P_p)$.}\\[2pt]
\chmtablerow{$\mathcal{R}(\omega)$}{Set of ordered encoder eigenvalue-index-path pairs $(J,J')$ with accumulated frequency $\omega$, retaining degeneracy; $|\mathcal{R}(\omega)|$ is the frequency redundancy, cf.\ Eq.~\eqref{eq:app_Romega_def}.}\\[2pt]
\chmtablerow{$\mathcal{R}_{\mathrm{eig}}(\omega)$}{Alternative label for $\mathcal{R}(\omega)$ when contrasting encoder multiplicity with circuit-resolved term counts; used in the comparison with Mhiri and Duffy--Jastrzebski.}\\[2pt]
\chmtablerow{$\mathcal{R}_{\mathrm{term}}(\omega)$}{Set of active trigonometric-expansion multi-indices contributing to $\omega$ in the Duffy--Jastrzebski representation; its cardinality counts terms before Fourier cancellations.}\\[2pt]
\chmtablerow{$c_{(J,J'),k}$}{Parameter-harmonic contribution of an encoder index-path pair, so $C_{\omega k}=\sum_{(J,J')\in \mathcal{R}(\omega)}c_{(J,J'),k}$.}\\[2pt]
\chmtablerow{$d_{\mathcal H}=2^n$}{Hilbert-space dimension.}\\[2pt]
\chmtablerow{$\widetilde{\mathcal{R}}(\omega)$}{Normalised redundancy $|\mathcal{R}(\omega)|/d_{\mathcal H}^2$ used in the comparison with Mhiri \emph{et al.}, cf.\ Eq.~\eqref{eq:mhiri_chm_exact_design}.}\\[2pt]
\hline\hline
\end{chmkeytable}
\end{table}

\FloatBarrier

\begin{table}[!htp]
\centering
\footnotesize
\renewcommand{\arraystretch}{0.94}
\caption{Notation summary (continued).}
\label{tab:notation_prelims-cont}
\vspace{1mm}
\begin{chmkeytable}
\hline\hline
\chmtablerow{\textbf{Symbol}}{\textbf{Meaning / definition}}\\
\hline
\chmtablerow{$\mathbb{E}_\theta[\cdot]$}{Uniform average over $\theta\sim\mathrm{Unif}(\mathbb{T}^m)$, cf.\ Eq.~\eqref{eq:prelims_Etheta_def}.}\\[2pt]
\chmtablerow{$\mathbb{E}_{\mathrm{data}}[\cdot]$}{Empirical average over the training inputs $\{x_i\}_{i=1}^N$, cf.\ Eq.~\eqref{eq:prelims_Edata_def}.}\\[2pt]
\chmtablerow{$\mathrm{Cov}_\theta(\cdot,\cdot)$}{Covariance under $\theta$-averaging, cf.\ Eq.~\eqref{eq:prelims_cov_def}.}\\[2pt]
\chmtablerow{$\tilde a(\theta)$}{Centred coefficient vector $\tilde a(\theta)=a(\theta)-\mathbb{E}_\theta[a(\theta)]$, cf.\ Eq.~\eqref{eq:coeffstats_centered_a}.}\\[2pt]
\chmtablerow{$X,\ J$}{Coefficient and data Jacobians, with $J=VX$; initially $X=\partial_\theta a$, then $X=\partial_\vartheta a$ in learning equations and numerical QNTKs.}\\[2pt]
\chmtablerow{$q_{\mathrm{cs}},\ T_{\mathrm{cs}}$}{Independent real cosine--sine coefficients and their map to complex coefficients: $a=T_{\mathrm{cs}}q_{\mathrm{cs}}$, $a_{\pm\omega}=(c_\omega\mp i s_\omega)/2$.}\\[2pt]
\chmtablerow{$\Sigma_{\mathbb R},\ H_{\mathbb R}$}{Covariance and tangent kernel in the unscaled real cosine--sine basis, distinct from the complex-basis matrices $\Sigma,H$.}\\[2pt]
\chmtablerow{$R_{\omega\mu}$}{Normalised Hermitian coefficient correlation $\Sigma_{\omega\mu}/\sqrt{v_\omega v_\mu}$ on nonzero-variance rows; distinct from the redundancy set $\mathcal{R}(\omega)$.}\\[2pt]
\chmtablerow{$\mathcal S$}{Numerical variance support after the stated tolerance.}\\[2pt]
\chmtablerow{$D,\ r_{\mathrm{num}},\ r_{\mathrm{eff}}$}{Number of real Fourier features, thresholded numerical kernel rank, and participation-ratio effective rank; the last depends on basis normalisation.}\\[2pt]
\chmtablerow{$V$}{Fourier design matrix, $V_{i\omega}=e^{i \omega\cdot x_i}$, evaluating the frequency basis on the finite input set, cf.\ Eq.~\eqref{eq:kernels_F_Va}.}\\[2pt]
\chmtablerow{$K$}{Data-space QNTK, $K=VHV^\dagger$, at the stated parameter point and in the same parameter metric as $H$, cf.\ Eq.~\eqref{eq:kernels_K_VHV}.}\\[2pt]
\chmtablerow{$K_{\mathrm{fr}}$}{Pointwise data-space QNTK held fixed at a chosen initialisation or checkpoint for the frozen-kernel approximation, cf.\ Eq.~\eqref{eq:kernels_frozen_residual_solution}.}\\[2pt]
\chmtablerow{$H$}{Coefficient-space QNTK, $H=XX^\dagger=CMC^\dagger$, using the stated parameter coordinates, cf.\ Eq.~\eqref{eq:kernels_H_definition}.}\\[2pt]
\chmtablerow{$\bar H$}{Uniform full-period parameter average $\mathbb E[H]$ in the stated parameter metric.}\\[2pt]
\chmtablerow{$\bar K$}{Ensemble-mean data-space QNTK, $V\bar H V^\dagger=\mathbb E[K]$.}\\[2pt]
\chmtablerow{$\gamma_\omega$}{Ratio $\bar H_{\omega\omega}/v_\omega$ on the nonzero-variance support; mean squared parameter sensitivity relative to coefficient variation, cf.\ Eq.~\eqref{eq:kernels_sensitivity_variation_ratio}.}\\[2pt]
\chmtablerow{$M(\theta)$}{Character-gradient Gram matrix, cf.\ Eq.~\eqref{eq:kernels_M_definition}.}\\[2pt]
\chmtablerow{$C^\circ$}{Nonconstant parameter-harmonic submatrix of $C$, retaining only columns with $k\neq0$, cf.\ Eq.~\eqref{eq:coeffstats_P_def}.}\\[2pt]
\chmtablerow{$v_\omega$}{Coefficient variance $v_\omega=\mathrm{Var}_\theta[a_\omega]$ under the stated parameter ensemble, cf.\ Eq.~\eqref{eq:coeffstats_variance}.}\\[2pt]
\chmtablerow{$\Omega_+$}{Exact positive-variance set $\{\omega\in\Omega:v_\omega>0\}$ on which complex correlation is defined, cf.\ Eq.~\eqref{eq:coeffstats_corr_D}.}\\[2pt]
\chmtablerow{$\bar a$}{Coefficient mean $\bar a=\mathbb{E}_\theta[a(\theta)]=C_{\bullet 0}$.}\\[2pt]
\chmtablerow{$S$}{Uncentred second-moment matrix $S_{\omega\mu}=\mathbb{E}_\theta[a_\omega\overline{a_\mu}]$, cf.\ Appendix~\ref{appendix:twirl}.}\\[2pt]
\chmtablerow{$\Sigma$}{Centred coefficient covariance $\Sigma=S-\bar a\,\bar a^\dagger=C^\circ C^{\circ\dagger}$, cf.\ Eq.~\eqref{eq:coeffstats_cov_CPCdag}.}\\[2pt]
\chmtablerow{$(Q_1,Q_2,\omega,\nu)\mapsto\Gamma$}{Entry of the two-copy pair dictionary in the torus twirl: Pauli-string pair, input-frequency labels, and accumulated weight, cf.\ Appendix~\ref{appendix:twirl}.}\\[2pt]
\hline\hline
\end{chmkeytable}
\end{table}

\FloatBarrier
\section{Experimental protocol and numerical validation}
\label{appendix:numerical_details}
\label{appendix:numerical_methods}

This appendix gives the experimental settings and the checks used to validate the calculations in the main text and Appendices~\ref{appendix:coefficient_results} and~\ref{appendix:learning_results}.

\subsection{Complete experimental specification}\label{subsec:appendix_protocol}

Table~\ref{tab:full-experimental-specification} records the exact-statistics sweep, and Table~\ref{tab:learning-experimental-specification} records the learning and initial-QNTK protocols. Circuit-specific gate sequences and parameter counts are shown in Figure~\ref{fig:appendix-circuit-families}. The 16 circuit--encoder combinations are the eight families shown there, each with \(R_X\) and \(R_Y\) encoding. The complete model encodes the fixed input state once and then applies its trainable blocks.

\begin{table}[!ht]
  \centering
  \caption{Experimental settings for the full exact-statistics sweep.}
  \label{tab:full-experimental-specification}
  \begin{chmkeytable}
    \toprule
    \chmtablerow{Quantity}{Setting}\\
    \midrule
    \chmtablerow{Input state}{$\ket{0}^{\otimes n}$}\\
    \chmtablerow{Observable}{$O_Z=n^{-1}\sum_j Z_j$}\\
    \chmtablerow{Input dimension}{one scalar input $x$}\\
    \chmtablerow{Encoders}{$R_X(x)^{\otimes n}$ and $R_Y(x)^{\otimes n}$}\\
    \chmtablerow{Encoder scale}{$1$}\\
    \chmtablerow{Encoding layers}{$L=1$}\\
    \chmtablerow{Qubits}{$n=2,3,4,5,6$}\\
    \chmtablerow{Trainable depths}{$d=1,2,3,4,5$}\\
    \chmtablerow{Parameter ensemble}{independent full-period coordinates: $[0,2\pi)$
      for Pauli rotations and $[0,4\pi)$ for controlled rotations}\\
    \chmtablerow{Integration in saved runs}{$[0,4\pi)$ for all angles; this covers
      two equivalent periods for local Pauli expectations}\\
    \chmtablerow{Intermediate pruning}{absolute propagation tolerance \(10^{-14}\)}\\
    \chmtablerow{Numerical support tolerances}{$\tau_{\mathrm{abs}}=10^{-14}$,
      $\tau_{\mathrm{rel}}=10^{-12}$}\\
    \chmtablerow{Numerically diagonal correlation}{maximum supported
      off-diagonal magnitude at most \(10^{-10}\)}\\
    \chmtablerow{Aggregate runs}{$400$ exact runs; $200$ complete RX--RY pairs}\\
    \bottomrule
  \end{chmkeytable}
\end{table}

\begin{table}[!ht]
  \centering
  \caption{Learning and pointwise-QNTK settings. The full-batch Adam bank and the frozen-kernel gradient-descent control are separate experiments.}
  \label{tab:learning-experimental-specification}
  \begin{chmkeytable}
    \toprule
    \chmtablerow{Quantity}{Setting}\\
    \midrule
    \chmtablerow{Learning cohorts}{\(n=6,\ d=1,\ldots,5\), and
      \(d=5,\ n=2,\ldots,5\): nine cohorts, each with 16 identities}\\
    \chmtablerow{Optimiser}{full-batch Adam, learning rate \(0.02\),
      \(\beta_1=0.9,\ \beta_2=0.999,\ \epsilon=10^{-8}\), 500 steps}\\
    \chmtablerow{Input grids}{64 training and 256 evaluation points, uniform on
      \([0,2\pi)\), endpoint excluded}\\
    \chmtablerow{Initialisation}{independent uniform physical-angle coordinates
      with periods \(2\pi\) for local and \(4\pi\) for controlled rotations}\\
    \chmtablerow{Training seeds}{0--9, matched across targets for each circuit}\\
    \chmtablerow{Saved histories}{every five steps, including steps 0 and 500}\\
    \chmtablerow{Single-mode targets}{\(0.8\cos(\omega x+\phi)\),
      \(\omega=1,\ldots,n\), \(\phi\in\{0,\pi/2\}\)}\\
    \chmtablerow{Multi-mode targets}{equal-amplitude sums on \(1,\ldots,n\) and
      \(1,\ldots,10\); three fixed phase draws per band (seed 271828),
      each rescaled to peak magnitude \(0.8\)}\\
    \chmtablerow{Two-mode targets}{all same-parity pairs in \(1,\ldots,n\);
      four quadrature combinations, both relative signs, fixed amplitudes
      balanced using pooled pilot variances; details in
      Appendix~\ref{subsec:appendix_covariance_learning}}\\
    \chmtablerow{Learning-bank size}{7,008 circuit--target--cohort conditions,
      ten seeds each: 70,080 trajectories}\\
    \chmtablerow{Initial QNTKs}{50 per identity and cohort (seeds 0--49),
      including the training seeds: 7,200 kernels}\\
    \chmtablerow{Reported resampling analyses}{2,000 bootstrap replicates or
      residual permutations; matched-sign reference tests enumerate all
      16 or 128 circuit--encoder cluster-sign patterns;
      assumptions in Appendix~\ref{subsec:appendix_covariance_learning}}\\
    \chmtablerow{Frozen-QNTK control}{YZY-entangling/\(R_X\), \(n=6,d=3\);
      full-batch gradient descent, learning rate \(0.002\), 500 steps;
      the same single-mode targets, ten seeds, and 64/256 grids}\\
    \bottomrule
  \end{chmkeytable}
\end{table}

\paragraph{Loss and learning summaries.}
The implementation minimises the unhalved mean-squared error
\begin{equation}
  \mathcal L_s^{\mathrm{train}}
  =\frac{1}{N}\sum_{i=1}^{N}
  \bigl[f(x_i;\vartheta_s)-y(x_i)\bigr]^2,
  \qquad N=64.
  \label{eq:appendix_training_mse}
\end{equation}
This is twice the normalisation in Eq.~\eqref{eq:kernels_empirical_mse}; the frozen-GD update therefore uses \(I-2\eta K_s/N\). The learning summaries use the analogous evaluation-grid MSE, denoted \(\mathcal L_s^{\mathrm{eval}}\). Write \(u_s=\mathcal L_s^{\mathrm{eval}}/\mathcal L_0^{\mathrm{eval}}\). The reported early decay rate is minus the ordinary least-squares slope of \(\log u_s\) against step number at \(s=0,5,10,15,20,25\). It measures how quickly the loss falls during the first 25 Adam steps. Before taking logarithms, we floor the MSE at \(10^{-300}\). Subtracting the initial log loss then gives \(\log u_s\) without changing the fitted slope. We also report the final relative loss \(u_{500}\) and the normalised area under the learning curve (AULC), computed as the trapezoidal time average
\begin{equation}
  \operatorname{AULC}
  =\frac{1}{500}\sum_{j=0}^{99}
  \frac{5}{2}\bigl(u_{5j}+u_{5(j+1)}\bigr).
  \label{eq:appendix_normalised_aulc}
\end{equation}
The paired-sign control uses the same trapezoidal average before division by the initial loss, giving the time-averaged MSE \(\overline{\mathcal L}\). For each condition, we summarise learning by the median over the ten seeds. In a matched-sign comparison, however, we first calculate the normalised difference in Eq.~\eqref{eq:learning-loss-asymmetry} within each seed, preserving the pairing between the two targets, and then summarise those differences.

Interquartile bands show the spread across the indicated seeds, conditions or circuit identities. Bootstrap intervals give percentile confidence intervals for the stated statistic, with the resampling unit specified for each analysis.

\paragraph{What the evaluation grid measures.}
The evaluation grid measures the full-period error in a noiseless Fourier-regression problem. The model bandwidth is at most six and the target bandwidth at most ten, so the squared residual contains only frequencies of magnitude at most 20. Both periodic grids integrate its mean exactly in ideal arithmetic, because neither aliases a nonzero squared-residual frequency to the constant term. The denser grid thus provides a check of the same error and resolves its Fourier components. It does not test generalisation to an unknown data distribution.

\begin{figure}[!htbp]
  \centering
  \begingroup
\tikzset{
  chmwire/.style={line width=0.45pt},
  chmgate/.style={draw,rounded corners=0.8pt,fill=darkorange!8,
    minimum height=12pt,inner xsep=2.2pt,inner ysep=1.2pt},
  chmcrgate/.style={draw,rounded corners=0.8pt,fill=darkblue!7,
    minimum width=16pt,minimum height=12pt,inner sep=1pt},
  chmlink/.style={line width=0.5pt}
}
\newcommand{\chmbase}[1]{%
  \foreach \chmq in {0,1,2,3}{%
    \pgfmathsetmacro{\chmy}{-0.42*\chmq}%
    \node[anchor=east,font=\scriptsize] at (-0.08,\chmy) {$q_{\chmq}$};%
    \draw[chmwire] (0,\chmy) -- (#1,\chmy);%
  }%
}
\newcommand{\chmgateat}[3]{%
  \pgfmathsetmacro{\chmy}{-0.42*(#2)}%
  \node[chmgate] at (#1,\chmy) {$\scriptstyle #3$};%
}
\newcommand{\chmcnot}[3]{%
  \pgfmathsetmacro{\chmyc}{-0.42*(#2)}%
  \pgfmathsetmacro{\chmyt}{-0.42*(#3)}%
  \draw[chmlink] (#1,\chmyc) -- (#1,\chmyt);%
  \fill (#1,\chmyc) circle (1.35pt);%
  \node[inner sep=0pt,font=\small] at (#1,\chmyt) {$\oplus$};%
}
\newcommand{\chmcontrolledrot}[4]{%
  \pgfmathsetmacro{\chmyc}{-0.42*(#2)}%
  \pgfmathsetmacro{\chmyt}{-0.42*(#3)}%
  \draw[chmlink] (#1,\chmyc) -- (#1,\chmyt);%
  \fill (#1,\chmyc) circle (1.35pt);%
  \node[chmcrgate] at (#1,\chmyt) {$\scriptstyle R_{#4}$};%
}

\centering
\small
\begin{tikzpicture}[baseline=(current bounding box.center)]
  \node (in) {$\ket{0}^{\otimes n}$};
  \node[draw,rounded corners,fill=darkblue!7,right=5mm of in,
        inner xsep=5pt,inner ysep=3pt] (enc)
        {$S_\alpha(x)=R_\alpha(x)^{\otimes n}$};
  \node[draw,rounded corners,fill=darkorange!9,right=5mm of enc,
        inner xsep=5pt,inner ysep=3pt] (train)
        {$W(\vartheta^{(1)})\to\cdots\to W(\vartheta^{(d)})$};
  \node[right=5mm of train] (out) {$\langle O_Z\rangle$};
  \draw[-{Latex[length=2mm]}] (in) -- (enc);
  \draw[-{Latex[length=2mm]}] (enc) -- (train);
  \draw[-{Latex[length=2mm]}] (train) -- (out);
  \node[below=1mm of enc,font=\scriptsize] {$\alpha\in\{X,Y\}$};
  \node[below=1mm of train,font=\scriptsize]
    {independent parameters in every repetition};
\end{tikzpicture}

\vspace{2mm}

\begin{minipage}[t]{0.31\textwidth}
  \centering
  \textbf{YZY}\par\smallskip
  \begin{tikzpicture}
    \chmbase{2.25}
    \foreach \chmq in {0,1,2,3}{\chmgateat{1.10}{\chmq}{R_Y\!\to R_Z\!\to R_Y}}
  \end{tikzpicture}
  \par\smallskip{$3n$ parameters per block.\par}
\end{minipage}
\hfill
\begin{minipage}[t]{0.31\textwidth}
  \centering
  \textbf{YZY--entangling}\par\smallskip
  \begin{tikzpicture}
    \chmbase{4.75}
    \foreach \chmq in {0,1,2,3}{\chmgateat{0.95}{\chmq}{R_Y\!\to R_Z\!\to R_Y}}
    \chmcnot{2.00}{0}{1}
    \chmcnot{2.45}{0}{2}
    \chmcnot{2.90}{0}{3}
    \chmcnot{3.35}{1}{2}
    \chmcnot{3.80}{1}{3}
    \chmcnot{4.25}{2}{3}
  \end{tikzpicture}
  \par\smallskip{$3n$ parameters per block.\par}
\end{minipage}
\hfill
\begin{minipage}[t]{0.31\textwidth}
  \centering
  \textbf{HEA}\par\smallskip
  \begin{tikzpicture}
    \chmbase{3.75}
    \foreach \chmq in {0,1,2,3}{\chmgateat{0.95}{\chmq}{R_Y\!\to R_Z\!\to R_Y}}
    \chmcnot{2.00}{0}{1}
    \chmcnot{2.00}{2}{3}
    \chmcnot{2.60}{1}{2}
    \chmcnot{3.20}{3}{0}
  \end{tikzpicture}
  \par\smallskip{$3n$ parameters per block.\par}
\end{minipage}

\par\medskip

\begin{minipage}[t]{0.31\textwidth}
  \centering
  \textbf{Circuit 15}\par\smallskip
  \begin{tikzpicture}
    \chmbase{4.75}
    \foreach \chmq in {0,1,2,3}{\chmgateat{0.43}{\chmq}{R_Y}}
    \chmcnot{1.05}{3}{0}
    \chmcnot{1.40}{2}{3}
    \chmcnot{1.75}{1}{2}
    \chmcnot{2.10}{0}{1}
    \foreach \chmq in {0,1,2,3}{\chmgateat{2.73}{\chmq}{R_Y}}
    \chmcnot{3.35}{3}{2}
    \chmcnot{3.70}{0}{3}
    \chmcnot{4.05}{1}{0}
    \chmcnot{4.40}{2}{1}
  \end{tikzpicture}
  \par\smallskip{$2n$ parameters per block.\par}
\end{minipage}
\hfill
\begin{minipage}[t]{0.31\textwidth}
  \centering
  \textbf{Circuit 16}\par\smallskip
  \begin{tikzpicture}
    \chmbase{3.95}
    \foreach \chmq in {0,1,2,3}{\chmgateat{0.68}{\chmq}{R_X\!\to R_Z}}
    \chmcontrolledrot{1.65}{1}{0}{Z}
    \chmcontrolledrot{2.55}{3}{2}{Z}
    \chmcontrolledrot{3.45}{2}{1}{Z}
  \end{tikzpicture}
  \par\smallskip{$3n-1$ parameters per block.\par}
\end{minipage}
\hfill
\begin{minipage}[t]{0.31\textwidth}
  \centering
  \textbf{Circuit 17}\par\smallskip
  \begin{tikzpicture}
    \chmbase{3.95}
    \foreach \chmq in {0,1,2,3}{\chmgateat{0.68}{\chmq}{R_X\!\to R_Z}}
    \chmcontrolledrot{1.65}{1}{0}{X}
    \chmcontrolledrot{2.55}{3}{2}{X}
    \chmcontrolledrot{3.45}{2}{1}{X}
  \end{tikzpicture}
  \par\smallskip{$3n-1$ parameters per block.\par}
\end{minipage}

\par\medskip

\begin{minipage}[t]{0.31\textwidth}
  \centering
  \textbf{Circuit 18}\par\smallskip
  \begin{tikzpicture}
    \chmbase{4.85}
    \foreach \chmq in {0,1,2,3}{\chmgateat{0.62}{\chmq}{R_X\!\to R_Z}}
    \chmcontrolledrot{1.85}{3}{0}{Z}
    \chmcontrolledrot{2.70}{2}{3}{Z}
    \chmcontrolledrot{3.55}{1}{2}{Z}
    \chmcontrolledrot{4.40}{0}{1}{Z}
  \end{tikzpicture}
  \par\smallskip{$3n$ parameters per block.\par}
\end{minipage}
\hfill
\begin{minipage}[t]{0.31\textwidth}
  \centering
  \textbf{Circuit 19}\par\smallskip
  \begin{tikzpicture}
    \chmbase{4.85}
    \foreach \chmq in {0,1,2,3}{\chmgateat{0.62}{\chmq}{R_X\!\to R_Z}}
    \chmcontrolledrot{1.85}{3}{0}{X}
    \chmcontrolledrot{2.70}{2}{3}{X}
    \chmcontrolledrot{3.55}{1}{2}{X}
    \chmcontrolledrot{4.40}{0}{1}{X}
  \end{tikzpicture}
  \par\smallskip{$3n$ parameters per block.\par}
\end{minipage}
\hfill
\begin{minipage}[t]{0.31\textwidth}
  \centering
  \textbf{Controlled-pair order}\par\smallskip
  \begin{minipage}[t][1.75cm][c]{0.96\linewidth}
    \centering\scriptsize
    C16/C17: $(1,0),(3,2),(2,1)$\par\smallskip
    C18/C19: $(3,0),(2,3),(1,2),(0,1)$
  \end{minipage}
  \par\smallskip{Pairs are (control, target) for the four-wire exemplar.\par}
\end{minipage}
\endgroup
  \caption{Circuit families and parameter counts. The eight circuit panels show an $n=4$ exemplar of one trainable block $W$ for each family in the aggregate sweep; the complete model applies one $R_X$ or $R_Y$ encoder and then repeats that block $d$ times, so its trainable parameter count is $d$ times the per-block count shown. Arrows inside a composite gate box denote temporal order from left to right. The YZY--entangling panel uses the fixed triangular all-to-all sequence $\operatorname{CNOT}(i\!\to\!j)$ for $i<j$. HEA uses the displayed CNOT ring. Circuit~15 uses two ordered CNOT rings with opposite directions, separated by an $R_Y$ layer. Circuits~16--19 controlled gates are applied from left to right, with wires ordered $q_0,\ldots,q_3$ from top to bottom.  The aggregate circuit choices follow the benchmark family of Strobl \emph{et al.}~\cite{strobl2025}.}
\label{fig:appendix-circuit-families}
\end{figure}

\FloatBarrier

\subsection{Checks of exact coefficient statistics}
\label{subsec:appendix_validation}

Appendix~\ref{appendix:twirl} proves the contraction in exact arithmetic; the comparisons here test its floating-point implementation. For a nonzero reference matrix $A_{\mathrm{ref}}$, reported relative Frobenius errors use
\begin{equation}
 \epsilon_A=
 \frac{\|A_{\mathrm{test}}-A_{\mathrm{ref}}\|_{\mathrm F}}
      {\|A_{\mathrm{ref}}\|_{\mathrm F}}.
 \label{eq:appendix_validation_relative_error}
\end{equation}
The explicit-CHM calculation, sampled circuit coefficients and sampled coefficient Jacobians test different routes to the same moments.

\begin{table}[!ht]
  \centering
  \caption{Direct validation of the ordinary and derivative-weighted two-copy contractions for the $R_X$-encoded YZY-entangling circuit.  The explicit columns of $C$ give $\Sigma=C^\circ C^{\circ\dagger}$ and $\bar H=C\,\mathrm{diag}(\|k\|_2^2)C^\dagger$.  The first two errors compare the two-copy output with these explicit contractions.  The final column compares $\bar H$ with an independent average of Fourier-coefficient Jacobian Gram matrices from $16\,384$ uniform parameter samples.}
  \label{tab:explicit-chm-validation}
  \begin{tabular}{@{}cccccc@{}}
    \toprule
    $n$ & $d$ & $m$ & $\epsilon_{\Sigma}^{\mathrm{CHM}}$
    & $\epsilon_{\bar H}^{\mathrm{CHM}}$
    & $\epsilon_{\bar H}^{\mathrm{MC}}$ \\
    \midrule
    $2$ & $1$ & $6$  & $0$ & $0$ & $1.25\times10^{-2}$ \\
    $2$ & $2$ & $12$ & $0$ & $0$ & $9.11\times10^{-3}$ \\
    $3$ & $1$ & $9$  & $3.02\times10^{-16}$
        & $4.38\times10^{-16}$ & $7.62\times10^{-3}$ \\
    \bottomrule
  \end{tabular}
\end{table}

The two checks in Table~\ref{tab:explicit-chm-validation} reach the ensemble moments by different routes. The first constructs the CHM explicitly and contracts its columns, giving a sampling-free reference for the two-copy accumulators. The second uses automatic differentiation of circuit outputs on a 64-point input grid, Fourier transforms each Jacobian, and averages the resulting coefficient-space Gram matrices. The explicit-CHM discrepancies are at floating-point precision, whereas the Monte Carlo discrepancies also contain finite-sample error.

Controlled rotations require a further check because their physical period differs from that of local rotations. We test the differentiated controlled-$R_X$ local map on Circuit~19 at $(n,d)=(2,1)$ over the full $[0,4\pi)$ period of its controlled angles.  The independently sampled Fourier-Jacobian average has relative Frobenius error $8.56\times10^{-3}$ and normalised Frobenius alignment $0.999963$ with the exact full-torus kernel.  This independently checks both the controlled-rotation period handling and the derivative weight in the averaged-QNTK contraction.

Figure~\ref{fig:exact-mc-validation} checks the larger YZY-entangling example at $n=6$. The relative variance and covariance errors are below $0.7\%$ at every tested depth. This matrix-level agreement does not imply that sampling resolves every small covariance: the weakest off-diagonal entries can be obscured by finite-sample fluctuations. The final panel shows the finite-sample correlations at pairs classified as covariance zeros by the exact computation. The exact support is fixed before either comparison.

\begin{figure}[!htbp]
  \centering
  \includegraphics[width=.88\textwidth]{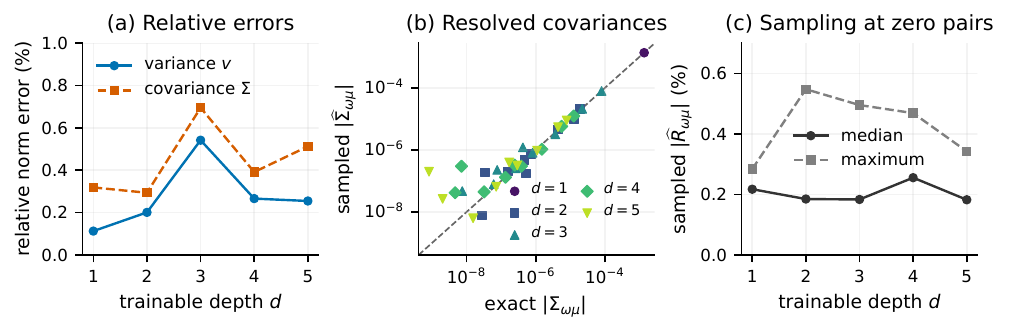}
  \caption{\justifying{Exact--Monte Carlo validation for the $R_X$-encoded YZY-entangling circuit at $n=6$ and $d=1,\ldots,5$, using 200,000 parameter samples per setting. (a) Relative vector-$2$-norm and matrix-Frobenius errors on the supported nonnegative-frequency grid, including the covariance diagonal. (b) Magnitudes of the resolved off-diagonal covariances, counting each pair $\omega<\mu$ once; the dashed line denotes exact agreement. (c) Median and maximum sampled correlation magnitudes over pairs whose exact covariance magnitude is at most $\max(10^{-14},10^{-12}\max_{\omega,\mu}|\Sigma_{\omega\mu}|)$. There are two such pairs at $d=1$ and twelve at each later depth. Sampled correlations use sampled marginal variances, restricted throughout to the fixed exact variance support.}}
  \label{fig:exact-mc-validation}
\end{figure}

\begin{figure}[!htbp]
  \centering
  \includegraphics[width=0.50\textwidth]
  {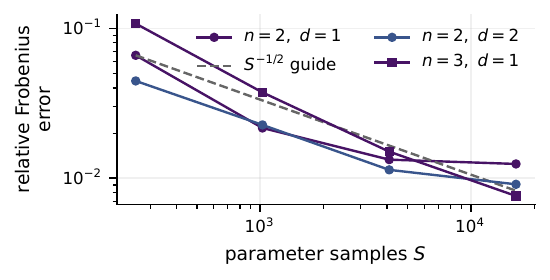}
  \caption{\justifying{Validation of the derivative-weighted two-copy contraction by an independent sampled Fourier-Jacobian average. For the three YZY-entangling/$R_X$ cases in Table~\ref{tab:explicit-chm-validation}, the relative Frobenius discrepancy from the exact $\bar H$ decreases as the sample count increases. Colours denote depth, as in the main text; circles and squares distinguish $n=2$ and $n=3$. Points on each curve use nested subsets of the same draws. The dashed $S^{-1/2}$ guide is anchored to the first $n=2,d=1$ point. The table separately reports the sampling-free explicit-CHM check.}}
  \label{fig:qntk-contraction-validation}
\end{figure}

\FloatBarrier

\subsection{Data-space consistency and rank thresholds}
\label{subsec:appendix_kernel_checks}

For each saved training initialisation and target, we compare the data-space scalar
\begin{equation}
 \frac{r_0^{\mathsf T}K_0r_0}{N\,r_0^{\mathsf T}r_0},
 \qquad r_0=F(\vartheta_0)-y,\qquad N=64,
 \label{eq:appendix_coordinate_diagnostic}
\end{equation}
with its Fourier-coordinate evaluation. Multiplying this scalar by $4\eta$ gives the initial fractional loss-decay rate under Euclidean gradient flow on the unhalved MSE, with learning-rate factor $\eta$. The Fourier evaluation includes all model and target frequencies and uses the cosine--sine Gram weights $1,\tfrac12,\tfrac12,\ldots$ from Appendix~\ref{app:kernel_real_coordinates}. Target-only modes have zero model derivative but still contribute to the residual norm. Table~\ref{tab:qntk-coordinate-checks} resolves the agreement by target family. Its 70,080 comparisons reuse the ten training initialisations per condition. The separate rank study uses 50 initialisations. Agreement here verifies a scalar consequence of $K=VHV^\dagger$, including the coordinate normalisation.

\begin{table}[!htbp]
  \centering
  \caption{Coordinate-consistency checks on all saved target--initialisation comparisons. The relative discrepancy is the absolute difference between the two evaluations of Eq.~\eqref{eq:appendix_coordinate_diagnostic} divided by the larger magnitude; all compared values are positive.}
  \label{tab:qntk-coordinate-checks}
  \setlength{\tabcolsep}{8pt}
  \renewcommand{\arraystretch}{1.15}
  \begin{tabular}{@{}lrc@{}}
    \toprule
    Target family & Comparisons & Maximum relative discrepancy\\
    \midrule
    Single-mode & 14,080 & $9.42\times10^{-14}$\\
    Multi-mode & 8,640 & $2.03\times10^{-14}$\\
    Matched-sign two-mode & 47,360 & $1.08\times10^{-13}$\\
    \bottomrule
  \end{tabular}
\end{table}

To check numerical-rank sensitivity, we vary the relative eigenvalue cutoff over $10^{-6}$--$10^{-14}$ while retaining the documented absolute floor $10^{-13}$.  The qualitative separation between shallow rank deficiency and the full median rank of deeper blocks is stable across this range.

\begin{figure}[!htbp]
  \centering
  \includegraphics[width=0.55\textwidth]
  {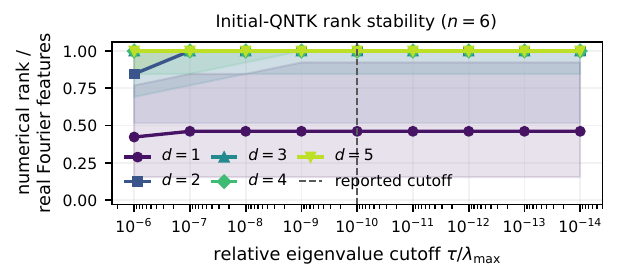}
  \caption{\justifying{Sensitivity of the initial-QNTK numerical-rank fraction to the relative eigenvalue cutoff at $n=6$.  Lines and bands show medians and interquartile ranges over circuit--encoder identities and 50 initialisations; the dashed line is the $10^{-10}$ cutoff used in the main text.  An absolute floor of $10^{-13}$ is retained throughout.}}
  \label{fig:qntk-rank-threshold-sensitivity}
\end{figure}

\FloatBarrier

\subsection{Numerical support and limits of validation}
\label{subsec:appendix_numerical_limits}

\paragraph{Propagation and final classification.}
Intermediate pruning at $10^{-14}$ discards small contributions during propagation, including controlled-map entries and accumulated weights. The algebraic contraction is exact before this numerical approximation; the threshold does not provide a bound on the accumulated error. A systematic intermediate-pruning sensitivity study has not been performed.

Final support and rank thresholds serve different purposes. Equation~\eqref{eq:support-threshold} classifies a coefficient's ensemble variance, while the eigenvalue cutoff in Fig.~\ref{fig:qntk-rank-threshold-sensitivity} classifies local tangent directions. Neither operation proves an exact algebraic zero, and varying these final cutoffs does not test intermediate-pruning error.

\paragraph{Correlation and target-score diagnostics.}
Correlations are formed only after fixing variance support, because normalising noise in an unsupported direction can produce misleadingly large values. The maximum-entry criterion in Table~\ref{tab:full-experimental-specification} classifies a correlation block as numerically diagonal. The conditional summary in Fig.~\ref{fig:aggregate-architecture-trends}(e) instead retains circuit settings whose \emph{mean} off-diagonal magnitude exceeds $10^{-10}$. Both summaries use supported nonnegative frequencies, including zero when supported, and require at least two such modes. The matched-sign comparison uses a third criterion, applied to the target-weighted cross-score difference as specified in Appendix~\ref{subsec:appendix_covariance_learning}. A circuit can qualify as a diagonal control for that comparison when its selected contrasts are negligible, even if its full covariance matrix is not algebraically diagonal.

All period-dependent checks use full physical periods; the controlled-rotation local rule is given in Appendix~\ref{app:twirl_controlled}. The separate frozen-QNTK experiment in Appendix~\ref{subsec:appendix_frozen_qntk} tests the accuracy of learning forecasts based on the initial kernel.

\FloatBarrier
\section{Supporting coefficient-structure results}
\label{appendix:coefficient_results}

This appendix shows how the coefficient structure changes across the full numerical sweep. We begin with the $R_X$-encoded YZY-entangling example at $n=2,\ldots,6$ and $d=1,\ldots,5$, extending the selected results in the main text to every size and depth. We then compare the variance profiles and correlation matrices of all circuits and encoders at six qubits, before examining the effect of size on fixed sets of coefficients. Appendix~\ref{appendix:learning_results} tests how these structural observations relate to learning.

\subsection{Support and marginal variation}
\label{subsec:appendix_yzy_dashboards}

Figure~\ref{fig:appendix-yzy-support} identifies the actual frequencies behind the support counts in Fig.~\ref{fig:yzy-exact-overview}. At depth one the numerical support is $|\omega|\leq2$; from depth two onwards every encoder-allowed frequency is supported in this tested grid. Grey frequencies are outside the encoder spectrum and cannot be activated by increasing trainable depth. Hatched frequencies are encoder-allowed but below the variance-support threshold for that circuit.

The nonentangling control has a stronger, structural restriction. Locality gives
\begin{equation}
    f(x;\theta)=\frac1n\sum_{j=1}^n f_j(x;\theta_j).
    \label{eq:nonentangling-decomposition}
\end{equation}
For the initial state $\ket{0}^{\otimes n}$ and one $R_X$ or $R_Y$ encoding on each qubit, each local Bloch vector contains only sine and cosine terms in $x$. Subsequent local trainable gates rotate this vector without introducing any other dependence on $x$. Each local $Z$ expectation therefore has the form \(f_j(x;\theta_j)=A_j(\theta_j)\cos x+B_j(\theta_j)\sin x\). Their sum contains neither a constant term nor higher harmonics, at any trainable depth:
\begin{equation}
    a_\omega(\theta)\equiv0\qquad\text{for }\omega\notin\{-1,1\}.
    \label{eq:nonentangling-support}
\end{equation}
In particular, $v_0=0$, and the zero-frequency row and column of the covariance matrix vanish. The saved coefficient means and variances obey this restriction at every tested size and depth for both encoding axes. The YZY row of Fig.~\ref{fig:variance-profile-atlas} shows the zero-frequency variance alongside the other coefficients for both encoders.

\begin{figure}[!htbp]
  \centering
  \includegraphics[width=.72\textwidth]{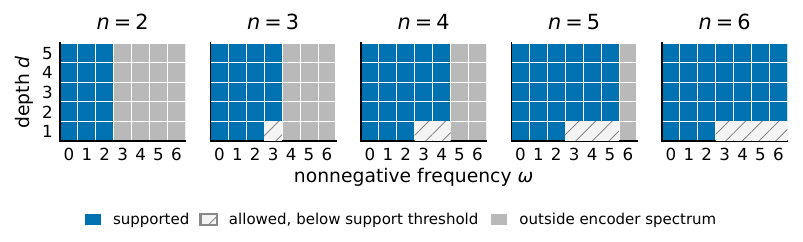}
  \caption{\justifying{Numerical variance support for YZY-entangling with $R_X$ encoding. Blue cells pass the threshold in Eq.~\eqref{eq:support-threshold}; hatched cells are encoder-allowed but fail this threshold; grey cells lie outside the encoder spectrum. Rows give depth and columns within each panel give frequency. Only nonnegative frequencies are needed because $v_{-\omega}=v_\omega$. Support is computed from the two-copy variances.}}
  \label{fig:appendix-yzy-support}
\end{figure}

Figure~\ref{fig:appendix-yzy-variance}(a)--(e) shows the absolute variance of every supported nonnegative-frequency coefficient across all 25 settings. The common vertical scale reveals that a newly supported high-frequency coefficient can have far less variance than the lower-frequency coefficients. Increasing depth also need not suppress each variance monotonically: a mode can first become active and then gain variance.

\begin{figure}[!htbp]
  \centering
  \includegraphics[width=.90\textwidth]{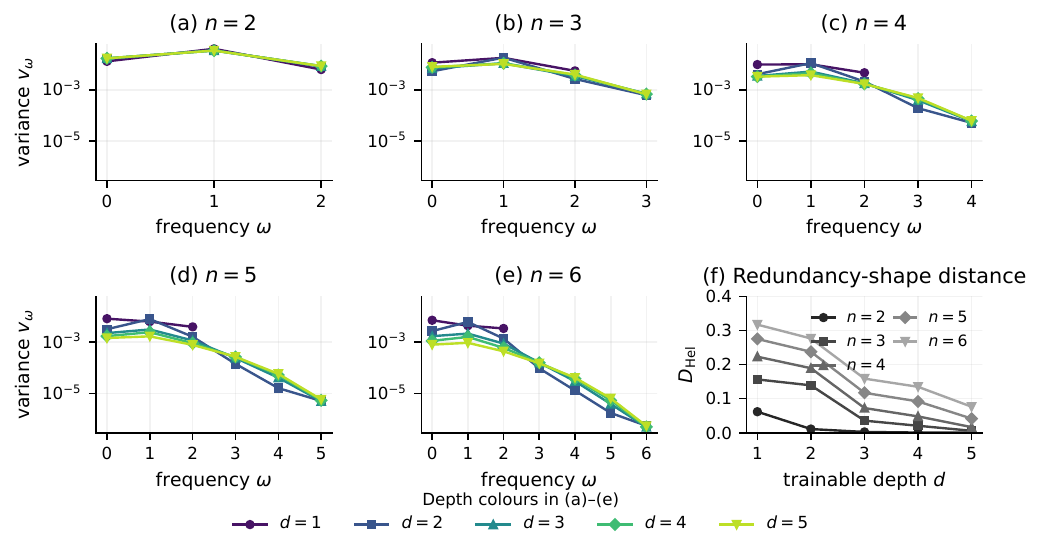}
  \caption{\justifying{Absolute coefficient variation and redundancy-shape agreement for YZY-entangling/$R_X$. (a)--(e) All five depths at each size, including the constant mode, on a common logarithmic variance scale. Markers are the discrete frequencies; lines are visual guides. Unsupported coefficients are omitted; their locations are given in Fig.~\ref{fig:appendix-yzy-support}. (f) The main-text Hellinger distance [Eq.~\eqref{eq:numerical-hellinger-distance}] at every size and depth. This shape metric uses all positive encoder-allowed frequencies, including zero variances, and excludes the constant mode. Depth colours apply to (a)--(e); symbols and grey shades in (f) distinguish system sizes.}}
  \label{fig:appendix-yzy-variance}
\end{figure}

\paragraph{Agreement with the redundancy profile.}
Panel (f) uses the same normalisation and distance as the main text. For $n=6$, the depth-one through depth-five distances are $0.317$, $0.276$, $0.159$, $0.135$, and $0.0763$. The approach to the reference shape is metric dependent. For example, the alternative relative least-squares residual
\begin{equation}
    D_{\mathrm{var}}
    =\min_{\alpha\geq0}
      \frac{\|v-\alpha(|\mathcal{R}(\omega)|)_\omega\|_2}{\|v\|_2},
    \label{eq:variance-profile-distance}
\end{equation}
on all positive encoder-allowed frequencies gives $0.260$, $0.391$, $0.239$, $0.242$, and $0.151$ on that same track. Both distances show a closer endpoint, but the latter is nonmonotonic. The apparent smoothness of the approach therefore depends on the chosen metric. Appendix~\ref{app:fixed_input_design_reference} derives the reference profile for our fixed input state and explains its relation to the design-based calculation of Ref.~\cite{mhiri2025a}.

\subsection{Variance profiles across circuit families}
\label{subsec:appendix_variance_profiles}

Figure~\ref{fig:variance-profile-atlas} compares the variance of each Fourier coefficient across all eight circuit families and both encoding axes at $n=6$ and $L=1$. Each row follows one circuit across trainable depths one to five, with the two encoding axes compared in each panel. The shared scale allows direct comparison of their absolute variances across circuits, while the separate zero rows identify coefficients that vanish in the saved calculation.

Nonentangling YZY has nonzero variance only at $\omega=1$ in the displayed half-spectrum; its constant coefficient and all higher frequencies remain absent at every depth, as the argument above shows. Circuits~16 and~18 also have zero variance at $\omega=0$ at depth one, but this coefficient becomes available at greater depths. The profiles therefore reveal which individual coefficients contribute to the largest and mean variances summarised in Fig.~\ref{fig:fixed-support-structure-depth}.

\begin{figure}[p]
  \centering
  \includegraphics[width=.96\textwidth]{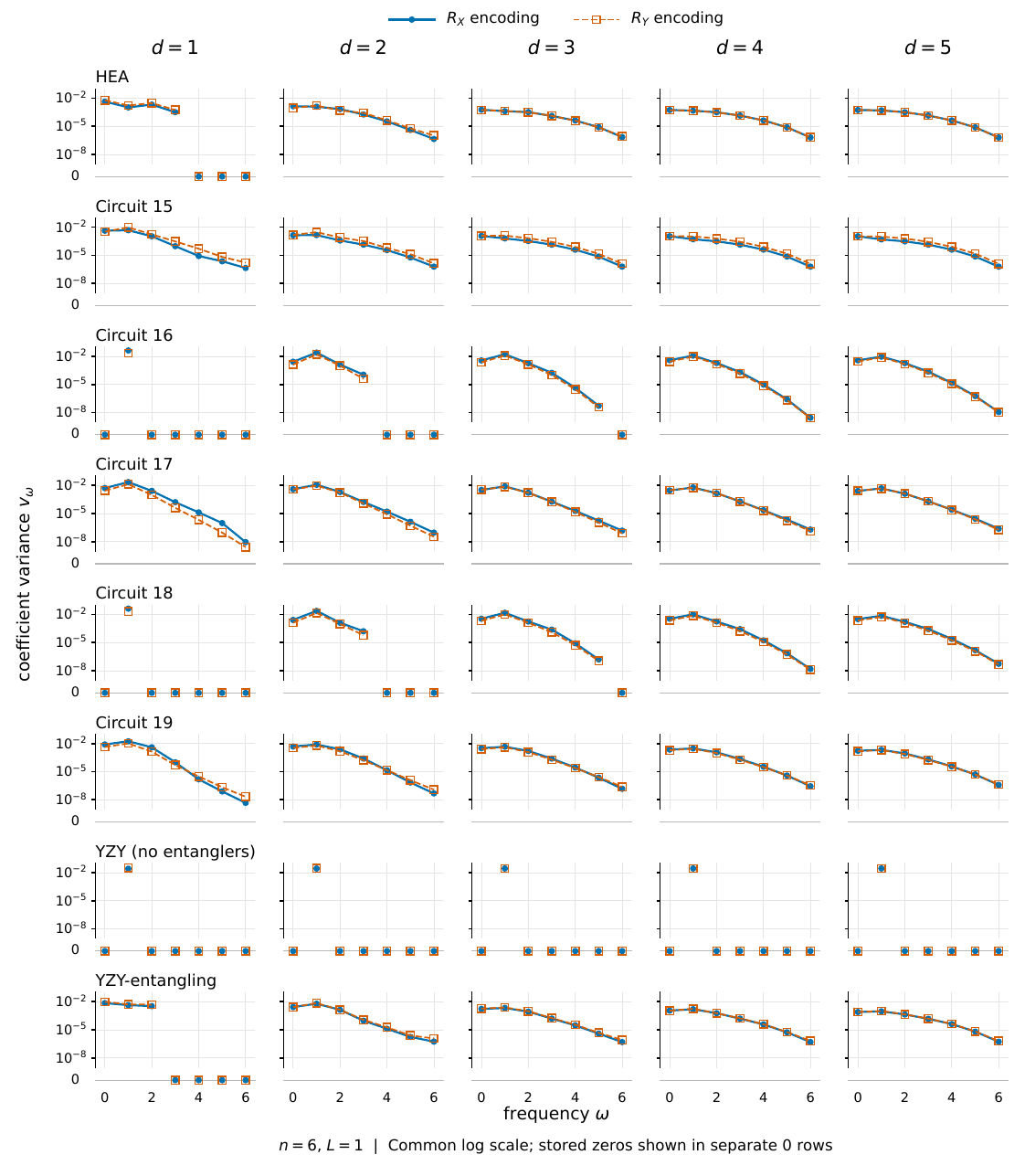}
  \caption{\justifying{Coefficient variance profiles for all circuit families at $n=6$, $L=1$. Rows identify the circuit and columns the trainable depth $d=1,\ldots,5$. Blue solid lines with filled circles denote $R_X$ encoding; orange dashed lines with open squares denote $R_Y$ encoding. Open squares leave the blue circles visible where profiles coincide. Positive variances share a logarithmic scale and retain their absolute scale. The separate rows labelled $0$ show stored zeros. Frequencies include $\omega=0$; the negative half-spectrum follows from $v_{-\omega}=v_\omega$. Lines guide the eye between adjacent positive values.}}
  \label{fig:variance-profile-atlas}
\end{figure}

\subsection{Cross-frequency correlations}

Figure~\ref{fig:appendix-yzy-correlation} extends the five-depth $n=6$ strip in Fig.~\ref{fig:yzy-exact-overview} to all 25 size--depth settings, using the same colour scale and masks. Normalisation by $\sqrt{v_\omega v_\mu}$ on the fixed variance support removes the marginal scales.

Resolved off-diagonal entries connect modes of the same parity. These patterns persist while their magnitudes generally become smaller at greater depth, as the common colour scale reveals. The unit diagonal is omitted and shown in grey. White cells are supported pairs below the plotting tolerance; hatched cells contain an unsupported coefficient, so their correlation is not defined.

\begin{figure}[p]
  \centering
  \includegraphics[width=\textwidth]{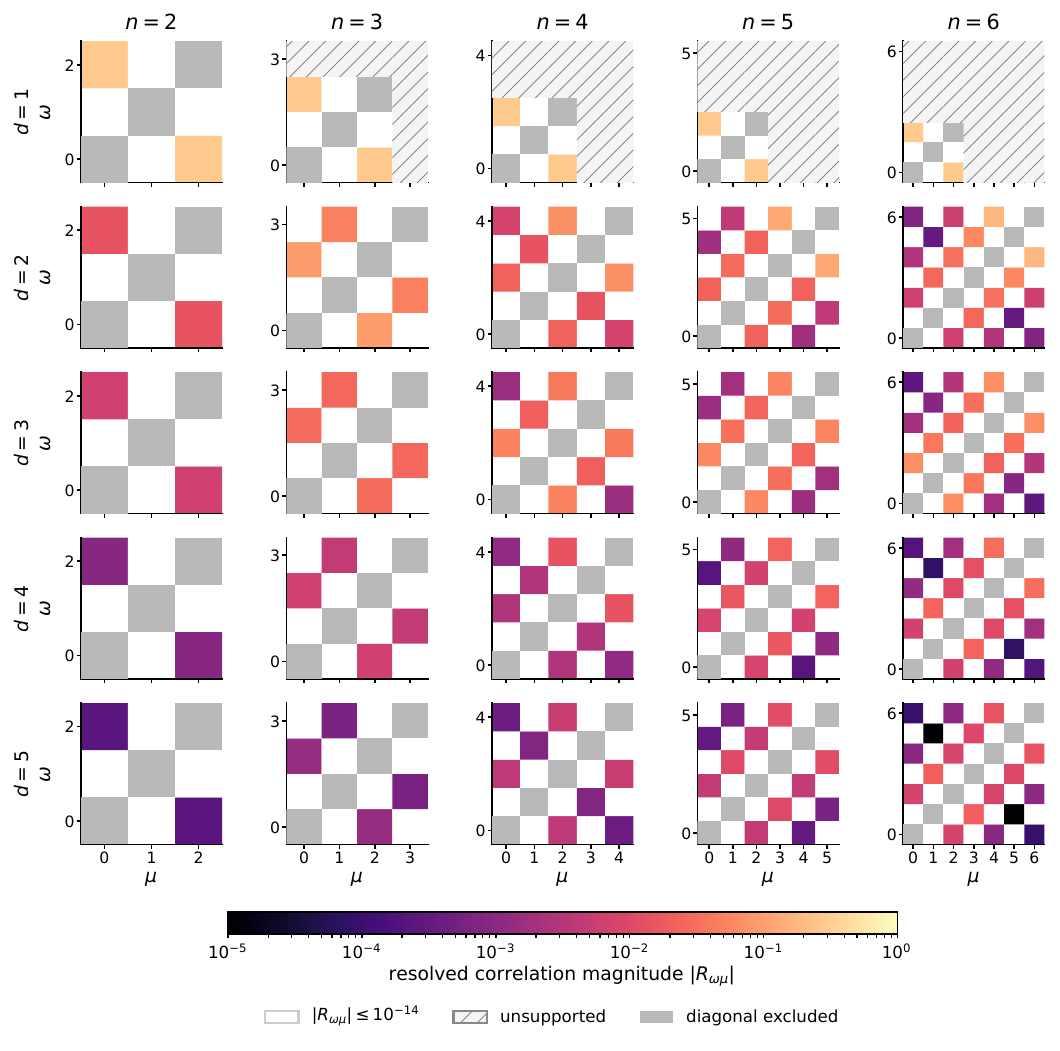}
  \caption{\justifying{Complete exact correlation-magnitude atlas for YZY-entangling/$R_X$. Columns give qubit number and rows give depth; each matrix runs over $\omega,\mu=0,\ldots,n$. All coloured cells use the same logarithmic scale ($10^{-5}$ to $1$), which contains every resolved magnitude in this sweep. White denotes supported off-diagonal pairs with $|R_{\omega\mu}|\leq10^{-14}$. Diagonal hatching marks unsupported frequencies, whose correlations are not defined; solid grey marks the omitted unit diagonal on supported frequencies.}}
  \label{fig:appendix-yzy-correlation}
\end{figure}

\subsection{Circuit-resolved correlation matrices across depth}
\label{subsec:appendix_correlation_depth}

Figures~\ref{fig:correlation-depth-atlas-rx} and~\ref{fig:correlation-depth-atlas-ry} extend the depth comparison to all 16 circuit--encoder identities at $n=6$. To relate these matrices to the learning targets, we use independent real coefficients rather than the complex nonnegative-frequency coefficients of the preceding atlas. The basis map in Appendix~\ref{app:kernel_real_coordinates} converts the full signed complex covariance to $(c_1,s_1,\ldots,c_6,s_6)$, excluding the constant coefficient. This reveals whether cosine and sine targets at a given frequency are both supported: a frequency may be present even when one of its two quadratures is missing.

For display, we place all cosine coordinates before all sine coordinates so that the corresponding blocks are easier to see. Correlations are normalised in this real basis. A coordinate is supported when its variance exceeds $\max(10^{-14},10^{-12}\lambda_{\max}^{\mathrm{block}})$, where $\lambda_{\max}^{\mathrm{block}}$ is the largest eigenvalue among that matrix's $2\times2$ cosine--sine frequency blocks. Correlation is undefined for a pair containing an unsupported coordinate.

The normalised matrices reveal two features that are absent from the fixed-set absolute moment summaries in Fig.~\ref{fig:fixed-support-structure-depth}: dependence after removing marginal scales, and the expansion of support with depth. For example, Circuits~16 and~18 support only frequency one at depth one. There is then no cross-frequency pair to compare, so an empty off-diagonal block at that depth cannot be interpreted as covariance having decayed to zero. The supported matrices are diagonal throughout the sweep for Circuit~15/$R_Y$, Circuits~16--19/$R_X$, HEA/$R_Y$ and YZY-entangling/$R_Y$. Nonentangling YZY is also diagonal on its supported real block, but has no two-frequency target pair.

The maximum-correlation statistic discussed in the main text tests whether even the strongest eligible cross-frequency pair weakens. On the fixed $d=2,\ldots,5$ coordinate sets, an RMS over all eligible cross-frequency entries also decreases at the endpoints in all seven coupled identities. For Circuit~18/$R_Y$, however, the RMS falls only from $0.0506$ to $0.0480$, even as the maximum increases. Circuit~16/$R_Y$ also shows a modest RMS decrease, from $0.0384$ to $0.0301$. The different summaries therefore agree that several families decorrelate strongly with depth, but they do not imply that every off-diagonal entry decays. All 400 saved exact matrices are included in the accompanying numerical audit; these two figures display the 80 matrices at the size used for the covariance-learning depth sweep.

\begin{figure}[p]
  \centering
  \includegraphics[width=\textwidth]{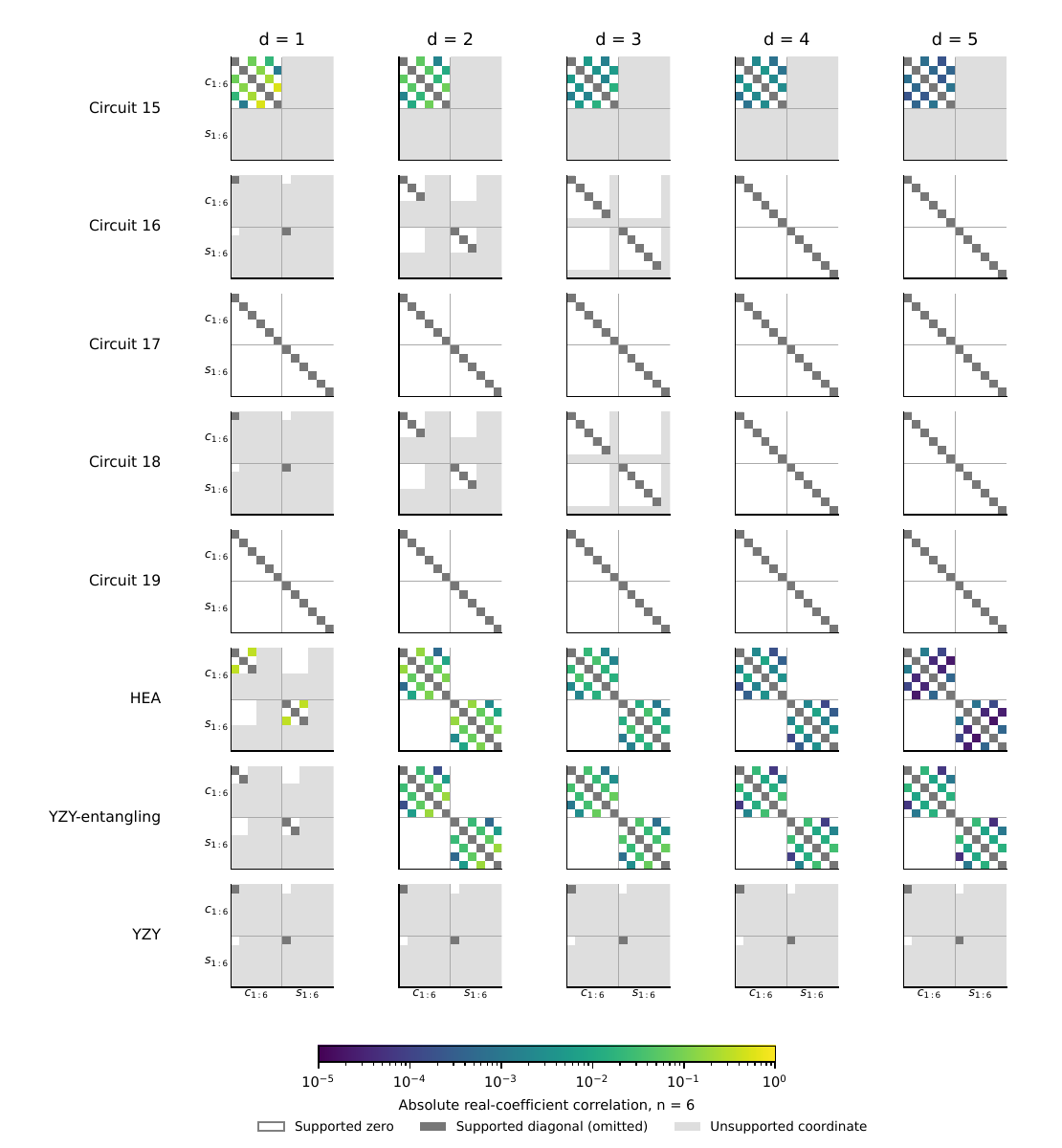}
  \caption{\justifying{Correlation matrices across depth for all eight $R_X$-encoded circuit families at $n=6$. Columns give depths one through five. Each matrix uses real coordinates $(c_1,\ldots,c_6,s_1,\ldots,s_6)$, increasing within each group from left to right and top to bottom. All panels share the logarithmic magnitude scale $10^{-5}$ to $1$, containing every resolved nonzero off-diagonal magnitude. White denotes supported off-diagonal entries of magnitude at most $10^{-14}$; dark grey omits the supported unit diagonal; light grey marks unsupported coordinates. Cosine and sine coordinates are separated to expose the available target phases, whereas Fig.~\ref{fig:appendix-yzy-correlation} uses complex coefficients to display the frequency structure. The matrices expose changing support, already-diagonal controls and decorrelation within the coupled families.}}
  \label{fig:correlation-depth-atlas-rx}
\end{figure}

\begin{figure}[p]
  \centering
  \includegraphics[width=\textwidth]{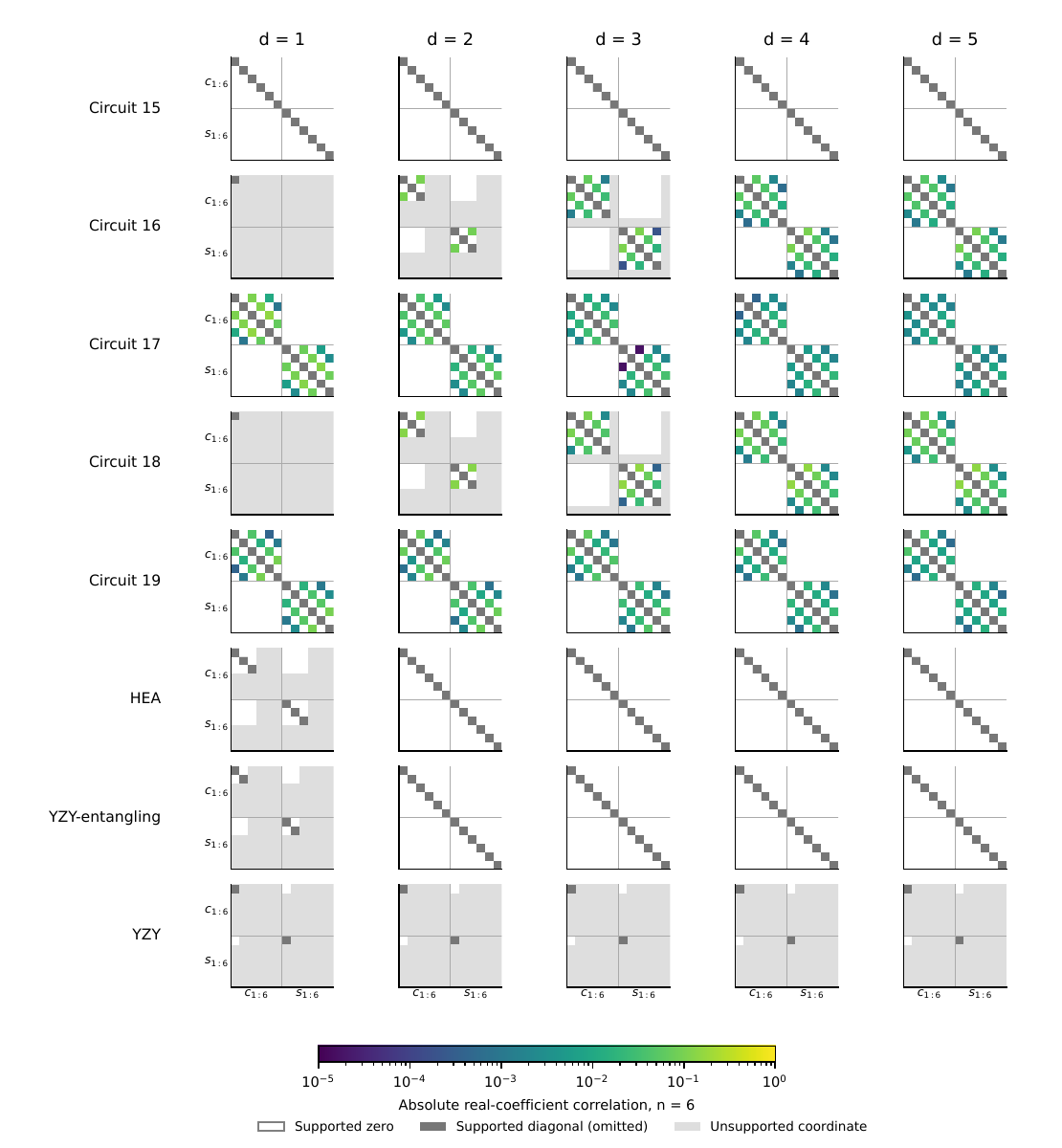}
  \caption{\justifying{Correlation matrices across depth for all eight $R_Y$-encoded circuit families at $n=6$, using exactly the coordinates, colour scale and masks of Fig.~\ref{fig:correlation-depth-atlas-rx}. Circuit~17 and Circuit~19 show weakening correlations, while Circuit~16 and Circuit~18 retain visible coupling at depth five. Circuit~15, HEA and YZY-entangling have diagonal supported matrices. Nonentangling YZY has only the frequency-one block.}}
  \label{fig:correlation-depth-atlas-ry}
\end{figure}

\FloatBarrier
\subsection{Size comparison on common coefficient directions}
\label{subsec:appendix_correlation_size}

Increasing $n$ changes the encoder envelope, so averaging correlations over every available pair can change both the values and the pair population. Figure~\ref{fig:fixed-support-structure-size} instead fixes the candidate frequencies to $1,2,3$ and intersects their directional support over $n=3,\ldots,6$ at $d=5$, separately within each identity. The lower limit $n=3$ retains a common same-parity frequency pair; the $n=2$ envelope lacks the $(1,3)$ comparison.

All 16 circuit--encoder combinations have smaller endpoint maximum and mean variances on these sets. This size dependence is consistent with the coefficient-variance concentration derived by Mhiri \emph{et al.}~\cite{mhiri2025a} under global unitary-design assumptions; Appendix~\ref{app:design_scope} explains the scope of that comparison. The covariance trends are less uniform. Both the maximum and mean absolute off-diagonal covariance decrease in Circuit~15/$R_X$ and Circuit~17/$R_Y$, but increase in the other five coupled identities. The colour-bar values show that off-diagonal fluctuations remain smaller than marginal variances even where they increase with size.

Normalising the same matrices gives a smaller endpoint maximum correlation in Circuit~15/$R_X$ and Circuit~17/$R_Y$; the other five coupled identities have larger ones. RMS correlations give the same endpoint directions. For example, Circuit~19/$R_Y$ has maximum correlation $0.000621$ at $n=3$ and $0.0412$ at $n=6$, while YZY-entangling/$R_X$ changes from $0.000612$ to $0.0231$. The depth-dependent decorrelation of those families therefore cannot be extended to increasing qubit count.

\begin{figure}[!htbp]
  \centering
  \includegraphics[width=\textwidth]{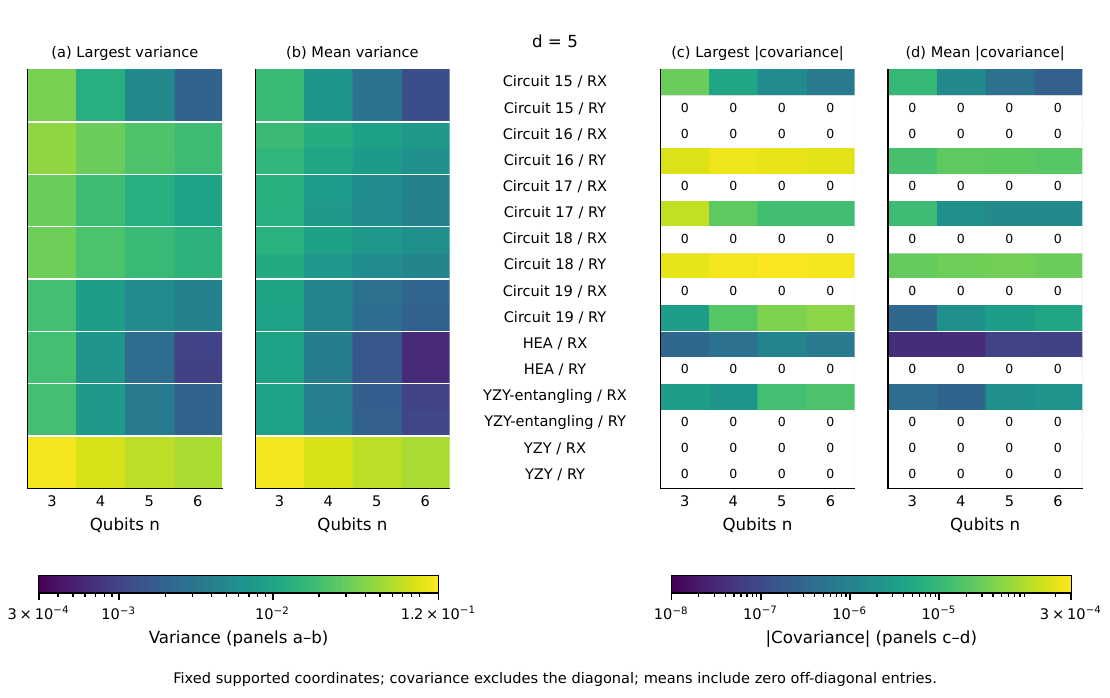}
  \caption{\justifying{Variance and covariance magnitudes across qubit count, with separate logarithmic scales for the variance and covariance pairs. At fixed $d=5$, each identity is compared on its common supported real coordinates at frequencies $1,2,3$ across $n=3,\ldots,6$. Panels (a,b) give the largest and mean variance; (c,d) give the largest and mean absolute off-diagonal covariance. Means include zero off-diagonal entries. Circuit--encoder labels are shared in the centre. The statistics, two colour scales and zero convention match Fig.~\ref{fig:fixed-support-structure-depth}.}}
  \label{fig:fixed-support-structure-size}
\end{figure}

\FloatBarrier
\section{Supporting learning results}
\label{appendix:learning_results}

This appendix develops the learning comparisons summarised in the main text. We first examine the association between marginal variance and learning, using an adjusted regression and trajectories for individual frequencies. We then ask whether covariance provides further information, both across circuits and in comparisons between targets with the same marginal score. These experiments use Adam and describe associations within the benchmark. Finally, a separate gradient-descent experiment tests how well the initial QNTK predicts subsequent loss when held fixed. All experiments follow the protocol in Appendix~\ref{subsec:appendix_protocol}.

\subsection{Variance association after categorical adjustment}
\label{subsec:appendix_variance_adjustment}

The association in Fig.~\ref{fig:variance-spectral-learning} could partly reflect differences between frequencies or between circuits. We therefore ask whether marginal variance remains associated with learning after accounting for these differences. We fit the single-mode conditions separately in each size--depth cohort. A condition \(i\) specifies a circuit--encoder identity, frequency \(\omega=1,\ldots,n\), and target phase \(\phi\in\{0,\pi/2\}\) in \(0.8\cos(\omega x+\phi)\). We retain numerically supported frequencies with \(A_{\mathrm{diag}}>\max(10^{-14},10^{-12}A_{\max})\), where \(A_{\max}\) is the largest single-mode score within that identity and cohort. This also excludes negligible-variance quadratures of a supported frequency. All retained conditions have complete ten-seed results and finite positive loss summaries. At \((n,d)=(6,3)\), the 158 retained conditions coincide with Fig.~\ref{fig:variance-spectral-learning}(a,b).

Let \(U_i\) be the median over seeds of the normalised AULC in Eq.~\eqref{eq:appendix_normalised_aulc}: each seed's relative-loss curve is integrated before taking the median. Lower \(U_i\) means less accumulated error relative to initial loss. We centre \(\log_{10}A_{\mathrm{diag}}(t_i)\) and divide by its population standard deviation over retained conditions to obtain \(z_i\). We then fit ordinary least squares:
\begin{equation}
 \log_{10}U_i=\beta_0+\beta_v z_i
 +\alpha_{\mathrm{identity}(i)}
 +\gamma_{\omega(i)}+\delta_{\phi(i)}+\epsilon_i.
 \label{eq:appendix_variance_adjustment}
\end{equation}
The categorical terms allow a separate offset for each identity, frequency and phase, with a reference level for each variable. Each condition contributes one equally weighted row; its ten seeds are already summarised by the median. Since only the predictor is standardised, \(\beta_v\) measures the change in fitted log AULC associated with a one-standard-deviation increase in log score within that cohort, after accounting for those offsets. The standard deviation can differ between cohorts, so the coefficients do not represent a common physical increment in variance.

\begin{table}[!htbp]
  \centering
  \caption{Marginal variance after additive adjustment for circuit--encoder identity, frequency and phase. \(N_{\mathrm{cond}}\) is the retained condition count; \(\beta_v\) is the coefficient in Eq.~\eqref{eq:appendix_variance_adjustment}. \(R^2_{\mathrm{cat}}\) is the fit using only the categorical terms and an intercept; \(\Delta R^2\) is the increase on adding the score. All values describe in-sample fits to condition medians.}
  \label{tab:appendix-variance-adjustment}
  \setlength{\tabcolsep}{9pt}
  \renewcommand{\arraystretch}{1.15}
  \begin{tabular}{@{}cccccc@{}}
    \toprule
    \(n\) & \(d\) & \(N_{\mathrm{cond}}\) & \(\beta_v\) &
    \(R^2_{\mathrm{cat}}\) & \(\Delta R^2\)\\
    \midrule
    6 & 1 &  96 & \(-0.3620\) & 0.88545 & 0.02388 \\
    6 & 2 & 142 & \(-0.2145\) & 0.97908 & 0.00671 \\
    6 & 3 & 158 & \(-0.2719\) & 0.95656 & 0.02050 \\
    6 & 4 & 166 & \(-0.2091\) & 0.97581 & 0.01003 \\
    6 & 5 & 166 & \(-0.2454\) & 0.97778 & 0.01024 \\
    \midrule
    2 & 5 &  58 & \(+0.2817\) & 0.99156 & 0.00134 \\
    3 & 5 &  85 & \(-0.0479\) & 0.99732 & 0.00012 \\
    4 & 5 & 112 & \(-0.0634\) & 0.99241 & 0.00036 \\
    5 & 5 & 139 & \(-0.1628\) & 0.98750 & 0.00338 \\
    \bottomrule
  \end{tabular}
\end{table}

At \(n=6\), all five fitted coefficients are negative and the score adds \(0.0067\)--\(0.0239\) to in-sample \(R^2\) (Table~\ref{tab:appendix-variance-adjustment}). Some association therefore remains after the additive frequency, identity and phase offsets have been accounted for. The result does not extend uniformly across sizes: at \((n,d)=(2,5)\), the coefficient is positive and the increment is only \(0.0013\). All nine regression designs have full column rank, but the categorical terms account for most of the score variation. Only \(4.2\%\)--\(8.7\%\) remains at \(n=6\), and just \(0.48\%\) at \((2,5)\).

The additive model removes common frequency and phase offsets, while interactions between identity, frequency and phase remain. Since both the score and the learning outcome depend on these joint circuit--target choices, the fitted coefficient does not isolate a causal effect of variance.

We do not report \(p\)-values for this coefficient. Permuting scores within an identity would disrupt their dependence on frequency and phase, so it would not provide a justified reference test for the adjusted model.

\FloatBarrier

\subsection{Depth and size with matched supported targets}
\label{subsec:appendix_matched_depth_size}

Figure~\ref{fig:matched-depth-size-learning} compares the single-mode training results on targets supported throughout each sweep. A condition is the circuit--encoder identity, frequency and phase in $0.8\cos(\omega x+\phi)$. Within each sweep, we retain the intersection of conditions whose requested real coefficient direction is supported at every setting under the directional variance rule in Appendix~\ref{subsec:appendix_correlation_depth}. Each condition has the same ten seed indices, optimiser settings and 500-step budget. Its early decay and final MSE are calculated per trajectory, then summarised by the ten-seed median before any circuit or pooled median. The plotted scores reproduce $0.64$ times the corresponding diagonal entry of the saved real covariance.

There are 96 common conditions over $n=6,d=1,\ldots,5$ and 58 over $d=5,n=2,\ldots,6$. The first intersection removes targets activated only at greater depth; the second uses only frequencies one and two. Conditions receive equal weight in the pooled curves. These intersections differ from the full supported population used separately within each cohort in the adjusted regression.

The endpoint counts in the main text summarise changes in condition medians, with frequencies and phases grouped within the same circuits. The absolute losses help put these changes in perspective: increasing size takes the frequency-one median MSE from $2.19\times10^{-24}$ to $2.55\times10^{-6}$, whereas the frequency-two median improves from $0.0450$ to $0.0169$. Slower initial decay therefore need not produce a uniform trend in final MSE.

The depth comparison also has exceptions. The frequency-six median early decay falls from $0.00240$ to $0.00140$ per step, even though learning improves over the fixed target set as a whole. In either sweep, coefficient moments change alongside the parameter count and optimisation geometry.

\FloatBarrier

\subsection{Frequency-resolved spectral-bias trajectories}
\label{subsec:appendix_spectral_bias_dynamics}

Figure~\ref{fig:appendix-frequency-recovery} resolves the aggregate learning rates in Fig.~\ref{fig:variance-spectral-learning} into the individual Fourier components of a common multi-frequency target.  For each of three fixed phase realisations $r$, we use
\begin{equation}
  y_r(x)=a_r\sum_{\omega=1}^{6}
  \cos(\omega x+\phi_{r\omega}),
  \label{eq:appendix_multimode_target}
\end{equation}
where the modes have equal amplitudes within a realisation and $a_r$ rescales the complete sum to the observable range, with maximum absolute value $0.8$. This tests how a circuit learns the encoder-accessible frequencies when they are requested simultaneously.

Let \(\widehat{\boldsymbol c}_{\omega rj}(s)\in\mathbb R^2\) be the learned cosine--sine coefficient pair at training step $s$, for phase realisation $r$ and matched parameter seed $j$, and let \(\boldsymbol c^{\star}_{\omega r}\) be its target value.  The plotted recovery fraction is
\begin{equation}
  G_{\omega rj}(s)
  =\operatorname{clip}_{[0,1]}\!\left(
  1-\frac{\|\widehat{\boldsymbol c}_{\omega rj}(s)
  -\boldsymbol c^{\star}_{\omega r}\|_2}
  {\|\boldsymbol c^{\star}_{\omega r}\|_2}\right).
  \label{eq:appendix_frequency_recovery}
\end{equation}
In the unscaled cosine--sine basis, the contribution of this frequency to the full-period squared error is
\begin{equation}
 \mathcal L_{\omega rj}(s)
 =\tfrac12\|\widehat{\boldsymbol c}_{\omega rj}(s)
                 -\boldsymbol c^\star_{\omega r}\|_2^2.
 \label{eq:appendix_mode_mse}
\end{equation}
The plotted score measures amplitude recovery relative to the target. Total MSE instead sums the squared errors above together with the constant and any other residual modes. Clipping assigns zero recovery whenever the coefficient error equals or exceeds the target amplitude, so larger errors are indistinguishable on this scale.

Reading a vertical strip upward follows training time. Each cell summarises the three phase realisations and ten initialisations. The broad left-to-right progression shows that lower modes are usually recovered earlier and more completely, consistent with the aggregate variance ordering. The speed of this progression and its occasional reversals depend on the circuit and encoder. Hatched, numerically unsupported directions show no recovery. The tendency to learn lower frequencies first is consistent with the spectral-bias picture studied by Duffy and Jastrzebski~\cite{duffy2026}.

\begin{figure}[!htbp]
  \centering
  \includegraphics[width=\textwidth]
  {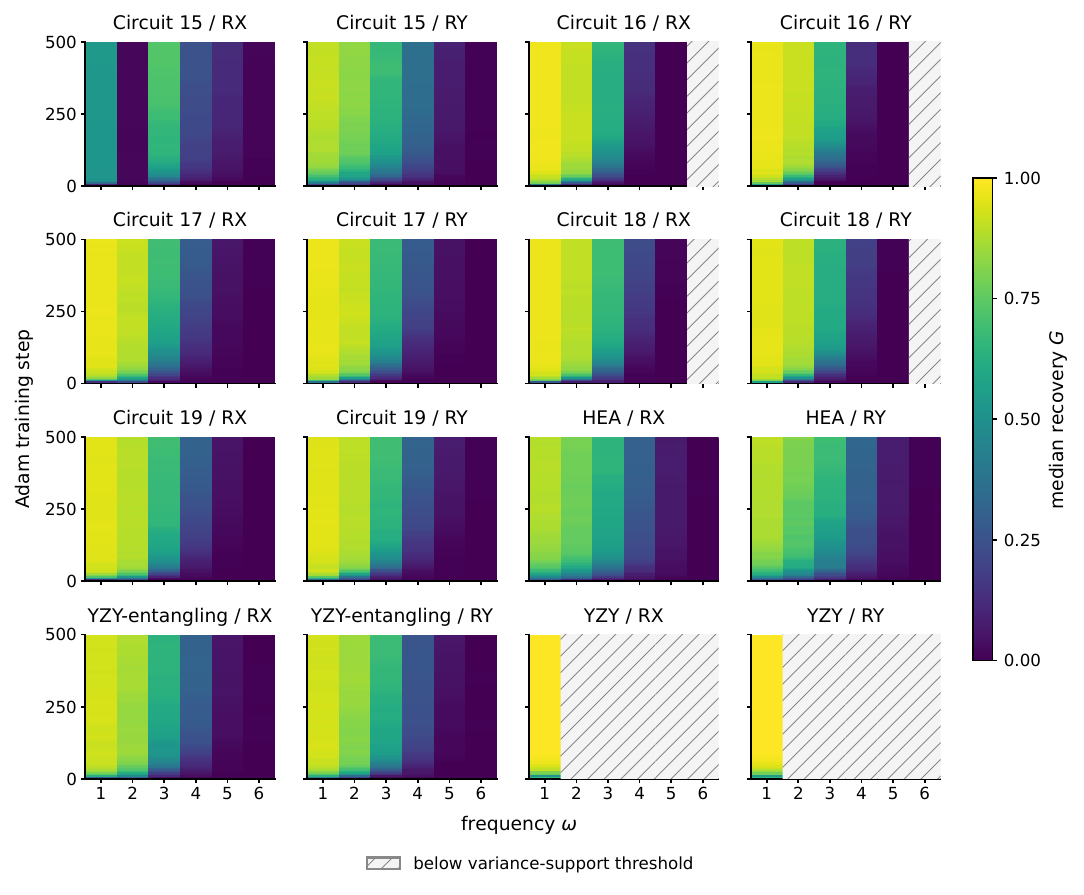}
  \caption{\justifying{Lower frequencies usually recover earlier, with circuit-dependent exceptions. A common in-band Fourier target is learned at \((n,d)=(6,3)\) by all 16 circuit--encoder identities. Colour gives median clipped amplitude-norm recovery in Eq.~\eqref{eq:appendix_frequency_recovery} over three target-phase realisations and ten matched initialisations, at five-step checkpoints. Yellow denotes nearly complete recovery; dark purple denotes $G=0$, including errors at least as large as the target amplitude. Hatching marks encoder-accessible frequencies below the numerical variance-support threshold.}}
  \label{fig:appendix-frequency-recovery}
\end{figure}

\FloatBarrier

\subsection{Learning comparisons beyond marginal variance}
\label{subsec:appendix_covariance_learning}

We approach the role of covariance in three ways. First, we compare mean coefficient correlation with multimode learning across circuits. We then hold the marginal score fixed in matched-sign comparisons, and finally ask whether covariance improves a multimode regression that already includes the marginal score. Each analysis concerns the chosen target bank and architectures, which are fixed experimental choices rather than a random sample of learning problems.

\subsubsection{Mean correlation across circuits}
\label{subsec:appendix_mean_correlation_learning}

Figure~\ref{fig:correlation-learning-overview} uses all 80 settings at $n=6$, with eight circuit families, two encoder axes and depths one through five, subject to the correlation statistic being defined. Let $S$ contain the nonnegative complex frequencies whose exact-computation variances exceed $\max(10^{-14},10^{-12}v_{\max})$, where $v_{\max}$ is the largest physical-mode variance in that setting. For $m=|S|\geq2$, the statistic is the pair mean
\[
  \bar r=\frac{2}{m(m-1)}
    \sum_{\substack{\omega,\mu\in S\\\omega<\mu}}
    \left|\frac{\Sigma_{\omega\mu}}
    {\sqrt{\Sigma_{\omega\omega}\Sigma_{\mu\mu}}}\right|.
\]
All unique pairs enter, including zero entries. If the largest supported off-diagonal correlation is at most $10^{-10}$, the mean is set to zero, as in the diagonal-matrix classification. The constant mode enters when supported; conjugate negative-frequency duplicates do not. This is a normalised complex-coefficient statistic, whereas the target scores below use real cosine--sine coordinates. Its support and pair-count conventions determine how it compares numerically with other FCC definitions.

Fourteen settings have $m=1$ and therefore an undefined pair mean: nonentangling YZY with both encoders at every depth, and Circuits 16 and 18 with both encoders at depth one. They are excluded rather than assigned zero, leaving 66 settings from seven circuit families. Of these, 33 have zero resolved mean correlation and 49 support all twelve positive-frequency cosine--sine coordinates. Real-coordinate coverage is the fraction of those twelve diagonal variances above the same absolute/relative rule, with the relative threshold taken from their largest variance. Coverage measures marginal availability, not simultaneous independent tunability.

The learning population comprises three in-band targets on frequencies one through six, with identical coefficient vectors across all circuits and depths. Each target is trained from ten initialisations for 500 Adam steps at step size $0.02$. Final MSE and early decay are first summarised by the ten-seed median for each target, then by the median over the three targets. The early rate is minus the fitted log-relative-loss slope over steps 0--25, as in Appendix~\ref{appendix:numerical_methods}. This gives 2,400 training runs, of which 420 belong to settings without a defined correlation mean. Each setting contributes one point. Unsupported target components are retained in the loss, which makes coverage a relevant covariate.

Table~\ref{tab:correlation-learning-adjustment} reports Spearman associations and descriptive partial-rank associations. The latter correlate residuals after regressing predictor and outcome ranks on the same controls. Depth and circuit--encoder identity use categorical indicators; continuous controls use ranks. Variance trace sums the twelve positive-frequency real-coordinate variances. The target marginal score is the median across the same three targets of the within-frequency quadratic score defined below. These controls have the target's positive-frequency scope, whereas $\bar r$ includes DC when supported.

\begin{table}[!htbp]
  \caption{\justifying{Association of mean absolute correlation with multimode learning across 66 settings. Entries give the rank association and a 95\% percentile interval from 2,000 bootstrap resamples of whole circuit families. Higher final MSE is worse; higher early decay is faster learning. Each row lists the full adjustment set.}}
  \label{tab:correlation-learning-adjustment}
  \centering
  \begin{tabular}{@{}lcc@{}}
\toprule
Adjustment & Final MSE & Early decay rate\\
\midrule
None & $0.285\;[0.207, 0.452]$ & $-0.445\;[-0.606, -0.315]$\\
Depth + coverage & $0.146\;[-0.035, 0.353]$ & $-0.404\;[-0.665, -0.225]$\\
Depth + coverage + variance trace & $0.148\;[-0.021, 0.602]$ & $-0.411\;[-0.693, -0.220]$\\
Depth + coverage + target marginal score & $0.144\;[-0.031, 0.527]$ & $-0.431\;[-0.694, -0.231]$\\
Depth + coverage + variance trace + identity & $-0.126\;[-0.651, 0.496]$ & $0.026\;[-0.442, 0.419]$\\
\bottomrule
\end{tabular}

\end{table}

Each bootstrap resample keeps both encoders and all depths of a circuit family together (seed 140926), while holding the saved medians over targets and initialisations fixed. With only seven eligible families, the intervals show how sensitive the association is to the composition of this benchmark for the fixed target bank. After adding identity controls, the intervals for both outcomes span zero. Restricting to the 49 settings with complete target-coordinate support gives a raw final-MSE association of $0.165$ and a depth/trace-adjusted association of $0.209$. These exploratory comparisons use our training protocol, which differs from that of Strobl \emph{et al.}~\cite{strobl2025}.

We can also compare depths while holding both mean correlation and target coverage fixed. Exactly three identities have $\bar r=0$ and coverage one at all five depths; Table~\ref{tab:diagonal-depth-learning} includes all three. On the same three multimode targets, each has lower final MSE and faster early decay at depth five. Thus learning can improve without further decorrelation or a gain in marginal target-coordinate support.

The improvement is measured between the endpoints; intermediate depths can vary nonmonotonically. Although marginal support is fixed, local control of those coefficients can still change with depth, motivating the tangent-kernel analysis in the main text.

\begin{table}[!htbp]
  \centering
  \caption{Depth improvement with zero resolved mean correlation and complete target-coordinate support throughout $d=1,\ldots,5$ at $n=6$. Each entry is a median over the same three multimode targets after taking the ten-initialisation median for each target. Final MSE uses 500 Adam steps; the early decay rate uses steps 0--25, in units of inverse steps. Higher decay rate means faster learning.}
  \label{tab:diagonal-depth-learning}
  \setlength{\tabcolsep}{7pt}
  \renewcommand{\arraystretch}{1.1}
  \begin{tabular}{@{}lcccc@{}}
    \toprule
    & \multicolumn{2}{c}{Final MSE} & \multicolumn{2}{c}{Early decay rate}\\
    Circuit / encoder & $d=1$ & $d=5$ & $d=1$ & $d=5$\\
    \midrule
    Circuit 15 / $R_Y$ & 0.07563 & 0.04136 & 0.01410 & 0.02678\\
    Circuit 17 / $R_X$ & 0.07620 & 0.04544 & 0.01949 & 0.02502\\
    Circuit 19 / $R_X$ & 0.07706 & 0.04341 & 0.01795 & 0.02617\\
    \bottomrule
  \end{tabular}
\end{table}

\subsubsection{Real-coordinate target scores}

All target scores in this appendix use the real cosine--sine coefficient vector. With the basis map of Appendix~\ref{app:kernel_real_coordinates},
\begin{equation}
 \Sigma=T_{\mathrm{cs}}\Sigma_{\mathbb R}T_{\mathrm{cs}}^\dagger,
 \qquad
 A_{\mathrm{cross}}(t)=
 t^{\mathsf T}\bigl[\Sigma_{\mathbb R}
      -\operatorname{diag}(\Sigma_{\mathbb R})\bigr]t.
 \label{eq:appendix_real_covariance_score}
\end{equation}
Thus $A_{\mathrm{full}}=A_{\mathrm{diag}}+A_{\mathrm{cross}}$.

The vectors \(q_{\mathrm{cs}}\) and \(t\) here include the constant coordinate, with \(t_0=0\) for every target in these comparisons. Covariances involving the constant coefficient therefore do not contribute to these target scores. The implementation stores only the positive-frequency cosine--sine block for these scores; this gives exactly the same quadratic forms. It constructs the marginal matrix by retaining each frequency's \(2\times2\) cosine--sine block. Within-frequency cosine--sine covariances are numerically zero in all 400 stored exact-statistics matrices, so on this bank that operation agrees with \(\operatorname{diag}(\Sigma_{\mathbb R})\).

\subsubsection{Matched-sign, seed-paired comparison}

At \(n=6\), the target bank contains every same-parity frequency pair in \(\{(1,3),(1,5),(2,4),(2,6),(3,5),(4,6)\}\); smaller-size cohorts retain all such pairs within \(1,\ldots,n\). For each pair and each of the four real cosine--sine quadrature pairs, we select amplitudes once using pilot directional variances pooled across the circuits in each cohort. The amplitudes are inversely proportional to the square roots of the cohort medians, using only above-threshold variances for each direction. The ratio is capped at ten and the amplitudes are rescaled to sum to \(0.8\), which bounds the target's absolute value by \(0.8\). If either direction has no supported reference variance, equal amplitudes are used. We retain the complete same-parity bank, including weak covariance contrasts. These target functions are then held fixed, and their scores are evaluated using each circuit's exact covariance matrix. The balancing reduces amplitude imbalance relative to the pilot variances, but the two marginal contributions need not be equal within each circuit. To form a matched pair, we use the same amplitudes for \(t_+\) and \(t_-\) and reverse only their relative sign. This leaves \(A_{\mathrm{diag}}\) unchanged while reversing the selected covariance cross term.

For example, targets \(y_\pm(x)=\alpha\cos(\omega x)\pm\beta\cos(\mu x)\), with \(\alpha,\beta>0\) and \(\omega\ne\mu\), give
\[
 \chi=
 \frac{2\alpha\beta\,\operatorname{Cov}_\theta(c_\omega,c_\mu)}
 {\alpha^2\operatorname{Var}_\theta(c_\omega)
  +\beta^2\operatorname{Var}_\theta(c_\mu)}.
\]
The other quadrature pairs use the corresponding real coefficients. Thus \(\chi\) compares a selected, target-weighted covariance with marginal variance. Its magnitude and its rank association with learning are distinct: a uniform positive rescaling of all contrasts leaves the Spearman association unchanged.

For each seed $j$, we train both members from the same parameter initialisation and form
\begin{equation}
  \Delta_{\mathcal L,j}
  =\frac{\overline{\mathcal L}_{-,j}
  -\overline{\mathcal L}_{+,j}}
  {\tfrac12(\overline{\mathcal L}_{-,j}
  +\overline{\mathcal L}_{+,j})}.
  \label{eq:appendix_seed_paired_covariance_response}
\end{equation}
We then take the median of Eq.~\eqref{eq:appendix_seed_paired_covariance_response} over ten seeds to obtain each condition's response in Figs.~\ref{fig:covariance-learning-controls} and~\ref{fig:appendix-covariance-pairwise}. Forming the ratio within each seed preserves the matched initialisation. A positive value of $\operatorname{sign}(\chi)\Delta_{\mathcal L}$ indicates lower time-averaged loss for the covariance-favoured target. Taking the absolute response instead measures the size of the training difference without regard to its direction.

The numerical support rule leaves 146, 250, 298, 330, and 330 conditions at depths one through five. Of these, respectively 32, 58, 70, 78, and 78 are classified as coupled; the remaining 114, 192, 228, 252, and 252 are diagonal controls. Here support requires both frequency blocks and a nonnegligible total marginal score; it does not require both selected real quadratures to have nonzero variance. The matrix profiles below distinguish these cases. A condition is classified as coupled when $|A_{\mathrm{full}}(t_+)-A_{\mathrm{full}}(t_-)|>2\times10^{-12}$; this threshold applies to the target-weighted cross-score difference. To track how contrast changes with depth, we compute its RMS over the same 146 supported conditions throughout, including zero contrasts. For the learning associations and pooled effects, we instead use every supported condition at each depth. The fixed intersection allows us to distinguish a reduction in exact contrast from a change in the target directions admitted by the support rule.

For the Spearman reference test we hold $\chi$ fixed and reverse all responses belonging to one circuit--encoder identity together. There are four coupled identities at depth one and seven at each later depth, so we enumerate all $2^4=16$ or $2^7=128$ sign patterns. The two-sided $p$-value is the fraction giving $|\rho_{\mathrm{null}}|\geq|\rho_{\mathrm{observed}}|$. This preserves dependence among frequency-pair and quadrature responses within each identity. Its validity nevertheless requires the joint response distribution under the null to be invariant under independent cluster sign reversals, conditional on the exact contrasts. The experiment does not guarantee this symmetry or independence between identities. A literal relabelling of $t_+$ and $t_-$ would reverse both $\chi$ and $\Delta_{\mathcal L}$, so this is a model-based reference test, not a randomised causal test.

We obtain the pointwise $95\%$ percentile intervals in Table~\ref{tab:covariance-rank-reference} from 2,000 resamples of whole circuit--encoder identities, keeping each identity's conditions and ten-seed medians together. Like the earlier family bootstrap, this measures sensitivity to the composition of a small architecture bank. The seed medians remain fixed during resampling, and each interval applies to one depth.

For the five matched-sign depth tests, we retrospectively apply Holm's step-down correction \cite{holm1979}. Provided the individual tests are valid, this controls the family of five tests without requiring independent depths. The small number of circuit groups limits the resolution of this comparison. Both the observed sign pattern and its global reversal give the same absolute rank association, so the smallest possible two-sided reference value is $2/128=0.015625$ at depths two through five, and $2/16=0.125$ at depth one. The first Holm comparison multiplies the smallest value by five; consequently, no matched-sign result can reach the adjusted $5\%$ threshold in this experiment. Adjusted non-significance alone therefore cannot establish that covariance has no learning association.

At $d=2$ and $d=5$, the signed rank associations are $\rho=0.344$ and $0.347$, with raw $p=0.0625$ and $p_{\mathrm{Holm}}=0.3125$ in both cases. The pooled median aligned responses are $0.205\%$ and $0.130\%$, respectively. No depth passes the adjusted $5\%$ threshold. These signed rank tests are distinct from the aligned medians in Fig.~\ref{fig:covariance-learning-controls}(b,d); the latter are descriptive response summaries.

\begin{table}[!htbp]
  \centering
  \caption{\justifying{Signed rank associations between contrast $\chi$ and paired loss difference $\Delta_{\mathcal L}$ at $n=6$. Intervals are pointwise circuit--encoder bootstrap percentiles; $p$-values use the separate cluster sign-flip reference, with Holm adjustment across five depths. Their assumptions and populations are specified in the text.}}
  \label{tab:covariance-rank-reference}
  \begin{tabular}{ccccc}
    \toprule
    $d$ & $\rho(\chi,\Delta_{\mathcal L})$ & $95\%$ interval & $p_{\mathrm{raw}}$ & $p_{\mathrm{Holm}}$ \\
    \midrule
    1 & 0.397 & $[-0.619,\ 0.718]$ & 0.500000 & 1.000000 \\
    2 & 0.344 & $[\phantom{-}0.004,\ 0.626]$ & 0.062500 & 0.312500 \\
    3 & 0.084 & $[-0.211,\ 0.428]$ & 0.656250 & 1.000000 \\
    4 & 0.034 & $[-0.193,\ 0.260]$ & 0.796875 & 1.000000 \\
    5 & 0.347 & $[\phantom{-}0.128,\ 0.498]$ & 0.062500 & 0.312500 \\
    \bottomrule
  \end{tabular}
\end{table}

The diagonal controls repeat the same training comparison where the selected cross-score difference is below threshold. They can still show a training preference: zero ensemble covariance does not force the two targets to have identical residuals, tangent kernels along training, or nonlinear optimisation trajectories. Figure~\ref{fig:covariance-learning-controls}(c) compares the magnitude of these control responses with that of the coupled conditions.

\paragraph{Direction and contrast strength.}
A covariance preference corresponds to upper-right or lower-left points in a plot of signed contrast against signed loss difference. We combine both signs by defining
\begin{equation}
  \Delta_{\mathrm{align}}
    =\operatorname{sign}(\chi)\Delta_{\mathcal L},
  \qquad
  \rho_{\mathrm{mag}}
    =\operatorname{Spearman}(|\chi|,\Delta_{\mathrm{align}}).
  \label{eq:covariance-aligned-response}
\end{equation}
Plotting $\Delta_{\mathrm{align}}$ against $|\chi|$ reflects the negative-contrast conditions through the origin. Positive responses then mean that the covariance-favoured target has lower time-averaged loss on either side of the original comparison. Interchanging the names $t_+$ and $t_-$ reverses both $\chi$ and $\Delta_{\mathcal L}$ and leaves this diagnostic unchanged. Diagonal controls have no resolved score preference and are excluded from the alignment summary.

Figure~\ref{fig:covariance-aligned-response}(a)--(e) combines all individual coupled responses with their contrast-binned summaries. The three bins are within-depth contrast-magnitude terciles, with tied magnitudes kept together. Each coloured marker gives the median contrast magnitude and median aligned response; bars give the interquartile spread of condition responses. The unbinned $\rho_{\mathrm{mag}}$ measures whether larger contrasts accompany larger aligned advantages, separately from the original signed association $\rho(\chi,\Delta_{\mathcal L})$. Its five values are $0.196$, $0.356$, $0.182$, $0.013$ and $0.246$. The binned medians do not rise consistently at every depth, and the strength ordering is negligible at $d=4$.

Panel (f) gives pooled median aligned responses of $0.035\%$, $0.205\%$, $0.011\%$, $-0.042\%$ and $0.130\%$. Conditions receive equal weight, so identities with more coupled targets contribute more; the negative-contrast groups contain only $9$, $8$, $18$, $25$ and $30$ conditions, respectively. These summaries mix variation between circuits with variation among targets within a circuit. The profiles below separate those two levels. The earlier sign-flip tests concern the signed rank association; $\rho_{\mathrm{mag}}$ and the aligned medians are descriptive summaries.

\begin{figure}[!htbp]
  \centering
  \includegraphics[width=0.95\textwidth]{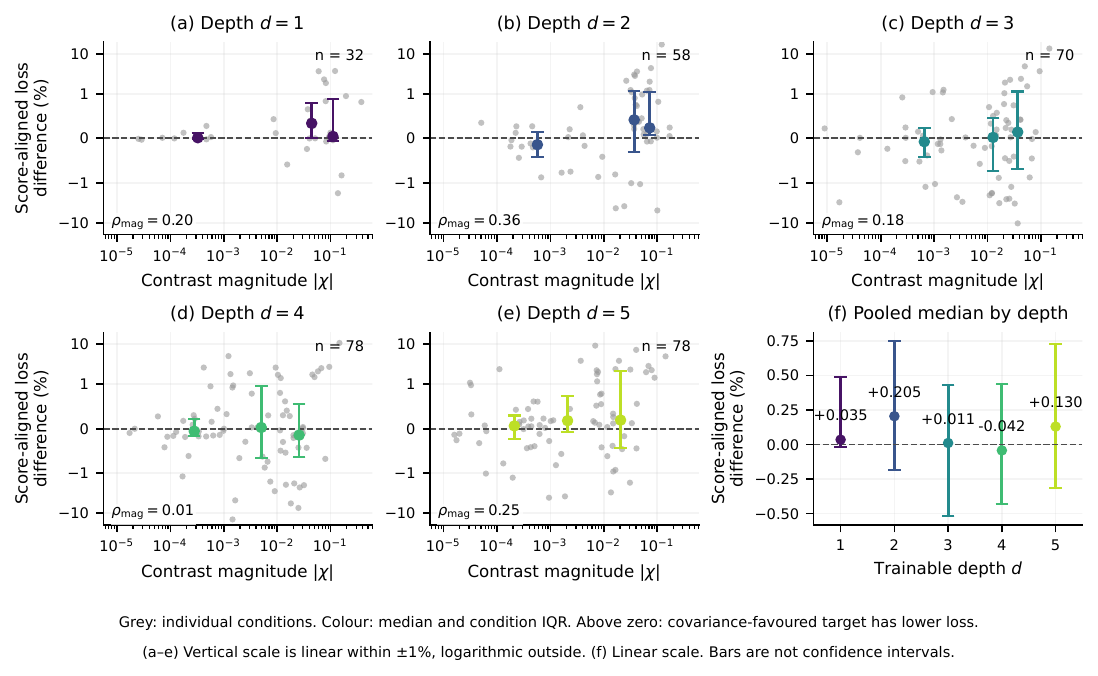}
  \caption{\justifying{Small pooled covariance preferences coexist with heterogeneous condition-level responses. Positive $100\Delta_{\mathrm{align}}$ means that the covariance-favoured target has lower time-averaged MSE. (a)--(e) Grey points retain all $32$, $58$, $70$, $78$ and $78$ coupled conditions at depths one through five. Coloured markers give the median response and median $|\chi|$ in three within-depth contrast terciles; bars show the 25th--75th response percentiles. Tied contrasts stay together. The common vertical scale is linear within $\pm1\%$ and logarithmic outside, retaining the full response range; the horizontal scale is logarithmic. The reported $\rho_{\mathrm{mag}}$ uses all conditions before binning. (f) Pooled median and condition interquartile spread at each depth on a linear scale, with medians labelled in percent. Each condition first takes the median of its ten seed-paired loss differences. Diagonal-control magnitudes are shown in Fig.~\ref{fig:covariance-learning-controls}(c).}}
  \label{fig:appendix-covariance-pairwise}
  \label{fig:covariance-aligned-response}
\end{figure}

\subsubsection{Circuit-resolved depth profiles}
\label{subsec:appendix_correlations}

Figures~\ref{fig:architecture-covariance-heterogeneity} and~\ref{fig:architecture-covariance-heterogeneity-ry} place each circuit's correlation geometry beside its learning responses, separated by encoder axis. The matrices normalise the positive-frequency real cosine--sine covariance used by the target scores, with coordinates $(c_1,\ldots,c_6,s_1,\ldots,s_6)$. Grey cells contain an unsupported coefficient; supported diagonal entries equal one. These coordinates resolve the cosine and sine targets whose scores are being tested; grouping them separately exposes phase-specific support and coupling. The earlier complex-frequency matrices serve the different purpose of displaying the encoder's frequency structure.

Each snapshot uses the earliest depth attaining the maximum number of supported positive frequencies over $d=1,\ldots,5$, independently of covariance magnitude or learning effect. Either supported real quadrature makes a frequency present. This gives $d=4$ for Circuits~16 and~18, $d=2$ for HEA and YZY-entangling, and $d=1$ for Circuits~15, 17 and~19, under both encoders. These snapshots all support frequencies one through six. Nonentangling YZY supports only frequency one throughout, so its earliest maximal-support snapshot is $d=1$. Full frequency support does not require both quadratures: the unsupported sine block of Circuit~15/$R_X$ remains grey. The dotted guide marks the chosen depth on each learning profile; Figs.~\ref{fig:correlation-depth-atlas-rx} and~\ref{fig:correlation-depth-atlas-ry} retain the complete depth sequences. These snapshots use all directions available at their chosen depth; the moment trends in Fig.~\ref{fig:fixed-support-structure-depth} instead hold the coordinate set fixed across depths to make the means comparable.

All seven identities without a covariance-score preference have numerically diagonal correlation matrices on the supported real coefficients at every depth $d=1,\ldots,5$. YZY/$R_X$ and YZY/$R_Y$ instead support only frequency one and have no eligible two-frequency comparison. The upper profiles retain the covariance-aligned responses: positive values favour the covariance score, while a zero score contrast leaves the aligned response undefined.

\begin{figure}[p]
  \centering
  \includegraphics[width=.975\textwidth]{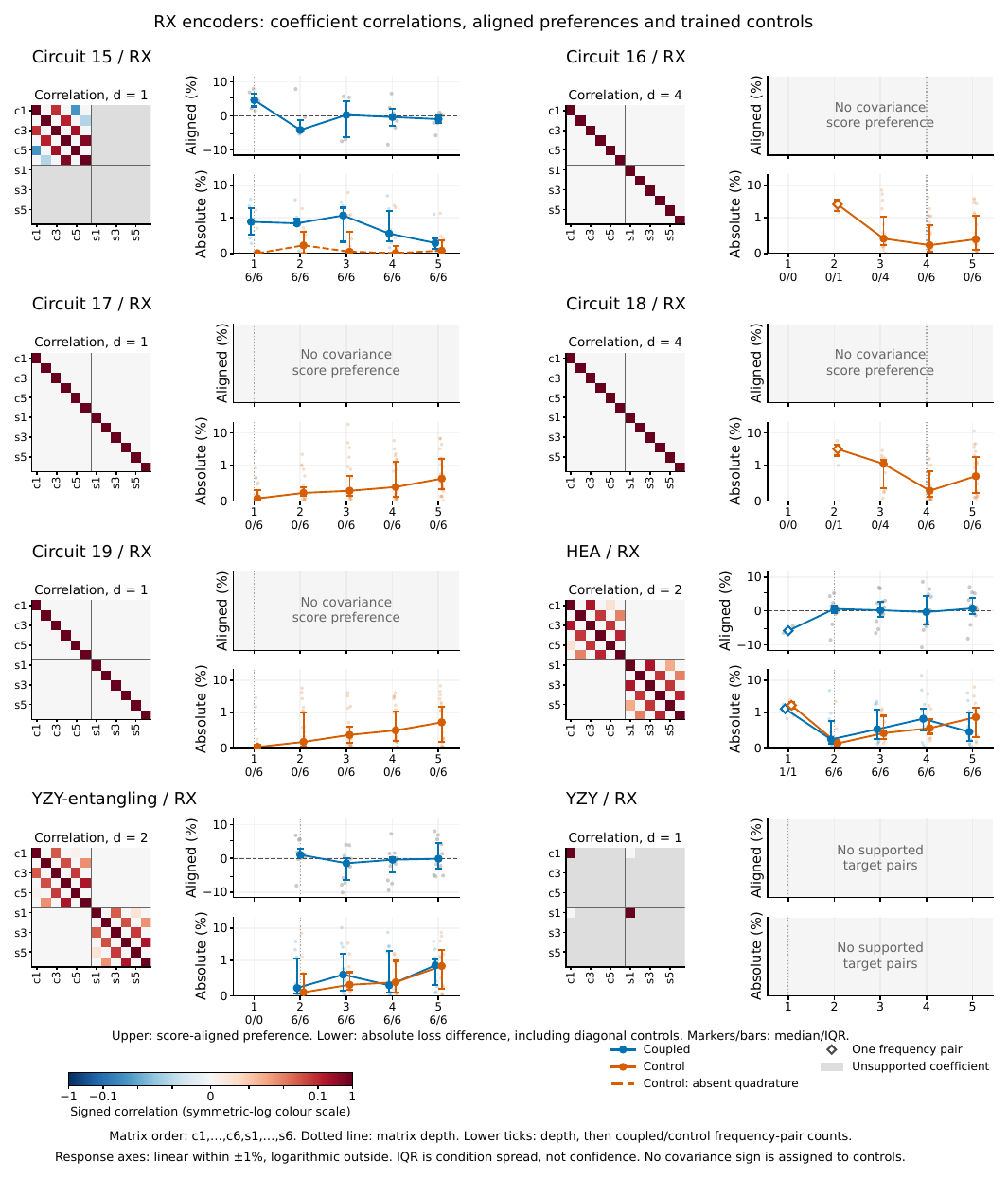}
  \caption{\justifying{Correlation matrices and trained controls for $R_X$ encoders at $n=6$. Matrices use the earliest depth attaining maximal positive-frequency support (six frequencies, except YZY's one); titles and dotted guides mark this depth. Cosine--sine coordinates match the target phases; grey marks unsupported quadratures. Upper: covariance-aligned loss differences for coupled conditions. Lower: absolute loss differences for coupled conditions (blue) and diagonal controls (orange; dashed if one quadrature is unsupported). Points show conditions; markers and bars give medians and condition IQRs after ten-seed aggregation. Open diamonds mark one frequency pair. Lower ticks give depth and coupled/control pair counts. Scales are shared.}}
  \label{fig:architecture-covariance-heterogeneity}
\end{figure}

\begin{figure}[p]
  \centering
  \includegraphics[width=.975\textwidth]{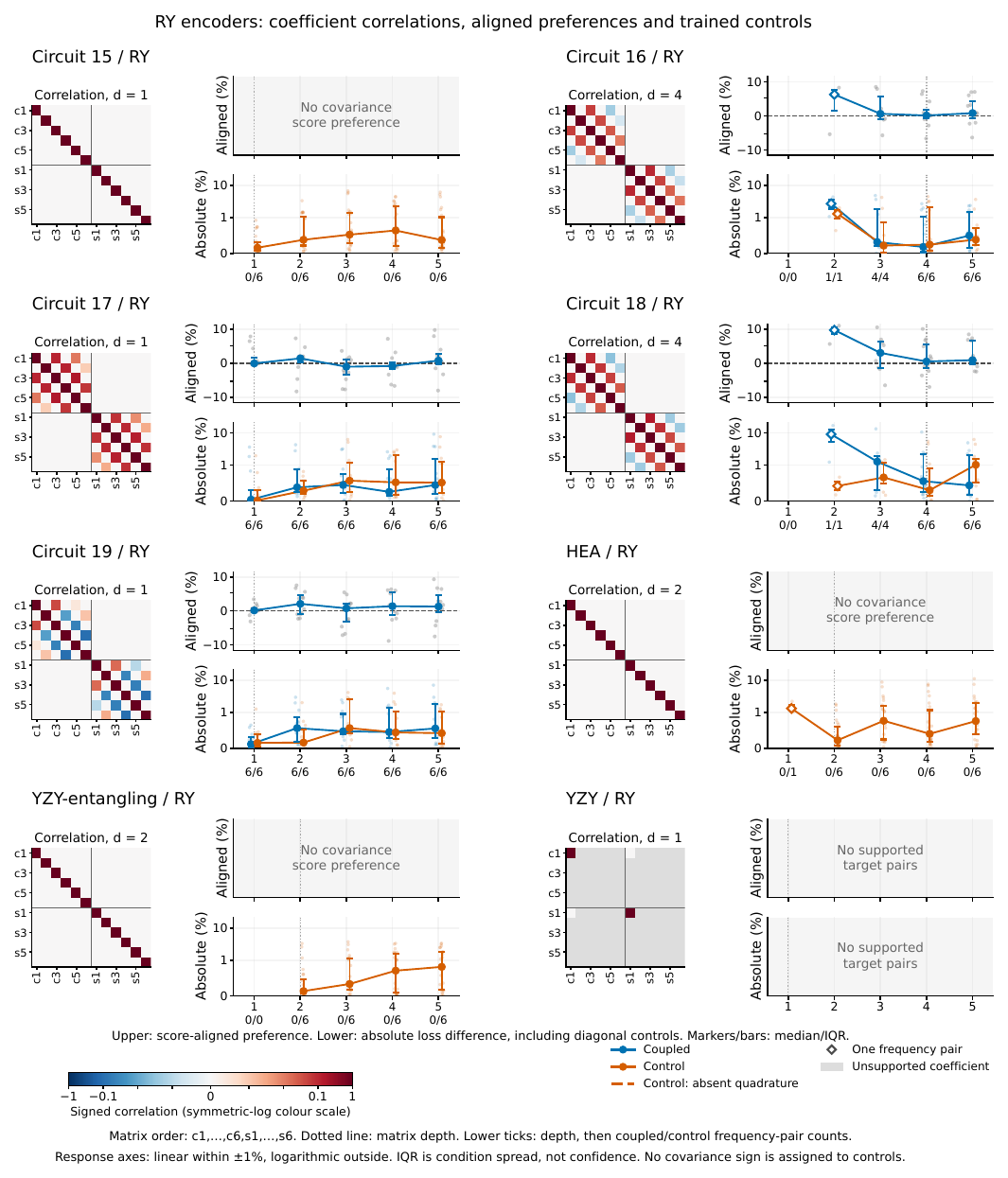}
  \caption{\justifying{Correlation matrices and trained controls for $R_Y$ encoders at $n=6$, with the same definitions, scales and markers as Fig.~\ref{fig:architecture-covariance-heterogeneity}. Matrices use the real cosine--sine coordinates of the target scores. Upper profiles measure covariance-aligned preferences; lower profiles compare absolute loss differences in coupled conditions and diagonal controls. Matrix snapshots use the earliest depth attaining maximal positive-frequency support, including $d=1$ for YZY's single supported frequency.}}
  \label{fig:architecture-covariance-heterogeneity-ry}
\end{figure}

\FloatBarrier

The diagonal controls were also trained on both target signs with ten paired initialisations per condition. There are 1,038 supported control conditions alongside 316 coupled conditions. Of these controls, 978 have both real coefficient directions available; 60 mix an available cosine direction with an unsupported sine direction. Dashed control curves distinguish the latter support limitation. The lower profiles display $100|\operatorname{median}_j\Delta_{\mathcal L,j}|$ for both populations, using the same estimator as Fig.~\ref{fig:covariance-learning-controls}(c). Taking the absolute value after the seed median exposes target-sign asymmetry without assigning an arbitrary covariance preference to a control. Frequency-pair counts and open diamonds identify limited target coverage; quadratures and seeds are not independent replications across frequency pairs.

The circuit profiles reveal different responses within the pooled results. Circuit~19/$R_Y$ favours the covariance score at every depth, whereas Circuit~15/$R_X$ has positive medians at $d=1,3$ and negative medians at $d=2,4,5$. Circuit~17/$R_Y$, HEA/$R_X$ and YZY-entangling/$R_X$ show mixed signs or medians close to zero. Circuits~16 and~18 with $R_Y$ have positive medians wherever supported, but their largest values occur at $d=2$, where only one frequency pair is available. These are descriptive patterns: the eligible targets change with depth and differ between the coupled and control groups.

The pooled control median absolute responses are $0.076\%$, $0.235\%$, $0.425\%$, $0.385\%$ and $0.540\%$ at depths one through five. Training asymmetry therefore persists even in these diagonal controls. A systematic excess of covariance-aligned preferences over appropriate controls would support a covariance association, but cross-circuit comparisons mix architectural and target-support differences and do not isolate a causal covariance effect.

\subsubsection{Increment beyond marginal variance for multi-mode targets}

The matched-sign construction reverses one cross term at a time. As a complementary analysis, the multi-mode target bank asks whether the sum of all covariance terms improves a descriptive regression of final loss. At each depth we use 48 in-band circuit--encoder--phase conditions (16 identities and three target phases), with final MSE summarised by the median over ten seeds, and fit the reduced model
\begin{equation}
  \log_{10}\mathcal L_{\mathrm{final}}
  =\beta_0+\beta_v\log_{10}A_{\mathrm{diag}}
  +\alpha_{\mathrm{architecture}}+\alpha_{\mathrm{phase}}+\epsilon,
  \label{eq:appendix_multimode_covariance_reduced}
\end{equation}
and then add the signed cross fraction \(A_{\mathrm{cross}}/A_{\mathrm{diag}}\). The architecture fixed effect is one intercept per circuit--encoder identity; the phase effect is shared across identities. Continuous regressors are centred and scaled to unit standard deviation. A Freedman--Lane test \cite{freedman1983} performs 2,000 permutations of reduced-model residuals within each circuit--encoder identity, adds them back to its fitted values, refits both models, and compares their increment in \(R^2\). This residual-permutation reference requires exchangeability within these blocks after fitting the reduced model. There are only three phase conditions per block, shared phase choices across identities, and matched seed banks; fitting fixed effects alone does not establish the required exchangeability.

The largest observed increment occurs at \(d=1\): the reduced and full values are \(0.988664\) and \(0.992295\), respectively, so \(\Delta R^2=3.63\times10^{-3}\), with standardised covariance coefficient \(-0.0265\) and raw \(p=0.0290\). Applying Holm adjustment separately to the five multimode depth tests gives \(p_{\mathrm{Holm}}=0.1159\) at this depth. The smallest adjusted value is \(0.1049\), at \(d=5\), so no depth passes the adjusted $5\%$ threshold. Table~\ref{tab:appendix-covariance-multimode} reports all five depths, including the fitted cross-fraction coefficient. The large baseline \(R^2\) includes the circuit--encoder and phase intercepts; covariance adds only a small increment to this in-sample fit.

\begin{table}[!htbp]
  \centering
  \caption{Additional in-sample association of covariance for multi-mode targets at $n=6$. Each depth uses 48 circuit--encoder--phase conditions, with final MSE summarised over ten seeds. The reduced model includes marginal score and circuit--encoder and phase effects [Eq.~\eqref{eq:appendix_multimode_covariance_reduced}]; the full model adds the standardised cross fraction, whose coefficient is $\beta_{\mathrm{cross}}$. The $p$-values are the saved residual-permutation references and their five-depth Holm adjustments. None passes the adjusted $5\%$ threshold.}
  \label{tab:appendix-covariance-multimode}
  \setlength{\tabcolsep}{6pt}
  \renewcommand{\arraystretch}{1.15}
  \begin{tabular}{@{}ccccccc@{}}
    \toprule
    $d$ & $R^2_{\mathrm{reduced}}$ & $R^2_{\mathrm{full}}$ &
    $\Delta R^2$ & $\beta_{\mathrm{cross}}$ &
    $p_{\mathrm{raw}}$ & $p_{\mathrm{Holm}}$\\
    \midrule
    1 & 0.988664 & 0.992295 & $3.63\times10^{-3}$ & $-0.0265$ & 0.0290 & 0.1159 \\
    2 & 0.988175 & 0.989283 & $1.11\times10^{-3}$ & $-0.0242$ & 0.1599 & 0.4798 \\
    3 & 0.994743 & 0.994862 & $1.19\times10^{-4}$ & $-0.0038$ & 0.6112 & 1.0000 \\
    4 & 0.996062 & 0.996102 & $3.96\times10^{-5}$ & $+0.0041$ & 0.6347 & 1.0000 \\
    5 & 0.994388 & 0.995128 & $7.40\times10^{-4}$ & $+0.0188$ & 0.0210 & 0.1049 \\
    \bottomrule
  \end{tabular}
\end{table}

\FloatBarrier

\subsection{Frozen-QNTK forecasts under gradient descent}
\label{subsec:appendix_frozen_qntk}

We test how accurately the QNTK at a reference parameter point predicts subsequent learning when held fixed. The experiment uses full-batch gradient descent, whose update matches the frozen-kernel calculation, with the coordinate conventions checked in Appendix~\ref{subsec:appendix_kernel_checks}.

Table~\ref{tab:learning-experimental-specification} specifies this 120-trajectory control for YZY-entangling/\(R_X\), \(n=6,d=3\): two target phases at each frequency \(\omega=1,\ldots,6\), with ten initialisations each. At each reference step \(s\in\{0,10,50,100,T\}\), it constructs \(K_s=V_{\mathbb R}H_{\mathbb R}(\vartheta_s)V_{\mathbb R}^{\mathsf T}\) and iterates \(r_{s+j}^{\mathrm{fr}}=(I-2\eta K_s/N)^j r_s\).  Evaluation-grid predictions use the corresponding cross-kernel. Figure~\ref{fig:frozen-qntk-validation} in the main text focuses on the initialisation-frozen forecast; source tables retain each later checkpoint forecast separately. Alongside the forecast, we track \(\|H_{\mathbb R}(\vartheta_s)-H_{\mathbb R}(\vartheta_0)\|_{\mathrm F}/
\|H_{\mathbb R}(\vartheta_0)\|_{\mathrm F}\) and the spectral radius of the discrete update.
This drift uses the unscaled real cosine--sine coordinates; changing the coefficient normalisation can change its value. The forecasts follow the frozen-QNTK framework of Liu \emph{et al.}~\cite{liu2022}.

The small learning rate makes this a conservative local experiment: median loss reductions at step 500 range from $8.26\%$ of the initial loss for $\omega=1$ to just $0.129\%$ for $\omega=6$. Comparing prediction error only with the remaining loss can therefore make a nearly static forecast look accurate. We instead compare with a no-change baseline, which keeps the initial model predictions fixed and hence predicts $\mathcal L_t^{\mathrm{nc}}=\mathcal L_0$. For each seed--phase trajectory, we compute the achieved reduction $D_t=\mathcal L_0-\mathcal L_t$ and the error ratio $e_t$ defined in Eq.~\eqref{eq:frozen-error-over-reduction}, using the observed evaluation-grid MSE $\mathcal L_t$. The no-change forecast has \(e_t^{\mathrm{nc}}=1\), whereas an exact loss forecast has \(e_t=0\); values above one would be worse than no change. Ratios are formed before taking medians over the ten seeds and two phases per frequency. We leave the ratio undefined when \(D_t\leq10^{-12}\mathcal L_0\). The zero-horizon comparison is excluded, but none of the saved positive-horizon comparisons requires exclusion. The same definition applies to later checkpoints with their own reference loss; the final checkpoint itself has no forward prediction horizon.

At step 500, the frequency-wise median errors \(e_t\), expressed as percentages, are \(2.44,\ 2.22,\ 0.98,\ 0.98,\ 0.41,\ 0.28\%\) for \(\omega=1,\ldots,6\), respectively. All 120 initialisation-frozen forecasts beat the no-change baseline at that step; the largest individual ratio is \(10.87\%\), so even that forecast error is smaller than the observed reduction.

The kernel nevertheless changes during training. Its final relative drift has median $3.33\%$, interquartile range $1.57\%$--$5.97\%$, and maximum $15.15\%$. This measures change in the whole matrix, whereas the forecast error concerns its effect on the loss for a particular target. Each coloured curve in Fig.~\ref{fig:frozen-qntk-validation}(d) groups trajectories by the frequency they were trained to learn; it does not project the kernel drift onto that frequency direction.

The frozen kernel gives an accurate \emph{loss forecast} relative to the learning achieved in this small-step, full-batch-gradient-descent experiment. That accuracy can coexist with kernel drift and does not establish agreement of the full residual trajectories. The experiment also leaves open whether the approximation remains accurate after large loss reductions or under the main Adam protocol. The evaluation-grid interpretation is given in Appendix~\ref{subsec:appendix_protocol}.

\FloatBarrier

\FloatBarrier
\section{Encoder harmonics and frequency redundancy}\label{appendix:qfms}

The encoder determines which frequencies can appear, but a given frequency may arise from several combinations of its eigenvalues. Starting from the expansion of Schuld \emph{et al.}~\cite{schuld2021}, we count these combinations and derive the binomial redundancy reference used in our figures. We then explain how the count relates to CHM rows and to the redundancy definitions of Mhiri \emph{et al.}~\cite{mhiri2025a} and Duffy--Jastrzebski~\cite{duffy2026}. The derivation allows re-uploading, although our experiments use one encoding layer, \(L=1\).

\subsection{Encoder spectra and a single encoding layer}
\label{app:qfms_assumption}
\label{app:qfms_single_insertion}

On a Hilbert space of dimension \(d_{\mathcal H}=2^n\), assume
\begin{equation}
S_\ell(x)=\exp\!\bigl(-i\,x\cdot G^{(\ell)}\bigr),
\qquad
G^{(\ell)}=(G^{(\ell)}_1,\dots,G^{(\ell)}_{d_x}),
\label{eq:app_encoder_form}
\end{equation}
where the Hermitian generators commute within each block. They admit joint eigenvectors \(\ket{\lambda_j^{(\ell)}}\), with eigenvalue vectors \(\lambda_j^{(\ell)}\in\mathbb R^{d_x}\)~\cite{horn2017}. Each basis vector has its own index \(j=1,\ldots,d_{\mathcal H}\), including vectors with the same joint eigenvalue. Different encoder blocks need not commute or share an eigenbasis.

Using the main text's \(e^{+i\omega\cdot x}\) convention, define
\begin{equation}
\Omega^{(\ell)} := \{\lambda^{(\ell)}_{j'}-\lambda^{(\ell)}_j:\ j,j'\}.
\label{eq:app_diffset}
\end{equation}
The set contains only distinct values: it is finite, contains zero, and is symmetric under \(\omega\mapsto-\omega\). The following identity explains why differences, rather than generator eigenvalues themselves, label the output harmonics.

\begin{lemma}[Difference-frequency expansion]\label{lem:app_diff_frequency_expansion}
For one encoder \(S(x)=\exp(-i\,x\cdot G)\) and any input-independent operators \(A,B\),
\begin{equation}
\Tr\!\bigl[A\,S(x)\,B\,S(x)^\dagger\bigr]
=\sum_{j,j'} A_{j'j}B_{jj'}\,e^{i(\lambda_{j'}-\lambda_j)\cdot x},
\label{eq:app_single_layer_trace_expansion}
\end{equation}
where \(A_{j'j}=\langle\lambda_{j'}|A|\lambda_j\rangle\) and \(B_{jj'}=\langle\lambda_j|B|\lambda_{j'}\rangle\). In particular, for \(U=WS\), take \(A=W^\dagger O W\) and \(B=\rho_0\).
\end{lemma}

\begin{proof}
Set \(P_j=\ket{\lambda_j}\!\bra{\lambda_j}\) and insert \(S(x)=\sum_j e^{-i\lambda_j\cdot x}P_j\) and its adjoint. The phase of each pair is \(e^{i(\lambda_{j'}-\lambda_j)\cdot x}\), and \(\Tr(AP_jBP_{j'})=A_{j'j}B_{jj'}\).
\end{proof}

\subsection{Re-uploading and the encoder envelope}\label{app:qfms_minkowski}

We use the circuit with \(L\) encoding layers and the expectation value in Eqs.~\eqref{eq:prelims_reupload}--\eqref{eq:prelims_model}. Re-uploading multiplies encoder phases, so their differences add~\cite{schuld2021,perez-salinas2020}.

\begin{proposition}[Encoder-accessible harmonic set]\label{prop:app_Omega_minkowski}
Let \(\oplus\) denote the Minkowski sum, \(A\oplus B=\{a+b:a\in A,\ b\in B\}\). Under Eq.~\eqref{eq:app_encoder_form}, for every \(\theta\),
\begin{align}
\operatorname{supp}_x f(\,\cdot\,;\theta)
&\subseteq \Omega := \Omega^{(1)} \oplus \cdots \oplus \Omega^{(L)},
\label{eq:app_Omega_subset}
\\
f(x;\theta)&=\sum_{\omega\in\Omega} a_\omega(\theta)\,e^{i\omega\cdot x}.
\label{eq:app_qfm_expansion}
\end{align}
The envelope \(\Omega\) is independent of \(\theta\).
\end{proposition}

\begin{proof}
Let \(J=(j_1,\ldots,j_L)\) index one eigenvector for each encoding layer, and write \(P_j^{(\ell)}=\ket{\lambda_j^{(\ell)}}
\!\bra{\lambda_j^{(\ell)}}\). Define
\begin{equation}
\Lambda_J=\sum_{\ell=1}^L\lambda_{j_\ell}^{(\ell)},\qquad
T_J(\theta)=W_L P_{j_L}^{(L)}\cdots W_1 P_{j_1}^{(1)}.
\label{eq:app_encoder_path_operator}
\end{equation}
Here and below, the parameter arguments of the \(W_\ell\) are suppressed. Inserting the spectral decomposition of each encoder gives \(U=\sum_J e^{-i\Lambda_J\cdot x}T_J\), hence
\begin{equation}
\begin{aligned}
f(x;\theta)&=\sum_{J,J'} A_{J,J'}(\theta)
                    e^{i(\Lambda_{J'}-\Lambda_J)\cdot x},\\
A_{J,J'}(\theta)&=\Tr[O\,T_J(\theta)\rho_0T_{J'}(\theta)^\dagger].
\end{aligned}
\label{eq:app_encoder_path_amplitude}
\end{equation}
Every difference \(\Lambda_{J'}-\Lambda_J\) is a sum of layer-wise differences. Grouping equal frequencies gives Eq.~\eqref{eq:app_qfm_expansion}, including zero coefficients for unpopulated rows.
\end{proof}

The operator product retains changes of eigenbasis between encoding layers. A pre-encoding trainable block, as in the design-based reference of Appendix~\ref{app:mhiri_chm}, would append a factor \(W_0\) on the right of every \(T_J\), without changing \(\Omega\). Our benchmark has no such block.

For arbitrary real spectra this is a finite exponential sum, not necessarily a periodic Fourier series. Integer frequencies, as in our Pauli encoders, give a series on the \(2\pi\)-periodic input torus~\cite{schuld2021}. Uniqueness of the expansion for continuous \(x\) gives \(a_{-\omega}=\overline{a_\omega}\) when \(f\) is real.

\subsection{Counting encoder paths}\label{app:qfms_paths}

Frequency redundancy counts the eigenvector-path pairs that produce each frequency in Eq.~\eqref{eq:app_encoder_path_amplitude}, giving each pair equal weight~\cite{mhiri2025a}.

\begin{definition}[Index-path set and redundancy]\label{def:app_Romega}
With the complete eigenbasis indexing above, define
\begin{equation}
\mathcal{R}(\omega):=\bigl\{(J,J'):J,J'\in\{1,\ldots,d_{\mathcal H}\}^{L},
\ \Lambda_{J'}-\Lambda_J=\omega\bigr\},
\label{eq:app_Romega_def}
\end{equation}
and call \(|\mathcal{R}(\omega)|\) the frequency redundancy.
\end{definition}

This is the frequency-generator definition of Mhiri \emph{et al.}~\cite{mhiri2025a}, with the ordered pairs oriented to match our positive-exponential convention. Reversing all pairs sends \(\mathcal{R}(\omega)\) to \(\mathcal{R}(-\omega)\), so the counts agree at opposite frequencies. Changing the basis inside a degenerate eigenspace can redistribute path amplitudes but does not change the count.

For a layer-wise difference \(\delta\), let \(r_\ell(\delta)=|\{(j,j'):\lambda_{j'}^{(\ell)}
-\lambda_j^{(\ell)}=\delta\}|\). Then
\begin{equation}
|\mathcal{R}(\omega)|
=\sum_{\substack{\delta^{(1)}+\cdots+\delta^{(L)}=\omega\\
                  \delta^{(\ell)}\in\Omega^{(\ell)}}}
  \prod_{\ell=1}^L r_\ell(\delta^{(\ell)}).
\label{eq:app_redundancy_convolution}
\end{equation}
For a fixed decomposition of the total frequency, the eigenvector indices for different encoding layers can be chosen independently, so their counts multiply. Summing over all decompositions gives the convolution above. The total count is \(\sum_{\omega\in\Omega}|\mathcal{R}(\omega)|=d_{\mathcal H}^{2L}\), the number of ordered path pairs. Counting only distinct differences at each layer would set every \(r_\ell\) to one and lose the multiplicities due to degeneracy.

With this notation, each coefficient admits the decomposition
\begin{equation}
a_\omega(\theta)=\sum_{(J,J')\in \mathcal{R}(\omega)} A_{J,J'}(\theta),
\label{eq:app_path_decomp}
\end{equation}
with the amplitudes specified in Eq.~\eqref{eq:app_encoder_path_amplitude}. An allowed pair may have zero amplitude; nonzero contributions may cancel. For example, if the fixed input state commutes with a single encoder, then \(S(x)\rho_0S(x)^\dagger=\rho_0\), so no nonzero input frequency appears, regardless of its encoder redundancy.

\subsection{Pauli encoders and the numerical reference}\label{app:qfms_examples}

For \(S(x)=e^{-ixZ/2}\), the eigenvalues are \(\pm\tfrac12\). Its four ordered eigenvector pairs give multiplicities \(1,2,1\) at frequencies \(-1,0,1\): the two equal-index pairs both give zero. For \(n\) parallel encodings of the same scalar input, \(S(x)=\exp[-ix\sum_{q=1}^n Z_q/2]\), an eigenvector of Hamming weight \(w\) has eigenvalue \(n/2-w\), with multiplicity \(\binom nw\). Since \(\omega=w-w'\), one encoding layer gives
\begin{equation}
|\mathcal{R}(\omega)|=\sum_{w=0}^{n}
\binom{n}{w}\binom{n}{w-\omega}
=\binom{2n}{n+\omega},
\label{eq:app_pauli_redundancy}
\end{equation}
where out-of-range binomial coefficients are zero. Equivalently, the multiplicities are the coefficients of \((z^{-1}+2+z)^n=z^{-n}(1+z)^{2n}\), proving the binomial identity. For \(n=2\), the redundancies at \(\omega=-2,-1,0,1,2\) are \(1,4,6,4,1\), summing to \(4^2=16\).

With \(L\) such encoding layers, containing \(nL\) individual encoding gates in total, convolution gives
\begin{equation}
\Omega=\{-nL,\ldots,nL\},\qquad
|\mathcal{R}(\omega)|=\binom{2nL}{nL+\omega}.
\label{eq:app_repeated_pauli_redundancy}
\end{equation}
This counts all possible index pairs, even when intervening blocks make some amplitudes vanish.

Parallel \(R_X\) and \(R_Y\) encoders have the same eigenvalue multiplicities as \(R_Z\), so the same count applies. In our benchmark, \(L=1\) and the encoder scale is one: Eq.~\eqref{eq:app_pauli_redundancy} is precisely the reference used in the figures. The profile \(b_\omega\) in Eq.~\eqref{eq:numerical-hellinger-distance} normalises this count over the positive frequencies \(1,\ldots,n\). These counts remain fixed as the trainable depth \(d\) changes.

For weighted encoders \(\sum_q\alpha_q Z_q/2\), the local differences are \(0,\pm\alpha_q\); the general convolution still applies, but the simple binomial formula need not~\cite{schuld2021}.

\subsection{From encoder paths to CHM rows}
\label{app:encoder_path_chm}

Under the finite parameter-harmonic assumptions of Section~\ref{subsec:param_harmonics}, choose a common ambient \(\mathcal K\) large enough to expand every path amplitude. Inserting these expansions into Eq.~\eqref{eq:app_path_decomp} and collecting the same input frequency and parameter character gives
\begin{equation}
\begin{aligned}
 A_{J,J'}(\theta)&=\sum_{k\in\mathcal K}c_{(J,J'),k}e^{ik\cdot\theta},\\
 C_{\omega k}&=\sum_{(J,J')\in \mathcal{R}(\omega)}c_{(J,J'),k}.
\end{aligned}
 \label{eq:app_schuld_to_C}
\end{equation}
Each \(c_{(J,J'),k}\) is a signed, generally complex weight, and weights with the same indices are summed before taking CHM row norms or overlaps. Equal redundancy therefore need not give equal variance or covariance [Eqs.~\eqref{eq:coeffstats_variance} and~\eqref{eq:coeffstats_cov_components}].

\subsection{Encoder-path and circuit-resolved redundancy}
\label{app:redundancy_dictionary}

Different representations of the same model lead to different redundancy counts. In their general eigenbasis expansion [Section~3, Eqs.~(5)--(7)], Duffy and Jastrzebski~\cite{duffy2026} use the same kind of encoder-path multiplicity as Mhiri \emph{et al.}~\cite{mhiri2025a}. We denote this count by \(\mathcal{R}_{\mathrm{eig}}(\omega):=\mathcal{R}(\omega)\) to distinguish it from the circuit-resolved count introduced later in their analysis.

After Eq.~(13), Duffy and Jastrzebski define redundancy as the number of active multi-indices \((s,c,s',c')\) in their Clifford--Pauli trigonometric expansion. Here \(s,c\) count input sine/cosine factors and \(s',c'\) count parameter sine/cosine factors. A multi-index is active when its grouped circuit weight times its input-frequency expansion factor is nonzero. We denote this set by \(\mathcal{R}_{\mathrm{term}}(\omega)\); their circuit-resolved redundancy is its cardinality.

The latter count records nonvanishing summands in the trigonometric representation, before all cancellations have been collected. Different \((s,c)\) can accompany the same parameter polynomial, and expanding those polynomials into characters can produce further contributions to the same CHM entry. Consequently, neither \(|\mathcal{R}_{\mathrm{eig}}(\omega)|\) nor \(|\mathcal{R}_{\mathrm{term}}(\omega)|\) generally equals the number of nonzero CHM entries. To obtain those entries, one must retain the weights and sum them as in Eq.~\eqref{eq:app_schuld_to_C}.

Our variance-profile figures use \(|\mathcal{R}_{\mathrm{eig}}|\) for the design-based reference in Appendix~\ref{app:mhiri_chm}. Appendix~\ref{appendix:pauli_propagation} next constructs the circuit-dependent weights in \(C\).

\FloatBarrier
\section{Constructing the circuit harmonic matrix}
\label{appendix:pauli_propagation}

To construct the CHM, we propagate the observable backwards through the circuit and express the resulting trigonometric terms in input and parameter Fourier coordinates. This appendix derives the procedure used in Section~\ref{subsec:C_paths}. Its computational tree follows Nemkov \emph{et al.}, with the separation of input and parameter dependence developed by Wiedmann \emph{et al.} \cite{nemkovFourierExpansionVariational2023,wiedmann2024}. A key step is to combine branches before interpreting their coefficients, because intermediate contributions can cancel. Appendix~\ref{appendix:twirl} then shows how to compute moments without constructing the complete matrix.

\subsection{Gate convention and propagation order}
\label{app:pp_setup}

We use the expectation model $f(x;\theta)=\Tr[\rho_0 U(x;\theta)^\dagger O U(x;\theta)]$, with fixed state $\rho_0$ and Hermitian observable $O$. The construction applies when the elementary gates are fixed Cliffords or Pauli rotations $R_P(\alpha)=e^{-i\alpha P/2}$. For the single-use trainable structure, a block can be written
\begin{equation}
 W_\ell(\theta^{(\ell)})
 =\left(\prod_{r\in\mathcal I_\ell}^{\mathrm{ordered}}
          R_{P_r}(\theta_r)\right)V_\ell,
 \qquad V_\ell\ \text{Clifford}.
 \label{eq:app_pp_circuit}
\end{equation}
Here the product has a specified order and each independent parameter appears in exactly one rotation across all blocks. Fixed Clifford interleavings can be collected into $V_\ell$ by conjugating the Pauli generators; the calculation may instead retain the original gate list. The recursion below also accommodates shared parameters.

Include the encoder gates in the forward list $U(x;\theta)=h_G\cdots h_1$, and set $g_t=h_{G-t+1}$. Then
\begin{equation}
 O_0=O,\qquad O_t=g_t^\dagger O_{t-1}g_t,
 \quad t=1,\ldots,G,\qquad O_G=U^\dagger O U.
 \label{eq:app_pp_backprop}
\end{equation}
Thus $g_1$ is the gate adjacent to the observable, and $t$ counts elementary gates processed backwards, as in Appendix~\ref{appendix:twirl}.

\subsection{From Pauli branching to complete paths}
\label{app:pp_dichotomy}

The local identity fixes both the branch factors and their signs.
\begin{lemma}[Pauli conjugation branching]
\label{lem:pauli_conjugation}
Let $P$ and $Q$ be Hermitian Pauli strings. If $[P,Q]=0$, then $R_P(\alpha)^\dagger Q R_P(\alpha)=Q$. If $\{P,Q\}=0$, then
\begin{equation}
 R_P(\alpha)^\dagger Q R_P(\alpha)
 =Q\cos\alpha+(iPQ)\sin\alpha,
 \label{eq:app_pp_dichotomy}
\end{equation}
where $iPQ$ is a Hermitian Pauli string up to a sign.
\end{lemma}
\begin{proof}
Since $P^2=I$, write $R_P(\alpha)=cI-isP$ with $c=\cos(\alpha/2)$ and $s=\sin(\alpha/2)$. Direct multiplication gives $c^2Q+s^2PQP+ics(PQ-QP)$. Commutation gives $Q$; anticommutation gives $(c^2-s^2)Q+2icsPQ$, which is Eq.~\eqref{eq:app_pp_dichotomy}. In the latter case $(iPQ)^\dagger=iPQ$ and $iPQ$ equals a signed Pauli string.
\end{proof}

Expand $O=\sum_a o_aQ_a$ in the unsigned Hermitian Pauli basis $\{I,X,Y,Z\}^{\otimes n}$, with $o_a\in\mathbb R$. A Clifford maps each basis string to another with a sign $\pm1$ \cite{gottesman1997,aaronson2004,dehaene2003}. A rotation either leaves it unchanged or creates the two branches in Lemma~\ref{lem:pauli_conjugation}. A complete path $p$ follows these choices through all $G$ reversed gates.

For an encoder angle $\lambda_g\cdot x$ and a trainable angle $\theta_{r(g)}$, let $E_p^{\sin},E_p^{\cos}$ and $T_p^{\sin},T_p^{\cos}$ be the corresponding gate occurrences at which the path takes a sine or cosine branch. Define
\begin{align}
 t_p^{(x)}(x)
 &=\prod_{g\in E_p^{\sin}}\sin(\lambda_g\cdot x)
   \prod_{g\in E_p^{\cos}}\cos(\lambda_g\cdot x),
 \label{eq:C_txp_explicit}\\
 t_p^{(\theta)}(\theta)
 &=\prod_{g\in T_p^{\sin}}\sin\theta_{r(g)}
   \prod_{g\in T_p^{\cos}}\cos\theta_{r(g)}.
 \label{eq:C_tthetap_explicit}
\end{align}
Empty products are one. Keeping gate occurrences distinct allows several gates to share a parameter index $r(g)$. If $P_p$ is the terminal unsigned Pauli and $\sigma_p$ contains the initial observable coefficient and all propagation signs, then
\begin{equation}
 O_G=\sum_{p\in\mathcal P}\sigma_pP_p\,t_p^{(x)}t_p^{(\theta)},
 \qquad \alpha_p=\sigma_p\Tr(\rho_0P_p).
 \label{eq:app_pp_OL_path_sum}
\end{equation}
Taking the trace gives the scalar path expansion Eq.~\eqref{eq:path_expansion_fx}. For $\rho_0=|0\rangle\langle0|^{\otimes n}$, the terminal overlap is one for strings containing only $I$ and $Z$, and zero otherwise. For a general state, the overlap lies in $[-1,1]$.

\subsection{Grouped recursion and its invariant}
\label{app:pp_kappa_recursion}

Complete paths need not be stored separately. For coordinate encoder angles $x_j$, define the trigonometric monomial
\begin{equation}
 \tau_{u,v}(z)=\prod_j\sin(z_j)^{u_j}\cos(z_j)^{v_j}.
 \label{eq:app_trig_monomial}
\end{equation}
The vectors $s,c\in\mathbb N_0^{d_x}$ count input sine and cosine factors; $s',c'\in\mathbb N_0^m$ count parameter factors. Maintain a dictionary
\begin{equation}
 \mathcal D_t:(Q,s,c,s',c')\longmapsto\gamma,\qquad
 \mathcal D_0(Q_a,0,0,0,0)\mathrel{+}=o_a.
 \label{eq:app_recursive_dictionary}
\end{equation}
Each dictionary entry stores a real weight until the Fourier conversion. Entries are merged when all components of their keys match; the Pauli label is retained because it determines how subsequent gates act.

\begin{definition}[Grouped computational-tree recursion]
\label{rec:kappa}
For each gate, start an empty $\mathcal D_{t+1}$ and add the contributions of every entry $\mathcal D_t(Q,s,c,s',c')=\gamma$ as follows.
\begin{enumerate}
\item For a Clifford $V$, write $V^\dagger QV=\chi_V(Q)Q_V$ and add $\gamma\chi_V(Q)$ to the key $(Q_V,s,c,s',c')$.
\item For an encoder $R_P(x_j)$ commuting with $Q$, copy the entry. Otherwise write $iPQ=\chi(P,Q)Q_{P,Q}$ and add
\begin{align}
 \mathcal D_{t+1}(Q,s,c+e_j,s',c')&\mathrel{+}=\gamma,\nonumber\\
 \mathcal D_{t+1}(Q_{P,Q},s+e_j,c,s',c')
     &\mathrel{+}=\gamma\chi(P,Q).
 \label{eq:app_recursive_encoder_update}
\end{align}
\item For a trainable rotation $R_P(\theta_r)$ commuting with $Q$, copy the entry. Otherwise add
\begin{align}
 \mathcal D_{t+1}(Q,s,c,s',c'+e_r)&\mathrel{+}=\gamma,\nonumber\\
 \mathcal D_{t+1}(Q_{P,Q},s,c,s'+e_r,c')
     &\mathrel{+}=\gamma\chi(P,Q).
 \label{eq:app_recursive_parameter_update}
\end{align}
\end{enumerate}
Here $e_j,e_r$ are the appropriate standard basis vectors and all $\chi$ are signs $\pm1$. A negative angle can equivalently be handled by negating its Pauli generator. Contributions to an existing key are added.
\end{definition}

The correctness statement is the operator invariant
\begin{equation}
 O_t(x;\theta)=
 \sum_{Q,s,c,s',c'}\mathcal D_t(Q,s,c,s',c')\,
 Q\,\tau_{s,c}(x)\tau_{s',c'}(\theta).
 \label{eq:app_recursive_invariant}
\end{equation}
\begin{proof}
Initialisation gives the Pauli expansion of $O_0$. If the identity holds at step $t$, conjugate every term by $g_{t+1}$. Clifford relabelling or Lemma~\ref{lem:pauli_conjugation} gives exactly the stated updates. Combining identical keys preserves the operator sum, proving the invariant by induction.
\end{proof}

Only after the last gate can the Pauli label be removed by tracing:
\begin{equation}
 \kappa_{s,c,s',c'}
 =\sum_Q\mathcal D_G(Q,s,c,s',c')\Tr(\rho_0Q),
 \qquad
 f=\sum_{s,c,s',c'}\kappa_{s,c,s',c'}
       \tau_{s,c}(x)\tau_{s',c'}(\theta).
 \label{eq:app_kappa_terminal}
\end{equation}
The resulting grouped coefficients are those of Wiedmann's representation, including the state, observable and propagation signs. The multiplicity labels help organise the calculation, but the functions they label need not be linearly independent. For example, $\sin^2z+\cos^2z=1$ relates terms with different labels. Further cancellations can therefore appear when the terms are collected into Fourier coefficients, even if none of their dictionary keys match.

The implementation's floating-point pruning tolerance is discussed in Appendix~\ref{subsec:appendix_numerical_limits}.

\subsection{Finite Fourier weights and CHM assembly}
\label{app:pp_to_C}

For coordinate angles, both monomials in Eq.~\eqref{eq:app_kappa_terminal} use the same finite transform. Expanding $\sin z=(e^{iz}-e^{-iz})/(2i)$ and $\cos z=(e^{iz}+e^{-iz})/2$ gives the signed binomial factor
\begin{equation}
 p(u,v;q)=\prod_{j=1}^{D}
 \left[
 \sum_{a=0}^{u_j}\sum_{b=0}^{v_j}
 \binom{u_j}{a}\binom{v_j}{b}
 (-1)^{u_j-a}\delta_{q_j,\,2a+2b-u_j-v_j}
 \right],
 \label{eq:app_wiedmann_p}
\end{equation}
where $u,v\in\mathbb N_0^D$, $q\in\mathbb Z^D$, and $\delta$ is the Kronecker delta. Choosing $a$ positive exponentials from the sine factors and $b$ from the cosine factors gives the exponent $2a+2b-u_j-v_j$ and the sign $(-1)^{u_j-a}$. Thus the \emph{complete} Fourier weight is
\begin{equation}
 \mathfrak f(u,v;q)
 =2^{-\lVert u+v\rVert_1}(-i)^{\lVert u\rVert_1}p(u,v;q),
 \qquad
 \tau_{u,v}(z)=\sum_q\mathfrak f(u,v;q)e^{iq\cdot z}.
 \label{eq:app_q_definition}
\end{equation}
The integer factor $p$ is the one used by Wiedmann \emph{et al.} \cite{wiedmann2024}; $\mathfrak f$ includes its normalisation and phase.

Collecting equal joint harmonics now yields
\begin{equation}
 C_{\omega k}
 =\sum_{s,c,s',c'}\kappa_{s,c,s',c'}\,
      \mathfrak f(s,c;\omega)\mathfrak f(s',c';k).
 \label{eq:app_C_explicit_kappa}
\end{equation}
This is Eq.~\eqref{eq:C_path_factorisation} after grouping complete paths by their multiplicities. The recursion determines $\kappa$, and the binomial transform determines every Fourier weight.

For a weighted encoder angle $\lambda_g\cdot x$, it is simpler to store input frequencies directly. A cosine branch shifts the frequency by $\pm\lambda_g$ with weight $1/2$; a sine branch uses weights $1/(2i)$ and $-1/(2i)$ for the positive and negative shifts. Equal frequency-labelled keys are again combined by addition. This algebraic construction permits $\lambda_g\in\mathbb R^{d_x}$; an integer frequency grid gives the usual periodic input Fourier series. The two-copy algorithm uses this frequency-labelled version.

\subsection{Intermediate frequencies and the encoder envelope}
\label{app:pp_intermediate_support}

The encoder envelope constrains the complete model, but it need not constrain every intermediate path. Appendix~\ref{appendix:qfms} derives $\Omega$ from the spectra of the \emph{complete encoder blocks}. When those blocks are decomposed into individual gates, a Pauli path can contain frequencies outside $\Omega$ that cancel in the final sum. The sums in Eqs.~\eqref{eq:C_txp_fourier_expansion} and
\eqref{eq:C_tthetap_fourier_expansion} therefore use each path
factor's own finite support, with zero coefficients elsewhere.

For a simple cancellation, regard $S(x)=R_Y(-x)R_Y(x)=I$ as one encoder block. Its generator is zero, so $\Omega=\{0\}$. Propagating $O=Z$ through its two rotations and tracing against $|0\rangle\langle0|$ gives two surviving path contributions, $\cos^2x$ and $\sin^2x$. Each contains frequencies $\pm2$, but their sum is one. The extra frequencies cancel and are not CHM rows of the complete model.

Consequently, intermediate labels must be retained until the remaining shifts and coherent sums are resolved. Under the encoder assumptions of Appendix~\ref{appendix:qfms}, the final sum is supported in $\Omega$; only then is it identified with the CHM indexed by that envelope. The benchmark's parallel unit-scale $R_X$ or $R_Y$ encoder has $\Omega=\{-n,\ldots,n\}$ and needs no larger final grid.

\subsection{Single-use support and parameter sharing}
\label{app:pp_trig_to_char}
\label{appendix:k_mode_scaling}

Under single-use, a complete path has disjoint sets of sine and cosine parameter indices, $A_s(p)$ and $A_c(p)$. Put $A_p=A_s(p)\cup A_c(p)$ and $b_p=|A_p|$.
\begin{lemma}[Character support of a single-use path]
\label{lem:app_pp_path_support}
The coefficient of $e^{ik\cdot\theta}$ in $t_p^{(\theta)}$ is
\begin{equation}
 \widehat t_p^{(\theta)}(k)
 =2^{-b_p}(-i)^{|A_s(p)|}\prod_{r\in A_s(p)}k_r
 \label{eq:app_pp_path_weight}
\end{equation}
when $k_r=\pm1$ on $A_p$ and $k_r=0$ elsewhere, and is zero otherwise. Hence the path generates exactly $2^{b_p}$ characters, each of magnitude $2^{-b_p}$.
\end{lemma}
\begin{proof}
Each cosine contributes coefficient $1/2$ at both signs, and each sine contributes $-ik_r/2$. Distinct parameter indices make the sign choices independent, so multiplying these factors gives the formula and count.
\end{proof}

The union of these path supports lies in $\{-1,0,1\}^m$, so the realised parameter support of the CHM lies there as well. Let $\mathcal P_{\mathrm{nz}}$ consist of paths with $\alpha_p\neq0$ and nonzero input factor. Define
\begin{equation}
 S_{\mathrm{gen}}=
 \left|\bigcup_{p\in\mathcal P_{\mathrm{nz}}}
       \operatorname{supp}\widehat t_p^{(\theta)}\right|,
 \qquad
 \operatorname{supp}_k C=\{k:\exists\,\omega,\ C_{\omega k}\neq0\}.
 \label{eq:app_generated_support}
\end{equation}
For nonempty $\mathcal P_{\mathrm{nz}}$, let $b_{\max}=\max_{p\in\mathcal P_{\mathrm{nz}}}b_p$. Set inclusion gives
\begin{equation}
 2^{b_{\max}}\leq S_{\mathrm{gen}}\leq3^m,
 \qquad |\operatorname{supp}_k C|\leq S_{\mathrm{gen}}.
 \label{eq:app_generated_and_realised_support}
\end{equation}
If there are no such paths, both supports are empty. The bounds count generated characters and nonzero CHM columns over the model as a whole. Cancellations can reduce the realised support below the lower bound on generated characters.

Repeated parameters remain valid in the grouped recursion, but their multiplicities can exceed one: two uses of $\theta_r$ can produce $k_r=\pm2$. In particular, the controlled-rotation decomposition Eq.~\eqref{eq:cr_pauli_decomposition} consists of two Pauli rotations with angles $\theta=\vartheta/2$ and $-\theta$, sharing one parameter. Its one-copy harmonics can reach $\pm2$ in this $2\pi$ character coordinate; the physical angle $\vartheta$ has period $4\pi$. The finite CHM construction still applies, but the single-use support bound does not. Both factors must retain the same parameter label so that their dependence on the shared angle is preserved when averaging. Remark~\ref{rem:cr_gates} gives the joint local averaging rule used for these gates.

\FloatBarrier
\section{Coefficient moments and the design-based reference}
\label{appendix:covariance_proofs}
\label{appendix:literature_bridge}

We first derive the coefficient-moment identities used in Section~\ref{sec:coeff_stats} from the CHM. We then examine the additional design assumptions used by Mhiri \emph{et al.}~\cite{mhiri2025a} and adapt their moment calculation to a fixed input state. This explains why the same reference shape for nonzero-frequency variances is relevant to our benchmark, despite the difference in circuit layout.

\subsection{Ensemble space}
\label{app:char_orth}

The independent character coordinates are uniform on $\mathbb T^m=[0,2\pi)^m$, with probability measure $\mathrm d\mu(\theta)=(2\pi)^{-m}\prod_r\mathrm d\theta_r$. These are the rescaled coordinates of Eq.~\eqref{eq:prelims_period_rescaling}: for a physical controlled-rotation angle of period $4\pi$, $\theta_r=\vartheta_r/2$.

\begin{lemma}[Orthogonality of parameter characters]
\label{lem:app_char_orth}
For $\psi_k(\theta)=e^{ik\cdot\theta}$ and $k,l\in\mathbb Z^m$,
\begin{equation}
 \mathbb E_\theta[\psi_k\overline{\psi_l}]=\delta_{kl},
 \qquad \mathbb E_\theta[\psi_k]=\delta_{k0}.
 \label{eq:app_char_orth}
\end{equation}
\end{lemma}
\begin{proof}
The first expectation factors into $\prod_r(2\pi)^{-1}\int_0^{2\pi}e^{i(k_r-l_r)\theta_r}\mathrm d\theta_r$. Each factor is one when $k_r=l_r$ and zero otherwise. Taking $l=0$ gives the second identity \cite{katznelson2004,stein2003}.
\end{proof}

\subsection{Means, covariance and total coefficient weight}
\label{app:cov_CCdag}

Write the finite coefficient expansion as
\begin{equation}
 a(\theta)=C\psi(\theta),\qquad
 a_\omega(\theta)=\sum_{k\in\mathcal K}C_{\omega k}e^{ik\cdot\theta}.
 \label{eq:app_a_expansion}
\end{equation}
As in the main text, $\mathcal K$ includes $0$ by convention; its column may be zero. Let $\mathcal K^\circ=\mathcal K\setminus\{0\}$ and let $C^\circ$ omit that column.

\begin{proposition}[Coefficient moments from the CHM]
\label{prop:app_moments_from_C}
Under the uniform parameter-torus ensemble,
\begin{align}
 \bar a:=\mathbb E_\theta[a]&=C_{\bullet0},
 \label{eq:app_mean}\\
 S:=\mathbb E_\theta[aa^\dagger]&=CC^\dagger,
 \label{eq:app_second_moment}\\
 \Sigma:=\mathbb E_\theta[(a-\bar a)(a-\bar a)^\dagger]
 &=C^\circ C^{\circ\dagger}.
 \label{eq:app_cov}
\end{align}
\end{proposition}
\begin{proof}
The mean selects $k=0$ by Lemma~\ref{lem:app_char_orth}. Componentwise,
\begin{equation}
 S_{\omega\mu}
 =\sum_{k,l}C_{\omega k}\overline{C_{\mu l}}\,
       \mathbb E_\theta[e^{i(k-l)\cdot\theta}]
 =\sum_k C_{\omega k}\overline{C_{\mu k}}.
\end{equation}
Subtracting $\bar a\bar a^\dagger$ removes exactly the constant-column term, giving
\begin{equation}
 \Sigma_{\omega\mu}
 =\sum_{k\in\mathcal K^\circ}C_{\omega k}\overline{C_{\mu k}}.
 \label{eq:app_cov_components}
\end{equation}
\end{proof}

The complex variance is $v_\omega=\mathbb E_\theta[|a_\omega-\bar a_\omega|^2]$. Taking the diagonal gives the rowwise Parseval identities
\begin{align}
 v_\omega&=\|C^\circ_{\omega\bullet}\|_2^2,
 \label{eq:app_variance}\\
 \mathbb E_\theta[|a_\omega|^2]
 &=\|C_{\omega\bullet}\|_2^2=|\bar a_\omega|^2+v_\omega.
 \label{eq:app_parseval}
\end{align}
Because the squared row norm is a sum of nonnegative terms, zero variance forces every nonconstant character weight in that row to vanish. The coefficient is then constant in the parameters, although its mean can still be nonzero and the frequency need not be absent.

\paragraph{Total coefficient weight.}
The operator bound $|f(x;\theta)|\leq\|O\|_\infty$ also constrains how much squared weight can be distributed among the coefficients \cite{mhiri2025a}. For the integer input frequencies used here, Parseval on the input torus gives, for each $\theta$, $\sum_\omega|a_\omega(\theta)|^2
=\mathbb E_x[|f(x;\theta)|^2]\leq\|O\|_\infty^2$.
Averaging over parameters then yields
\begin{equation}
 \|C\|_{\mathrm F}^2\leq\|O\|_\infty^2,
 \qquad
 \sum_\omega v_\omega=\|C^\circ\|_{\mathrm F}^2
 \leq\|O\|_\infty^2-\|\bar a\|_2^2.
 \label{eq:app_coefficient_budget}
\end{equation}
Here $\mathbb E_x$ is the uniform input integral. For arbitrary distinct real input frequencies, the same bound follows by averaging over expanding boxes in $\mathbb R^{d_x}$: each unequal-frequency exponential has limiting average zero. The bound constrains the total coefficient weight; individual variances can remain finite as the number of frequencies grows.

\subsection{Covariance, normalised correlation and their limits}
\label{app:correlation_gram}

Equation~\eqref{eq:app_cov_components} is a weighted overlap of nonconstant CHM rows. Shared characters are necessary for a nonzero overlap but not sufficient: rows proportional to $(1,1)$ and $(1,-1)$ share both columns and have zero inner product. For $\Omega_+=\{\omega:v_\omega>0\}$, define
\begin{equation}
 C^{\mathrm{unit}}_{\omega k}
 =\frac{C^\circ_{\omega k}}{\sqrt{v_\omega}},
 \qquad \omega\in\Omega_+,\quad k\in\mathcal K^\circ.
 \label{eq:coeffstats_C_hat}
\end{equation}
Then the correlation matrix of Eq.~\eqref{eq:coeffstats_corr_def} is
\begin{equation}
 R=C^{\mathrm{unit}}C^{\mathrm{unit}\dagger}.
 \label{eq:app_correlation_gram}
\end{equation}
It is positive semidefinite, has unit diagonal, and obeys $|R_{\omega\mu}|\leq1$ by Cauchy--Schwarz. Equality holds precisely when the two nonzero centred rows are proportional.

Normalisation removes the size of the marginal fluctuations: positively rescaling either row changes covariance but leaves correlation unchanged. Two very small fluctuations can therefore be strongly correlated, and a decrease in covariance need not produce a decrease in correlation.

Zero correlation can coexist with functional dependence: for example, $e^{i\theta}$ and $e^{2i\theta}$ have zero Hermitian covariance, although the second is determined by the first through squaring.

\subsection{The independent-block design reference}
\label{app:mhiri_chm}

Mhiri \emph{et al.}'s $c_\omega$ is our $a_\omega$. Their encoder redundancy $|\mathcal{R}(\omega)|$ counts the ordered eigenvector index-path pairs defined in Appendix~\ref{appendix:qfms}. For a pair $p=(J,J')$, let $c_{p,k}$ be its parameter-character weight in Eq.~\eqref{eq:app_schuld_to_C}. The general identity is
\begin{equation}
 v_\omega=\sum_{k\neq0}
     \left|\sum_{p\in \mathcal{R}(\omega)}c_{p,k}\right|^2
 =\sum_{p,q\in \mathcal{R}(\omega)}\sum_{k\neq0}
       c_{p,k}\overline{c_{q,k}}.
 \label{eq:app_mhiri_weighted_pairs}
\end{equation}
The variance therefore depends on the weights of pairs of contributions and on their interference. Under suitable design symmetry, these weighted sums simplify to redundancy multiplied by a frequency-independent factor.

For one encoder, consider $U=W_{\mathrm{post}}S(x)W_{\mathrm{pre}}$ with a pure initial state and independently sampled trainable blocks forming exact unitary $2$-designs on dimension $d_{\mathcal H}=2^n$. A unitary $2$-design matches Haar averages of polynomials containing at most two unitary entries and two conjugate entries, which is sufficient for the moments below \cite{dankert2009a,mele2024}. In CHM language, the assumption is on the distributions induced by the parameters of the two blocks; uniform angles alone do not imply it. Define
\begin{equation}
 A_O=d_{\mathcal H}\Tr(O^2)-\Tr(O)^2,\qquad
 \zeta_O=A_O/d_{\mathcal H}^2.
 \label{eq:app_observable_scale}
\end{equation}
Since $O$ is Hermitian, $\Tr(O^2)$ is its squared Hilbert--Schmidt norm; $A_O\geq0$ and is unchanged by adding a multiple of the identity. The single-encoder moments are
\begin{align}
 \mathbb E_\theta[a_\omega]
 &=\frac{\Tr(O)}{d_{\mathcal H}}\delta_{\omega0},
 \label{eq:app_mhiri_mean}\\
 v_\omega
 &=\frac{A_O\bigl(|\mathcal{R}(\omega)|-\delta_{\omega0}\bigr)}
 {d_{\mathcal H}^2(d_{\mathcal H}+1)(d_{\mathcal H}^2-1)}.
 \label{eq:app_mhiri_exact_variance}
\end{align}
The nonzero-frequency formula is the result of Ref.~\cite{mhiri2025a}. The derivation below also includes the variance of the input-independent offset $a_0$.

\paragraph{The two moment contractions.}
Let $\rho=W_{\mathrm{pre}}\rho_0W_{\mathrm{pre}}^\dagger$ and, in the encoder eigenbasis, define
\begin{equation}
 \rho_\omega=\sum_{(j,j')\in \mathcal{R}(\omega)}
       \rho_{jj'}|j\rangle\langle j'|,
 \qquad
 S(x)\rho S(x)^\dagger=\sum_\omega\rho_\omega e^{i\omega\cdot x}.
 \label{eq:app_encoded_state_modes}
\end{equation}
For this single encoding layer, $\mathcal{R}(\omega)$ consists of pairs with $\lambda_{j'}-\lambda_j=\omega$, and $\Tr(\rho_\omega)=\delta_{\omega0}$. Put $B=W_{\mathrm{post}}^\dagger (O-\Tr(O)I/d_{\mathcal H})W_{\mathrm{post}}$. The post-block mean is $\mathbb E[B]=0$. Its second moment is
\begin{equation}
 \mathbb E[B\otimes B]
 =\frac{A_O}{d_{\mathcal H}(d_{\mathcal H}^2-1)}
       \left(F-\frac{I\otimes I}{d_{\mathcal H}}\right),
 \label{eq:app_design_observable_twirl}
\end{equation}
where $F$ swaps the two copies. Indeed, Haar invariance restricts this average to a linear combination of $I\otimes I$ and $F$; its trace is zero and its contraction with $F$ is $\Tr(B^2)=A_O/d_{\mathcal H}$. These two conditions fix the coefficients \cite{mele2024}. Using $\Tr[F(X\otimes Y)]=\Tr(XY)$ gives, for any fixed $\rho$,
\begin{equation}
 \mathbb E_{\mathrm{post}}
       [|\Tr(B\rho_\omega)|^2\mid\rho]
 =\frac{A_O}{d_{\mathcal H}(d_{\mathcal H}^2-1)}
       \left(\Tr(\rho_\omega\rho_\omega^\dagger)
                    -\frac{\delta_{\omega0}}{d_{\mathcal H}}\right).
 \label{eq:app_post_design_contraction}
\end{equation}
The pre-block pure-state second moment is $\mathbb E_{\mathrm{pre}}[\rho\otimes\rho]
=(I\otimes I+F)/[d_{\mathcal H}(d_{\mathcal H}+1)]$.
Consequently $\mathbb E_{\mathrm{pre}}[|\rho_{jj'}|^2]
=(1+\delta_{jj'})/[d_{\mathcal H}(d_{\mathcal H}+1)]$, and
\begin{equation}
 \mathbb E_{\mathrm{pre}}
       [\Tr(\rho_\omega\rho_\omega^\dagger)]
 =\frac{|\mathcal{R}(\omega)|+d_{\mathcal H}\delta_{\omega0}}
 {d_{\mathcal H}(d_{\mathcal H}+1)}.
 \label{eq:app_mhiri_centered_contraction}
\end{equation}
Combining the two independent averages proves Eq.~\eqref{eq:app_mhiri_exact_variance}, retaining eigenvalue degeneracy. In particular, the bracket in Eq.~\eqref{eq:app_post_design_contraction} averages to
\[
 \frac{|\mathcal{R}(\omega)|+d_{\mathcal H}\delta_{\omega0}}
      {d_{\mathcal H}(d_{\mathcal H}+1)}
 -\frac{\delta_{\omega0}}{d_{\mathcal H}}
 =\frac{|\mathcal{R}(\omega)|-\delta_{\omega0}}
       {d_{\mathcal H}(d_{\mathcal H}+1)}.
\]
The $d_{\mathcal H}$ diagonal state entries all belong to zero input frequency and have twice the second moment of an off-diagonal entry. Retaining their contribution is essential to this cancellation.

Equation~\eqref{eq:app_mhiri_exact_variance} also follows from the single-layer specialisation of Eqs.~(B8), (B12), (B19) and (B20) of Ref.~\cite{mhiri2025a}, followed by subtraction of the squared mean. Specifically, the term $\delta_{\omega0}/(d_{\mathcal H}+1)$ in their Eq.~(B20) contributes $A_O\delta_{\omega0}/[d_{\mathcal H}(d_{\mathcal H}+1) (d_{\mathcal H}^2-1)]$. This contribution is absent from their Eq.~(B1) as printed; restoring it gives the expression above. This discrepancy concerns only the zero-frequency term of that displayed formula, not the nonzero-frequency redundancy result. Our plotted redundancy-shape comparisons use only positive frequencies and are unaffected.

With $\widetilde{\mathcal{R}}(\omega)=|\mathcal{R}(\omega)|/d_{\mathcal H}^2$, the same statement in CHM notation is
\begin{equation}
 \|C^\circ_{\omega\bullet}\|_2^2
 =\frac{\zeta_O\bigl(d_{\mathcal H}^2\widetilde{\mathcal{R}}(\omega)
                         -\delta_{\omega0}\bigr)}
 {(d_{\mathcal H}+1)(d_{\mathcal H}^2-1)}.
 \label{eq:mhiri_chm_exact_design}
\end{equation}

\subsection{Why the same reference shape applies to our encoding}
\label{app:fixed_input_design_reference}

Our benchmark has no trainable pre-encoding block. Its state is $|0\rangle^{\otimes n}$ and its single encoder is a parallel unit-scale $R_X$ or $R_Y$ layer. In either product encoder eigenbasis, $|\langle j|0^{\otimes n}\rangle|^2=1/d_{\mathcal H}$, so every $|\rho_{jj'}|^2=1/d_{\mathcal H}^2$. Therefore, for this fixed state,
\begin{equation}
 \Tr(\rho_\omega\rho_\omega^\dagger)
 =\frac{|\mathcal{R}(\omega)|}{d_{\mathcal H}^2}.
 \label{eq:app_flat_encoder_state}
\end{equation}
If the \emph{post-encoding} trainable block alone formed an exact unitary $2$-design, Eq.~\eqref{eq:app_post_design_contraction} would give
\begin{equation}
 v_\omega^{\mathrm{post}}
 =\frac{A_O\bigl(|\mathcal{R}(\omega)|-d_{\mathcal H}\delta_{\omega0}\bigr)}
 {d_{\mathcal H}^3(d_{\mathcal H}^2-1)}.
 \label{eq:app_fixed_input_design_variance}
\end{equation}
At $\omega\neq0$, this formula has the redundancy shape of Eq.~\eqref{eq:app_mhiri_exact_variance}, with a prefactor larger by $(d_{\mathcal H}+1)/d_{\mathcal H}$. The zero-frequency corrections also differ. The common shape depends on the fixed input used here: for a general fixed input, the weights $|\rho_{jj'}|^2$ need not be equal, and a post-block design alone would not make raw redundancy determine the profile.

This motivates the comparison in Eq.~\eqref{eq:numerical-hellinger-distance}: normalising over $\omega=1,\ldots,n$ removes either frequency-independent prefactor and gives the same binomial reference $b_\omega\propto\binom{2n}{n+\omega}$. Agreement with this profile tests the predicted shape of the variances; establishing a unitary $2$-design would require agreement of all second unitary moments.

\subsection{Scope of the comparison and size scaling}
\label{app:design_scope}

For the two-block reference and $\omega\neq0$,
\begin{equation}
 v_\omega
 =\zeta_O\frac{d_{\mathcal H}^2}
 {(d_{\mathcal H}+1)(d_{\mathcal H}^2-1)}
       \widetilde{\mathcal{R}}(\omega).
 \label{eq:app_mhiri_asymptotic}
\end{equation}
Here $0\leq\widetilde{\mathcal{R}}(\omega)\leq1$ for one encoding layer and $\zeta_O\leq\|O\|_\infty^2$. Thus a sequence of bounded-norm observables gives $v_{\omega_n}=O(2^{-n})$ under those global-design assumptions; Eq.~\eqref{eq:app_fixed_input_design_variance} has the same order. Large raw multiplicities are compatible with small absolute variances because the dimension-dependent prefactor also changes.

Mhiri \emph{et al.}'s vanishing expressivity concerns the absolute variance along a sequence of frequencies $\omega_n\in\Omega_n$, such as a fixed interior frequency or the spectral edge. Normalised profiles discard the scale needed to assess this concentration. Its consequences for optimisation also depend on whether training can reach parameter settings that are rare under the ensemble.

Extending the reference requires further assumptions. With re-uploading, Mhiri \emph{et al.}'s calculation uses partial redundancy sets for intervals of encoder layers, so the final count $|\mathcal{R}(\omega)|$ alone is insufficient. Their approximate-design bounds contain a correction governed by a specified design distance. Their local-design result assumes local trainable blocks and an observable light cone. Applying these extensions to our circuits would require evaluating the corresponding partial redundancies, design distances or light-cone structure.

For the actual circuit ensembles, Eq.~\eqref{eq:app_mhiri_weighted_pairs} retains the circuit-dependent weights and cancellations, and Eq.~\eqref{eq:app_cov_components} retains their row overlaps. The two-copy calculation in Appendix~\ref{appendix:twirl} evaluates these moments directly.

\FloatBarrier
\section{Tangent spaces, coordinates and frozen predictions}
\label{app:resolved_kernel_tensors}

A coefficient can vary over the parameter ensemble without being independently adjustable at a particular parameter point. This appendix develops that distinction for the kernels in Section~\ref{sec:kernels}, derives exact restrictions on attainable outputs, and specifies the coordinate conventions needed to predict learning. We then obtain frozen predictions directly from sampled labels. The frozen dynamics use the standard tangent-kernel construction \cite{jacot2018,liu2022}; the CHM provides its coefficient-space factorisation and connection to ensemble moments. The corresponding numerical checks and forecasts appear in Appendices~\ref{subsec:appendix_kernel_checks} and~\ref{subsec:appendix_frozen_qntk}.

\subsection{Real Fourier coordinates and local control}
\label{app:schuld_chm}

For a real univariate output, write \(\Omega=\{0\}\cup\Omega_{\mathrm{pos}}\cup(-\Omega_{\mathrm{pos}})\), where \(\Omega_{\mathrm{pos}}\) contains the \(K_+\) positive input frequencies, regardless of their coefficient variances. Use the same independent real coordinates as the numerical study:
\begin{equation}
\begin{aligned}
f(x)&=a_0+\sum_{\omega\in\Omega_{\mathrm{pos}}}
\bigl(c_\omega\cos(\omega x)+s_\omega\sin(\omega x)\bigr),\\
q_{\mathrm{cs}}&=(a_0,c_{\omega^{(1)}},s_{\omega^{(1)}},\ldots)^{\mathsf T}
\in\mathbb R^{D},\qquad D=1+2K_+.
\end{aligned}
\label{eq:app_schuld_real_coeff_vector}
\end{equation}
Here \(c_\omega=2\Re a_\omega\) and \(s_\omega=-2\Im a_\omega\). For coefficient order \((a_0,a_{\omega^{(1)}},a_{-\omega^{(1)}},\ldots)\), the fixed invertible matrix \(T_{\mathrm{cs}}\) is specified by
\begin{equation}
a=T_{\mathrm{cs}}q_{\mathrm{cs}},\qquad
a_{\pm\omega}=\frac{c_\omega\mp i s_\omega}{2}.
\label{eq:kernel_real_basis_map}
\end{equation}

Initially take derivatives in character coordinates \(\theta\), and write \(X_{\mathbb R}=\partial_\theta q_{\mathrm{cs}}\). Differentiating Eq.~\eqref{eq:kernel_real_basis_map} and \(a=C\psi\) gives
\begin{equation}
T_{\mathrm{cs}}X_{\mathbb R}=X=C^\circ G^\circ,\qquad
G^\circ_{kr}=ik_r e^{ik\cdot\theta}.
\label{eq:app_schuld_real_jacobian}
\end{equation}
The constant character has zero derivative. Since \(X_{\mathbb R}\) is real and \(T_{\mathrm{cs}}\) is invertible,
\begin{equation}
\begin{aligned}
\operatorname{rank}_{\mathbb R}X_{\mathbb R}
&=\operatorname{rank}_{\mathbb C}X\\
&\leq\min\{m,\operatorname{rank}_{\mathbb C}C^\circ,D\}.
\end{aligned}
\label{eq:app_schuld_rank_bound}
\end{equation}
The basis map \(T_{\mathrm{cs}}\) converts assembled coefficient vectors and Jacobian columns, which obey conjugate symmetry. A single CHM column need not have that symmetry on its own: reality relates both \(\omega\leftrightarrow-\omega\) and \(k\leftrightarrow-k\). One therefore cannot convert each CHM column independently into real coefficients.

Independent first-order control of all \(D\) real coordinates means that \(X_{\mathbb R}\) has full row rank at the parameter point. The parameter count \(m\geq D\) discussed by Schuld \emph{et al.}~\cite{schuld2021} is necessary for that property, but not sufficient: the architecture's column space or the current character derivatives may impose further restrictions. By the submersion theorem, full row rank allows sufficiently small independent coefficient changes near that point; it does not guarantee global coefficient controllability. The universality theorem of Ref.~\cite{schuld2021} allows the encoder size, prepared state and effective observable to be chosen for the target. Its representational existence result in that flexible family does not establish independent control or successful optimisation for a fixed shallow ansatz.

\subsection{Pointwise and averaged tangent spaces}
\label{app:kernel_rank_geometry}

The Jacobian factorisation immediately gives
\begin{equation}
H=XX^\dagger,\qquad J=VX\in\mathbb R^{N\times m},\qquad
K=JJ^{\mathsf T}=VHV^\dagger.
\label{eq:app_kernel_factorisation}
\end{equation}
For any matrix \(A\), \(AA^\dagger\) has the same column space and rank as \(A\). Consequently,
\begin{equation}
\begin{aligned}
\operatorname{col}H&=\operatorname{col}X
\subseteq\operatorname{col}C^\circ,\\
\operatorname{rank}H&\leq\min\{m,\operatorname{rank}C^\circ\},\\
\operatorname{col}K&=\operatorname{col}(VX)
\subseteq\operatorname{col}(VC^\circ),\\
\operatorname{rank}K&\leq\min\{m,\operatorname{rank}(VC^\circ)\}.
\end{aligned}
\label{eq:app_kernel_pointwise_ranges}
\end{equation}
Thus the dataset may hide tangent directions that are present in coefficient space.

For independent uniform full-period parameters, character orthogonality gives
\begin{equation}
\begin{aligned}
\mathbb E_\theta[(G^\circ G^{\circ\dagger})_{kl}]
&=(k\cdot l)\mathbb E_\theta[e^{i(k-l)\cdot\theta}]\\
&=\|k\|_2^2\delta_{kl}.
\end{aligned}
\label{eq:app_kernel_character_average}
\end{equation}
Hence \(\bar H=C^\circ W_{\mathcal K}C^{\circ\dagger}\) and \(\bar K=V\bar H V^\dagger\), with \(W_{\mathcal K}=\operatorname{diag}_{k\in\mathcal K^\circ} (\|k\|_2^2)\). Every retained weight is strictly positive, so for any \(z\),
\begin{equation}
z^\dagger\bar H z
=\sum_{k\in\mathcal K^\circ}
\|k\|_2^2\,|(C^{\circ\dagger}z)_k|^2.
\label{eq:app_kernel_positive_weight_proof}
\end{equation}
It vanishes exactly when \(C^{\circ\dagger}z=0\). The covariance \(\Sigma=C^\circ C^{\circ\dagger}\) obeys the same argument with unit weights. Therefore
\begin{equation}
\begin{aligned}
\ker\bar H&=\ker\Sigma=\ker C^{\circ\dagger},\\
\operatorname{col}\bar H&=\operatorname{col}\Sigma
=\operatorname{col}C^\circ,\\
\operatorname{rank}\bar H&=\operatorname{rank}\Sigma
=\operatorname{rank}C^\circ.
\end{aligned}
\label{eq:kernels_covariance_average_kernel}
\end{equation}
Replacing \(C^\circ\) by \(VC^\circ\) proves the corresponding data-space statements, including \(\operatorname{col}\bar K=\operatorname{col}(VC^\circ)\) and \(\operatorname{rank}\bar K=\operatorname{rank}(VC^\circ)\).

These identities describe the complex linear span of tangent directions across parameter points. Combining these directions allows the averaged rank to exceed the rank at any individual point, as in Eq.~\eqref{eq:kernels_one_qubit_example}.

For predictions, averaging the matrix exponentials of individual kernels generally differs from exponentiating their average. A mean kernel therefore introduces an additional approximation to ensemble learning dynamics.

\subsection{Affine output restrictions and training error}
\label{app:kernel_affine_restrictions}

Some target error can be unavoidable regardless of the optimiser. The coefficient identity \(a-\bar a=C^\circ\psi^\circ\) gives \(F-\bar F=B\psi^\circ\), where \(B=VC^\circ\) and \(\bar F=V\bar a\). Every attainable output therefore lies in \(\bar F+\operatorname{col}B\). Orthogonal projection onto this affine space yields the error bound in Corollary~\ref{cor:kernels_affine_error}.

\begin{proof}[Proof of Corollary~\ref{cor:kernels_affine_error}]
For \(P_B=BB^+\), the orthogonal component of every residual is \((I-P_B)(F-y)=-(I-P_B)(y-\bar F)\). The squared norm of the full residual is at least the squared norm of this component. Dividing by \(2N\) and taking the infimum over parameters proves the bound.
\end{proof}

Since \(BB^\dagger\) and \(B\) have the same column space,
\begin{equation}
\begin{aligned}
\Sigma_F:=\operatorname{Cov}_\theta[F]
&=V\Sigma V^\dagger=BB^\dagger,\\
P_B&=\Sigma_F\Sigma_F^+.
\end{aligned}
\label{eq:app_kernel_affine_projector}
\end{equation}
The exact full-period covariance therefore suffices to obtain the projector without constructing the CHM. Although \(B\) uses complex characters, \(\Sigma_F\) is the covariance of real outputs and its projector is real.

The pseudoinverse includes singular cases. If \(B=0\), then \(P_B=0\) and the bound is the loss of the fixed output \(\bar F\), attained for every parameter choice. If \(B\) has full row rank, the bound vanishes for every label vector. More generally, a zero bound only places the labels in the affine span, not necessarily on the circuit's attainable set. For example, the one-parameter circuit in Eq.~\eqref{eq:kernels_one_qubit_example} traces a circle in real coefficient space. On a non-aliasing grid, zero labels lie in its affine output span but cannot be attained.

Discarding small but nonzero singular values shrinks the computed span and can raise the estimated error floor above the true minimum. A thresholded numerical rank therefore cannot certify the bound without separate control of the discarded directions.

The affine restriction also differs from the residual floor of a frozen kernel. At a given point, \(\operatorname{col}K(\vartheta_{\mathrm{fr}})\) can be smaller than \(\operatorname{col}B\), and its tangent directions can change as training proceeds.

\subsection{Fourier normalisation and parameter metrics}
\label{app:kernel_real_coordinates}
\label{app:kernel_coordinate_metrics}

Kernel calculations depend on how we represent both the Fourier coefficients and the trainable parameters. We first change the coefficient basis using Eq.~\eqref{eq:kernel_real_basis_map}, while keeping the parameter metric fixed. This gives
\begin{equation}
\begin{aligned}
H_{\mathbb R}&=X_{\mathbb R}X_{\mathbb R}^{\mathsf T},\qquad
H=T_{\mathrm{cs}}H_{\mathbb R}T_{\mathrm{cs}}^\dagger,\\
V_{\mathbb R}&=VT_{\mathrm{cs}},\qquad
K=V_{\mathbb R}H_{\mathbb R}V_{\mathbb R}^{\mathsf T}.
\end{aligned}
\label{eq:kernel_real_basis_contraction}
\end{equation}
The columns of \(V_{\mathbb R}\) are \(1,\cos(\omega x_i),\sin(\omega x_i)\). The data-space kernel is unchanged by a matched basis conversion. Ranks are preserved, but coefficient-kernel eigenvalues and effective rank need not be: \(T_{\mathrm{cs}}\) is not unitary.

On a non-aliasing uniform grid for integer input frequencies,
\begin{equation}
\begin{aligned}
V^\dagger V&=NI,\\
V_{\mathbb R}^{\mathsf T}V_{\mathbb R}
&=N\operatorname{diag}(1,\tfrac12,\tfrac12,\ldots).
\end{aligned}
\label{eq:app_kernel_grid_metric}
\end{equation}
Here \(V/\sqrt N\) is an isometry, even if rectangular, so \(K/N\) and \(H\) have the same nonzero eigenvalues. This does not generally hold for the unscaled \(H_{\mathbb R}\). Thus a coefficient-kernel spectrum cannot be interpreted as data-space decay rates without specifying the evaluation map and basis normalisation.

Second, changing the \emph{parameter} coordinates changes the Euclidean-gradient dynamics unless the metric is transformed too. For physical periods \(T_r\), put \(D_T=\operatorname{diag}_r(2\pi/T_r)\) and \(\theta=D_T\vartheta\). The chain rule gives
\begin{equation}
\begin{aligned}
X^{(\vartheta)}&=X^{(\theta)}D_T,\\
H^{(\vartheta)}&=X^{(\theta)}D_T^2X^{(\theta)\dagger}.
\end{aligned}
\label{eq:app_kernel_parameter_metric}
\end{equation}
Physical-angle gradient flow, expressed in character coordinates, is \(\dot\theta=-\eta D_T^2\nabla_\theta\mathcal L\). For mixed \(2\pi\) and \(4\pi\) periods, this is not in general absorbed by one global learning-rate rescaling. The corresponding ensemble weights are
\begin{equation}
w_k^{(\vartheta)}
=\sum_r(2\pi k_r/T_r)^2,
\label{eq:app_kernel_physical_weights}
\end{equation}
so a controlled-angle component contributes \(k_r^2/4\). All nonconstant weights remain positive, preserving the averaged range and nullspace results above. In the remaining subsections, \(X,J,H,K\) use the physical optimiser coordinates \(\vartheta\), as in the numerical forecasts.

\subsection{Frozen predictions on training inputs}
\label{app:kernel_frozen_predictions}

For the halved MSE \(\mathcal L=\|r\|_2^2/(2N)\), \(r=F-y\), Euclidean gradient flow gives the exact relations for the full circuit model,
\begin{equation}
\dot\vartheta=-\frac{\eta}{N}J^{\mathsf T}r,\qquad
\dot r=-\frac{\eta}{N}K(\vartheta)r.
\label{eq:kernels_gradient_flow}
\end{equation}
A forecast requires an approximation to the changing Jacobian. At a reference checkpoint \(\vartheta_{\mathrm{fr}}\), define \(F_{\mathrm{fr}}=F(\vartheta_{\mathrm{fr}})\), \(a_{\mathrm{fr}}=a(\vartheta_{\mathrm{fr}})\), and \(r_{\mathrm{fr}}=F_{\mathrm{fr}}-y\). The frozen model is
\begin{equation}
\begin{aligned}
\widehat F&=F_{\mathrm{fr}}+J_{\mathrm{fr}}\delta\vartheta,\\
\widehat a&=a_{\mathrm{fr}}+X_{\mathrm{fr}}\delta\vartheta,
\qquad \delta\vartheta(0)=0.
\end{aligned}
\label{eq:app_kernel_frozen_linearisation}
\end{equation}
Holding these Jacobians fixed recovers the usual tangent-kernel predictions \cite{jacot2018,liu2022}:
\begin{equation}
\begin{aligned}
\widehat r(\tau)&=e^{-\eta\tau K_{\mathrm{fr}}/N}r_{\mathrm{fr}},\\
\widehat r_j&=(I-\eta K_{\mathrm{fr}}/N)^j r_{\mathrm{fr}}.
\end{aligned}
\label{eq:app_kernel_frozen_solutions}
\end{equation}
The first line is continuous-time flow; the second is full-batch gradient descent on the linearised model. For the \emph{unhalved} MSE minimised in the code, every occurrence of \(\eta/N\) in these dynamics is replaced by \(2\eta/N\).

An eigenmode of \(K_{\mathrm{fr}}\) with eigenvalue \(\lambda>0\) has residual amplitude decaying as \(e^{-\eta\lambda\tau/N}\); its squared contribution to the loss decays with twice that exponent. A zero-eigenvalue residual persists. Discrete convergence in every positive-eigenvalue direction requires
\begin{equation}
0<\frac{\eta\lambda_{\max}(K_{\mathrm{fr}})}{N}<2,
\label{eq:kernels_frozen_stability}
\end{equation}
with the same loss-normalisation conversion. If \(K_{\mathrm{fr}}=0\), the frozen outputs do not change. Kernel eigenmodes need not be individual Fourier frequencies.

The accuracy of these predictions depends on how much the kernel changes over the forecast interval. Freezing at a later checkpoint gives a new local model for the subsequent evolution. Our numerical controls assess these forecasts under full-batch gradient descent. Adam's adaptive updates require a different dynamical model. Appendix~\ref{app:kernel_test_predictions} derives the corresponding forecasts for coefficients and test inputs.

\subsection{Coefficient and test-input predictions}
\label{app:kernel_test_predictions}

The same model predicts coefficient changes directly from the sampled residuals. Since \(J_{\mathrm{fr}}=VX_{\mathrm{fr}}\), \(\dot{\widehat a}=-(\eta/N)H_{\mathrm{fr}}V^\dagger\widehat r\). Integrating gives
\begin{equation}
\widehat a(\tau)=a_{\mathrm{fr}}
-H_{\mathrm{fr}}V^\dagger K_{\mathrm{fr}}^+
(I-e^{-\eta\tau K_{\mathrm{fr}}/N})r_{\mathrm{fr}}.
\label{eq:kernels_frozen_coefficient_prediction}
\end{equation}
Here \(+\) denotes the Moore--Penrose pseudoinverse. To see why this also covers singular training kernels, note that
\begin{equation}
\ker K_{\mathrm{fr}}=\ker J_{\mathrm{fr}}^{\mathsf T},\qquad
H_{\mathrm{fr}}V^\dagger=X_{\mathrm{fr}}J_{\mathrm{fr}}^{\mathsf T}.
\label{eq:app_kernel_nullspace_annihilation}
\end{equation}
The nondecaying residual component therefore makes no coefficient update. Each positive eigenvalue contributes \((1-e^{-\eta\tau\lambda/N})/\lambda\), yielding Eq.~\eqref{eq:kernels_frozen_coefficient_prediction}. Multiplying by \(V\) recovers \(\widehat F=y+\widehat r\).

For a test input \(x_*\), let \(v(x_*)_\omega=e^{i\omega x_*}\). The required cross-kernel is
\begin{equation}
k(x_*,\mathsf X)=v(x_*)H_{\mathrm{fr}}V^\dagger,\qquad
\mathsf X=\{x_i\}_{i=1}^N.
\label{eq:kernels_test_cross_kernel}
\end{equation}
Equivalently, this is the overlap of the test-input parameter gradient with each training-input gradient. Evaluating the coefficient forecast gives
\begin{equation}
\widehat f_*(\tau)=f(x_*;\vartheta_{\mathrm{fr}})
-k(x_*,\mathsf X)K_{\mathrm{fr}}^+
(I-e^{-\eta\tau K_{\mathrm{fr}}/N})r_{\mathrm{fr}}.
\label{eq:kernels_frozen_test_finite}
\end{equation}
Its continuous-time limit is \(\widehat f_*(\infty)=f(x_*;\vartheta_{\mathrm{fr}})
-k(x_*,\mathsf X)K_{\mathrm{fr}}^+r_{\mathrm{fr}}\), the usual
frozen-kernel extrapolation \cite{liu2022}, including singular kernels. For discrete descent, replace the exponential in both forecasts by \((I-\eta K_{\mathrm{fr}}/N)^j\); the same limit holds under Eq.~\eqref{eq:kernels_frozen_stability}.

These forecasts use the reference model, its tangent kernel and the sampled residuals. Their accuracy at test inputs depends on the evolution of the cross-kernel as well as the training kernel.

\subsection{Prediction versus inverse reconstruction}
\label{app:kernel_coefficient_recovery}

Reconstructing coefficients from predicted sample values alone is a different operation. Let \(P=V^+V\) be the orthogonal projector onto the dataset-visible coefficient subspace. Then
\begin{equation}
\begin{aligned}
\Delta a_{\mathrm{rec}}
&:=V^+[\widehat F(\tau)-F_{\mathrm{fr}}]\\
&=P[\widehat a(\tau)-a_{\mathrm{fr}}],
\end{aligned}
\label{eq:kernels_pseudoinverse_trajectory}
\end{equation}
and the reconstructed kernel satisfies
\begin{equation}
V^+K(V^+)^\dagger=PHP.
\label{eq:kernels_pseudoinverse_recovery}
\end{equation}
Full column rank of \(V\) gives \(P=I\) and exact recovery. Otherwise, sample values cannot determine changes in \(\ker V\). Equation~\eqref{eq:kernels_frozen_coefficient_prediction} can specify them because it uses the known coefficient kernel, not just \(K\) and the predicted samples.

For a small algebraic illustration, retain only the real constant and cosine coordinates: \(\widehat f(x)=q_0+q_1\cos x\). At the sole training input \(x=0\), \(V_{\mathbb R}=(1,1)\). If \(H_{\mathbb R}=\operatorname{diag}(1,4)\), a frozen training-output change of \(-1\) corresponds to
\begin{equation}
\Delta q=-\frac15\binom{1}{4},\qquad
\Delta q_{\mathrm{rec}}=-\frac12\binom{1}{1}.
\label{eq:app_kernel_reconstruction_example}
\end{equation}
Both coefficient changes give the prescribed output change \(-1\) at \(x=0\), yet at \(x=\pi\) they give \(+3/5\) and \(0\), respectively. Their difference lies in \(\ker V_{\mathbb R}\), so the training sample cannot distinguish them. The coefficient kernel supplies the information needed to predict which change the frozen model makes.

Likewise, \(V^+y\) is a minimum-norm least-squares fit in the chosen coefficient normalisation, not generally the Fourier series of a continuous target. Small singular values amplify errors even when an inverse exists.

\FloatBarrier
\section{Two-copy computation of covariance and averaged QNTKs}
\label{appendix:twirl}

Ensemble moments can be computed without constructing every CHM column. This appendix derives the local averaging maps used in Section~\ref{sec:ensemble_computation}, proves that their propagation gives the required moments, and bounds the computational cost. The construction follows the Pauli propagation of Nemkov \emph{et al.} and the coefficient-moment framework of Wiedmann \emph{et al.} \cite{nemkovFourierExpansionVariational2023,wiedmann2024}. We include both ordinary and differentiated maps for controlled rotations over their full physical periods. The proof assumes exact arithmetic; the final subsection discusses the limits of the numerical implementation.

\subsection{Scope, outputs and frequency labels}
\label{appendix:twirl_assumptions}

The local averaging requires independent, single-use parameters: each physical angle \(\vartheta_r\) belongs to one trainable gate and is uniform over a full period. Trainable gates may be Pauli rotations or the controlled rotations treated below; fixed interleavings are Clifford gates. Encoders are Pauli rotations \(R_P(s\cdot x)=e^{-i(s\cdot x)P/2}\), with \(s\in\mathbb Z^{d_x}\). The observable has a finite Hermitian Pauli expansion \(O=\sum_s c_sP_s\), with real \(c_s\). The input state \(\rho\) is fixed; the experiments use \(\rho=\ket{0}\!\bra{0}^{\otimes n}\). Repeated use of an angle in different gates, or correlated parameter distributions, generally prevents this gate-local elimination.

Write the forward circuit as \(U=h_G\cdots h_1\), and process the reversed list \(g_t=h_{G-t+1}\). Define
\begin{equation}
O_0=O,\qquad O_t(x)=g_t(x)^\dagger O_{t-1}(x)g_t(x).
\label{eq:twirl_reversed_observable}
\end{equation}
Then \(f(x;\vartheta)=\operatorname{Tr}[\rho O_G(x)]\). The required mean and uncentred second moment are \(\bar a=\mathbb E[a]\) and \(S=\mathbb E[aa^\dagger]\); covariance is \(\Sigma=S-\bar a\bar a^\dagger\). Because the output is real,
\begin{equation}
\begin{aligned}
F_2(x,x')&:=\mathbb E_\vartheta[f(x;\vartheta)f(x';\vartheta)]\\
&=\sum_{\omega,\mu}S_{\omega\mu}
e^{i\omega\cdot x}e^{-i\mu\cdot x'}.
\end{aligned}
\label{eq:F2_target}
\end{equation}
The inputs differ between the two copies, but their parameters are the \emph{same} random variables. To preserve this dependence, we start with \(O\otimes O\) and average the pair jointly. Averaging each observable independently would instead give a product of means.

The implementation stores positive character signs on \emph{both} copies:
\begin{equation}
(Q_1,Q_2,\omega,\nu)\mapsto\Gamma
\quad\longleftrightarrow\quad
\Gamma e^{i\omega\cdot x+i\nu\cdot x'}Q_1\otimes Q_2.
\label{eq:twirl_stored_key}
\end{equation}
Its terminal array therefore has entries \(S^{\mathrm{raw}}_{\omega\nu}=\mathbb E[a_\omega a_\nu]\). The final assembly reverses the second index:
\begin{equation}
S_{\omega\mu}=S^{\mathrm{raw}}_{\omega,-\mu}.
\label{eq:twirl_second_index}
\end{equation}
This uses \(a_{-\mu}=\overline{a_\mu}\). The same relabelling is used for the derivative accumulator.

\subsection{Local propagation and shared-angle averaging}
\label{app:twirl_rules}

All Pauli labels below are unsigned Hermitian strings; scalar signs and phases belong to the stored coefficients. A Clifford conjugation maps \(Q_i\) to \(\sigma_iQ_i'\), with \(\sigma_i=\pm1\), so a pair receives the factor \(\sigma_1\sigma_2\) without frequency branching \cite{gottesman1998,dehaene2003}.

For a trainable Pauli rotation \(R_P(\vartheta)\), the one-copy rule is
\begin{equation}
R_P^\dagger Q R_P=
\begin{cases}
Q,&[P,Q]=0,\\
Q\cos\vartheta+R\sin\vartheta,&\{P,Q\}=0,
\end{cases}
\qquad R:=iPQ.
\label{eq:twirl_pauli_rule}
\end{equation}
In the anticommuting case \(R\) is a \emph{signed} Hermitian Pauli string. Uniform full-period averaging gives
\begin{equation}
\langle\cos\vartheta\rangle=\langle\sin\vartheta\rangle
=\langle\cos\vartheta\sin\vartheta\rangle=0,\qquad
\langle\cos^2\vartheta\rangle
=\langle\sin^2\vartheta\rangle=\tfrac12.
\label{eq:twirl_local_inner_products}
\end{equation}
Here brackets denote a normalised angle average. The one-copy mean pass consequently keeps commuting strings and deletes anticommuting strings. The two-copy map \(\mathcal T_P\) instead follows Table~\ref{tab:pairrule}.

\begin{table}[!htbp]
\centering
\renewcommand{\arraystretch}{1.25}
\setlength{\tabcolsep}{8pt}
\begin{tabular}{lll}
\toprule
Copy 1 against \(P\) & Copy 2 against \(P\) &
\(\mathcal T_P(Q_1\otimes Q_2)\)\\
\midrule
Commutes & Commutes & \(Q_1\otimes Q_2\)\\
Commutes & Anticommutes & \(0\)\\
Anticommutes & Commutes & \(0\)\\
Anticommutes & Anticommutes &
\(\tfrac12(Q_1\otimes Q_2+R_1\otimes R_2)\)\\
\bottomrule
\end{tabular}
\caption{\justifying Shared-angle two-copy average at a Pauli rotation. Here \(R_i=iPQ_i\) includes its sign. When stored as \(R_i=\sigma_iQ_i'\), the rotated pair has coefficient \(\sigma_1\sigma_2/2\).}
\label{tab:pairrule}
\end{table}

For example, \(R_Z^\dagger X R_Z=X\cos\vartheta-Y\sin\vartheta\) has zero mean but a nonzero shared-angle second moment:
\begin{equation}
\mathbb E_\vartheta[
(R_Z^\dagger X R_Z)\otimes(R_Z^\dagger X R_Z)]
=\tfrac12(X\otimes X+Y\otimes Y).
\label{eq:twirl_shared_angle_example}
\end{equation}

Encoders use Eq.~\eqref{eq:twirl_pauli_rule} without averaging. For an anticommuting string,
\begin{equation}
R_P(s\cdot x)^\dagger Q R_P(s\cdot x)
=\sum_{\epsilon=\pm1}
\left(\frac Q2+\frac{\epsilon R}{2i}\right)
e^{i\epsilon s\cdot x}.
\label{eq:twirl_encoder_rule}
\end{equation}
Each copy therefore produces at most four Pauli--character terms. Both stored frequency labels receive the same sign convention \(\omega\mapsto\omega+\epsilon s\) on the first copy and \(\nu\mapsto\nu+\epsilon' s\) on the second, with independent branch signs \(\epsilon,\epsilon'\in\{-1,1\}\). For a multiqubit Pauli generator, the same rule applies after testing commutation of the whole strings. After each gate, we sum contributions with identical Pauli and frequency labels.

\subsection{Controlled rotations and full physical periods}
\label{app:twirl_controlled}

\begin{remark}[One shared angle in a controlled rotation]
\label{rem:cr_gates}
For \(\alpha\in\{X,Y,Z\}\),
\begin{equation}
\mathrm{CR}_{\alpha,c\to t}(\vartheta)
=e^{-i\vartheta\sigma_t^\alpha/4}
 e^{+i\vartheta Z_c\sigma_t^\alpha/4}.
\label{eq:cr_pauli_decomposition}
\end{equation}
The two commuting factors share one angle, so we average them jointly as a single two-qubit operation. Its conjugation action can contain half-angle terms and has natural physical period \(4\pi\); the character coordinate \(\theta=\vartheta/2\) has period \(2\pi\).
\end{remark}

For the two-qubit Hermitian Pauli basis \(\{P_a\}_{a=1}^{16}\), define the transfer coefficients in the code's input--output index order:
\begin{equation}
U_{\mathrm{CR}}^\dagger P_aU_{\mathrm{CR}}
=\sum_b A_{ab}(\vartheta)P_b,\qquad
A_{ab}(\vartheta)=\tfrac14\operatorname{Tr}
[P_bU_{\mathrm{CR}}^\dagger P_aU_{\mathrm{CR}}].
\label{eq:cr_ptm_coefficients}
\end{equation}
They are real because both operators in the trace are Hermitian. The mean map uses \(\langle A_{ab}\rangle\); the pair map uses
\begin{equation}
\mathcal T_{\mathrm{CR}}(P_a\otimes P_c)
=\sum_{b,d}\langle A_{ab}A_{cd}\rangle P_b\otimes P_d,\qquad
\langle u\rangle:=\frac1{4\pi}\int_0^{4\pi}u(\vartheta)\,d\vartheta.
\label{eq:cr_two_copy_local_map}
\end{equation}
It replaces the Pauli-rotation table only at a controlled gate.

In character coordinates, \(A_{ab}(2\theta)\) has Fourier degree at most \(2\); its two-copy product has degree at most \(4\). An equispaced full-period quadrature with \(M>4\) points therefore integrates these polynomials exactly in exact arithmetic. The implementation uses \(64\) midpoint points for its controlled \(R_X/R_Z\) maps. Differentiation does not increase the degree, so the same statement covers the derivative products below. For ordinary Pauli rotations, averaging over the common \(4\pi\) interval simply traverses their \(2\pi\) conjugation period twice. Using only \(2\pi\) for a controlled gate generally gives a different average.

\subsection{Derivative accumulator for the averaged QNTK}
\label{app:twirl_derivative}

The desired kernel is \(\bar H=\mathbb E_\vartheta[XX^\dagger]\), with \(X=\partial_\vartheta a\) in physical optimiser coordinates. For a single-use trainable gate \(g\), define two local pair maps: the ordinary average \(\mathcal T_g\), and
\begin{equation}
\mathcal D_g(Q_1\otimes Q_2)
:=\mathbb E_{\vartheta_r}\!\left[
\partial_{\vartheta_r}(g^\dagger Q_1g)
\otimes
\partial_{\vartheta_r}(g^\dagger Q_2g)\right].
\label{eq:twirl_differentiated_map}
\end{equation}
Each copy is differentiated before taking the product and averaging. Clifford and encoder steps have \(\mathcal D_g=0\).

Carry a pair of weights \((\Gamma,\Gamma_H)\) on each dictionary key, initialised with \(\Gamma_H=0\). For a particular output branch, let \(b_S\) and \(b_H\) be its coefficients in \(\mathcal T_g\) and \(\mathcal D_g\), respectively. The update is
\begin{equation}
\Gamma'=b_S\Gamma,\qquad
\Gamma_H'=b_S\Gamma_H+b_H\Gamma.
\label{eq:derivative_twirl_update}
\end{equation}
The first term transports contributions from earlier differentiated gates; the second adds the contribution from differentiating the current gate. This evaluates the sum over parameters in one paired propagation, rather than one full pass per parameter.

For a Pauli rotation, differentiating the trigonometric factors gives \(\partial1=0\), \(\partial\cos=-\sin\), and \(\partial\sin=\cos\). Thus \(\mathcal D_P\) is zero unless both strings anticommute. In that case its two signed branches equal those of \(\mathcal T_P\): the unrotated branch receives \((\Gamma',\Gamma_H')=(\Gamma/2,(\Gamma_H+\Gamma)/2)\); the rotated branch also carries the sign \(\sigma_1\sigma_2\). For a controlled gate,
\begin{equation}
\mathcal D_{\mathrm{CR}}(P_a\otimes P_c)
=\sum_{b,d}\langle
(\partial_\vartheta A_{ab})(\partial_\vartheta A_{cd})
\rangle P_b\otimes P_d.
\label{eq:cr_derivative_local_map}
\end{equation}
Physical differentiation automatically includes the half-angle factors. Its harmonic weights are
\begin{equation}
\bar H=C^\circ\operatorname{diag}_{k\in\mathcal K^\circ}
\left[\sum_r(2\pi k_r/T_r)^2\right]C^{\circ\dagger},
\label{eq:twirl_physical_kernel}
\end{equation}
where \(T_r\) is the physical period associated with the character coordinate. This reduces to the \(\|k\|_2^2\) contraction in character coordinates; a \(4\pi\) component contributes \(k_r^2/4\).

When keys are merged, both weights must be added independently. A key survives if either channel is nonzero, because cancellation of \(\Gamma\) need not cancel \(\Gamma_H\). The local ordinary and differentiated transfer maps must likewise retain every branch used by either channel.

At the terminal step, the derivative array uses the same second-index reversal as Eq.~\eqref{eq:twirl_second_index}. The derivative accumulator directly gives \(\bar H\).

\subsection{Algorithm and correctness}
\label{app:twirl_algorithm}

Algorithm~\ref{alg:twirl} collects these rules. Its maps retain the scalar signs specified above.

\begin{chmalgorithm}{Shared-parameter propagation of coefficient moments and the averaged QNTK.}
\label{alg:twirl}
\begin{algorithmic}[1]
\Require \(O=\sum_s c_sP_s\), reversed gate list \(\mathcal G\), fixed state \(\rho\)
\Ensure \(\bar a,\ S,\ \Sigma,\ \bar H\)
\State Initialise one-copy mean dictionary from \((P_s,0)\mapsto c_s\)
\State \(D\gets\{(P_s,P_{s'},0,0)\mapsto(c_sc_{s'},0)\}\)
\For{gate \(g\) in \(\mathcal G\)}
  \State Update the mean dictionary with the one-copy map; merge identical keys
  \State \(D_{\mathrm{new}}\gets\varnothing\)
  \ForAll{\(\kappa\mapsto(\Gamma,\Gamma_H)\) in \(D\)}
    \ForAll{output keys \(\kappa'\) with weights \((b_S,b_H)\) from the local maps}
      \State \(D_{\mathrm{new}}[\kappa']\mathrel{+}=
      (b_S\Gamma,\ b_S\Gamma_H+b_H\Gamma)\)
    \EndFor
  \EndFor
  \State \(D\gets D_{\mathrm{new}}\); remove keys only when both weights are zero
\EndFor
\State Contract mean weights with \(\operatorname{Tr}[\rho Q]\) to obtain \(\bar a\)
\State Contract each pair weight with \(\operatorname{Tr}[\rho Q_1]\operatorname{Tr}[\rho Q_2]\)
\State Reverse the second frequency index to assemble \(S\) and \(\bar H\)
\State \(\Sigma\gets S-\bar a\bar a^\dagger\)
\State \Return \(\bar a,S,\Sigma,\bar H\)
\end{algorithmic}
\end{chmalgorithm}

\begin{proposition}[Exact shared-parameter contraction]
\label{prop:twirl_correctness}
\label{lem:mean_pass}
\label{prop:derivative_twirl}
Under the assumptions of Appendix~\ref{appendix:twirl_assumptions}, Algorithm~\ref{alg:twirl} returns
\begin{equation}
\bar a=C_{\bullet0},\qquad S=CC^\dagger,\qquad
\Sigma=C^\circ C^{\circ\dagger},\qquad
\bar H=\mathbb E_\vartheta[XX^\dagger].
\label{eq:twirl_exact_outputs}
\end{equation}
\end{proposition}

\begin{proof}
Let \(\vartheta_{(t)}\) denote the independent parameters encountered after the first \(t\) reversed gates. The ordinary pair dictionary represents the invariant
\begin{equation}
\mathbb E_{\vartheta_{(t)}}[O_t(x)\otimes O_t(x')]
=\sum_{\kappa\in D_t}
\Gamma_\kappa e^{i\omega\cdot x+i\nu\cdot x'}
Q_1\otimes Q_2.
\label{eq:twirl_invariant}
\end{equation}
Here \(\kappa=(Q_1,Q_2,\omega,\nu)\) is a stored key. The invariant holds initially because \(D_0\) expands \(O\otimes O\). Clifford conjugation and encoder character expansion are exact linear maps. At a trainable gate, the new angle is independent of every previously encountered one, so its integration can be taken locally. The Pauli or controlled-rotation map is exactly that integral. Merging equal keys preserves the represented sum. These facts prove the invariant by induction.

The same induction on one copy represents \(\mathbb E_{\vartheta_{(t)}}[O_t(x)]\). At the terminal step, contraction with \(\rho\) gives the Fourier coefficients of \(\mathbb E[f]\), hence \(\bar a=C_{\bullet0}\). Contracting Eq.~\eqref{eq:twirl_invariant} with \(\rho\otimes\rho\) gives \(F_2\); uniqueness of its finite Fourier expansion and the second-index reversal give \(S=CC^\dagger\). Subtracting the mean outer product removes exactly the constant parameter column.

For the derivative channel, the corresponding operator invariant is
\begin{equation}
\sum_{r:\,\vartheta_r\in\vartheta_{(t)}}\mathbb E_{\vartheta_{(t)}}[
\partial_{\vartheta_r}O_t(x)\otimes
\partial_{\vartheta_r}O_t(x')].
\label{eq:twirl_derivative_invariant}
\end{equation}
It is initially zero. A deterministic step transports every existing term by its ordinary map. At a new trainable gate, all earlier derivatives are likewise transported by \(\mathcal T_g\), while differentiating the new angle contributes \(\mathcal D_g\) applied to the ordinary channel. These are exactly the two terms in Eq.~\eqref{eq:derivative_twirl_update}. Terminal contraction and index reversal therefore give \(\sum_r\mathbb E[\partial_{\vartheta_r}a\, (\partial_{\vartheta_r}a)^\dagger]=\bar H\).
\end{proof}

For the experimental state \(\ket0^{\otimes n}\), the final contraction keeps only strings made of \(I\) and \(Z\). This rule applies at the terminal step: a string containing \(X\) or \(Y\) earlier in propagation may still be rotated into one that contributes. The proof also permits other fixed states, although the current implementation uses this computational-basis rule. Numerical checks of these contractions are reported in Appendix~\ref{subsec:appendix_validation}.

\subsection{Computational cost and numerical limits}
\label{app:twirl_cost}

Let \(E_t\) count the individual encoding gates (Pauli rotations) encountered by step \(t\), and define \(B_{t,j}=\sum_{e\leq t}|(s_e)_j|\), summing over those encoders. There are \(4^n\) unsigned Pauli strings on each copy, and each stored frequency obeys \(|\omega_j|,|\nu_j|\leq B_{t,j}\). Consequently, the merged dictionary obeys
\begin{equation}
|D_t|\leq16^n\prod_{j=1}^{d_x}(2B_{t,j}+1)^2
\leq16^n(2\bar sE_t+1)^{2d_x},
\qquad \bar s:=\max_e\|s_e\|_\infty.
\label{eq:dictionary_bound}
\end{equation}
The derivative channel changes the number of stored weights per key, not this key-count bound. If there are no encoders, take \(\bar s=0\).

For bounded-size local gates, each entry has a bounded number of children: at most two for a trainable Pauli rotation, sixteen for a Pauli encoder, and \(16^2\) for a two-qubit controlled pair map. Cliffords only relabel. Writing \(D_{\max}=\max_t|D_t|\), the resulting unit-cost estimate is
\begin{equation}
\text{work}=O(GD_{\max}),\qquad
\text{storage}=O(D_{\max}),
\label{eq:twirl_work_storage}
\end{equation}
including simultaneous old and new dictionaries up to a constant branching factor. The one-copy mean pass satisfies the analogous smaller key-space bound. These estimates treat Pauli-label operations, hash-table updates and terminal Pauli expectations as unit cost. Bit operations, dictionary overhead, local-map preparation and the cost of an arbitrary state's expectations are additional.

For our single-encoding-layer experiments, \(d_x=1\), \(E_G=n\), and \(\bar s=1\), giving
\begin{equation}
D_{\max}\leq16^n(2n+1)^2.
\label{eq:twirl_benchmark_bound}
\end{equation}
Increasing trainable depth adds gates to \(G\) and can populate more of the allowed keys, but leaves the encoder frequency labels unchanged. For a re-uploading circuit with \(L\) such encoding layers, each containing \(n\) encoding gates, the total gate count is \(E_G=nL\). The same argument then gives \(16^n(2nL+1)^2\).

The bound is exponential in qubit number and, at fixed input dimension, polynomial in the encoding-gate count. Actual cost depends on how many of the allowed keys are populated during propagation.

The implementation uses floating-point weights and prunes small intermediate contributions, including controlled-map entries. Appendix~\ref{subsec:appendix_numerical_limits} discusses the resulting pruning error. Hermitian symmetrisation at final assembly removes the antisymmetric roundoff component.

A dictionary-size safeguard aborts calculations that grow too large. Because it is checked after each gate update, peak memory can exceed the safeguard's nominal threshold.

\end{document}